\documentclass[acmsmall,screen]{acmart}

\usepackage{syntax}

\usepackage{prftree}

\usepackage{listings}
\usepackage{xcolor}
\usepackage{caption}
\usepackage{subcaption}
\usepackage{fancyvrb}
\usepackage{enumitem}
\usepackage{string-diagrams}
\usepackage{cancel}
\usepackage{thmtools}
\usetikzlibrary{calc}

\definecolor{codegreen}{rgb}{0,0.6,0}
\definecolor{codegray}{rgb}{0.5,0.5,0.5}
\definecolor{codepurple}{rgb}{0.58,0,0.82}
\definecolor{backcolour}{rgb}{0.95,0.95,0.92}

\lstdefinestyle{mystyle}{
    commentstyle=\color{codegreen},
    keywordstyle=\color{magenta},
    numberstyle=\tiny\color{codegray},
    stringstyle=\color{codepurple},
    basicstyle=\ttfamily\footnotesize,
    breakatwhitespace=false,         
    breaklines=true,                 
    captionpos=b,                    
    keepspaces=true,                 
    numbers=left,                    
    numbersep=5pt,                  
    showspaces=false,                
    showstringspaces=false,
    showtabs=false,                  
    tabsize=2
}

\newcounter{todos}

\newcommand{\mc}[1]{\ensuremath{\mathcal{#1}}}
\newcommand{\mb}[1]{\ensuremath{\mathbf{#1}}}
\newcommand{\ms}[1]{\ensuremath{\mathsf{#1}}}

\newcommand{\nats}{\mathbb{N}}

\newcommand{\lto}{:}
\newcommand{\linl}[1]{\iota_l\;{#1}}
\newcommand{\linr}[1]{\iota_r\;{#1}}
\newcommand{\labort}[1]{\ms{abort}\;{#1}}

\newcommand{\letexpr}[3]{\ensuremath{\ms{let}\;#1 = #2;\;#3}}
\newcommand{\caseexpr}[5]{\ms{case}\;#1\;\{\linl{#2} \lto #3, \linr{#4} \lto #5\}}
\newcommand{\letstmt}[3]{\ensuremath{\ms{let}\;#1 = #2; #3}}
\newcommand{\brb}[2]{\ms{br}\;#1\;#2}
\newcommand{\ite}[3]{\ms{if}\;#1\;\{#2\}\;\ms{else}\;\{#3\}}
\newcommand{\casestmt}[5]{\ms{case}\;#1\;\{\linl{#2} \lto #3, \linr{#4} \lto #5\}}
\newcommand{\where}[2]{#1\;\ms{where}\;#2}
\newcommand{\wbranch}[3]{#1(#2) \lto \{#3\}}
\newcommand{\cfgsubst}[1]{\ms{cfgs}\;\{#1\}}

\newcommand{\rupg}[1]{{#1}^\upharpoonright}
\newcommand{\lupg}[1]{{#1}^\upharpoonleft}

\newcommand{\thyp}[3]{#1 : {#2}^{#3}}
\newcommand{\bhyp}[2]{#1 : #2}
\newcommand{\lhyp}[2]{#1(#2)}
\newcommand{\rle}[1]{{\scriptsize\textsf{#1}}}

\newcommand{\hasty}[4]{#1 \vdash_{#2} #3: {#4}}
\newcommand{\haslb}[3]{#1 \vdash #2 \rhd #3}

\newcommand{\ahasty}[4]{#1 \vdash_{#2}^{\ms{anf}} #3 : {#4}}
\newcommand{\thaslb}[3]{#1 \vdash^{\ms{t}}_{\ms{ssa}} #2 \rhd #3}
\newcommand{\ahaslb}[3]{#1 \vdash^{\ms{anf}} #2 \rhd #3}
\newcommand{\bhaslb}[3]{#1 \vdash^{\ms{b}}_{\ms{ssa}} #2 \rhd #3}

\newcommand{\shaslb}[3]{#1 \vdash^{\ms{s}} #2 \rhd #3}

\newcommand{\isop}[4]{#1 \in \mc{I}_{#4}(#2, #3)}
\newcommand{\issubst}[3]{#1: #2 \mapsto #3}
\newcommand{\lbsubst}[4]{#1 \vdash #2: #3 \rightsquigarrow #4}
\newcommand{\teqv}{\approx}
\newcommand{\tmeq}[5]{#1 \vdash_{#2} #3 \teqv #4 : {#5}}
\newcommand{\lbeq}[4]{#1 \vdash #2 \teqv #3 \rhd {#4}}
\newcommand{\tmseq}[4]{\issubst{#1 \teqv #2}{#3}{#4}}
\newcommand{\lbseq}[5]{\lbsubst{#1 \teqv #2}{#3}{#4}{#5}}
\newcommand{\brle}[1]{{\textsf{#1}}}

\newcommand{\toanf}[1]{\ms{ANF}(#1)}
\newcommand{\letanf}[3]{\ms{ANF}_{\ms{let}}(#1, #2, #3)}
\newcommand{\tossa}[1]{\ms{SSA}(#1)}
\newcommand{\ssawhere}[2]{\ms{SSA}_{\ms{a}}(#1, #2)}
\newcommand{\toentry}[1]{\ms{entry}(#1)}
\newcommand{\todom}[1]{\ms{children}(#1)}
\newcommand{\tocfg}[1]{\ms{cfg}(#1)}
\newcommand{\adddom}[2]{\ms{bb}(#1, #2)}
\newcommand{\toreg}[1]{\ms{reg}(#1)}

\newcommand{\towhile}[2]{\ms{WH}_{#1}(#2)}
\newcommand{\topwhile}[2]{\ms{PW}_{#1}(#2)}

\newcommand{\dnt}[1]{\llbracket{#1}\rrbracket}

\newcommand{\tmor}[1]{{!}_{#1}}
\newcommand{\dmor}[1]{{\Delta}_{#1}}
\newcommand{\entrymor}[3]{\ms{esem}_{#1, #3}(#2)}
\newcommand{\loopmor}[3]{\ms{lsem}_{#1, #3}(#2)}
\newcommand{\substpure}[1]{#1\;\ms{pure}}

\newcommand{\lmor}[1]{\ms{let}(#1)}
\newcommand{\envcom}[2]{{#1}_{#2 \otimes \cdot}}
\newcommand{\rlmor}[1]{\ms{rlet}(#1)}
\newcommand{\rcase}[1]{\ms{rcase}(#1)}
\newcommand{\rfix}[1]{\ms{rfix}(#1)}
\newcommand{\rseq}[3]{#2 \gg_{#1} #3}

\newcommand{\toenv}[2]{\ms{env}_{#1}(#2)}
\newcommand{\envcop}[3]{[#2, #3]_{#1}}
\newcommand{\envinr}[1]{\iota^{#1}_{r}}
\newcommand{\envinl}[1]{\iota^{#1}_{l}}
\newcommand{\envtn}[3]{{#2} \otimes_{#1} {#3}}

\newcommand{\invar}{\square}
\newcommand{\outlb}{\blacksquare}
\newcommand{\pckd}[1]{\langle #1 \rangle}

\newcommand{\bufloc}[1]{\overline{#1}}

\newcommand{\isotopessa}{\(\lambda_{\ms{SSA}}\)}

\setcopyright{acmcopyright}
\copyrightyear{2024}
\acmYear{2024}
\acmDOI{XXXXXXX.XXXXXXX}

\begin{document}

\title{The Denotational Semantics of SSA}

\author{Jad Ghalayini}
\email{jeg74@cl.cam.ac.uk}
\orcid{0000-0002-6905-1303}

\author{Neel Krishnaswami}
\email{nk480@cl.cam.ac.uk}
\orcid{0000-0003-2838-5865}

\begin{abstract}
  Static single assignment form, or SSA, has been the dominant
  compiler intermediate representation for decades. In this paper, we
  give a type theory for a variant of SSA, including its equational
  theory, which are strong enough to validate a variety of control and
  data flow transformations. We also give a categorical semantics for
  SSA, and show that the type theory is sound and complete with
  respect to the categorical axiomatization. We demonstrate the
  utility of our model by exhibiting a variety of concrete models
  satisfying our axioms, including in particular a model of TSO weak
  memory. The correctness of the syntactic metatheory, as well as the
  completeness proof has been mechanized in the Lean proof assistant.
\end{abstract}





\maketitle

\emph{This paper is dedicated to the memory of Alan Jeffrey, who
  taught us about both premonoidal categories and the semantics of weak memory, and
who never shied away from either theory or implementation.}

\section{Introduction}

Static single assignment form, or SSA form, has been the dominant compiler intermediate
representation since its introduction by \citet{alpern-ssa-original-88} and \citet{rosen-gvn-1988}
in the late 1980s. Most major compilers -- GCC, Clang, MLIR, Cranelift -- use this representation,
because it makes many optimizations much easier to do than traditional 3-address code IRs.

The key idea behind SSA is to adapt an idea from functional programming: namely, every variable is
defined only once. This means that substitution is unconditionally valid, without first requiring a
dataflow analysis to compute where definitions reach. Unlike in functional programming, though,
scoping of definitions in SSA is traditionally not lexical. Instead, scoping is determined by
\emph{dominance}: every variable occurrence must be dominated by a single assignment in the control
flow graph.

The semantics of SSA has traditionally been handled quite informally, because conceptually, it is a
simple first-order imperative programming language. As a result, whether a rewrite is sound or not
is usually obvious, without having to do a complex correctness argument.

Unfortunately, computers are no longer as simple as they were in the late 1980s. Modern computers
are typically multicore, and feature many levels of caching, and as a result the semantics of memory
is no longer correctly modelled as a big array of bytes. Finding good semantics for modern weak
memory systems remains an ongoing challenge.

As a result, it is not correct to justify compiler optimizations in terms of a simple imperative
model, and it is an open question which equations should hold of an SSA program. This is a
particularly fraught question, because it is also unclear which equations weak memory models should
satisfy.

What we would like to know is which equations any SSA representation should satisfy. This would let
us establish a contract between compiler writers and hardware designers. The compiler writers could
rely upon the equational theory of SSA when justifying optimizations, without needing to know all
the details of the memory model at all times.  Conversely, memory models could be validated by
seeing if they satisfy the equations of SSA, without needing to study every possible compiler
optimization.

Concretely, our contributions are as follows: 

\begin{itemize}
\item First, we give a type-theoretic presentation of SSA, with both typing rules (in
  Section~\ref{sec:typing}) and an equational theory (in Section~\ref{sec:equations}) for
  well-typed terms. We also prove the correctness of suitable substitution properties for this
  calculus. 
  
\item Next, in Section~\ref{sec:densem}, we give a categorical semantics for this type theory, in
  terms of distributive Elgot categories. We show that any denotational model with this categorical
  structure is also a model of SSA. This shows that all of the equations we give are sound with
  respect to the categorical structure. 

\item We also show, in Section~\ref{ssec:completeness}, that syntax quotiented by the equational
  theory yields the initial distributive Elgot category. This establishes that our set of syntactic
  equations is complete, and that there are no equations which the denotational semantics validates,
  but which cannot be proved syntactically. 

\item We proceed in Section~\ref{sec:concrete} to show that this denotational axiomatization is
  useful in practice, by giving a variety of concrete models, including a model of TSO weak memory
  based on~\citet{sparky} in Section~\ref{ssec:tso}. This demonstrates that it is possible to give
  realistic weak memory models which do not disturb the structure of SSA in fundamental ways.

\item Finally, we have substantially mechanized our proofs using the Lean 4 proof assistant. We have
  mechanized proofs of substitution for our type theory, as well as proofs that the syntax forms the
  initial model, and that the SPARC TSO semantics forms a valid model of SSA. The denotational
  semantics and its proof of the soundness of substitution are done on paper. 

\end{itemize}

\section{Static Single Assignment Form}

In this section, we describe SSA form and the isomorphism between the standard $\phi$-node-based
presentation and the more functional \emph{basic blocks with arguments} format. We then discuss
standard dominance-based scoping and how it can be recast as lexical scoping to make it more
amenable to standard type-theoretic treatment. We further generalize this format to allow branching
to arbitrary code rather than only labels, obtaining \emph{A-normal form}
(ANF)~\cite{flanagan-93-anf}, analogously to the transformation described by
\citet{chakravarty-functional-ssa-2003}. Finally, to allow for substitution, we relax our syntax to
permit arbitrary expression nesting and \ms{let}-expressions, resulting in \emph{type-theoretic
SSA}, or \isotopessa{}, which will be the focus of the rest of this paper. A straightforward
argument shows that these extensions add no expressive power; we give a more formal version of this
argument, as well as algorithms to interconvert between \isotopessa{}, standard SSA, and ANF, in
Section~\ref{ssec:ssa-normal}.

As a running example, consider the simple imperative program to compute $10!$ given in
Figure~\ref{fig:fact-program}. Operating directly on an imperative language can be challenging,
since having a uniform representation of code friendly to mechanized optimization and analysis is
often in tension with features designed to improve readability and programmer productivity, such as
syntactic sugar. Early work on compiler intermediate representations, notably by Frances
Allen~\cite{allen-70-cfa}, introduced \emph{three-address code}, also known as \emph{register
transfer language (RTL)},  to normalize programs into a form more suitable for analysis and
optimization. We can normalize our code into 3-address code, as in Figure~\ref{fig:fact-3addr}, by:
\begin{itemize}
  \item Converting structured control flow (e.g., \ms{while}) into unstructured jumps between basic
  blocks labelled \ms{start}, \ms{loop}, and \ms{body}.
  \item Replacing subexpressions like $i + 1$ in $a * (i + 1)$ with \ms{let}-bindings so that every
  expression in our program is atomic.
\end{itemize}

\begin{figure}
  \begin{subfigure}[t]{.5\textwidth}
    \begin{align*}
      & \ms{let}\;n = 10; \\
      & \ms{let\;mut}\;i = 1; \\
      & \ms{let\;mut}\;a = 1; \\
      & \ms{while}\;i_0 < n\;\{ \\
      & \quad a = a * (i + 1) \\
      & \quad i = i + 1; \\
      & \} \\
      & \ms{ret}\;a \\
    \end{align*}
    \caption{As an imperative program}
    \label{fig:fact-imp}
  \end{subfigure}%
  \begin{subfigure}[t]{.5\textwidth}
    \begin{align*}
      \ms{start}:\quad  & \ms{let}\;n = 10; \\
                        & \ms{let\;mut}\;i = 1; \\
                        & \ms{let\;mut}\;a = 1; \\
                        & \ms{br}\;\ms{loop} \\
      \ms{loop}: \quad  & \ms{if}\;i < n\;
                          \{\;\ms{br}\;\ms{body}\;\}\;
                          \ms{else}\;\{\;\ms{ret}\;a\;\} \\
      \ms{body}: \quad  & \ms{let}\;t = i + 1; \\
                        & a = a * t; \\
                        & i = i + 1; \\
                        & \ms{br}\;\ms{loop}
    \end{align*}
    \caption{As 3-address code}
    \label{fig:fact-3addr}
  \end{subfigure}
  \caption{
    A simple, slightly suboptimal program to compute $10!$ via multiplication in a loop, represented
    as typical imperative code and in 3-address code.
  }
  \Description{}
  \label{fig:fact-program}
\end{figure}

While functional languages typically rely on \emph{lexical scoping}, where the scope of a variable
is determined by its position within the code's nested structure, 3-address code uses a different
scoping mechanism based on \emph{dominance}. In particular, a variable $x$ is considered to be in
scope at a specific point $P$ if and only if all execution paths from the program's entry point to
$P$ pass through a definition $D$ for $x$. In this case, we say that the definition $D$
\emph{strictly dominates} $P$. The relation on basic blocks ``$A$ strictly dominates $B$"
intersected with ``$A$ is a \emph{direct predecessor} of $B$" forms a tree called the
\emph{dominance tree} of the CFG -- this can be computed in nearly linear time
\cite{cytron-ssa-intro-91}.

Despite this normalization, many optimizations remain difficult to express in this format because a
variable's value may be set by multiple definitions throughout the program's execution. To improve
our ability to reason about programs, we introduce the \emph{static single assignment} restriction,
originally proposed by \citet{alpern-ssa-original-88}, which states that every variable must be
defined at exactly one point in the program. We can intuitively represent this as every variable
being given by an immutable \ms{let}-binding. In particular, given a variable $x$, which we can now
associate to its unique definition $D_x$, $x$ is in scope at a point $P$ if and only if $D_x$
strictly dominates $P$.

One might attempt to convert programs to SSA form by numbering each definition of a variable,
effectively changing references to $x$ to references to $x_t$, i.e. ``$x$ at time $t$.'' For
example, we could rewrite
\begin{multline}
  \letexpr{x}{3y + 5}{\letexpr{x}{3x + 2}{\letexpr{x}{3x + 1}{\ms{ret}\;x}}}
  \\ \teqv \letexpr{x_0}{3y + 5}{\letexpr{x_1}{3x_0 + 2}{\letexpr{x_2}{3x_1 + 1}{\ms{ret}\;x_2}}}
\end{multline}
This transformation enables algebraic reasoning about expressions involving each $x_t$. However,
since we can only define a variable once in SSA form, expressing programs with loops and branches
becomes challenging. For example, na\"ively trying to lower the program in
Figure~\ref{fig:fact-3addr} into SSA form would not work, since the $i$ definition $i = i + 1$ can
refer to \emph{either} the previous value of $i$ from the last iteration of the loop \emph{or} the
original value $i = 1$. The classical solution is to introduce \emph{$\phi$-nodes}, which select a
value based on the predecessor block from which control arrived. We give the lowering of our program
into SSA with $\phi$-nodes in Figure~\ref{fig:fact-ssa}. 

\citet{cytron-ssa-intro-91} introduced the first efficient algorithm to lower a program in 3-address
code to valid SSA while introducing a minimum number of $\phi$-nodes, making SSA practical for
widespread use as an intermediate representation. Unfortunately, $\phi$-nodes do not have an obvious
operational semantics. Additionally, they require us to adopt more complex scoping rules than simple
dominance-based scoping. For example, in basic block \ms{loop} in Figure~\ref{fig:fact-ssa}, $i_0$
evaluates to 1 if we came from \ms{start} and to $i_1$ if we came from \ms{body}. Similarly, $a_0$
evaluates to either 1 or $a_1$ based on the predecessor block. This does not obey dominance-based
scoping, since $i_0$ and $i_1$ are defined \emph{after} the $\phi$-nodes $i_0$, $a_0$ that reference
them, which seems counterintuitive -- after all, variables are typically used after they are
defined. In fact, since the value of a $\phi$-node is determined by which basic block is our
immediate predecessor, we instead need to use the rule that expressions in $\phi$-node branches with
source $S$ can use any variable $y$ defined at the \emph{end} of $S$. Note that this is a strict
superset of the variables visible for a normal instruction $x$, which can only use variables $y$
which \emph{dominate} $x$ -- i.e., such that \emph{every} path from the entry block to the
definition of $x$ goes through $y$, rather than only those paths which also go through $S$.

\begin{figure}
  \begin{subfigure}[t]{.5\textwidth}
    \begin{align*}
      \ms{start}:\quad  & \ms{let}\;n = 10; \\
                        & \ms{let\;mut}\;i = 1; \\
                        & \ms{let\;mut}\;a = 1; \\
                        & \ms{br}\;\ms{loop} \\
      \ms{loop}: \quad  & \ms{if}\;i < n\;
                          \{\;\ms{br}\;\ms{body}\;\}\;
                          \ms{else}\;\{\;\ms{ret}\;a\;\} \\
      \ms{body}: \quad  & \ms{let}\;t = i + 1; \\
                        & a = a * t; \\
                        & i = i + 1; \\
                        & \ms{br}\;\ms{loop}
    \end{align*}
    \caption{3-address code}
  \end{subfigure}%
  \begin{subfigure}[t]{.5\textwidth}
    \begin{align*}
      \ms{start}:\quad & \ms{let}\;n = 10; \\
      & \ms{br}\;\ms{loop} \\
      \ms{loop}: \quad  & \ms{let}\;i_0 = \phi(\ms{start}: 1, \ms{body}: i_1) \\
                        & \ms{let}\;a_0 = \phi(\ms{start}: 1, \ms{body}: a_1) \\
                        & \ms{if}\;i_0 < n\;
                          \{\;\ms{br}\;\ms{body}\;\}\;
                          \ms{else}\;\{\;\ms{ret}\;a_0\;\} \\
      \ms{body}: \quad  & \ms{let}\;t = i_0 + 1 \\
                        & \ms{let}\;a_1 = a_0 * t \\
                        & \ms{let}\;i_1 = i_0 + 1 \\
                        & \ms{br}\;\ms{loop}
    \end{align*}
    \caption{Converted to SSA form}
    \label{fig:fact-ssa}
  \end{subfigure}
  \caption{
    Conversion of three address code for the program in Figure~\ref{fig:fact-program} to SSA 
    form, requring the insertion of $\phi$-nodes for $i$ and $a$ due to control-flow dependent
    updates. Note how SSA-form can be viewed as ``three address code in which all 
    \ms{let}-bindings are immutable.''
  }
  \Description{}
\end{figure}
 
While this rule can be quite confusing, and in particular makes it non-obvious how to assign an
operational semantics to $\phi$-nodes, the fact that the scoping for $\phi$-node branches is based
on the source block, rather than the block in which the $\phi$-node itself appears, hints at a
possible solution. By \emph{moving} the expression in each branch to the \emph{call-site}, we can
transition to an isomorphic syntax called basic blocks with arguments (BBA), as illustrated in
Figure \ref{fig:fact-bba}. In this approach, each $\phi$-node -- since it lacks side effects and has
scoping rules independent of its position in the basic block, depending only on the source of each
branch -- can be moved to the top of the block. This reorganization allows us to treat each
$\phi$-node as equivalent to an argument for the basic block, with the corresponding values passed
at the jump site. Converting a program from BBA format back to standard SSA form with $\phi$-nodes
is straightforward: introduce a $\phi$-node for each argument of a basic block, and for each branch
corresponding to the $\phi$-node, add an argument to the jump instruction from the appropriate
source block. 

This allows us to use standard dominance-based scoping without any special cases for $\phi$-nodes.
When considering basic blocks, this means that a variable is visible within the block $D$ where it
is defined, starting from the point of its definition. It continues to be visible in all subsequent
blocks $P$ that are strictly dominated by $D$ in the control-flow graph (CFG). For example, in
Figure~\ref{fig:fact-bba}:
\begin{itemize}
  \item \ms{start} strictly dominates \ms{loop} and \ms{body}; thus, the variable $n$ defined in
  \ms{start} is visible in \ms{loop} and \ms{body}.
  \item \ms{loop} strictly dominates \ms{body}; therefore, the parameters $i_0$, $a_0$ to \ms{loop}
  are visible in \ms{body} without the need to pass them as parameters.
  \item \ms{body} does \emph{not} strictly dominate \ms{loop}, since there is a path from \ms{start}
  to \ms{loop} that does not pass through \ms{body}.
\end{itemize}

\begin{figure}
  \begin{subfigure}[t]{.5\textwidth}
    \centering
    \begin{align*}
      \ms{start}:\quad  & \ms{let}\;n = 10; \\
                        & \ms{br}\;\ms{loop} \\
      \ms{loop}: \quad  & \begingroup \color{red}
                          \ms{let}\;i_0 = \phi(\ms{start}: 1, \ms{body}: i_1) 
                          \endgroup \\
                        & \begingroup \color{blue}
                          \ms{let}\;a_0 = \phi(\ms{start}: 1, \ms{body}: a_1) 
                          \endgroup \\
                        & \ms{if}\;i_0 < n\;\{\;\ms{br}\;\ms{body}\;\} \\
                        & \ms{else}\;\{\;\ms{ret}\;a_0\;\} \\
      \ms{body}: \quad  & \ms{let}\;t = i_0 + 1 \\
                        & \ms{let}\;a_1 = a_0 * t \\
                        & \ms{let}\;i_1 = i_0 + 1 \\
                        & \ms{br}\;\ms{loop}
    \end{align*}
    \caption{With $\phi$-nodes}
    \label{fig:fact-phi}
  \end{subfigure}%
  \begin{subfigure}[t]{.5\textwidth}
    \centering
    \begin{align*}
      \ms{start}:\quad            & \ms{let}\;n = 10; \\
                                  & \ms{br}\;\ms{loop}(\textcolor{red}{1}, \textcolor{blue}{1}) \\
      \ms{loop}(\textcolor{red}{i_0}, \textcolor{blue}{a_0}): \quad  
                                  & \ms{if}\;i_0 < n\; \{\;\ms{br}\;\ms{body}\;\} \\
                                  & \ms{else}\;\{\;\ms{ret}\;a_0\;\} \\
      \ms{body}: \quad            & \ms{let}\;t = i_0 + 1 \\
                                  & \ms{let}\;a_1 = a_0 * t \\
                                  & \ms{let}\;i_1 = i_0 + 1 \\
                                  & \ms{br}\;\ms{loop}(\textcolor{red}{i_1}, \textcolor{blue}{a_1}) 
                                  \\ \\
    \end{align*}
    \caption{Basic-blocks with arguments}
    \label{fig:fact-bba}
  \end{subfigure}
  
  \caption{
    The program in Figure \ref{fig:fact-program} written in standard SSA (using $\phi$ nodes),
    like in LLVM \cite{llvm}, and in basic-blocks with arguments SSA, like in MLIR \cite{mlir} and
    Cranelift \cite{cranelift}. The arguments $i_0, a_0$ corresponding to the $\phi$-nodes $i_0,
    a_0$ are colored in \textcolor{red}{red} and \textcolor{blue}{blue}, respectively.
  }

  \Description{}
\end{figure}

An important insight provided by the BBA format, as discussed by \citet{appel-ssa} and
\citet{kelsey-95-cps}, is that a program in SSA form can in this way be interpreted as a collection
of tail-recursive functions, where each basic block and branch correspond to a function and tail
call, respectively. This interpretation offers a natural framework for defining the semantics of SSA
and reasoning about optimizations. However, there is a subtle difference between the scoping rules
in this format and the actual scoping used in traditional SSA, which requires careful consideration.
By topologically sorting the basic blocks in the CFG according to this partial order and inserting
brackets based on the dominance tree, we can convert dominance-based scoping to lexical scoping. In
this arrangement, a variable is in lexical scope if and only if it is in scope under dominance-based
scoping, as shown in Figure~\ref{fig:dominance-to-lexical}. This transformation is straightforward,
and standard SSA can be recovered by removing the inserted \ms{where}-blocks.

\begin{figure}
  \centering
  \begin{subfigure}[t]{.5\textwidth}
    \begin{align*}
      \ms{start}:\quad            & \ms{let}\;n = 10; \\
                                  & \ms{br}\;\ms{loop}(1, 1) \\
      \ms{loop}(i_0, a_0): \quad  & \ms{if}\;i_0 < n\; \{\;\ms{br}\;\ms{body}\;\} \\
                                  & \ms{else}\;\{\;\ms{ret}\;a_0\;\} \\
      \ms{body}: \quad            & \ms{let}\;t = i_0 + 1 \\
                                  & \ms{let}\;a_1 = a_0 * t \\
                                  & \ms{let}\;i_1 = i_0 + 1 \\
                                  & \ms{br}\;\ms{loop}(i_1, a_1) \\ \\ \\ \\ 
    \end{align*}
    \caption{Dominance-based scoping}
  \end{subfigure}%
  \begin{subfigure}[t]{.5\textwidth}
    \begin{align*}
      & \ms{let}\;n = 10; \\
      & \ms{br}\;\ms{loop}(1, 1) \\
      & \ms{where}\;\ms{loop}(i_0, a_0): \{ \\
      & \quad \ms{if}\;i_0 < n\;\{\;\ms{br}\;\ms{body}\;\} \\
      & \quad \ms{else}\;\{\;\ms{ret}\;a_0\;\} \\
      & \quad \ms{where}\;\ms{body}: \{\\ 
      & \qquad \ms{let}\;t = i_0 + 1 \\
      & \qquad \ms{let}\;a_1 = a_0 * t \\
      & \qquad \ms{let}\;i_1 = i_0 + 1 \\
      & \qquad \ms{br}\;\ms{loop}(i_1, a_1) \\
      & \quad \} \\
      & \}
    \end{align*}
    \caption{Lexical scoping}
  \end{subfigure}
  \caption{Conversion of an SSA program from dominance-based scoping to explicit lexical scoping}
  \Description{}
  \label{fig:dominance-to-lexical}
\end{figure}

Lexical scoping allows us to apply many of the techniques developed in theoretical computer science
and functional programming for reasoning about and developing optimizations and analysis passes --
in particular, the result of our conversion to lexical scoping looks a lot like the correspondence
between SSA and CPS described in \citet{kelsey-95-cps}. We can use this correspondence to guide us
in developing an \textit{equational theory} for SSA programs, with the goal of enabling
compositional reasoning about program transformations such as:
\begin{itemize}
  \item \textit{Control-flow rewrites}, such as jump-threading or fusing two identical branches of
  an \ms{if}-statement
  \item \textit{Algebraic rewrites}, such as simplifying arithmetic expressions
  \item Combinations of the two, such as rewriting $\ms{if}\;x > 0\;\ms{then}\;1 - x\;\ms{else}\;1 +
  x$ to $1 + \ms{abs}(x)$.
\end{itemize}
We can work towards making these easier to express by generalizing our syntax to allow the branches
of if-statements to contain arbitrary code, rather than just unconditional branches, as in
Figure~\ref{fig:bba-to-anf}. This clearly adds no additional expressive power, since:
\begin{itemize}
  \item This syntax clearly generalizes the previous syntax, so no conversion into it is necessary
  \item To revert back to the less general syntax, one must simply introduce new anonymous basic
  blocks for each branch of the if-statement, likeso:
  \begin{equation}
    \ms{if}\;e\;\{s\}\;\ms{else}\;\{t\}
    \to (\ms{if}\;e\;\{\ms{br}\;\ell_\top\}\;\ms{else}\;\{\ms{br}\;\ell_\bot\})\;
        \ms{where}\;\ell_\top: \{s\},\;\ell_\bot: \{t\}
  \end{equation}
\end{itemize}
What we end up with is something which looks a lot like \textit{administrative normal form} (ANF),
with our transformation analogous to that described in \citet{chakravarty-functional-ssa-2003}. The
key difference is that, in our format (which is strictly first order), we require an explicit
\ms{ret} instruction (rather than adopting an expression-oriented language), and write
``$\ms{let\;rec}\;f(x) = e; t$" as ``$\where{t}{\wbranch{f}{x}{e}}$."

\begin{figure}
  \centering
  \begin{subfigure}[t]{.5\textwidth}
    \begin{align*}
      & \ms{let}\;n = 10; \\
      & \ms{br}\;\ms{loop}(1, 1) \\
      & \ms{where}\;\ms{loop}(i_0, a_0): \{ \\
      & \quad \ms{if}\;i_0 < n\;\{\;\ms{br}\;\ms{body}\;\} \\
      & \quad \ms{else}\;\{\;\ms{ret}\;a_0\;\} \\
      & \quad \ms{where}\;\ms{body}: \{\\ 
      & \qquad \ms{let}\;t = i_0 + 1 \\
      & \qquad \ms{let}\;a_1 = a_0 * t \\
      & \qquad \ms{let}\;i_1 = i_0 + 1 \\
      & \qquad \ms{br}\;\ms{loop}(i_1, a_1) \\
      & \quad \} \\
      & \}
    \end{align*}
  \end{subfigure}%
  \begin{subfigure}[t]{.5\textwidth}
    \begin{align*}
      & \ms{let}\;n = 10; \\
      & \ms{br}\;\ms{loop}(1, 1) \\
      & \ms{where}\;\ms{loop}(i_0, a_0): \{\\
      & \quad \ms{if}\;i_0 < n\;\{ \\
      & \qquad \ms{let}\;t = i_0 + 1 \\
      & \qquad \ms{let}\;a_1 = a_0 * t \\
      & \qquad \ms{let}\;i_1 = i_0 + 1 \\
      & \qquad \ms{br}\;\ms{loop}(i_1, a_1) \\
      & \quad \}\;\ms{else}\;\{ \\
      & \qquad \ms{ret}\;a_0 \\
      & \quad \} \\
      & \}
    \end{align*}
  \end{subfigure}
  \caption{Allowing if-statements to jump to arbitrary instructions, rather than a terminator}
  \Description{}
  \label{fig:bba-to-anf}
\end{figure}

ANF, however, lacks a good substitution property, since substituting a value for a variable can take
you out of ANF, making it difficult to express optimizations like $(i + 1) - 1 \to i$ as rewrite
rules. To fix this, we can simply relax the restriction that expressions in a program must be
atomic. This can again trivially be seen to add no excessive power, since we can always introduce
temporary variables via \ms{let}-bindings to make any expression atomic. For full generality, we
will also allow \ms{let}-bindings and \ms{if}-statements \textit{inside} expressions, which again
can be eliminated in the obvious manner, such as by taking
\begin{align*}
  \ms{let}\;x = (\ms{if}\;e\;\{a\}\;\ms{else}\;\{b\}); t &
    \to \ms{if}\;e\;\{\ms{let}\;x = a; t\}\;\ms{else}\;\{\ms{let}\;x = b; t\} \\ 
  & \to \ms{if}\;e\;\{\ms{br}\;\ell(x)\}\;\ms{else}\;\{\ms{br}\;\ell(x)\}\;
        \ms{where}\;\ell(x): \{t\}
\end{align*}

\begin{figure}
  \centering
  \begin{subfigure}[t]{.31\textwidth}
    \begin{align*}
      & \ms{let}\;n = 10; \\
      & \ms{br}\;\ms{loop}(1, 1) \\
      & \ms{where}\;\ms{loop}(i_0, a_0): \{\\
      & \quad \ms{if}\;i_0 < n\;\{ \\
      & \qquad \ms{let}\;t = i_0 + 1 \\
      & \qquad \ms{let}\;a_1 = a_0 * t \\
      & \qquad \ms{let}\;i_1 = i_0 + 1 \\
      & \qquad \ms{br}\;\ms{loop}(i_1, a_1) \\
      & \quad \}\;\ms{else}\;\{ \\
      & \qquad \ms{ret}\;a_0 \\
      & \quad \} \\
      & \}
    \end{align*}
    \caption{Program in ANF}
    \label{fig:fact-anf}
  \end{subfigure}%
  \begin{subfigure}[t]{.35\textwidth}
    \begin{align*}
      & \ms{let}\;n = 10; \\
      & \ms{br}\;\ms{loop}(1, 1) \\
      & \ms{where}\;\ms{loop}(i_0, a_0): \{\\
      & \quad \ms{if}\;i_0 < n\;\{ \\
      & \qquad \ms{br}\;\ms{loop}(i_0 + 1, a_0 * (i_0 + 1)) \\
      & \quad \}\;\ms{else}\;\{ \\
      & \qquad \ms{ret}\;a_0 \\
      & \quad \} \\
      & \}  \\ \\ \\
    \end{align*}
    \caption{
      Programs \ref{fig:fact-anf} and \ref{fig:fact-subst} after substitution;
      since the result is the same, both programs must be equivalent.
    }
    \label{fig:fact-subst}
  \end{subfigure}\hspace{1em}%
  \begin{subfigure}[t]{.31\textwidth}
    \begin{align*}
      & \ms{let}\;n = 10; \\
      & \ms{br}\;\ms{loop}(1, 1) \\
      & \ms{where}\;\ms{loop}(i_0, a_0): \{\\
      & \quad \ms{if}\;i_0 < n\;\{ \\
      & \qquad \ms{let}\;i_1 = i_0 + 1 \\
      & \qquad \ms{let}\;a_1 = a_0 * i_1 \\
      & \qquad \ms{br}\;\ms{loop}(i_1, a_1) \\
      & \quad \}\;\ms{else}\;\{ \\
      & \qquad \ms{ret}\;a_0 \\
      & \quad \} \\
      & \} \\
    \end{align*}
    \caption{Optimized ANF program}
    \label{fig:fact-opt}
  \end{subfigure}
  \caption{
    Adding support for expressions, allowing us to perform substitutions of (pure) expressions.
    Optimizations such as common subexpression elimination can be built using substitution as a
    building block.
  }
  \Description{}
  \label{fig:fact-cse}
\end{figure}

This gives us our final \isotopessa{} calculus, which we will call \emph{type-theoretic SSA}. In
this language, we can safely perform \textit{substitutions}, like in Figure~\ref{fig:fact-cse}.
These can then be used to build up optimizations such as \textit{common-subexpression elimination}.
More generally, substitution lets us do \textit{algebra}. For example, since we know that:
\begin{align*}
  (i_0 + 1, a_0 * (i_0 + 1)) &= (\ms{let}\;(x, y) = (i_0, a_0)\;\ms{in}\;(x + 1, y * (x + 1))) \\
  (1, 1) &= (\ms{let}\;(x, y) = (0, 1)\;\ms{in}\;(x + 1, y * (x + 1)))
\end{align*} 
we can rewrite the program in Figure~\ref{fig:fact-subst-2} to that in Figure~\ref{fig:fact-dinat}.
We can then apply general rewrite rules such as \textit{dinaturality} to rewrite
Figure~\ref{fig:fact-dinat} to Figure~\ref{fig:fact-zero}. This allows us to build up justifications
for complex optimizations, such as rewriting \ref{fig:fact-zero} to \ref{fig:fact-opt}, in terms of
simple rewriting steps. In particular, we can do \textit{complex}, \textit{error-prone} loop and
control-flow graph optimizations by breaking them down into closed set of simple algebraic steps,
with each step rigorously justified via our denotational semantics

\begin{figure}
  \begin{minipage}{.5\textwidth}
    \begin{subfigure}{\textwidth}
      \begin{align*}
        & \ms{let}\;n = 10; \\
        & \ms{br}\;\ms{loop}(1, 1) \\
        & \ms{where}\;\ms{loop}(i_0, a_0): \{\\
        & \quad \ms{if}\;i_0 < n\;\{ \\
        & \qquad \ms{br}\;\ms{loop}(i_0 + 1, a_0 * (i_0 + 1)) \\
        & \quad \}\;\ms{else}\;\{ \\
        & \qquad \ms{ret}\;a_0 \\
        & \quad \} \\
        & \}
      \end{align*}
      \caption{Substituted program from Figure \ref{fig:fact-subst}}
      \label{fig:fact-subst-2}
    \end{subfigure}
    \begin{subfigure}{\textwidth}
      \begin{align*}
        & \ms{let}\;n = 10; \\
        & \ms{br}\;\ms{loop}\;(0, 1) \\
        & \ms{where}\;\ms{loop}(x, y): \{\\
        & \quad \ms{let}\;(i_0, a_0) = (x + 1, y * (x + 1)); \\
        & \quad \ms{if}\;i_0 < n\;\{ \\
        & \qquad \ms{br}\;\ms{loop}(i_0, a_0) \\
        & \quad \}\;\ms{else}\;\{ \\
        & \qquad \ms{ret}\;a_0 \\
        & \quad \} \\
        & \}
      \end{align*}
      \caption{Equivalent to Figure \ref{fig:fact-zero} by \textit{dinaturality}}
      \label{fig:fact-dinat}
    \end{subfigure}
  \end{minipage}%
  \begin{subfigure}[c]{.5\textwidth}
    \begin{align*}
      & \ms{let}\;n = 10; \\
      & \ms{br}\;\ms{loop}( \\
      & \quad \ms{let}\;(x, y) = (0, 1); \\
      & \quad(x + 1, y * (x + 1)) \\
      & ) \\
      & \ms{where}\;\ms{loop}(i_0, a_0): \{\\
      & \quad \ms{if}\;i_0 < n\;\{ \\
      & \qquad \ms{br}\;\ms{loop}( \\
      & \qquad \quad \ms{let}\;(x, y) = (i_0, a_0); \\
      & \qquad \quad (x + 1, y * (x + 1)) \\ 
      & \qquad ) \\
      & \quad \}\;\ms{else}\;\{ \\
      & \qquad \ms{ret}\;a_0 \\
      & \quad \} \\
      & \}
    \end{align*}
    \caption{Equivalent to Figure \ref{fig:fact-subst-2} by substitution}
    \label{fig:fact-zero}
  \end{subfigure}
  \caption{
    Decomposing multi-block rewrites (from \ref{fig:fact-zero} to
    \ref{fig:fact-subst-2}, and therefore to the more optimal program 
    \ref{fig:fact-opt}) into simple algebraic steps. By verifying each step, we can
    verify complex optimizations through decomposition.
  } 
  \Description{}
  \label{fig:fact-dinat-rewrites}
\end{figure}

\section{Type Theory}

\label{sec:typing}

We now give a formal account of \isotopessa{}, starting with the types. Our types are first order,
and consists of binary sums $A + B$, products $A \otimes B$, the unit type $\mathbf{1}$, and the
empty type $\mb{0}$, all parameterised over a set of base types $X \in \mc{T}$. We write our set of
types as $\ms{Ty}(X)$. We also parameterise over:
\begin{itemize}
  
  \item A set of effects $\epsilon \in \mc{E}$, forming a join-semilattice with bottom element $\bot
  \in \mc{E}$
  
  \item For each pair $A, B \in \ms{Ty}(X)$ and effect $\epsilon \in \mc{E}$, a
  set of \textit{primitive instructions} $f \in \mc{I}_\epsilon(A, B)$, where
  $\epsilon \leq \epsilon' \implies \mc{I}_\epsilon(A, B) \subseteq
  \mc{I}_{\epsilon'}(A, B)$. 
  
  We write $\mc{I}(A, B) = \bigcup_\epsilon\mc{I}_\epsilon(A, B)$,
  $\mc{I}_\epsilon = \bigcup_{A, B}\mc{I}_\epsilon(A, B)$, and $\mc{I} =
  \bigcup_\epsilon\mc{I}_\epsilon$.

\end{itemize}
We'll call a tuple $Sg = (\mc{E}, \mc{T}, \mc{I})$ of types and instructions over these types an
\emph{\isotopessa{}-signature}.

A (variable )\textit{context} $\Gamma$ is a list of \textit{typing hypotheses}
$\thyp{x}{A}{\epsilon}$, where $x$ is a variable name, $A$ is the type of that variable, and
$\epsilon$ is the effect of using that variable (used when filling holes with effectful
expressions). If $\epsilon = \bot$, we often omit it, writing $\bhyp{x}{A}$. Similarly, we define a
\textit{label-context} to be a list of \textit{labels} $\lhyp{\ell}{A}$, where $A$ is the parameter
type that must be passed on a jump to the label $\ell$.

\begin{figure}[H]
  \begin{center}
    \begin{grammar}
      <\(A, B, C\)> ::= 
      \(X\)
      \;|\; \(A \otimes B\)
      \;|\; \(\mathbf{1}\)
      \;|\; \(A + B\)
      \;|\; \(\mathbf{0}\)

      <\(a, b, c, e\)> ::= \(x\) 
      \;|\;  \(f\;a\)
      \;|\; \(\letexpr{x}{a}{e}\)
      \alt  \(()\)
      \;|\; \((a, b)\)
      \;|\; \(\letexpr{(x, y)}{a}{e}\)
      \alt  \(\linl{a}\) 
      \;|\; \(\linr{a}\)
      \;|\; \(\labort{a}\)
      \;|\; \(\caseexpr{e}{x}{s}{y}{t}\)
      
      <\(s, t\)> ::= \(\brb{\ell}{a}\) 
      \alt  \(\letstmt{x}{a}{t}\)
      \;|\; \(\letstmt{(x, y)}{a}{t}\)
      \;|\; \(\casestmt{e}{x}{s}{y}{t}\)
      \alt  \(\where{t}{(\wbranch{\ell_i}{x_i}{t_i},)_i}\)

      <\(\Gamma\)> ::= \(\cdot\) \;|\; \(\Gamma, \thyp{x}{A}{\epsilon}\)

      <\(\ms{L}\)> ::= \(\cdot\) \;|\; \(\ms{L}, \lhyp{\ell}{A}\)
    \end{grammar}
  \end{center}
  \caption{Grammar for \isotopessa{}, parametrized over an \isotopessa{} signature}
  \Description{}
  \label{fig:ssa-grammar}
\end{figure}

As shown in Figure~\ref{fig:ssa-grammar}, \isotopessa{} terms are divided into two syntactic
categories, each of with associated with a judgement:
\begin{itemize}
  \item \emph{Expressions} $a, b, c, e$ typed with the judgement $\hasty{\Gamma}{\epsilon}{a}{A}$,
  which says that under the typing context $\Gamma$, the expression $a$ has type $A$ and effect
  $\epsilon$. We say a term is \emph{pure} if it has effect $\bot$; note that whether an expression
  is pure or not depends both on the expression itself and on the purity of the variables used in
  the expression; this is to allow reasoning about impure substitutions.
  \item \emph{Regions} $r, s, t$, which recursively define a lexically-scoped SSA program with a
  single entry and (potentially) multiple exits. This is typed with the judgement
  $\haslb{\Gamma}{r}{\ms{L}}$, which states that given that $\Gamma$ is live at the unique entry
  point, $r$ will either loop forever or branch to one of the exit labels in $\ell(A) \in \ms{L}$
  with an argument of type $A$.
\end{itemize}

The typing rules for expressions are given in Figure~\ref{fig:ssa-expr-rules}. In particular,
expressions may be built up from the following fairly standard primitives:
\begin{itemize}
  \item A variable $x$ in the context $\Gamma$, as typed by \brle{var}. We write $(A, \epsilon) \leq
  (B, \epsilon') \iff A = B \and \epsilon \leq \epsilon'$.
  \item An \emph{primitive instruction} $f \in \mc{I}_\epsilon(A, B)$ applied to an expression
  $\hasty{\Gamma}{\epsilon}{a}{A}$, typed by \brle{op}
  \item Unary and binary \emph{let-bindings}, typed by \brle{let$_1$} and \brle{let$_2$}
  respectively
  \item A \emph{pair} of expressions $\hasty{\Gamma}{\epsilon}{a}{A}$,
  $\hasty{\Gamma}{\epsilon}{b}{B}$, typed by \brle{pair}. Operationally, we interpret this as
  executing $a$, and then $b$, and returning the pair of their values.
  \item An empty tuple $()$, which types in any context by \brle{unit}
  \item Injections, typed by \brle{inl} and \brle{inr}
  \item Pattern matching on sum types, typed by \brle{case}. Operationally, we interpret this as
  executing $e$, and then, if $e$ is a left injection $\iota_l\;x$, executing $a$ with its value
  ($x$), otherwise executing $b$.
  \item An operator $\ms{abort}\;e$ allowing us to abort execution if given a value of the empty type. Since the empty type is a 0-ary sum type, $\ms{abort}$ can be seen as a $\ms{case}$ with no branches. Since the empty type is uninhabited, execution can never reach an $\ms{abort}$. This can be viewed as a typesafe version of the \texttt{unreachable} instruction in LLVM IR. 
\end{itemize}

Traditional presentations of SSA use a boolean type instead of sum types. Naturally, booleans can be encoded with sum types as $\mb{1} + \mb{1}$. If-then-else is then a $\ms{case}$ which ignores the unit payloads, so that
$\ite{e_1}{e_2}{e_3} := \caseexpr{e_1}{()}{e_2}{()}{e_3}$.

\begin{figure}
  \begin{gather*}
    \boxed{\hasty{\Gamma}{\epsilon}{a}{A}} \\
    \prftree[r]{\rle{var}}{\Gamma\;x \leq (A, \epsilon)}{\hasty{\Gamma}{\epsilon}{x}{A}} \qquad
    \prftree[r]{\rle{op}}{\isop{f}{A}{B}{\epsilon}}{\hasty{\Gamma}{\epsilon}{a}{A}}
      {\hasty{\Gamma}{\epsilon}{f\;a}{B}} \qquad
    \prftree[r]{\rle{let$_1$}}
      {\hasty{\Gamma}{\epsilon}{a}{A}}
      {\hasty{\Gamma, \bhyp{x}{A}}{\epsilon}{b}{B}}
      {\hasty{\Gamma}{\epsilon}{\letexpr{x}{a}{b}}{B}} \\
    \prftree[r]{\rle{unit}}{\hasty{\Gamma}{\epsilon}{()}{\mb{1}}} \qquad
    \prftree[r]{\rle{pair}}{\hasty{\Gamma}{\epsilon}{a}{A}}{\hasty{\Gamma}{\epsilon}{b}{B}}
      {\hasty{\Gamma}{\epsilon}{(a, b)}{A \otimes B}} \\
    \prftree[r]{\rle{let$_2$}}
      {\hasty{\Gamma}{\epsilon}{e}{A \otimes B}}
      {\hasty{\Gamma, \bhyp{x}{A}, \bhyp{y}{B}}{\epsilon}{c}{C}}
      {\hasty{\Gamma}{\epsilon}{\letexpr{(x, y)}{e}{c}}{C}} \\
    \prftree[r]{\rle{inl}}{\hasty{\Gamma}{\epsilon}{a}{A}}
      {\hasty{\Gamma}{\epsilon}{\linl{a}}{A + B}} \qquad
    \prftree[r]{\rle{inr}}{\hasty{\Gamma}{\epsilon}{b}{B}}
      {\hasty{\Gamma}{\epsilon}{\linr{b}}{A + B}} \qquad
    \prftree[r]{\rle{abort}}{\hasty{\Gamma}{\epsilon}{a}{\mb{0}}}
      {\hasty{\Gamma}{\epsilon}{\labort{a}}{A}} \\
    \prftree[r]{\rle{case}}
      {\hasty{\Gamma}{\epsilon}{e}{A + B}}
      {\hasty{\Gamma, \bhyp{x}{A}}{\epsilon}{a}{C}}
      {\hasty{\Gamma, \bhyp{y}{A}}{\epsilon}{b}{C}}
      {\hasty{\Gamma}{\epsilon}{\caseexpr{e}{x}{a}{y}{b}}{C}}
  \end{gather*}
  \caption{Rules for typing \isotopessa{} expressions}
  \Description{}
  \label{fig:ssa-expr-rules}
\end{figure}

We now move on to \emph{regions}, which can be built up as follows:
\begin{itemize}
  \item A branch to a label $\ell$ with pure argument $a$, typed with \brle{br}.
  
  \item Unary and binary \emph{let-bindings}, typed by \brle{let$_1$} and \brle{let$_2$}
  respectively
  
  \item Pattern matching on sum types, typed by \brle{case}. Operationally, we interpret this as
  executing the (potentially effectful) expression $e$, and then, if $e$ is a left injection
  $\iota_l\;x$, executing $r$ with its value ($x$), otherwise executing $s$.
  
  \item \emph{\ms{where}-statements} of the form ``$\where{r}{(\wbranch{\ell_i}{x_i}{t_i})_i}$",
  which consist of a collection of mutually recursive regions $\wbranch{\ell_i}{x_i}{t_i}$ and a
  \emph{terminator region} $r$ which may branch to one of $\ell_i$ or an exit label.
\end{itemize}
We previously described SSA programs as being made up out of \emph{basic blocks}, each of which is
made up of a sequence of instructions followed by a \emph{terminator} and, potentially, a list of
strictly dominated basic blocks this terminator may jump to. Basic blocks are an \emph{implicit}
feature of our grammar: we can view each as a list of unary or binary let-bindings, until we reach a
terminator, which is either an unconditional branch, a \ms{case}-statement or a
\ms{where}-statement, as follows:
\begin{gather*}
  \ms{defs}(\letstmt{x}{e}{r}) = (x, e)::r \qquad 
  \ms{defs}(\letstmt{(x, y)}{e}{r}) = ((x, y), e)::r \\
  \ms{defs}(r) = [] \quad \text{otherwise} \\
  \ms{terminator}(\letstmt{x}{e}{r}) 
  = \ms{terminator}(\letstmt{(x, y)}{e}{r}) 
  = \ms{terminator}(r) \\
  \ms{terminator}(r) = r \quad \text{otherwise} \\
  \ms{bb}(r) = (\ms{defs}(r), \ms{terminator}(r))
\end{gather*}
Note in particular that the region $r$ in a \ms{where}-statement
$\where{r}{(\wbranch{\ell_i}{x_i}{t_i})_i}$ is best interpreted as a terminator rather than an
entry-block even if of the form, e.g., $r = \letstmt{x}{e}{s}$. The entry block is instead the
\emph{implicit} basic-block made up of any let-bindings surrounding the \ms{where}-statement. The
key difference is that variables defined in $r$ are \emph{not} visible in the blocks $t_i$, whereas
variables defined in the entry-block are. While at first glance this can be unintuitive, the
additional generality greatly simplifies rewriting, and we can always ``normalize away'' this
feature, since our equational theory generally admits that, \emph{if} both sides are well-typed and
there is no shadowing,
\begin{align*}
  \where{(\letexpr{x}{e}{r})}{(\wbranch{\ell_i}{x_i}{t_i})_i}
  &= (\letexpr{x}{e}{(\where{r}{(\wbranch{\ell_i}{x_i}{t_i})_i})}) \\
  \where{(\letexpr{(x, y)}{e}{r})}{(\wbranch{\ell_i}{x_i}{t_i})_i}
  &= (\letexpr{(x, y)}{e}{(\where{r}{(\wbranch{\ell_i}{x_i}{t_i})_i})})
\end{align*}

\begin{figure}
  \begin{gather*}
    \boxed{\haslb{\Gamma}{r}{\ms{L}}} \\
    \prftree[r]{\rle{br}}{\hasty{\Gamma}{\bot}{a}{A}}{\ms{L}\;\ell = A}
      {\haslb{\Gamma}{\brb{\ell}{a}}{\ms{L}}} \qquad
    \prftree[r]{\rle{let$_1$-r}}
      {\hasty{\Gamma}{\epsilon}{a}{A}}
      {\haslb{\Gamma, \bhyp{x}{A}}{r}{\ms{L}}}
      {\haslb{\Gamma}{\letstmt{x}{a}{r}}{\ms{L}}} \\
    \prftree[r]{\rle{let$_2$-r}}
      {\hasty{\Gamma}{\epsilon}{e}{A \otimes B}}
      {\haslb{\Gamma, \bhyp{x}{A}, \bhyp{y}{B}}{r}{\ms{L}}}
      {\haslb{\Gamma}{\letstmt{(x, y)}{e}{r}}{\ms{L}}} \\
    \prftree[r]{\rle{case-r}}
      {\hasty{\Gamma}{\epsilon}{e}{A + B}}
      {\haslb{\Gamma, \bhyp{x}{A}}{r}{\ms{L}}}
      {\haslb{\Gamma, \bhyp{y}{B}}{s}{\ms{L}}}
      {\haslb{\Gamma}{\casestmt{e}{x}{r}{y}{s}}{\ms{L}}} \\
    \prftree[r]{\rle{cfg}}
      {\haslb{\Gamma}{r}{\ms{L}, (\lhyp{\ell_i}{A_i},)_i}}
      {\forall i. \haslb{\Gamma, \bhyp{x_i}{A_i}}{t_i}{\ms{L}, (\lhyp{\ell_j}{A_j},)_j}}
      {\haslb{\Gamma}{\where{r}{(\wbranch{\ell_i}{x_i}{t_i},)_i}}{\ms{L}}}
  \end{gather*}
  \caption{Rules for typing \isotopessa{} regions}
  \Description{}
  \label{fig:ssa-reg-rules}
\end{figure}

\subsection{Metatheory}

We can now begin to state the syntactic metatheory of \isotopessa{}. One of the
most important metatheorems, and a basic sanity check of our type theory, is
\emph{weakening}; essentially, if things typecheck in a context $\Delta$, and
$\Gamma$ contains all the variables of $\Delta$ (written $\Gamma \leq \Delta$,
pronounced ``$\Gamma$ \emph{weakens} $\Delta$''), then these things should
typecheck in the context $\Gamma$ as well. That is, ``$\Gamma$ typechecks more
terms than $\Delta$''.
 
Making things more formal, we introduce the rules for weakening $\Gamma \leq
\Delta$ in the first part of Figure \ref{fig:ssa-meta-rules}: \brle{wk-nil} says
that the empty context weakens itself, \brle{wk-skip} says that if $\Gamma$
weakens $\Delta$, then $\Gamma$ with an arbitrary variable added also weakens
$\Delta$, and \brle{wk-cons} says that if $\Gamma$ weakens $\Delta$ and
$\epsilon \leq \epsilon'$, then $\Gamma$ with $\thyp{x}{A}{\epsilon}$ added
weakens $\Delta, \thyp{x}{A}{\epsilon'}$. It is easy to see that weakening
defined in this manner induces a partial order on contexts.

We may go further and introduce weakening for \emph{label contexts} $\ms{L} \leq
\ms{K}$ analogously, except that we will flip the ordering such that $\ms{L}$
weakens $\ms{K}$ if it contains \emph{less}, rather than \emph{more}, labels
than $\ms{K}$. This is a bit unconventional, since in this case $\ms{L}$ types
\emph{less} regions than $\ms{K}$, but it will make our metatheory and
denotational semantics come out more clearly, since it corresponds to
label-contexts being ``on the right'' (with variable contexts ``on the left'').
In particular, in Figure \ref{fig:ssa-meta-rules}, we introduce the rules
\brle{lwk-nil}, which says that the empty label context weakens itself,
\brle{lwk-skip}, which says that if $\ms{L}$ weakens $\ms{K}$, then $\ms{L}$
weakens $\ms{K}$ with an arbitrary label added, and \brle{lwk-cons}, which says
that if $\ms{L}$ weakens $\ms{K}$, then $\ms{L}$ with a label $\lhyp{\ell}{A}$
added weakens $\ms{K}$ with the same label added. It is easy to see that this
induces a partial order on label contexts.

\begin{figure}
  \begin{gather*}
    \boxed{\Gamma \leq \Delta} \\
    \prftree[r]{\rle{wk-nil}}{}{\cdot \leq \cdot} \qquad
    \prftree[r]{\rle{wk-skip}}{\Gamma \leq \Delta}{\Gamma, \thyp{x}{A}{\epsilon} \leq \Delta} \qquad
    \prftree[r]{\rle{wk-cons}}{\Gamma \leq \Delta}{\epsilon \leq \epsilon'}
      {\Gamma, \thyp{x}{A}{\epsilon} \leq \Delta, \thyp{x}{A}{\epsilon'}} \\
    \boxed{\ms{L} \leq \ms{K}} \\
    \prftree[r]{\rle{lwk-nil}}{}{\cdot \leq \cdot} \qquad
    \prftree[r]{\rle{lwk-skip}}{\ms{L} \leq \ms{K}}{\ms{L} \leq \ms{K}, \lhyp{\ell}{A}} \qquad
    \prftree[r]{\rle{lwk-cons}}{\ms{L} \leq \ms{K}}
      {\ms{L}, \lhyp{\ell}{A} \leq \ms{K}, \lhyp{\ell}{A}} \\
    \boxed{\issubst{\gamma}{\Gamma}{\Delta}} \\
    \prftree[r]{\rle{sb-nil}}{}{\issubst{\cdot}{\Gamma}{\cdot}} \qquad
    \prftree[r]{\rle{sb-cons}}{\issubst{\gamma}{\Gamma}{\Delta}}{\hasty{\Gamma}{\epsilon}{e}{A}}
      {\issubst{\gamma, x \mapsto e}{\Gamma}{\Delta, \thyp{x}{A}{\epsilon}}}
       \\
    \boxed{\lbsubst{\Gamma}{\sigma}{\ms{L}}{\ms{K}}} \\
    \prftree[r]{\rle{ls-nil}}{}{\lbsubst{\Gamma}{\cdot}{\cdot}{\ms{K}}}
    \\
    \prftree[r]{\rle{ls-cons}}
      {\lbsubst{\Gamma}{\sigma}{\ms{L}}{\ms{K}}}{\haslb{\Gamma, \bhyp{x}{A}}{r}{\ms{K}}}
      {\lbsubst{\Gamma}{\sigma}{\ms{L}}{\ms{K}}}
      {\lbsubst{\Gamma}{\sigma, \ell(x) \mapsto r}{\ms{L}, \lhyp{\ell}{A}}{\ms{K}}}
  \end{gather*}
  \caption{Rules for typing \isotopessa{} weakening and substitution}
  \Description{}
  \label{fig:ssa-meta-rules}
\end{figure}

We can now state weakening formally as follows:
\begin{lemma}[Weakening]
  Given $\Gamma \leq \Delta$, $\epsilon \leq \epsilon'$, and $\ms{L} \leq \ms{K}$, we have that:
  \begin{enumerate}[label=(\alph*)]
    \item $\hasty{\Delta}{\epsilon}{a}{A} \implies \hasty{\Gamma}{\epsilon'}{a}{A}$
    \item $\haslb{\Delta}{r}{\ms{L}} \implies \haslb{\Gamma}{r}{\ms{K}}$
    \item $\issubst{\sigma}{\Delta}{\Xi} \implies \issubst{\sigma}{\Gamma}{\Xi}$
    \item $\lbsubst{\Delta}{\sigma}{\ms{L}}{\ms{K}} \implies \lbsubst{\Gamma}{\sigma}{\ms{L}}{\ms{K}}$
  \end{enumerate}
\end{lemma}
\begin{proof}
  These are formalized as:
  \begin{enumerate}[label=(\alph*)]
    \item \texttt{Term.Wf.wk} in \texttt{Typing/Term/Basic.lean}
    \item \texttt{Region.Wf.wk} in \texttt{Typing/Region/Basic.lean}
    \item Follows from \texttt{Term.Subst.Wf.comp} in \texttt{Typing/Term/Subst.lean}
    \item Follows from \texttt{Region.Subst.Wf.vsubst} in \texttt{Typing/Region/LSubst.lean}
  \end{enumerate}
\end{proof}

If we look at the proof of variable weakening, we might arrive at the alternate statement that all
the variables in $\Delta$ are also available with the same type in $\Gamma$, i.e., if
$\hasty{\Delta}{\epsilon}{x}{A} \implies \hasty{\Gamma}{\epsilon}{x}{A}$, then anything which can be
typed in $\Delta$ can be typed in $\Gamma$. More generally, we might ask if every variable
$\hasty{\Delta}{\epsilon}{x}{A}$ in $\Delta$ can be associated with a term
$\hasty{\Gamma}{\epsilon}{\gamma_x}{A}$ which is well-typed in $\Gamma$. An assignment of such
variables $\gamma : x \mapsto \gamma_x$ is called a \emph{substitution}, which we can type with the
judgement $\issubst{\gamma}{\Gamma}{\Delta}$ as per the rules given in Figure
\ref{fig:ssa-reg-rules}. In particular,
\begin{itemize}
  \item \brle{sb-nil} says that the empty substitution takes every context to the empty context.
  \item \brle{sb-cons} says that if $\gamma$ takes $\Gamma$ to $\Delta$ and
  $\hasty{\Gamma}{\epsilon}{e}{A}$, then $\gamma$ with the additional substitution $x \mapsto e$
  adjoined takes $\Gamma$ to $\Delta, \thyp{x}{A}{\epsilon}$
\end{itemize}
To \emph{use} a substitution, we simply need to perform standard capture-avoiding substitution, as
made explicit in Figure \ref{fig:ssa-subst-def}.
\begin{figure}
  \begin{gather*}
    (\gamma, x \mapsto e)(x) = e \qquad
    (\gamma, y \mapsto e)(x) = \gamma(x) \qquad
    (\cdot)(x) = x
    \\ \\
    [\gamma]x = \gamma(x) \qquad
    [\gamma](\letexpr{x}{a}{e}) = \letexpr{x}{[\gamma]a}{[\gamma]e} \qquad
    [\gamma](a, b) = ([\gamma]a, [\gamma]b) \qquad
    [\gamma]() = () \\
    [\gamma](\letexpr{(x, y)}{a}{e})
    = \letexpr{(x, y)}{[\gamma]a}{[\gamma]e} \qquad
    [\gamma](\linl{a}) = \linl{[\gamma]a} \qquad
    [\gamma](\linr{b}) = \linr{[\gamma]b} \\
    [\gamma](\caseexpr{e}{x}{a}{y}{b}) =
    \caseexpr{[\gamma]e}{x}{[\gamma]a}{y}{[\gamma]b} \\
    [\gamma](\labort{a}) = \labort{[\gamma]a} 
    \\ \\
    [\gamma](\brb{\ell}{a}) = \brb{\ell}{[\gamma]a} \qquad
    [\gamma](\letstmt{x}{a}{r}) = \letstmt{x}{[\gamma]a}{[\gamma]r} \\
    [\gamma](\letstmt{(x, y)}{e}{r}) = \letstmt{(x, y)}{[\gamma]e}{[\gamma]r} \\
    [\gamma](\casestmt{e}{x}{r}{y}{s}) 
    = \casestmt{[\gamma]e}{x}{[\gamma]r}{y}{[\gamma]s} \\
    [\gamma](\where{r}{(\wbranch{\ell_i}{x_i}{t_i},)_i}) =
    \where{[\gamma]r}{(\wbranch{\ell_i}{x_i}{[\gamma]t_i},)_i} 
    \\ \\
    [\gamma](\cdot) = \cdot \qquad
    [\gamma](\gamma', x \mapsto e) 
    = ([\gamma]\gamma', x \mapsto [\gamma]e)
    \\ \\
    [\gamma](\cdot) = \cdot \qquad
    [\gamma](\sigma, \ell(x) \mapsto r) 
    = ([\gamma]\sigma, \ell(x) \mapsto [\gamma]r)
  \end{gather*}
  \caption{ 
    Capture-avoiding substititon for \isotopessa{} terms, regions, and
    (label) substitutions; in particular, we assume bound variables and labels
    are $\alpha$-converted so as not to appear in $\gamma$/$\sigma$. 
  }
  \Description{}
  \label{fig:ssa-subst-def}
\end{figure}

This gives us everything we need to state the \emph{substitution lemma}, which
is as follows:
\begin{lemma}[Substitution]
  Given $\issubst{\gamma}{\Gamma}{\Delta}$, we have that:
  \begin{enumerate}[label=(\alph*)]
    \item $\hasty{\Delta}{\epsilon}{a}{A} \implies \hasty{\Gamma}{\epsilon}{[\gamma]a}{A}$ 
    \item $\haslb{\Delta}{r}{\ms{L}} \implies \haslb{\Gamma}{[\gamma]r}{\ms{L}}$
    \item $\issubst{\rho}{\Delta}{\Xi} \implies \issubst{[\gamma]\rho}{\Gamma}{\Xi}$
    \item $\lbsubst{\sigma}{\Gamma}{\ms{L}}{\ms{K}} \implies \lbsubst{[\gamma]\sigma}{\Delta}{\ms{L}}{\ms{K}}$
  \end{enumerate}
\end{lemma}
\begin{proof}
  These are formalized as:
  \begin{enumerate}[label=(\alph*)]
    \item \texttt{Term.Wf.subst} in \texttt{Typing/Term/Subst.lean}
    \item \texttt{Region.Wf.vsubst} in \texttt{Typing/Region/VSubst.lean}
    \item \texttt{Term.Subst.Wf.comp} in \texttt{Typing/Term/Subst.lean}
    \item \texttt{Region.Subst.Wf.vsubst} in \texttt{Typing/Region/LSubst.lean}
  \end{enumerate}
\end{proof}
Note in particular that this allows us to take the \emph{composition}
$\issubst{[\gamma']\gamma}{\Gamma'}{\Delta}$ of substitutions $\issubst{\gamma'}{\Gamma'}{\Gamma}$
and $\issubst{\gamma}{\Gamma}{\Delta}$; the composition associates as expected:
$[[\gamma_1]\gamma_2]\gamma_3 = [\gamma_1]([\gamma_2]\gamma_3)$, and has identity $[\ms{id}]\gamma =
\gamma$, yielding a category of substitutions with variable contexts $\Gamma$ as objects.

We will define a ``left extension'' operation $\lupg{\cdot}_\Xi$ yielding
$\issubst{\lupg{\gamma}_{\Xi}}{\Xi, \Gamma}{\Xi, \Delta}$ which appends the identity substitution
for each variable in $\Xi$ in the obvious manner:
\begin{equation}
  \lupg{\gamma}_{\cdot} = \gamma \qquad 
  \lupg{\gamma}_{\Xi, \thyp{x}{A}{\epsilon}} = x \mapsto x, \lupg{\gamma}_{\Xi}
\end{equation}
We may similarly define a ``right extension'' operation $\rupg{\cdot}_\Xi$ yielding
$\issubst{\rupg{\gamma}_{\Xi}}{\Gamma, \Xi}{\Delta, \Xi}$ as follows:
\begin{equation}
  \rupg{\gamma}_{\cdot} = \gamma \qquad 
  \rupg{\gamma}_{\Xi, \thyp{x}{A}{\epsilon}} = \rupg{\gamma}_{\Xi}, x \mapsto x
\end{equation}
We will usually infer $\Xi$ from context; sometimes, for convenience, we may even omit the
$\rupg{\cdot}, \lupg{\cdot}$ operators entirely, due to the fact that for any $\gamma$, we have
$[\gamma]a = [\lupg{\Gamma}_\Xi]a = [\rupg{\gamma}_\Xi]a$. The identity substitution on $\Gamma$ can
therefore be written as $\rupg{\cdot}_{\Gamma}$. One other particularly important form of
substitution is that of substituting an expression $a$ for an individual variable $x$, which we will
write $[a/x] := \lupg{(x \mapsto a)}$.

Finally, just as we can generalize weakening by substituting expressions for variables via
substitution, we can generalize label weakening by substituting \emph{labels} for
\emph{(parametrized) regions} via \emph{label substitution}. In particular, a label-substitution
$\lbsubst{\sigma}{\Gamma}{\ms{L}}{\ms{K}}$ maps every label $\ell(A) \in \ms{L}$ to a region
$\haslb{\Gamma, x : A}{r}{\ms{K}}$ parametrized by $x : A$. As shown in Figure
\ref{fig:ssa-label-subst-def}, we may then define label-substitution recursively in the obvious
manner, mapping $\ms{br}\;\ell\;a$ to $[a/x]r$ as a base case. Composition of label-substitutions is
pointwise. This allows us to state \emph{label substitution} as follows:

\begin{lemma}[Label substitution]
  Given $\lbsubst{\Gamma}{\sigma}{\ms{L}}{\ms{K}}$, we have that
  \begin{enumerate}[label=(\alph*)]
    \item $\haslb{\Gamma}{r}{\ms{L}} \implies \haslb{\Gamma}{[\sigma]r}{\ms{K}}$
    \item $\lbsubst{\Gamma}{\kappa}{\ms{L}}{\ms{J}} 
      \implies \lbsubst{\Gamma}{[\sigma]\kappa}{\ms{K}}{\ms{J}}$
  \end{enumerate}
\end{lemma}
\begin{proof}
  These are formalized as:
  \begin{enumerate}[label=(\alph*)]
    \item \texttt{Region.Wf.lsubst} in \texttt{Typing/Region/LSubst.lean}
    \item \texttt{Region.Subst.Wf.comp} in \texttt{Typing/Region/LSubst.lean}
  \end{enumerate}
\end{proof}

We may similarly define left and right extensions $\lbsubst{\Gamma}{\lupg{\sigma}_{\ms{K}}}{\ms{L},
\ms{J}}{\ms{K}, \ms{J}}$ and $\lbsubst{\Gamma}{\rupg{\sigma}_{\ms{K}}}{\ms{L}, \ms{J}}{\ms{K},
\ms{J}}$ and for label substitutions $\lbsubst{\Gamma}{\sigma}{\ms{L}}{\ms{K}}$ in the obvious
manner:
\begin{gather}
  \rupg{\sigma}_{\cdot} = \sigma \qquad 
  \rupg{\sigma}_{\ms{K}, \ell(A)} = \rupg{\sigma}_{\ms{K}}, \ell(x) \mapsto \brb{\ell}{x} \\
  \lupg{\sigma}_{\cdot} = \sigma \qquad
  \lupg{\sigma}_{\ms{K}, \ell(A)} = \ell(x) \mapsto \brb{\ell}{x}, \lupg{\sigma}_{\ms{K}}
\end{gather}
As for variable substitutions, we will often omit $\ms{L}$ when it is clear from the context, or
even omit the $\rupg{\cdot}$ and $\lupg{\cdot}$ operators entirely, due to the fact that for any
$\sigma$, we have $[\sigma]r = [\rupg{\sigma}_{\ms{K}}]r = [\lupg{\sigma}_{\ms{K}}]r$. We also
define the shorthand $[\ell / \kappa] = [\lupg{\kappa(x) \mapsto \brb{\ell}{x}}]$.

\begin{figure}
  \begin{gather*}
    (\sigma, \ell(x) \mapsto r)(\ell, a) = [a/x]r \qquad
    (\sigma, \kappa(x) \mapsto r)(\ell, a) = \sigma(\ell, a) \qquad
    (\cdot)(\ell, a) = \brb{\ell}{a}
    \\ \\
    [\sigma](\brb{\ell}{a}) = \sigma(\ell, a) \qquad
    [\sigma](\letstmt{x}{a}{r}) = \letstmt{x}{a}{[\sigma]r} \\
    [\sigma](\letstmt{(x, y)}{e}{r}) = \letstmt{(x, y)}{e}{[\sigma]r} \\
    [\sigma](\casestmt{e}{x}{r}{y}{s}) = \casestmt{e}{x}{[\sigma]r}{y}{[\sigma]s} \\
    [\sigma](\where{r}{(\wbranch{\ell_i}{x_i}{t_i},)_i}) =
    \where{([\sigma]r)}{(\wbranch{\ell_i}{x_i}{[\sigma]t_i},)_i} 
    \\ \\
    [\sigma](\cdot) = \cdot \qquad
    [\sigma](\sigma', \ell(x) \mapsto r) 
    = ([\sigma]\sigma', \ell(x) \mapsto [\sigma]r)
  \end{gather*}
  \caption{ 
    Capture-avoiding label substititon for \isotopessa{} regions and label substitutions; in 
    particular, we assume bound variables and labels are $\alpha$-converted  so as not to appear in 
    $\sigma$. 
  } 
  \Description{}
  \label{fig:ssa-label-subst-def}
\end{figure}

\section{Equational Theory}

\label{sec:equations}

\subsection{Expressions}

We can now give an equational theory for \isotopessa{} expressions. In particular,
we will inductively define an equivalence relation
$
\tmeq{\Gamma}{\epsilon}{a}{a'}{A}
$
on terms $a, a'$ for each context $\Gamma$, effect $\epsilon$, and type $A$. For each of the rules
we will present, we assume the rule is valid if and only if \emph{both sides} of the rule are
well-typed. We also assume that variables are $\alpha$-converted as appropriate to avoid shadowing;
our formalization uses de Bruijn indices, but we stick with names in this exposition for simplicity.

The rules for this relation can be roughly split into \emph{rewriting rules}, which denote when two
particular expressions have equivalent semantics, and \emph{congruence rules}, which govern how
rewrites can be composed to enable equational reasoning. In particular, our congruence rules, given
in Figure~\ref{fig:ssa-expr-congr-rules}, consist of:
\begin{itemize}
  \item \brle{refl}, \brle{symm}, \brle{trans}, which state that
  $\tmeq{\Gamma}{\epsilon}{\cdot}{\cdot}{A}$ is reflexive, transitive, and symmetric respectively
  for each choice of $\Gamma, \epsilon, A$, and therefore an equivalence relation.
  \item \brle{let$_1$}, \brle{let$_2$}, \brle{pair}, \brle{inl}, \brle{inr}, \brle{case}, and
  \brle{abort}, which state that $\tmeq{\Gamma}{\epsilon}{\cdot}{\cdot}{A}$ is a \emph{congruence}
  with respect to the corresponding expression constructor, and, in particular, that the expression
  constructors are well-defined functions on the quotient of expressions up to $\teqv$.
\end{itemize} 
We also include the following \emph{type-directed} rules as part of our congruence relation:
\begin{itemize}
  \item \brle{initial}, which equates \emph{all} terms in a context containing the empty type
  $\mb{0}$, since we will deem any such context to be \emph{unreachable} by control flow. In
  particular, any instruction or function call returning $\mb{0}$ is assumed to diverge, similarly
  to Rust's ``never type'' ``$!$".
  \item \brle{terminal}, which equates all \emph{pure} terms of unit type $\mb{1}$. Note that
  \emph{impure} terms may be disequal, since while their result values are the same, their side
  effects may differ!
\end{itemize}

\begin{figure}
  \begin{gather*}
    \prftree[r]{\rle{refl}}{\hasty{\Gamma}{\epsilon}{a}{A}}{\tmeq{\Gamma}{\epsilon}{a}{a}{A}} \qquad
    \prftree[r]{\rle{trans}}
      {\tmeq{\Gamma}{\epsilon}{a}{b}{A}}
      {\tmeq{\Gamma}{\epsilon}{b}{c}{A}} 
      {\tmeq{\Gamma}{\epsilon}{a}{c}{A}} \qquad
    \prftree[r]{\rle{symm}}
      {\tmeq{\Gamma}{\epsilon}{a}{b}{A}}
      {\tmeq{\Gamma}{\epsilon}{b}{a}{A}}
    \\
    \prftree[r]{\rle{let$_1$}}
      {\tmeq{\Gamma}{\epsilon}{a}{a'}{A}}
      {\tmeq{\Gamma, \bhyp{x}{A}}{\epsilon}{b}{b'}{B}}
      {\tmeq{\Gamma}{\epsilon}{\letexpr{x}{a}{b}}{\letexpr{x}{a'}{b'}}{B}} 
    \\
    \prftree[r]{\rle{pair}}
      {\tmeq{\Gamma}{\epsilon}{a}{a'}{A}}
      {\tmeq{\Gamma}{\epsilon}{b}{b'}{B}}
      {\tmeq{\Gamma}{\epsilon}{(a, b)}{(a', b)}{A \otimes B}}
    \\
    \prftree[r]{\rle{let$_2$}}
      {\tmeq{\Gamma}{\epsilon}{e}{e'}{A \otimes B}}
      {\tmeq{\Gamma, \bhyp{x}{A}, \bhyp{y}{B}}{\epsilon}{c}{c'}{C}}
      {\tmeq{\Gamma}{\epsilon}{\letexpr{(x, y)}{e}{c}}{\letexpr{(x, y)}{e'}{c'}}{C}}
    \\
    \prftree[r]{\rle{inl}}
      {\tmeq{\Gamma}{\epsilon}{a}{a'}{A}}
      {\tmeq{\Gamma}{\epsilon}{\linl{a}}{\linl{a'}}{A + B}} \qquad
    \prftree[r]{\rle{inr}}
      {\tmeq{\Gamma}{\epsilon}{b}{b'}{B}}
      {\tmeq{\Gamma}{\epsilon}{\linr{b}}{\linr{b'}}{A + B}} \qquad
    \\
    \prftree[r]{\rle{case}}
      {\tmeq{\Gamma}{\epsilon}{e}{e'}{A + B}}
      {\tmeq{\Gamma, \bhyp{x}{A}}{\epsilon}{a}{a'}{C}}
      {\tmeq{\Gamma, \bhyp{y}{B}}{\epsilon}{b}{b'}{C}}
      {\tmeq{\Gamma}{\epsilon}{\caseexpr{e}{x}{a}{y}{b}}{\caseexpr{e'}{x}{a'}{y}{b'}}{C}}
    \\
    \prftree[r]{\rle{abort}}
      {\tmeq{\Gamma}{\epsilon}{a}{a'}{\mb{0}}}
      {\tmeq{\Gamma}{\epsilon}{\labort{a}}{\labort{a'}}{A}}
    \\
    \prftree[r]{\rle{initial}} 
      {\hasty{\Gamma}{\epsilon}{a}{A}}
      {\hasty{\Gamma}{\epsilon}{a'}{A}}
      {\exists x, \Gamma\;x = (\mb{0}, \bot)}
      {\tmeq{\Gamma}{\epsilon}{a}{a'}{A}}
      \qquad
    \prftree[r]{\rle{terminal}}
      {\hasty{\Gamma}{\bot}{a}{\mb{1}}}
      {\hasty{\Gamma}{\bot}{a'}{\mb{1}}}
      {\tmeq{\Gamma}{\epsilon}{a}{a'}{\mb{1}}}
  \end{gather*}
  \caption{Congruence rules for \isotopessa{} expressions}
  \Description{}
  \label{fig:ssa-expr-congr-rules}
\end{figure}

We may group the rest of our rules according to the relevant constructor, i.e. $\ms{let}$ (unary and
binary) and $\ms{case}$. In particular, for unary $\ms{let}$, we have the following rules,
summarized in Figure~\ref{fig:ssa-unary-let-expr}:
\begin{itemize}
  \item \brle{let$_1$-$\beta$}, which allows us to substitute the bound variable in $x$ the
  let-statement $\letexpr{x}{a}{b}$ with its definition $a$, yielding $[a/x]b$. Note that we require
  $\hasty{\Gamma}{\bot}{a}{A}$; i.e., $a$ must be \emph{pure}.
  \brle{let$_1$-$\beta$} can be combined with \brle{initial} to derive the more convenient
  rule \brle{initial-expr}:
  \begin{gather*}
    \prftree[r]{\rle{initial-expr}} 
      {\hasty{\Gamma}{\epsilon}{a}{A}}
      {\hasty{\Gamma}{\epsilon}{a'}{A}}
      {\exists e, \hasty{\Gamma}{\bot}{e}{\mb{0}}}
      {\tmeq{\Gamma}{\epsilon}{a}{a'}{A}}
  \end{gather*}
  This states that any context in which a \emph{pure} term of empty type can be constructed equates
  all terms. Note in particular that this means divergent instructions must have impure effect!

  \item \brle{let$_1$-$\eta$}, which is the standard $\eta$-rule for \ms{let}. This is included as a
  separate rule since, while it follows trivially from $\beta$ for pure $a$, we also want to
  consider \emph{impure} expressions with some effect $e \neq \bot$.
  
  \item Rules \brle{let$_1$-op}, \brle{let$_1$-let$_1$}, \brle{let$_1$-let$_2$},
  \brle{let$_1$-abort}, and \brle{let$_1$-case} which allow us to ``pull'' a let-statement out of
  any of the other expression constructors; operationally, this is saying that the bound expression
  we pull out is evaluated before the rest of the \ms{let}-binding.
  
  For example, \brle{let$_1$-case} says that, if both
  $\letexpr{z}{\caseexpr{e}{x}{a}{y}{b}}{d}$ and
  $\caseexpr{e}{x}{\letexpr{z}{a}{d}}{y}{\letexpr{z}{b}{d}}{y}$,
  are well typed, then both must have the same behaviour:
  \begin{enumerate}
    \item Compute $e$
    \item If $e = \linl{e_l}$, compute $[e_l/x]a$, else, if $e = \linr{e_r}$, compute $[e_r/y]b$;
          store this value as $z$
    \item Compute $d$ given our value for $z$
  \end{enumerate}
  Note in particular that, since both sides are well-typed, $d$ cannot depend on either $x$ or $y$.
\end{itemize}

\begin{figure}
  \begin{gather*}
    \prftree[r]{\rle{let$_1$-$\beta$}}
      {\hasty{\Gamma}{\bot}{a}{A}}
      {\hasty{\Gamma, \bhyp{x}{A}}{\epsilon}{b}{B}}
      {\tmeq{\Gamma}{\epsilon}{\letexpr{x}{a}{b}}{[b/x]a}{B}}
    \qquad
    \prftree[r]{\rle{let$_1$-$\eta$}}
      {\hasty{\Gamma}{\epsilon}{a}{A}}
      {\tmeq{\Gamma}{\epsilon}{\letexpr{x}{a}{x}}{a}{A}} 
    \\
    \prftree[r]{\rle{let$_1$-op}}
      {\isop{f}{A}{B}{\epsilon}}
      {\hasty{\Gamma}{\epsilon}{a}{A}}
      {\hasty{\Gamma, \bhyp{y}{B}}{\epsilon}{c}{C}}
      {\tmeq{\Gamma}{\epsilon}{\letexpr{y}{f\;a}{c}}
      {\letexpr{x}{a}{\letexpr{y}{f\;x}{c}}}{C}}
    \\
    \prftree[r]{\rle{let$_1$-let$_1$}}
      {\hasty{\Gamma}{\epsilon}{a}{A}}
      {\hasty{\Gamma, \bhyp{x}{A}}{\epsilon}{b}{B}}
      {\hasty{\Gamma, \bhyp{y}{B}}{\epsilon}{c}{C}}
      {\tmeq{\Gamma}{\epsilon}
        {\letexpr{y}{(\letexpr{x}{a}{b})}{c}}
        {\letexpr{x}{a}{\letexpr{y}{b}{c}}}{C}}
    \\ 
    % \prftree[r]{\rle{let$_1$-pair}}
    %   {\hasty{\Gamma}{\epsilon}{a}{A}}
    %   {\hasty{\Gamma}{\epsilon}{b}{B}}
    %   {\hasty{\Gamma, \bhyp{z}{A \otimes B}}{\epsilon}{c}{C}}
    %   {\tmeq{\Gamma}{\epsilon}
    %     {\letexpr{z}{(a, b)}{c}}
    %     {\letexpr{x}{a}{
    %       \letexpr{y}{b}{\letexpr{z}{(x, y)}{c}}}}{C}}
    % \\
    \prftree[r]{\rle{let$_1$-let$_2$}}
      {\hasty{\Gamma}{\epsilon}{e}{A \times B}}
      {\hasty{\Gamma, \bhyp{x}{A}, \bhyp{y}{C}}{\epsilon}{c}{C}}
      {\hasty{\Gamma, \bhyp{z}{C}}{\epsilon}{d}{D}}
      {\tmeq{\Gamma}{\epsilon}
        {\letexpr{z}{(\letexpr{(x, y)}{e}{c})}{d}}
        {\letexpr{(x, y)}{e}{\letexpr{z}{c}{d}}}{D}}
    \\
    % \prftree[r]{\rle{let$_1$-inl}}
    %   {\hasty{\Gamma}{\epsilon}{a}{A}}
    %   {\hasty{\Gamma, \bhyp{x}{A + B}}{\epsilon}{c}{C}}
    %   {\hasty{\Gamma, \bhyp{y}{C}}{\epsilon}{d}{D}}
    %   {\tmeq{\Gamma}{\epsilon}{\letexpr{y}{\linl{a}}{c}}{\letexpr{x}{a}{\letexpr{y}{\linl{x}}{c}}}{C}}
    % \\
    % \prftree[r]{\rle{let$_1$-inr}}
    %   {\hasty{\Gamma}{\epsilon}{b}{B}}
    %   {\hasty{\Gamma, \bhyp{x}{A + B}}{\epsilon}{c}{C}}
    %   {\hasty{\Gamma, \bhyp{y}{C}}{\epsilon}{d}{D}}
    %   {\tmeq{\Gamma}{\epsilon}{\letexpr{y}{\linr{b}}{c}}{\letexpr{x}{b}{\letexpr{y}{\linr{x}}{c}}}{C}}
    % \\
    \prftree[r]{\rle{let$_1$-abort}}
      {\hasty{\Gamma}{\epsilon}{a}{\mb{0}}}
      {\hasty{\Gamma, \bhyp{y}{A}}{\epsilon}{b}{B}}
      {\tmeq{\Gamma}{\epsilon}
        {\letexpr{y}{\labort{b}}{b}}
        {\letexpr{x}{a}{\letexpr{y}{\labort{x}}{b}}}{B}}
    \\
    \prftree[r]{\rle{let$_1$-case}}
      {\hasty{\Gamma}{\epsilon}{e}{A + B}}
      {\hasty{\Gamma, \bhyp{x}{A}}{\epsilon}{a}{C}}
      {\hasty{\Gamma, \bhyp{y}{B}}{\epsilon}{b}{C}}
      {\hasty{\Gamma, \bhyp{z}{C}}{\epsilon}{d}{D}}
      { 
        \prfStackPremises
        {\Gamma \vdash_\epsilon \letexpr{z}{(\caseexpr{e}{x}{a}{y}{b})}{d}}
        {\hspace{6em} \teqv \caseexpr{e}{x}{\letexpr{z}{a}{d}}{y}{\letexpr{z}{b}{d}} : D}
      }
  \end{gather*}
  \Description{}
  \caption{Rewriting rules for \isotopessa{} unary \ms{let} expressions}
  \label{fig:ssa-unary-let-expr}
\end{figure}

Handling the other type constructors is a little simpler: by providing a ``binding'' rule, we
generally only need to specify how to interact with $\ms{let}_1$, as well as an $\eta$ and $\beta$
rule; interactions with the other constructors can then be derived. For example, consider the rules
for $\ms{let}_2$ given in \ref{fig:ssa-let2-case-expr}; we have:
\begin{itemize}
  \item \brle{let$_2$-$\eta$}, which is the standard $\eta$-rule for binary \ms{let}-bindings
  \item \brle{let$_2$-pair}, which acts like a slightly generalized $\beta$-rule, since we can
  derive $\beta$ reduction as follows: given pure $\hasty{\Gamma}{\bot}{a}{A}$ and
  $\hasty{\Gamma}{\bot}{b}{B}$, we have
  $$
  (\letexpr{(x, y)}{(a, b)}{c}) 
  \teqv (\letexpr{x}{a}{\letexpr{y}{b}{c}})
  \teqv ([a/x](\letexpr{y}{b}{c}))
  \teqv ([a/x][b/y]c)
  $$
  We state the rule in a more general form to allow for impure $a$ and $b$, as well as to simplify
  certain proofs.
  \item \brle{let$_2$-bind}, which allows us to ``pull'' out the bound value of a binary
  \ms{let}-expression into its own unary \ms{let}-expression; operationally, this just says that
  we execute the bound value before executing the binding itself.
\end{itemize}
This is enough to allow us to define our interactions with the other expression constructors: for
example, to show that we can lift an operation $f$ out of a binary $\ms{let}$-binding, rather than
adding a separate rule, we can instead derive (types omitted for simplicity) it from
\brle{let$_2$-bind} and \brle{let$_1$-op} as follows:
\begin{align*}
  (\letexpr{(x, y)}{f\;a}{b})
  &\teqv (\letexpr{z_f}{f\;a}{\letexpr{(x, y)}{z}{b}}) \\
  &\teqv (\letexpr{z_a}{a}{\letexpr{z_f}{f\;z_a}{\letexpr{(x, y)}{z}{b}}}) \\
  &\teqv (\letexpr{z_a}{a}{\letexpr{(x, y)}{f\;z_a}{b}})
\end{align*}

\begin{figure}
  \begin{gather*}
    \prftree[r]{\rle{let$_2$-pair}}
      {\hasty{\Gamma}{\epsilon}{a}{A}}
      {\hasty{\Gamma}{\epsilon}{b}{B}}
      {\hasty{\Gamma, \bhyp{x}{A}, \bhyp{y}{B}}{\epsilon}{c}{C}}
      {\tmeq{\Gamma}{\epsilon}{\letexpr{(x, y)}{(a, b)}{c}}{\letexpr{x}{a}{\letexpr{y}{b}{c}}}{C}}
    \\
    \prftree[r]{\rle{let$_2$-$\eta$}}
      {\hasty{\Gamma}{\epsilon}{e}{A \otimes B}}
      {\tmeq{\Gamma}{\epsilon}{\letexpr{(x, y)}{e}{(x, y)}}{e}{A \otimes B}} 
    \\
    \prftree[r]{\rle{let$_2$-bind}}
      {\hasty{\Gamma}{\epsilon}{e}{A \otimes B}}
      {\hasty{\Gamma, \bhyp{x}{A}, \bhyp{y}{B}}{\epsilon}{c}{C}}
      {\tmeq{\Gamma}{\epsilon}
        {\letexpr{(x, y)}{e}{c}}
        {\letexpr{z}{e}{\letexpr{(x, y)}{z}{c}}}{C}}
    \\
    \prftree[r]{\rle{case-inl}}
      {\hasty{\Gamma}{\epsilon}{a}{A}}
      {\hasty{\Gamma, \bhyp{x}{A}}{\epsilon}{c}{C}}
      {\hasty{\Gamma, \bhyp{y}{B}}{\epsilon}{d}{C}}
      {\tmeq{\Gamma}{\epsilon}{\caseexpr{\linl{a}}{x}{c}{y}{d}}{\letexpr{x}{a}{c}}{C}}
    \\
    \prftree[r]{\rle{case-inr}}
      {\hasty{\Gamma}{\epsilon}{b}{B}}
      {\hasty{\Gamma, \bhyp{x}{A}}{\epsilon}{c}{C}}
      {\hasty{\Gamma, \bhyp{y}{B}}{\epsilon}{d}{C}}
      {\tmeq{\Gamma}{\epsilon}{\caseexpr{\linr{b}}{x}{c}{y}{d}}{\letexpr{y}{b}{d}}{C}}
    \\
    \prftree[r]{\rle{case-$\eta$}}
      {\hasty{\Gamma}{\epsilon}{e}{A + B}}
      {\tmeq{\Gamma}{\epsilon}{\caseexpr{e}{x}{\linl{x}}{y}{\linr{y}}}{e}{A + B}}
    \\
    \prftree[r]{\rle{case-bind}}
      {\hasty{\Gamma}{\epsilon}{e}{A + B}}
      {\hasty{\Gamma, \bhyp{x}{A}}{\epsilon}{c}{C}}
      {\hasty{\Gamma, \bhyp{y}{B}}{\epsilon}{d}{C}}
      {\tmeq{\Gamma}{\epsilon}{\caseexpr{e}{x}{c}{y}{d}}
      {\letexpr{z}{e}{\caseexpr{z}{x}{c}{y}{d}}}{C}}
  \end{gather*}
  \Description{}
  \caption{Rewriting rules for \isotopessa{} binary \ms{let} and \ms{case} expressions}
  \label{fig:ssa-let2-case-expr}
\end{figure}

Similarly, it is enough to give $\eta$, $\beta$, and binding rules for \brle{case} expressions. 
In particular, we have that
\begin{itemize}
  \item \brle{case-inl} and \brle{case-inr} serve as $\beta$-reduction rules, telling us that
  \ms{case}-expressions given an injection as an argument have the expected operational behaviour.
  Note that we reduce to a \ms{let}-expression rather than perform a substitution to allow for
  impure discriminants.
  \item \brle{case-$\eta$} is the standard $\eta$-rule for \ms{case}-expressions.
  \item \brle{case-bind} allows us to ``pull'' out the bound value of the discriminant into
  it's own \ms{let}-expression; again, operationally, this just says that we need to evaluate
  the discriminant before executing the \ms{case}-expression.
\end{itemize}
It's interesting that this is enough, along with the \brle{let-case} rule and friends, to derive the
distributivity properties we would expect well-behaved \ms{case}-expressions to have. For example,
we have that
\begin{align*}
  f(\caseexpr{e}{x}{a}{y}{b}) 
  &\teqv (\letexpr{z}{\caseexpr{e}{x}{a}{y}{b}}{f\;z}) \\
  &\teqv \caseexpr{e}{x}{\letexpr{z}{a}{f\;z}}{y}{\letexpr{z}{b}{f\;z}} \\
  &\teqv \caseexpr{e}{x}{f\;a}{y}{f\;b}
\end{align*}
and likewise for more complicated distributivity properties involving, e.g., \ms{let}-bindings.

The case for other the other constructors is even more convenient: no additional rules are required
at all to handle operations, pairs, and injections. For example, we can derive the expected
bind-rule for operations as follows:
\begin{align*}
  f\;a \teqv (\letexpr{y}{f\;a}{y})
  \teqv (\letexpr{x}{a}{\letexpr{y}{f\;x}{y}})
  \teqv (\letexpr{x}{a}{f\;x})
\end{align*}

This completes the equational theory for \isotopessa{} terms; in Section \ref{ssec:completeness}, we
will show that this is enough to state a relatively powerful completeness theorem.

\subsection{Regions}

We now come to the equational theory for regions, which is similar to that for terms, except that we
also need to support control-flow graphs. As before, we will split our rules into a set of
\emph{congruence rules} and, for each region constructor, \emph{rewriting rules} based on that
constructor's semantics. Our congruence rules, given in Figure~\ref{fig:ssa-reg-congr-rules}, are
quite standard; we have:
\begin{itemize}
  \item As for terms, \brle{refl}, \brle{trans}, and \brle{symm} state that
  $\lbeq{\Gamma}{\cdot}{\cdot}{\ms{L}}$ is an equivalence relation for all $\Gamma$, $\ms{L}$.
  \item Similarly, \brle{let$_1$}, \brle{let$_2$}, \brle{case}, and \brle{cfg} state that
  $\lbeq{\Gamma}{\cdot}{\cdot}{\ms{L}}$ is a congruence over the respective region constructors;
  \emph{as well as} the equivalence relation on terms $\tmeq{\Gamma}{\epsilon}{\cdot}{\cdot}{A}$.
  \item \brle{initial} states that any context containing the empty type $\mb{0}$ equates all
  regions, by a similar reasoning to the rules for terms. Note that we do not require an analogue to
  the \brle{terminal} rule (for example, for regions targeting $\ms{L} = \lhyp{\ell}{\mb{1}}$),
  since it will follow from the version for terms; this is good, since the concept of a ``pure''
  region has not yet been defined.
\end{itemize}
Our rewriting rules for unary \ms{let}-statements, given in Figure~\ref{fig:ssa-reg-unary-let}, are
analogous to those for unary \ms{let}-expressions:
\begin{itemize}
  \item \brle{let$_1$-$\beta$} allows us to perform $\beta$-reduction of \emph{pure} expressions
  into regions; unlike for terms, we do not need an $\eta$-rule
  \item Exactly like for \ms{let}-expressions, \brle{let$_1$-op}, \brle{let$_1$-let$_1$},
  \brle{let$_1$-let$_2$}, \brle{let$_1$-abort}, and \brle{let$_1$-case} allow us to pull out nested
  subexpressions of the bound value of a \ms{let}-statement into their own unary \ms{let}-statement
\end{itemize}
Similarly to expressions, binary \ms{let}-statements and \ms{case}-statements need only the obvious
$\beta$ rule and binding rule, with all the interactions with other constructors derivable; these
rules are given in Figure~\ref{fig:ssa-reg-let2-case-expr}. Note in particular that $\eta$-rules are
not necessary, as these are derivable from binding and the $\eta$-rules for expressions.

\begin{figure}
  \begin{gather*}
    \prftree[r]{\rle{refl}}{\haslb{\Gamma}{r}{\ms{L}}}{\lbeq{\Gamma}{r}{r}{\ms{L}}} \qquad
    \prftree[r]{\rle{trans}}{\lbeq{\Gamma}{r}{s}{\ms{L}}}{\lbeq{\Gamma}{s}{t}{\ms{L}}}
      {\lbeq{\Gamma}{r}{t}{\ms{L}}} \qquad
    \prftree[r]{\rle{symm}}{\lbeq{\Gamma}{r}{s}{\ms{L}}}{\lbeq{\Gamma}{s}{r}{\ms{L}}}
    \\
    \prftree[r]{\rle{let}$_1$}
      {\tmeq{\Gamma}{\epsilon}{a}{a'}{A}}
      {\lbeq{\Gamma, \bhyp{x}{A}}{r}{r'}{\ms{L}}}
      {\lbeq{\Gamma}{\letstmt{x}{a}{r}}{\letstmt{x}{a'}{r'}}{\ms{L}}}
    \qquad
    \prftree[r]{\rle{let}$_2$}
      {\tmeq{\Gamma}{\epsilon}{e}{e'}{A \otimes B}}
      {\lbeq{\Gamma, \bhyp{x}{A}, \bhyp{y}{B}}{r}{r'}{\ms{L}}}
      {\lbeq{\Gamma}{\letstmt{(x, y)}{e}{r}}{\letstmt{(x, y)}{e'}{r'}}{\ms{L}}}
    \\
    \prftree[r]{\rle{case}}
      {\tmeq{\Gamma}{\epsilon}{e}{e'}{A + B}}
      {\lbeq{\Gamma, \bhyp{x}{A}}{r}{r'}{\ms{L}}}
      {\lbeq{\Gamma, \bhyp{y}{B}}{s}{s'}{\ms{L}}}
      {\lbeq{\Gamma}{\caseexpr{e}{x}{r}{y}{s}}{\caseexpr{e'}{x}{r'}{y}{s'}}{\ms{L}}}
    \\
    \prftree[r]{\rle{cfg}}
      {\lbeq{\Gamma}{r}{r'}{\ms{L}, (\lhyp{\ell_i}{A_i},)_i}}
      {\forall i. \lbeq{\Gamma, \bhyp{x_i}{A_i}}{t_i}{t_i'}{\ms{L}, (\lhyp{\ell_j}{A_j},)_j}}
      {\lbeq{\Gamma}
        {\where{r}{(\wbranch{\ell_i}{x_i: A_i}{t_i},)_i}}
        {\where{r'}{(\wbranch{\ell_i}{x_i: A_i}{t_i'},)_i}}
        {\ms{L}}
      }
    \\
    \prftree[r]{\rle{initial}}
      {\haslb{\Gamma}{r}{\ms{L}}}
      {\haslb{\Gamma}{s}{\ms{L}}}
      {\exists x, \Gamma\;x = \mb{0}}
      {\lbeq{\Gamma}{r}{s}{\ms{L}}}
  \end{gather*}
  \Description{}
  \caption{Congruence rules for \isotopessa{} regions}
  \label{fig:ssa-reg-congr-rules}
\end{figure}

\begin{figure}
  \begin{gather*}
    \prftree[r]{\rle{let$_1$-$\beta$}}
      {\hasty{\Gamma}{\bot}{a}{A}}
      {\haslb{\Gamma, \bhyp{x}{A}}{r}{\ms{L}}}
      {\lbeq{\Gamma}{\letstmt{x}{a}{r}}{[r/x]a}{\ms{L}}}
    \\
      \prftree[r]{\rle{let$_1$-op}}
      {\isop{f}{A}{B}{\epsilon}}
      {\hasty{\Gamma}{\epsilon}{a}{A}}
      {\haslb{\Gamma, \bhyp{y}{B}}{r}{\ms{L}}}
      {\lbeq{\Gamma}{\letstmt{y}{f\;a}{r}}{\letstmt{x}{a}{\letstmt{y}{f\;x}{r}}}{\ms{L}}}
    \\
    \prftree[r]{\rle{let$_{1}$-let$_1$}}
      {\hasty{\Gamma}{\epsilon}{a}{A}}
      {\hasty{\Gamma, \bhyp{x}{A}}{\epsilon}{b}{B}}
      {\haslb{\Gamma, \bhyp{y}{B}}{r}{\ms{L}}}
      {\lbeq{\Gamma}{\letstmt{y}{(\letexpr{x}{a}{b})}{r}}{\letstmt{x}{a}{\letstmt{y}{b}{r}}}{\ms{L}}}
    \\
    % \prftree[r]{\rle{let$_1$-pair}}
    %   {\hasty{\Gamma}{\epsilon}{a}{A}}
    %   {\hasty{\Gamma}{\epsilon}{b}{B}}
    %   {\haslb{\Gamma, \bhyp{z}{A \otimes B}}{r}{\ms{L}}}
    %   {\lbeq{\Gamma}
    %     {\letstmt{z}{(a, b)}{r}}
    %     {\letstmt{x}{a}{\letstmt{y}{b}{\letstmt{z}{(x, y)}{r}}}}
    %     {\ms{L}}}
    % \\
    \prftree[r]{\rle{let$_{1}$-let$_2$}}
      {\hasty{\Gamma}{\epsilon}{e}{A \otimes B}}
      {\hasty{\Gamma, \bhyp{x}{A}, \bhyp{y}{B}}{\epsilon}{c}{C}}
      {\haslb{\Gamma, \bhyp{z}{C}}{r}{\ms{L}}}
      {\lbeq{\Gamma}
        {\letstmt{z}{(\letexpr{(x, y)}{e}{c})}{r}}
        {\letstmt{(x, y)}{e}{\letstmt{z}{c}{r}}}
        {\ms{L}}}
    \\
    % \prftree[r]{\rle{let$_1$-inl}}
    %   {\hasty{\Gamma}{\epsilon}{a}{A}}
    %   {\haslb{\Gamma, \bhyp{y}{A + B}}{r}{\ms{L}}}
    %   {\lbeq{\Gamma}{\letstmt{y}{\linl{a}}{r}}{\letstmt{x}{a}{\letstmt{y}{\linl{x}}{r}}}{\ms{L}}}
    % \\
    % \prftree[r]{\rle{let$_1$-inr}}
    %   {\hasty{\Gamma}{\epsilon}{b}{B}}
    %   {\haslb{\Gamma, \bhyp{y}{A + B}}{r}{\ms{L}}}
    %   {\lbeq{\Gamma}{\letstmt{y}{\linr{b}}{r}}{\letstmt{x}{b}{\letstmt{y}{\linr{x}}{r}}}{\ms{L}}}
    % \\
    \\
    \prftree[r]{\rle{let$_1$-case}}
      {\hasty{\Gamma}{\epsilon}{e}{A + B}}
      {\hasty{\Gamma, \bhyp{x}{A}}{\epsilon}{a}{C}}
      {\hasty{\Gamma, \bhyp{y}{B}}{\epsilon}{b}{C}}
      {\haslb{\Gamma, \bhyp{z}{C}}{r}{\ms{L}}}
      { 
        \prfStackPremises
        {\Gamma \vdash \letstmt{z}{(\caseexpr{e}{x}{a}{y}{b})}{r}}
        {\hspace{6em} \teqv \casestmt{e}{x}{\letstmt{z}{a}{r}}{y}{\letstmt{z}{b}{r}} \rhd \ms{L}}
      }
    \\
    \prftree[r]{\rle{let$_1$-abort}}
      {\hasty{\Gamma}{\epsilon}{a}{\mb{0}}}
      {\haslb{\Gamma, \bhyp{y}{A}}{r}{\ms{L}}}
      {\lbeq{\Gamma}{\letstmt{y}{\labort{a}}{r}}
        {\letstmt{x}{a}{\letstmt{y}{\labort{x}}{r}}}{\ms{L}}}
  \end{gather*}
  \Description{}
  \caption{Rewriting rules for \isotopessa{} unary \ms{let}-statements}
  \label{fig:ssa-reg-unary-let}
\end{figure}

\begin{figure}
  \begin{gather*}
    \prftree[r]{\rle{let$_2$-pair}}
      {\hasty{\Gamma}{\epsilon}{a}{A}}
      {\hasty{\Gamma}{\epsilon}{b}{B}}
      {\haslb{\Gamma, \bhyp{x}{A}, \bhyp{y}{B}}{r}{\ms{L}}}
      {\lbeq{\Gamma}{\letstmt{(x, y)}{(a, b)}{r}}{\letstmt{x}{a}{\letstmt{y}{b}{r}}}{\ms{L}}}
    \\
    \prftree[r]{\rle{let$_2$-bind}}
      {\hasty{\Gamma}{\epsilon}{e}{A \otimes B}}
      {\haslb{\Gamma, \bhyp{x}{A}, \bhyp{y}{B}}{r}{\ms{L}}}
      {\lbeq{\Gamma}{\letstmt{(x, y)}{e}{r}}{\letstmt{z}{e}{\letstmt{(x, y)}{z}{r}}}{\ms{L}}}
    \\
    \prftree[r]{\rle{case-inl}}
      {\hasty{\Gamma}{\epsilon}{a}{A}}
      {\haslb{\Gamma, \bhyp{x}{A}}{r}{\ms{L}}}
      {\haslb{\Gamma, \bhyp{y}{B}}{s}{\ms{L}}}
      {\lbeq{\Gamma}{\caseexpr{\linl{a}}{x}{r}{y}{s}}{\letstmt{x}{a}{r}}{\ms{L}}}
    \\
    \prftree[r]{\rle{case-inr}}
      {\hasty{\Gamma}{\epsilon}{b}{B}}
      {\haslb{\Gamma, \bhyp{x}{A}}{r}{\ms{L}}}
      {\haslb{\Gamma, \bhyp{y}{B}}{s}{\ms{L}}}
      {\lbeq{\Gamma}{\caseexpr{\linr{b}}{x}{r}{y}{s}}{\letstmt{y}{b}{s}}{\ms{L}}}
    \\
    \prftree[r]{\rle{case-bind}}
    {\hasty{\Gamma}{\epsilon}{e}{A + B}}
    {\haslb{\Gamma, \bhyp{x}{A}}{r}{\ms{L}}}
    {\haslb{\Gamma, \bhyp{y}{B}}{s}{\ms{L}}}
    {\lbeq{\Gamma}{\caseexpr{e}{x}{r}{y}{s}}{\letstmt{z}{e}{\caseexpr{z}{x}{r}{y}{s}}}{\ms{L}}}
  \end{gather*}
  \Description{}
  \caption{Rewriting rules for \isotopessa{} binary \ms{let}-statements and \ms{case}-statements}
  \label{fig:ssa-reg-let2-case-expr}
\end{figure}

Dealing with \ms{where}-blocks, on the other hand, is a little bit more complicated, as shown by the
number of rules in Figure~\ref{fig:ssa-where-rules}. One difficulty is that, unlike the other region
constructors, we will need an $\eta$-rule as well as \emph{two} $\beta$-rules. The latter are simple
enough to state:
\begin{itemize}
  \item For $\ell_k$ defined in a \ms{where}-block, \brle{cfg-$\beta_1$} says that we can replace a
  branch to $\ell_k$ with argument $a$ with a \ms{let}-statement binding $a$ to the corresponding
  body $t_k$'s argument $x_k$.
  \item For $\kappa$ \emph{not} defined in a \ms{where}-block, \brle{cfg-$\beta_2$} says that
  a branch to $\kappa$ within the \ms{where}-block has the same semantics as if the \ms{where}-block
  was not there; hence, it can be removed.
\end{itemize}
To state our $\eta$-rule, however, we will need to introduce some more machinery. Given a mapping
from a set of labels $\ell_i$ to associated regions $t_i$, we may define the \emph{control-flow
graph substitution} $\cfgsubst{(\wbranch{\ell_i}{x_i}{t_i},)_i}$ pointwise as follows:
\begin{equation}
  \cfgsubst{(\wbranch{\ell_i}{x_i}{t_i},)_i}\;\kappa\;a
  := (\where{\brb{\kappa}{a}}{(\wbranch{\ell_i}{x_i}{t_i},)_i})
\end{equation}
In general, we may derive, for any label-context $\ms{L}$ (assuming $\cfgsubst{\cdot}$ acts uniformly
on the labels $\kappa$ in $\ms{L}$ as described above), the following rule:
\begin{equation}
  \prftree[r]{\rle{cfgs}}
    {\forall i. \haslb{\Gamma, \bhyp{x_i}{A_i}}{t_i}{\ms{L}, (\lhyp{\ell_j}{A_j},)_j}}
    {\lbsubst{\Gamma}
      {\cfgsubst{(\wbranch{\ell_i}{x_i}{t_i},)_i}}{\ms{L}, (\lhyp{\ell_j}{A_j},)_j}{\ms{L}}}
\end{equation}
Our $\eta$-rule, \brle{cfg-$\eta$}, says that any \ms{where}-block of the form
$\where{r}{(\wbranch{\ell_i}{x_i}{t_i},)_i}$ has the same semantics as the label-substitution
$[\cfgsubst{(\wbranch{\ell_i}{x_i}{t_i},)_i}]r$, which in effect propagates the where-block to the
branches of $r$, if any. While we called this rule \brle{cfg-$\eta$}, it also functions similarly
to a binding rule in that it allows us to derive many of the expected commutativity properties of
\ms{where}; for example, we have that
\begin{align*}
  \where{\letexpr{y}{a}{r}}{(\wbranch{\ell_i}{x_i}{\brb{\ell_j}{a_j}},)_i}
  &\teqv [\cfgsubst{(\wbranch{\ell_i}{x_i}{\brb{\ell_j}{a_j}},)_i}](\letexpr{y}{a}{r}) \\
  &\teqv \letexpr{y}{a}{[\cfgsubst{(\wbranch{\ell_i}{x_i}{\brb{\ell_j}{a_j}},)_i}]r} \\
  &\teqv \letexpr{y}{a}{\where{r}{(\wbranch{\ell_i}{x_i}{\brb{\ell_j}{a_j}},)_i}}
\end{align*}
One particularly important application of the $\eta$-rule for control-flow graphs is in validating
the rewrite
\begin{equation}
  \prftree[r]{\rle{case2cfg}}
    {\hasty{\Gamma}{\epsilon}{a}{A + B}}
    {\haslb{\Gamma, \bhyp{x}{A}}{s}{\ms{L}}}
    {\haslb{\Gamma, \bhyp{y}{B}}{t}{\ms{L}}}
    {
      \prfStackPremises{\Gamma \vdash \casestmt{a}{x}{s}{y}{t} 
        \teqv (\casestmt{a}{x}{\brb{\ell}{x}}{y}{\brb{\ell'}{y}})}{
        \where{\hspace{16em}}
          {\wbranch{\ell}{x}{s}, \wbranch{\ell'}{y}{t}} \rhd \ms{L}}
    }
\end{equation}
However, on it's own, it's not quite enough to get rid of nested \ms{where}-blocks. To enable this,
it turns out that we only need to add as an axiom the ability to get rid of a single, trivially
nested \ms{where}-block; this is given as the rule \brle{codiag}.

To be able to soundly perform equational rewriting, we will need the \emph{uniformity} property,
which is described by the rule \brle{uni}. In essence, this lets us commute pure expressions with
loop bodies, enabling rewrites like
\begin{equation}
  \ms{loop}\;\{ x = x + 1; \ms{if}\;p\;3x\;\{\ms{ret}\;3x\} \}
  \qquad \teqv \qquad 
  y = 3x; \ms{loop}\;\{ y = y + 3; \ms{if}\;p\;y\;\{\ms{ret}\;y\} \}
  \label{eqn:simple-loop-comm} 
\end{equation}
Note that substitution alone would not allow us to derive Equation~\ref{eqn:simple-loop-comm} above,
since $x$ and $y$ change each iteration, and hence, in SSA, would need to become parameters as
follows:
\begin{multline}
  \where{\brb{\ell}{x}}{\wbranch{\ell}{x}
    {\letexpr{x'}{x + 1}{\ms{if}\;p\;3x'\;\{\ms{ret}\;3x'\}\;\ms{else}\;\{\brb{\ell}{x'}\}}}}
  \\ \teqv
  \where{\letexpr{y}{3x}{\brb{\kappa}{y}}}{\wbranch{\kappa}{y}
    {\letexpr{y'}{y + 3}{\ms{if}\;p\;y'\;\{\ms{ret}\;y'\}\;\ms{else}\;\{\brb{\kappa}{y'}\}}}}
\end{multline}
The actual rule is quite complicated, so let's break it down point by point. Assume we are given:
\begin{itemize}
  \item A region $\haslb{\Gamma, \bhyp{y}{B}}{s}{\ms{L}, \kappa(B)}$ taking ``input'' $y$ of type
    $B$ and, as ``output,'' jumping to a label $\kappa$ with an argument of type $B$. We'll
    interpret branches to any other label as a (divergent) ``side effect.''
  \item A region $\haslb{\Gamma, \bhyp{x}{A}}{t}{\ms{L}, \ell(A)}$ taking ``input'' $x$ of type
    $A$ and, as ``output,'' jumping to a label $\ell$ with an argument of type $A$.
  \item A \emph{pure} expression $\hasty{\Gamma, \bhyp{x}{A}}{\bot}{e}{B}$ parameterised by a value
    $x$ of type $A$
\end{itemize}
Suppose further that the following condition holds:
$$
  \lbeq{\Gamma, \bhyp{x}{A}}{[e/y]s}{\where{t}{\wbranch{\ell}{x}{\brb{\kappa}{e}}}}
    {\ms{L}, \kappa(B)}
$$
That is, the following two programs are equivalent:
\begin{enumerate}[label=(\alph*)]
  \item Given input $x$, evaluate $e$ and, taking it's output to be input $y$, evaluate $s$,
  (implicitly) yielding as output a new value of $y$. In imperative pseudocode,
  $$
    y = e; x = s
  $$
  \item Given input $x$, evaluate $t$ and, taking it's output to be the \emph{new} value of $x$,
  evaluate $e$, (implicitly) yielding as output a new value $y$. In imperative pseudocode,
  $$
    x = t; y = e
  $$
\end{enumerate}

\emph{Then}, for any well-typed entry block $\haslb{\Gamma}{r}{\ms{L}, \ell(A)}$ (which can produce
an appropriate input $x : A$ at label $\ell$), we have that
$$
  \lbeq{\Gamma}{\where{(\where{r}{\wbranch{\ell}{x}{\brb{\kappa}{e}}})}
    {\wbranch{\kappa}{y}{s}}}{\where{r}{t}}{\ms{L}}
$$
i.e., in imperative pseudocode,
\begin{align}
  x = r; y = e; \ms{loop}\;\{ y = s \}  & \teqv x = r; \ms{loop}\;\{ x = t \} \\ 
  \intertext{
    where $s$ and $t$ may branch out of the loop. The reason why we require $e$ to be \emph{pure} is
    that impure expressions do not necessarily commute with infinite loops, even if they commute 
    with any finite number of iterations of the loop. For example, if $\ms{hi}$ is some effectful 
    operation (say, printing ``hello''), it is quite obvious that,
  }
  \ms{hi} ; x = x + 1 ; \ms{if}\;x = y\;\{\ms{ret}\;y\}
  & \teqv 
  x = x + 1 
  ; \ms{if}\;x = y\;\{\ms{hi} ; \ms{ret}\;y\} 
  ;  \ms{hi} \\
  \intertext{whereas}
  \ms{hi} ; \ms{loop} \{ x = x + 1 ; \ms{if}\;x = y\;\{\ms{ret}\;y\} \} 
  &\not\teqv
  \ms{loop} \{ x = x + 1 ;  \ms{if}\;x = y\;\{\ms{hi} ; \ms{ret}\;y\} \} ;
\end{align}
since, in particular, we may have $y \leq x$, in which case the loop will never exit and hence
$\ms{hi}$ will never be executed. 
Note that, due to \brle{let$_1$-$\beta$}, \brle{cfg-$\eta$}, and \brle{cfg-$\beta_1$}, this is
equivalent to
\begin{equation}
\prftree[r]{\rle{uni'}}
{
  \haslb{\Gamma}{r}{\ms{L}, \ell(A)}
}
{
  \lbeq{\Gamma, \bhyp{x}{A}}
    {[e/y]s}
    {[\ell(x) \mapsto \brb{\kappa}{e}]t}
    {\ms{L}, \kappa(B)}
}
{
  \lbeq{\Gamma}
    {(\where{([\ell(x) \mapsto \brb{\kappa}{e}]r)}{\wbranch{\kappa}{y}{s}})}
    {(\where{r}{t})}
    {\ms{L}}
}
\end{equation}

\begin{figure}
  \begin{gather*}
      \prftree[r]{\rle{cfg-$\beta_1$}}
        {\hasty{\Gamma}{\bot}{a}{A_k}}
        {\forall i. \haslb{\Gamma, \bhyp{x_i}{A_i}}{t_i}{\ms{L}, (\lhyp{\ell_j}{A_j},)_j}}
        {\lbeq{\Gamma}
          {\where{\brb{\ell_k}{a}}{(\wbranch{\ell_i}{x_i}{t_i},)_i}}
          {\where{(\letstmt{x_k}{a}{t_k})}{(\wbranch{\ell_i}{x_i}{t_i},)_i}}
          {\ms{L}}}
      \\
      \prftree[r]{\rle{cfg-$\beta_2$}}
        {\hasty{\Gamma}{\bot}{b}{B}}
        {\forall i. \haslb{\Gamma, \bhyp{x_i}{A_i}}{t_i}{\ms{L}, (\lhyp{\ell_j}{A_j},)_j}}
        {\ms{L}\;\kappa = B}
        {\kappa \notin \{(\ell_i,)_i\}}
        {\lbeq{\Gamma}
          {\where{\brb{\kappa}{b}}{(\wbranch{\ell_i}{x_i}{t_i},)_i}}
          {\brb{\kappa}{b}}
          {\ms{L}}}
      \\
        \prftree[r]{\rle{cfg-$\eta$}}
        {\haslb{\Gamma}{r}{\ms{L}, (\lhyp{\ell_i}{A_i},)_i}}
        {\forall i. \haslb{\Gamma, \bhyp{x_i}{A_i}}{t_i}{\ms{L}, (\lhyp{\ell_j}{A_j},)_j}}
        {
          \lbeq{\Gamma}
            {\where{r}{(\wbranch{\ell_i}{x_i}{t_i},)_i}}
            {[\cfgsubst{(\wbranch{\ell_i}{x_i}{t_i},)_i}]r}
            {\ms{L}}
        }
      \\
      \prftree[r]{\rle{codiag}}
        {\haslb{\Gamma}{r}{\ms{L}, \ell(A)}}
        {\haslb{\Gamma, \bhyp{y}{A}}{s}{\ms{L}, \ell(A), \kappa(A)}}
        {\lbeq{\Gamma}{\where{r}{\wbranch{\ell}{x}{\where{\brb{\kappa}{x}}
          {\wbranch{\kappa}{y}{s}}}}}
        {\where{r}{\wbranch{\ell}{y}{[\ell/\kappa]s}}}
        {\ms{L}}} 
      \\
      \\
      \prftree[r]{\rle{uni}}
        {
          \haslb{\Gamma}{r}{\ms{L}, \ell(A)}
        }
        % {
        %   \prfStackPremises
        %   {
        %     \hasty{\Gamma, \bhyp{x}{A}}{\bot}{e}{A}
        %   }{
        %     \haslb{\Gamma, \bhyp{x}{A}}{s}{\ms{L}, \kappa(B)}
        %   }{
        %     \haslb{\Gamma, \bhyp{x}{A}}{t}{\ms{L}, \ell(A)}
        %   }
        % }
        {
          \lbeq{\Gamma, \bhyp{x}{A}}
            {\letexpr{y}{e}{s}}
            {\where{t}{\wbranch{\ell}{x}{\brb{\kappa}{e}}}}
            {\ms{L}, \kappa(B)}
        }
        {
          \lbeq{\Gamma}
            {\where{(\where{r}{\wbranch{\ell}{x}{\brb{\kappa}{e}}})}
              {\wbranch{\kappa}{y}{s}}}
            {\where{r}{t}}
            {\ms{L}}
        }
      \\
      \text{where} \qquad
      {\hasty{\Gamma, \bhyp{x}{A}}{\bot}{e}{B}}, \quad
      {\haslb{\Gamma, \bhyp{y}{B}}{s}{\ms{L}, \kappa(B)}}, \quad \text{and} \quad
      {\haslb{\Gamma, \bhyp{x}{A}}{t}{\ms{L}, \ell(A)}}
      \\
      \\
      \prftree[r]{\rle{dinat}}
        {
          \haslb{\Gamma}{r}{\ms{L}, (\lhyp{\ell_i}{A_i},)_i}
        }{
          \lbsubst{\Gamma}{\sigma}{(\lhyp{\ell_i}{A_i},)_i}{(\lhyp{\kappa_j}{B_j},)_j}
        }{
          \forall i. \haslb{\Gamma, \bhyp{x_i}{B_i}}{t_i}{\ms{L}, (\lhyp{\ell_j}{A_j},)_j}
        }{
          \lbeq{\Gamma}
            {\where{([\lupg{\sigma}]r)}{(\wbranch{\kappa_i}{x_i}{[\lupg{\sigma}]t_i},)_i}}
            {\where{r}
              {(\wbranch{\ell_i}{x_i}{[\lupg{(\kappa_j(x_j) \mapsto t_j,)}](\sigma_i\;x_i)},)_i}}
            {\ms{L}}
        }
  \end{gather*}
  \Description{}
  \caption{Rewriting rules for \isotopessa{} \ms{where}-blocks}
  \label{fig:ssa-where-rules}
\end{figure}

This format illuminates a very important potential use for uniformity; namely, formalizing rewrites
like those in Figure~\ref{fig:fact-dinat-rewrites}. In particular, consider a program of the form
\begin{equation*}
  \haslb{\Gamma}{\where{([\ell(x) \mapsto \brb{\kappa}{e}]r)}{
    \wbranch{\kappa}{y}{[\ell(x) \mapsto \brb{\kappa}{e}]s}
  }}{\ms{L}}
\end{equation*}
where $\haslb{\Gamma}{r}{\ms{L}, \ell(A)}$, $\haslb{\Gamma, y : B}{s}{\ms{L}, \ell(A)}$, and
$\hasty{\Gamma, \bhyp{x}{A}}{\bot}{e}{B}$ is pure. Then we have that
\begin{multline*}
  [e/y][\ell(x) \mapsto \brb{\kappa}{(e)}]s 
  \\ \teqv
  [\ell(x) \mapsto \brb{\kappa}{[e/y](e)}][e/y]s
  \\ \teqv
  [\ell(x) \mapsto \brb{\kappa}{(e)}][e/y]s
\end{multline*}
and therefore that
\begin{multline*}
  \Gamma \vdash \where{([\ell(x) \mapsto \brb{\kappa}{e}]r)}{
    \wbranch{\kappa}{z}{[\ell(x) \mapsto \brb{\kappa}{e}]s}} \\ \teqv
    \where{r}{
    \wbranch{\ell}{x}{[e/y]s}} \\ \teqv
    \where{r}{
    \wbranch{\ell}{x}{\letstmt{y}{e}{s}}} 
\end{multline*}
In particular, for example, we can then easily derive the rewrite from Figure~\ref{fig:fact-dinat}
to Figure~\ref{fig:fact-zero} by noting the \emph{equalities} (an equivalence would be enough, of
course)
\begin{multline*}
  \ms{if}\;i_0 < n\;\{
    \ms{br}\;\ms{loop}(
      \letexpr{(x, y)}{(i_0, a_0)}{(x + 1, y * x + 1)}
    )
  \}\;\ms{else}\;\{
    \ms{ret}(a_0)
  \} \\
  =
  [\ms{loop}(i_0, a_0) \mapsto \letexpr{(x, y)}{(i_0, a_0)}{(x + 1, y * x + 1)}](
    \ms{if}\;i_0 < n\;\{
      \ms{br}\;\ms{loop}(i_0, a_0)
    \}\;\ms{else}\;\{
      \ms{ret}(a_0)
    \}
  )
\end{multline*}
and
\begin{multline*}
  \letstmt{n}{10}{\ms{br}\;\ms{loop}(
    \letstmt{(x, y)}{(0, 1)}{(x + 1, y * (x + 1))}
  )} \\
  =
  [\ms{loop}(i_0, a_0) \mapsto \letexpr{(x, y)}{(i_0, a_0)}{(x + 1, y * x + 1)}](
    \letstmt{n}{10}{\ms{br}\;\ms{loop}(0, 1)}
  )
\end{multline*}
Rewrites like this are an instance of the principle we call \emph{dinaturality}, which, for
structured control-flow, can be best expressed as an equivalence between the control-flow graphs
in Figure~\ref{fig:dinat-struct-cfg}. Unlike in the case of uniformity, however, this is true even
when the program fragment $P$ is \emph{impure}, since we do not commute $P$ over an infinite number
of iterations; hence, we cannot simply use uniformity to derive this rewrite rule in its most
general form. We instead postulate our final rewriting rule, \brle{dinat}, which generalises the
above rewrite from sequential composition on a structured control-flow graph to label substitution
on an arbitrary control-flow graph. 

\begin{figure}
  \begin{tikzpicture}
    \node[] (Al) at (0, 0) {};
    \node[box] (Pl1) at (0, -1) {P};
    \node[dot] (cdl) at (0, -1.5) {};
    \node[box=1/0/2/0] (Ql) at (0, -2) {\quad Q \quad};
    \node[box] (Pl2) at (0.5, -3) {P};
    \coordinate[] (cupl) at (1, -4) {};
    \coordinate[] (capl) at (1, -1.2) {};
    \node[] (Bl) at (0, -5) {};
    \wires{
      Al = { south = Pl1.north },
      Pl1 = { south = cdl.north, },
      cdl = { south = Ql.north },
      Ql = { south.1 = Bl.north, south.2 = Pl2.north},
      Pl2 = { south = cupl.west },
      cupl = { east = capl.east },
      capl = { west = cdl.east },
    }{}

    \node[] (eq) at (3.5, -2.5) {=};
    
    \node[] (Al) at (6, 0) {};
    \node[dot] (cdl) at (6, -1.25) {};
    \node[box] (Pl) at (6, -2) {P};
    \node[box=1/0/2/0] (Ql) at (6, -3) {\quad Q \quad};
    \coordinate[] (cupl) at (7, -4) {};
    \coordinate[] (capl) at (7, -0.5) {};
    \node[] (Bl) at (6, -5) {};
    \wires{
      Al = { south = cdl.north },
      cdl = { south = Pl.north },
      Pl = { south = Ql.north, },
      Ql = { south.1 = Bl.north, south.2 = cupl.west },
      cupl = { east = capl.east },
      capl = { west = cdl.east },
    }{}
  \end{tikzpicture}
  \caption{
    Dinaturality on a structured loop
  }
  \Description{}
  \label{fig:dinat-struct-cfg}
\end{figure}
It turns out that this being able to do this is essential, as it allows us to relate unary and
$n$-ary control-flow graphs and, in particular, use this relationship to interconvert between
data-flow and control-flow. This means we now have enough machinery to justify the elimination of
arbitrary nested \ms{where}-blocks, and hence to justify the rule
\begin{equation}
  \prftree[r]{\rle{cfg-fuse}}
    {\haslb{\Gamma}{r}{\ms{L}, (\lhyp{\ell_i}{A_i},)_i, (\lhyp{\kappa_j}{B_j},)_j}}
    {
      \prfStackPremises{
        \forall i. \haslb{\Gamma, \bhyp{x_i}{A_i}}{t_i}{
        \ms{L}, (\lhyp{\ell_j}{A_j},)_j}
      }{
        \forall i. \haslb{\Gamma, \bhyp{y_i}{B_i}}{s_i}{
        \ms{L}, (\lhyp{\ell_j}{A_j},)_j, (\lhyp{\kappa_k}{B_k},)_k}
      }
    }
    {
      \prfStackPremises{
        \Gamma \vdash 
          \where{(\where{r}{(\wbranch{\kappa_i}{y_i}{s_i}),})_i}{(\wbranch{\ell_j}{x_j}{t_j}),)_j}
      }{
        \hspace{8em}
        \teqv \where{r}{(\wbranch{\kappa_i}{y_i}{s_i},)_i, (\wbranch{\ell_j}{x_j}{t_j}),)_j}
        \rhd \ms{L}
      }
    }
    \label{eqn:where-fusion-1}
\end{equation}
Also very convenient is the variant
\begin{equation}
  \prftree[r]{\rle{cfg-fuse$_2$}}
    {\haslb{\Gamma}{r}{\ms{L}, (\lhyp{\ell_i}{A_i},)_i, \kappa(B)}}
    {
      \prfStackPremises{
        \forall i. \haslb{\Gamma, \bhyp{x_i}{A_i}}{t_i}{
          \ms{L}, (\lhyp{\ell_j}{A_j},)_j, \kappa(B)}
      }{
        \haslb{\Gamma, \bhyp{y}{B}}{s}{
          \ms{L}, (\lhyp{\ell_j}{A_j},)_j, \kappa(B), (\lhyp{\ell_j'}{A_j'},)_j,}
      }{
        \forall i. \haslb{\Gamma, \bhyp{x_i'}{A_i'}}{t_i'}{
          \ms{L}, (\lhyp{\ell_j}{A_j},)_j, \kappa(B), (\lhyp{\ell_j'}{A_j'},)_j}
      }
    }
    {
      \prfStackPremises{
        \Gamma \vdash 
          \where{r}{(\wbranch{\ell_i}{x_i}{t_i},)_i, 
            \wbranch{\kappa}{y}{\where{s}{(\wbranch{\ell_i'}{x_i'}{t_i'},)_i}}}
      }{
        \hspace{8em}
        \teqv \where{r}{(\wbranch{\ell_i}{x_i}{t_i},)_i, 
            \wbranch{\kappa}{y}{s}, (\wbranch{\ell_i'}{x_i'}{t_i'},)_i}
        \rhd \ms{L}
      }
    }
    \label{eqn:where-fusion-2}
\end{equation}
Rather than justify these rules directly, it turns out to be much more convenient to do so using our
denotational semantics, which, due to our completeness result in Section~\ref{ssec:completeness}, is
sound. A proof can be found in Lemma~\ref{lem:where-fusion} in the appendix.

There are some other basic rules we may want to use which turn out to be
derivable from our existing set. For example, while re-ordering labels in a \ms{where}-block looks
like a no-op in our named syntax, to rigorously justify the following rule actually requires
dinaturality (with the permutation done via a label-substitution):
\begin{equation}
  \prftree[r]{\rle{perm-cfg}}
    {\haslb{\Gamma}{r}{\ms{L}, (\lhyp{\ell_i}{A_i},)_i}}
    {\forall i. \haslb{\Gamma, \bhyp{x_i}{A_i}}{t_i}{\ms{L}, (\lhyp{\ell_j}{A_j},)_j}}
    {\sigma\;\text{permutation}}
    {\lbeq{\Gamma}
      {\where{r}{(\wbranch{\ell_i}{x_i}{t_i},)_i}}
      {\where{r}{(\wbranch{\ell_{\sigma_i}}{x_{\sigma_i}}{t_{\sigma_i}},)_i}}{\ms{L}}}
\end{equation}
Note the implicit use of the fact that if some region $r$ typechecks in some label-context $\ms{L}$,
then it typechecks in any permutation of $\ms{L}$, which is again proven by label-substitution.

\subsection{Metatheory}

We can now begin to investigate the metatheoretic properties of our equational theory. As a first
sanity check, we can verify that weakening, label-weakening, and loosening of effects all respect
our equivalence relation, as stated in the following lemma:
\begin{lemma}[Weakening (Rewriting)]
  Given $\Gamma \leq \Delta$, $\ms{L} \leq \ms{K}$, and $\epsilon \leq \epsilon'$, we have that
  \begin{enumerate}[label=(\alph*)]
    \item $\tmeq{\Delta}{\epsilon}{a}{a'}{A} \implies \tmeq{\Gamma}{\epsilon'}{a}{a'}{A}$
    \item $\lbeq{\Delta}{r}{r'}{\ms{L}} \implies \lbeq{\Gamma}{r}{r'}{\ms{L}'}$
  \end{enumerate}
\end{lemma}
\begin{proof}
  These are formalized as:
  \begin{enumerate}[label=(\alph*)]
    \item \texttt{Term.InS.wk_congr} and \texttt{Term.InS.wk_eff_congr} in 
    \texttt{Rewrite/Term/Setoid.lean}
    \item \texttt{Region.InS.vwk_congr} and \texttt{Region.InS.lwk_congr} in
    \texttt{Rewrite/Region/Setoid.lean}
  \end{enumerate}
\end{proof}
In particular, note that this lemma uses an equivalence relation on substitutions and
label-substitutions: this is just the obvious pointwise extension of the equivalence relation on
terms and regions respectively. We give the rules for this relation in
Figure~\ref{fig:ssa-subst-equiv} in the interests of explicitness. It is straightforward to verify
that these are indeed equivalence relations.  In fact, it turns out that substitution and
label-substitution both respect these equivalences, in the following precise sense:
\begin{lemma}[Congruence (Substitution)]
  Given $\tmseq{\gamma}{\gamma'}{\Gamma}{\Delta}$, we have that
  \begin{enumerate}[label=(\alph*)]
    \item $\tmeq{\Delta}{\epsilon}{a}{a'}{A} 
      \implies \tmeq{\Gamma}{\epsilon}{[\gamma]a}{[\gamma']a'}{A}$
    \item $\lbeq{\Delta}{r}{r'}{\ms{L}} 
      \implies \lbeq{\Gamma}{[\gamma]r}{[\gamma']r'}{\ms{L}}$
    \item $\tmseq{\rho}{\rho'}{\Delta}{\Xi}
      \implies \tmseq{[\gamma]\rho}{[\gamma']\rho'}{\Gamma}{\Xi}$
    \item $\lbseq{\sigma}{\sigma'}{\Delta}{\ms{L}}{\ms{K}}
      \implies \lbseq{[\gamma]\sigma}{[\gamma']\sigma'}{\Gamma}{\ms{L}}{\ms{K}}$
  \end{enumerate}
\end{lemma}
\begin{proof}
  These are formalized as:
  \begin{enumerate}[label=(\alph*)]
    \item \texttt{Term.InS.subst_congr} in \texttt{Rewrite/Term/Setoid.lean}
    \item \texttt{Region.InS.vsubst_congr} in \texttt{Rewrite/Region/Setoid.lean}
    \item \texttt{Term.Subst.InS.comp_congr} in \texttt{Rewrite/Term/Setoid.lean}
    \item \texttt{Region.Subst.InS.vsubst_congr} in \texttt{Rewrite/Region/LSubst.lean}
  \end{enumerate}
\end{proof}
\begin{lemma}[Congruence (Label Substitution)]
  Given $\lbseq{\sigma}{\sigma'}{\Gamma}{\ms{L}}{\ms{K}}$, we have that
  \begin{enumerate}[label=(\alph*)]
    \item $\lbeq{\Gamma}{r}{r'}{\ms{L}} \implies \lbeq{\Gamma}{[\sigma]r}{[\sigma']r'}{\ms{K}}$
    \item $\lbseq{\kappa}{\kappa'}{\Gamma}{\ms{L}}{\ms{J}}
      \implies \lbseq{[\sigma]\kappa}{[\sigma']\kappa'}{\Gamma}{\ms{K}}{\ms{J}}$
  \end{enumerate}
\end{lemma}
\begin{proof}
  These are formalized as:
  \begin{enumerate}[label=(\alph*)]
    \item \texttt{Region.InS.lsubst_congr} in \texttt{Rewrite/Region/LSubst.lean}
    \item \texttt{Region.LSubst.InS.comp_congr} in \texttt{Rewrite/Region/LSubst.lean}
  \end{enumerate}
\end{proof}
This means, in particular, that, substitution and label-substitution are well-defined operators on
equivalence classes of terms, which will come in handy later as we set out to prove completeness
in Section~\ref{ssec:completeness}.

\begin{figure}
  \begin{gather*}
    \prftree[r]{\rle{sb-nil}}{\tmseq{\cdot}{\cdot}{\Gamma}{\cdot}} \qquad
    \prftree[r]{\rle{sb-cons}}
      {\tmeq{\Gamma}{\epsilon}{a}{a'}{A}}{\tmseq{\gamma}{\gamma'}{\Gamma}{\Delta}}
      {\tmseq{\gamma, x \mapsto a}{\gamma', x \mapsto a'}{\Gamma}{\Delta, \thyp{x}{A}{\epsilon}}}
    \\
    \prftree[r]{\rle{sb-skip-l}}
      {\tmseq{\gamma}{\gamma'}{\Gamma}{\Delta}}
      {\tmseq{\gamma, x \mapsto a}{\gamma'}{\Gamma}{\Delta}} \qquad
    \prftree[r]{\rle{sb-skip-r}}
      {\tmseq{\gamma}{\gamma'}{\Gamma}{\Delta}}
      {\tmseq{\gamma}{\gamma', x \mapsto a'}{\Gamma}{\Delta}}
    \\
    \prftree[r]{\rle{ls-nil}}{\lbseq{\cdot}{\cdot}{\cdot}{\ms{K}}} \qquad
    \prftree[r]{\rle{ls-cons}}
      {\lbeq{\Gamma, \bhyp{x}{A}}{r}{r'}{\ms{K}}}
      {\lbseq{\sigma}{\sigma'}{\Gamma}{\ms{L}}{\ms{K}}}
      {\lbseq
        {\sigma, \ell(x) \mapsto r}{\sigma', \ell(x) \mapsto r'}{\Gamma}
        {\ms{L}, \ell(A)}{\ms{K}}}
    \\
    \prftree[r]{\rle{ls-skip-l}}
      {\lbseq{\sigma}{\sigma'}{\Gamma}{\ms{L}}{\ms{K}}}
      {\lbseq{\sigma, \ell(x) \mapsto r}{\sigma'}{\Gamma}{\ms{L}}{\ms{K}}}
      \qquad
    \prftree[r]{\rle{ls-skip-r}}
      {\lbseq{\sigma}{\sigma'}{\Gamma}{\ms{L}}{\ms{K}}}
      {\lbseq{\sigma}{\sigma', \ell(x) \mapsto r'}{\Gamma}{\ms{L}}{\ms{K}}}
    \\
    \prftree[r]{\rle{sb-id}}
      {\tmseq{\gamma}{\gamma'}{\Gamma, \thyp{x}{A}{\epsilon}}{\Delta, \thyp{x}{A}{\epsilon}}}
      {\tmseq{\gamma}{\gamma'}{\Gamma}{\Delta}} \qquad
    \prftree[r]{\rle{ls-id}}
      {\lbseq{\sigma}{\sigma'}{\Gamma}{\ms{L}, \lhyp{\ell}{A}}{\ms{K}, \lhyp{\ell}{A}}}
      {\lbseq{\sigma}{\sigma'}{\Gamma}{\ms{L}}{\ms{K}}}
  \end{gather*}
  \Description{}
  \caption{Rules for the equivalence relation on \isotopessa{} substitutions and label-substitutions}
  \label{fig:ssa-subst-equiv}
\end{figure}

\subsection{Strict SSA}

\label{ssec:ssa-normal}

The relaxation of SSA to \isotopessa{} allows us to state our equational theory and handle
substitution more conveniently. However, this approach may lead readers to question whether we have
truly provided a type-theoretic presentation of SSA as it is used in practice. To address this
concern, we introduce a subset of \isotopessa{} regions called \emph{strict regions}. We make the
following claims about these strict regions:
\begin{enumerate}
  \item Every \isotopessa{} region can be converted to an equivalent strict region 
  %using our equational theory 
  \label{claim:ssa-conv}
  \item Every strict region can be erased to a well-formed SSA program by removing
  ``\ms{where}'' \label{claim:ssa-erase}
  \item Every well-formed SSA program can be typed as an strict region purely by adding
  ``\ms{where}'' \label{claim:ssa-wf}
  \item Two strict regions which erase to the same SSA program are equivalent
  % according to our equational theory
  \label{claim:ssa-inj}
\end{enumerate}
We assert that strict regions correspond exactly to SSA, with the \ms{where}-bracketing serving as
syntactic sugar to make our typing rules syntax-directed. Specifically, \ms{where}-blocks can be
added algorithmically by computing dominance relationships between basic blocks.

Converting directly to strict SSA can be unwieldy. Therefore, we do
this in two stages, first converting \isotopessa{} into to a subset
corresponding to A-normal form (ANF) extended with mutually recursive
\ms{where}-bindings, and then converting the ANF regions into strict
regions.

\subsubsection{From \isotopessa{} to A-Normal Form}

Our first step is to extend the work of \citet{chakravarty-functional-ssa-2003}, by providing an
algorithm for converting between \isotopessa{} and ANF in an equivalence preserving way. We begin by
defining a typing judgement $\ahasty{\Gamma}{\epsilon}{a}{A}$ for \emph{atomic expressions} in
Figure~\ref{fig:ssa-ops}. An atomic expression is either a single variable, a function application,
an injection, a pair of variables, or the unit expression. We specifically forbid
\ms{case}-expressions and \ms{let}-bindings at this level; they are required to appear only at the
region level. It is obvious that $\ahasty{\Gamma}{\epsilon}{a}{A} \implies
\hasty{\Gamma}{\epsilon}{a}{A}$. 

We now present typing rules for ANF regions in Figure~\ref{fig:ssa-anf}. These rules mirror those
for regions but with additional constraints: each \ms{let}-binding in a region must bind an atomic
expression, and case-statements and branches must operate on variables. This ensures that each
binding corresponds to a single SSA instruction. It is similarly obvious that
$\ahaslb{\Gamma}{r}{\ms{L}}$. We can now define a syntactic function to convert an \isotopessa{}
region to ANF inductively as follows:
\begin{equation}
  \begin{aligned}
    \toanf{\brb{\ell}{a}} &= \letanf{x}{a}{\brb{\ell}{x}} \\
    \toanf{\letexpr{x}{a}{r}} &= \letanf{x}{a}{\toanf{r}} \\
    \toanf{\letexpr{(x, y)}{a}{r}} &= \letanf{z}{a}{\letstmt{(x, y)}{z}{\toanf{r}}} \\
    \toanf{\casestmt{a}{x}{r}{y}{s}} &= \letanf{z}{a}{\casestmt{z}{x}{\toanf{r}}{y}{\toanf{s}}} \\
    \toanf{\where{r}{(\wbranch{\ell_i}{x_i}{t_i},)_i}}
      &= \where{\toanf{r}}{(\wbranch{\ell_i}{x_i}{\toanf{t_i}},)_i}
  \end{aligned}
\end{equation}
where we define $\letanf{x}{a}{r}$ by induction on expressions $a$ as follows
\begin{equation}
  \begin{aligned}
    \letanf{x}{a}{r} &= (\letstmt{x}{a}{r}) \hspace{8em} \text{if}\;a\;\text{atomic} \\
    \letanf{x}{f\;e}{r} &= \letanf{y}{e}{\letstmt{x}{f\;y}{r}} \\
    \letanf{x}{(\letexpr{y}{e}{a})}{r} &= \letanf{y}{e}{\letanf{x}{a}{r}} \\
    \letanf{x}{(e_1, e_2)}{r} &= \letanf{y_1}{e_1}{\letanf{y_2}{e_2}{\letanf{x}{(y_1,y_2)}{r}}} \\
    \letanf{x}{(\letexpr{(y, z)}{e}{a})}{r} 
      &= \letanf{w}{e}{(\letstmt{(y, z)}{w}{\letanf{x}{a}{r}})} \\
    \letanf{x}{\linl{e}}{r} &= \letanf{y}{e}{(\letstmt{x}{\linl{y}}{r})} \\
    \letanf{x}{\linr{e}}{r} &= \letanf{y}{e}{(\letstmt{x}{\linr{y}}{r})} \\
    \letanf{x}{\caseexpr{e}{y}{a}{z}{b}}{r} &= 
      \letanf{w}{e}{\\ & \qquad \casestmt{w}{y}{\letanf{x}{a}{r}}{z}{\letanf{x}{b}{r}}} \\
    \letanf{x}{\labort{e}}{r} &= \letanf{y}{e}{(\letstmt{x}{\labort{y}}{r})}
  \end{aligned}
\end{equation}
We state the correctness of these functions in the following lemma:
\begin{lemma}[name=A-normalization, restate=anfconversion]
  For all $\haslb{\Gamma}{r}{\ms{L}}$, we have that that $\ahaslb{\Gamma}{\ms{ANF}(r)}{\ms{L}}$ and
  $\lbeq{\Gamma}{r}{\ms{ANF}(r)}{\ms{L}}$. Furthermore, given $\hasty{\Gamma}{\epsilon}{a}{A}$, we
  have that
  \begin{itemize}
    \item $\haslb{\Gamma, \bhyp{x}{A}}{r}{\ms{L}} \implies 
      \lbeq{\Gamma}{\letstmt{x}{a}{r}}{\letanf{x}{a}{r}}{\ms{L}}$
    \item $\ahaslb{\Gamma, \bhyp{x}{A}}{r}{\ms{L}} \implies
      \ahaslb{\Gamma}{\letanf{x}{a}{r}}{\ms{L}}$
  \end{itemize}
\end{lemma}
\begin{proof}
  See Appendix~\ref{proof:anf-conversion}
\end{proof}

\begin{figure}
  \begin{equation*}
    \boxed{\ahasty{\Gamma}{\epsilon}{a}{A}}
  \end{equation*}
  \begin{gather*}    
    \prftree[r]{\rle{var}}{\Gamma\;x \leq (A, \epsilon)}{\ahasty{\Gamma}{\epsilon}{x}{A}} \qquad
    \prftree[r]{\rle{op}}{\isop{f}{A}{B}{\epsilon}}{\Gamma\;x \leq (A, \epsilon)}
      {\ahasty{\Gamma}{\epsilon}{f\;x}{B}} \\
    \prftree[r]{\rle{unit}}{\ahasty{\Gamma}{\epsilon}{()}{\mb{1}}} \qquad
    \prftree[r]{\rle{pair}}{\Gamma\;x \leq (A, \epsilon)}{\Gamma\;y \leq (B, \epsilon)}
      {\hasty{\Gamma}{\epsilon}{(x, y)}{A \otimes B}} \\
    \prftree[r]{\rle{inl}}{\Gamma\;x \leq (A, \epsilon)}
      {\hasty{\Gamma}{\epsilon}{\linl{x}}{A + B}} \qquad
    \prftree[r]{\rle{inr}}{\Gamma\;y \leq (B, \epsilon)}
      {\hasty{\Gamma}{\epsilon}{\linr{y}}{A + B}} \qquad
    \prftree[r]{\rle{abort}}{\Gamma\;x \leq (\mb{0}, \epsilon)}
      {\hasty{\Gamma}{\epsilon}{\labort{x}}{A}}
  \end{gather*}
  \caption{Typing rules for atomic expressions}
  \Description{}
  \label{fig:ssa-ops}
\end{figure}

\begin{figure}
  \begin{equation*}
    \boxed{\ahaslb{\Gamma}{r}{\ms{L}}}
  \end{equation*}
  \begin{gather*}
    \prftree[r]{\rle{br}}{\Gamma\;x \leq (A, \epsilon)}{\ms{L}\;\ell = A}
      {\ahaslb{\Gamma}{\brb{\ell}{x}}{\ms{L}}} \qquad
    \prftree[r]{\rle{let$_1$-r}}
      {\ahasty{\Gamma}{\epsilon}{a}{A}}
      {\ahaslb{\Gamma, \bhyp{x}{A}}{r}{\ms{L}}}
      {\ahaslb{\Gamma}{\letstmt{x}{a}{r}}{\ms{L}}} \\
    \prftree[r]{\rle{let$_2$-r}}
      {\ahasty{\Gamma}{\epsilon}{e}{A \otimes B}}
      {\ahaslb{\Gamma, \bhyp{x}{A}, \bhyp{y}{B}}{r}{\ms{L}}}
      {\ahaslb{\Gamma}{\letstmt{(x, y)}{e}{r}}{\ms{L}}} \\
    \prftree[r]{\rle{case-r}}
      {\Gamma\;x \leq (A + B, \epsilon)}
      {\ahaslb{\Gamma, \bhyp{y}{A}}{r}{\ms{L}}}
      {\ahaslb{\Gamma, \bhyp{z}{B}}{s}{\ms{L}}}
      {\ahaslb{\Gamma}{\casestmt{x}{y}{r}{z}{s}}{\ms{L}}} \\
    \prftree[r]{\rle{cfg}}
      {\ahaslb{\Gamma}{r}{\ms{L}, (\lhyp{\ell_i}{A_i},)_i}}
      {\forall i. \ahaslb{\Gamma, \bhyp{x_i}{A_i}}{t_i}{\ms{L}, (\lhyp{\ell_j}{A_j},)_j}}
      {\ahaslb{\Gamma}{\where{r}{(\wbranch{\ell_i}{x_i}{t_i},)_i}}{\ms{L}}}
  \end{gather*}
  \caption{Typing rules for \isotopessa{} ANF regions}
  \Description{}
  \label{fig:ssa-anf}
\end{figure}

\subsubsection{From ANF to SSA}

We now wish to give typing rules for strict regions. Since SSA consists of \emph{basic blocks}, that is,
linear sequences of instructions followed by a terminator, we will introduce two mutually
recursive judgements:
\begin{itemize}
  \item $\thaslb{\Gamma}{r}{\ms{L}}$, which states that $r$ is a terminator, which we model
  as a tree of case-statements with unconditional branches at the leaves
\item $\bhaslb{\Gamma}{r}{\ms{L}}$, which states that $r$ is a basic block \emph{plus} the blocks
  which are its children in the dominator tree. The basic block itself is a linear
  sequence of unary and binary let-bindings followed by a terminator, which is wrapped in a where-block
  containing the child blocks. 
\end{itemize}
A strict region, then, is a basic block, with the block itself (made up of its component
let-bindings followed by its terminator) the distinguished \emph{entry block}. We give typing rules
for these judgements in Figure~\ref{fig:ssa-strict}. Note in particular that we require every block
to have an associated \ms{where}-block of dominated blocks; if there are no blocks dominated, then
we simply use the empty \ms{where}-block. This hints at the natural inductive representation
\begin{lstlisting}
    struct BasicBlock {
      instructions : List LetBinding,
      terminator : Terminator,
      children : List BasicBlock
    }
\end{lstlisting}
In particular, given a strict region $r$, we can define the following functions to extract the
region's entry block and a list of its children as follows
\begin{equation}
  \begin{aligned}
    \toentry{\letstmt{x}{a}{r}} &= (\letstmt{x}{a}{\toentry{r}}) \\
    \toentry{\letstmt{(x, y)}{a}{r}} &= (\letstmt{(x, y)}{a}{\toentry{r}}) \\
    \toentry{\where{r}{(\wbranch{\ell_i}{x_i}{t_i},)_i}} &= r
  \end{aligned}
\end{equation}
\begin{equation}
  \begin{aligned}
    \todom{\letstmt{x}{a}{r}} &= \todom{r} \\
    \todom{\letstmt{(x, y)}{a}{r}} &= \todom{r} \\
    \todom{\where{r}{(\wbranch{\ell_i}{x_i}{t_i},)_i}} &= [\wbranch{\ell_i}{x_i}{\todom{t_i}},]_i
  \end{aligned}
\end{equation}
We can similarly define a function to construct a strict region from an entry block and a list
of children as follows:
\begin{equation}
  \begin{aligned}
  \adddom{\letstmt{x}{a}{r}}{G} &= \letstmt{x}{a}{\adddom{r}}{G} \\
  \adddom{\letstmt{(x, y)}{a}{r}}{G} &= \letstmt{(x, y)}{a}{\adddom{r}}{G} \\
  \adddom{r}{(\wbranch{\ell_i}{x_i}{t_i},)_i} &= \where{r}{(\wbranch{\ell_i}{x_i}{t_i},)_i}
  \end{aligned}
\end{equation}
It is easy to see that these functions are mutually inverse: for any strict region $r$, we have
\begin{equation}
  r = \adddom{\toentry{r}}{\todom{r}}
\end{equation}
Some other useful facts about $\adddom{\cdot}{\cdot}$ include:
\begin{itemize}
  \item It is a congruence: if $r \teqv r'$ and each $t_i \teqv t_i'$,
  $\adddom{r}{(\wbranch{\ell_i}{x_i}{t_i},)_i} \teqv \adddom{r'}{(\wbranch{\ell_i}{x_i}{t_i'},)_i}$
%  \item If both sides are well-typed (in particular, if no $t_i$ use any of the variables in $r$),
%  $\wbranch{r}{(\wbranch{\ell_i}{x_i}{t_i},)_i} \teqv \adddom{r}{(\wbranch{\ell_i}{x_i}{t_i},)_i}$
  \item It is invariant up to permutations
  $\sigma$: $\adddom{r}{(\wbranch{\ell_i}{x_i}{t_i},)_i} \teqv 
    \adddom{r}{(\wbranch{\ell_{\sigma_i}}{x_{\sigma_i}}{t_{\sigma_i}},)_i}$
  \item Similarly, \emph{if both sides of the equation are well-typed}, i.e.,
  \begin{itemize}
    \item All $t_i$ do not use $y$ or any of the variables defined in $s$
    \item All branches to $\ell_i$ come from $\kappa$ or $\ell_j$ (and not from either $r$ or $G$)
  \end{itemize}
  \begin{equation}
    \adddom{r}{(G, \wbranch{\kappa}{y}{s}, (\wbranch{\ell_i}{x_i}{t_i},))} \teqv
    \adddom{r}{(G, \wbranch{\kappa}{y}{\where{s}{(\wbranch{\ell_i}{x_i}{t_i},)}})}
  \end{equation} 
  and hence
  \begin{equation}
    \adddom{r}{(G, \wbranch{\kappa}{y}{s}, (\wbranch{\ell_i}{x_i}{t_i},))} \teqv
    \adddom{r}{(G, \wbranch{\kappa}{y}{\adddom{s}{(\wbranch{\ell_i}{x_i}{t_i},)}})}
  \end{equation}
  In particular, we may apply this rule twice to obtain
  \begin{equation}
    \adddom{r}{(G, \wbranch{\kappa}{y}{\adddom{s}{G'}}, (\wbranch{\ell_i}{x_i}{t_i},))} \teqv
    \adddom{r}{(G, \wbranch{\kappa}{y}{\adddom{s}{G', (\wbranch{\ell_i}{x_i}{t_i},)}})}
    \label{eqn:pull-where}
  \end{equation}
\end{itemize}
Now that we have defined strict regions, we can give a function $\tossa{r}$ to convert a region $r$
to a strict SSA as follows:
\begin{equation}
  \begin{aligned}
    \tossa{r} &= \ssawhere{\toanf{r}}{\cdot}
      \\
    \ssawhere{t}{(\wbranch{\ell_i}{x_i}{t_i},)_i} 
      &= \where{t}{(\wbranch{\ell_i}{x_i}{t_i},)_i} \qquad \text{where}\;t\;\text{is a terminator}
      \\
    \ssawhere{(\letstmt{x}{a}{r})}{G} &= (\letstmt{x}{a}{\ssawhere{r}{G}}) \\
    \ssawhere{(\letstmt{(x, y)}{a}{r})}{G} &= (\letstmt{(x, y)}{a}{\ssawhere{r}{G}}) \\
    \ssawhere{(\casestmt{a}{x}{s}{y}{t})}{G} &=
     \where{(\casestmt{a}{x}{\brb{\ell_l}{x}}{y}
                    {\brb{\ell_r}{y}})  \\ & \qquad }
                    {G, \wbranch{\ell_l}{x}{\tossa{s}}, \wbranch{\ell_r}{y}{\tossa{t}}} \\
    \ssawhere{(\where{r}{(\wbranch{\ell_i}{x_i}{t_i},)_i})}{G} 
      &= \ssawhere{r}{(G, (\wbranch{\ell_i}{x_i}{\tossa{t_i}},)_i)}
  \end{aligned}
\end{equation}
We may state its correctness as follows:
\begin{lemma}[name=SSA conversion, restate=ssaconversion]
  If $\haslb{\Gamma}{r}{\ms{L}}$, then $\bhaslb{\Gamma}{\tossa{r}}{\ms{L}}$ and
  $\lbeq{\Gamma}{r}{\tossa{r}}{\ms{L}}$. In particular, given $\ahaslb{\Gamma}{r}{\ms{L},
  (\ell_i(A_i),)_i}$ and $\forall i, \bhaslb{\Gamma}{t_i}{\ms{L}, (\ell_i(A_i),)_i}$, we have that
  $\bhaslb{\Gamma}{\ssawhere{r}{(\wbranch{\ell_i}{x_i}{t_i},)_i}}{\ms{L}}$ and $\lbeq{\Gamma}
  {(\where{r}{(\wbranch{\ell_i}{x_i}{t_i},)_i})  }
  {\ssawhere{r}{(\wbranch{\ell_i}{x_i}{t_i},)_i}}{\ms{L}}$
\end{lemma}
\begin{proof}
  See Appendix~\ref{proof:ssa-conversion}
\end{proof}

\begin{figure}
  \begin{equation*}
    \boxed{\thaslb{\Gamma}{r}{\ms{L}}}
  \end{equation*}
  \begin{gather*}
    \prftree[r]{\rle{br}}{\Gamma\;x \leq (A, \epsilon)}{\ms{L}\;\ell = A}
      {\thaslb{\Gamma}{\brb{\ell}{x}}{\ms{L}}} \qquad
    \prftree[r]{\rle{case-t}}
      {\Gamma\;x \leq (A + B, \epsilon)}
      {\thaslb{\Gamma, \bhyp{y}{A}}{r}{\ms{L}}}
      {\thaslb{\Gamma, \bhyp{z}{B}}{s}{\ms{L}}}
      {\thaslb{\Gamma}{\casestmt{x}{y}{r}{z}{s}}{\ms{L}}} \\
  \end{gather*}
  \begin{equation*}
    \boxed{\bhaslb{\Gamma}{r}{\ms{L}}}
  \end{equation*}
  \begin{gather*}
    \prftree[r]{\rle{cfg}}
      {\thaslb{\Gamma}{r}{\ms{L}, (\lhyp{\ell_i}{A_i},)_i}}
      {\forall i. \bhaslb{\Gamma, \bhyp{x_i}{A_i}}{t_i}{\ms{L}, (\lhyp{\ell_j}{A_j},)_j}}
      {\bhaslb{\Gamma}{\where{r}{(\wbranch{\ell_i}{x_i}{t_i},)_i}}{\ms{L}}} \\
    \prftree[r]{\rle{let$_1$-c}}
      {\ahasty{\Gamma}{\epsilon}{a}{A}}
      {\bhaslb{\Gamma, \bhyp{x}{A}}{r}{\ms{L}}}
      {\bhaslb{\Gamma}{\letstmt{x}{a}{r}}{\ms{L}}} \\
    \prftree[r]{\rle{let$_2$-c}}
      {\ahasty{\Gamma}{\epsilon}{e}{A \otimes B}}
      {\bhaslb{\Gamma, \bhyp{x}{A}, \bhyp{y}{B}}{r}{\ms{L}}}
      {\bhaslb{\Gamma}{\letstmt{(x, y)}{e}{r}}{\ms{L}}} \\
  \end{gather*}
  \caption{Typing rules for \isotopessa{} terminators and strict regions}
  \Description{}
  \label{fig:ssa-strict}
\end{figure}

We now come to the final part of our argument: that strict regions are equivalent to standard SSA.
We begin by giving a grammar for standard SSA in Figure~\ref{fig:ssa-standard}. It is easy to see
that removing \ms{where}-blocks from a strict region, gives us a function $\tocfg{\cdot}$ from
strict regions $r$ to control-flow graphs $G$ defined as follows:
\begin{equation}
    \tocfg{r} = \toentry{r}, (\wbranch{\ell_i}{x_i}{\tocfg{t_i}},)_{
      (\wbranch{\ell_i}{x_i}{t_i}) \in \todom{r}}
\end{equation}
where
\begin{equation}
  G, \wbranch{\ell}{x}{\beta, G'} := G, \wbranch{\ell}{x}{\beta}, G' 
\end{equation}
An SSA program is \emph{well-formed} if:
\begin{itemize}
  \item All variable uses are well-typed
  \item All variable uses respect dominance-based scoping
\end{itemize}
It is easy to see that any erased program is well-formed, since our typing rules guarantee every
expression is well-typed, while our lexical scoping for values ensures that variables are only
visible in the children of the block $\beta$ in which they are defined, which the lexical scoping of
labels guarantees are dominated by $\beta$. On the other hand, we may give an algorithm to convert
any well-formed SSA program $G$ into a well-typed strict region $r = \toreg{G}$ as follows:
\begin{enumerate}
  \item Compute the dominance tree of $G$, rooted at its entry block $\beta$.
  \item For each child $\wbranch{\ell_i}{x_i}{\beta_i}$ of $\beta$, let $G_i$ denote the CFG
  composed of the descendants of $\beta_i$ (with $\beta_i$ as entry block), given in the order they
  appear in $G$. Recursively compute $r_i = \toreg{G_i}$.
  \item Return the program $\adddom{\beta}{(\wbranch{\ell_i}{x_i}{r_i},)_i}$
\end{enumerate}
We will write $G \simeq G'$ to mean ``$G$ is a permutation of $G'$'' (in particular, $G$ and $G'$
must have the same entry block!). It is easy to see that this algorithm gives a strict SSA region
which erases to a permutation of $G$, i.e., $\tocfg{\toreg{G}} \simeq G$, as desired. To complete
our argument, it hence suffices to show that, given strict regions $\shaslb{\Gamma}{r}{\ms{L}}$,
$\shaslb{\Gamma}{r'}{\ms{L}}$, such that $\tocfg{r} \simeq \tocfg{r'}$, we have that
$\lbeq{\Gamma}{r}{r'}{\ms{L}}$. We break this down into two lemmas:
\begin{lemma}[name=Permutation Invariance, restate=cfgperminvar]
  If $G \simeq G'$ and $\shaslb{\Gamma}{\toreg{G}}{\ms{L}}$, then 
    $\shaslb{\Gamma}{\toreg{G'}}{\ms{L}}$ and $\lbeq{\Gamma}{\toreg{G}}{\toreg{G'}}{\ms{L}}$
\end{lemma}
\begin{proof}
  See Appendix~\ref{proof:cfg-perm-invar}
\end{proof}
\begin{lemma}[name=CFG Conversion, restate=cfgconversion]
  If $\shaslb{\Gamma}{r}{\ms{L}}$, then $\lbeq{\Gamma}{r}{\toreg{\tocfg{r}}}{\ms{L}}$.
\end{lemma}
\begin{proof}
  See Appendix~\ref{proof:cfg-conversion}
\end{proof}

We may now conclude that, given $\shaslb{\Gamma}{r}{\ms{L}}$ and $\shaslb{\Gamma}{r'}{\ms{L}}$, if
$\tocfg{r} \simeq \tocfg{r'}$, we have that $\lbeq{\Gamma}{r}{r'}{\ms{L}}$, since in particular
\begin{equation}
  r \teqv \toreg{\tocfg{r}} \teqv \toreg{\tocfg{r'}} \teqv r'
\end{equation}

\begin{figure}[H]
  \begin{center}
    \begin{grammar}
      <\(a, b, c, e\)> ::= \(x\) 
      \;|\;  \(f\;x\)
      \alt  \(()\)
      \;|\; \((x, y)\)
      \alt  \(\linl{x}\) 
      \;|\; \(\linr{y}\)
      \;|\; \(\labort{x}\)
      
      <\(\tau\)> ::= \(\brb{\ell}{x}\) 
      \;|\; \(\casestmt{x}{y}{\tau}{z}{\tau'}\)

      <\(\beta\)> := \(\tau\)
      \;|\; \(\letstmt{x}{a}{\beta}\)
      \;|\; \(\letstmt{(x, y)}{a}{\beta}\)

      <\(G\)> ::= \(\beta\) \;|\; \(G; \wbranch{\ell}{x}{\beta}\)
    \end{grammar}
  \end{center}
  \caption{ Grammar for standard SSA, parametrized over an \isotopessa{} signature. A standard SSA
    program is made up of \emph{basic blocks} \(\beta\), and can be thought of as a
    \emph{control-flow graph} \(G\) (with a distinguished entry block $\beta$). } \Description{}
  \label{fig:ssa-standard}
\end{figure}

\subsection{Records and Enums}

\label{ssec:records-enums}

For simplicity, we defined our language to have only binary products and sums. However, we would
like to support records (named $n$-ary products) and enums (named $n$-ary sums). We will implement
these on top of our language using contexts; the resulting machinery will turn out to be crucial in
proving the B\"ohm-Jacopini theorem in Section~\ref{ssec:data-control} and completeness in
Section~\ref{ssec:completeness}. We begin by defining \emph{packing} of variable contexts as
follows:
\begin{equation}
  \pckd{\cdot} = \mb{1} \qquad \pckd{\Gamma, \thyp{x}{A}{\epsilon}} = \pckd{\Gamma} \otimes A
\end{equation}
While in general we destructure pairs using binary \ms{let}-bindings, it will be more convenient
to start defining operations on records in terms of projections. We begin by defining projections
on pairs in the obvious manner:
\begin{equation}
  \prftree[r]{\rle{$\pi_l$}}
    {\hasty{\Gamma}{\epsilon}{e}{A \otimes B}}
    {\hasty{\Gamma}{\epsilon}{\pi_l\;e := (\letexpr{(x, y)}{e}{x})}{A}}
  \qquad
  \prftree[r]{\rle{$\pi_r$}}
    {\hasty{\Gamma}{\epsilon}{e}{A \otimes B}}
    {\hasty{\Gamma}{\epsilon}{\pi_r\;e := (\letexpr{(x, y)}{e}{y})}{A}}
\end{equation}
We may then define projections on records as follows:
\begin{equation}
  \pi_{(\Delta, \thyp{x}{A}{\epsilon}), y}\;e = \pi_{\Delta, y}(\pi_l\;e) \qquad
  \pi_{(\Delta, \thyp{x}{A}{\epsilon}), x}\;e = \pi_r\;e
\end{equation}
with typing rule
\begin{equation}
  \prftree[r]{\rle{$\pi_{\ms{rec}}$}}
    {\hasty{\Gamma}{\epsilon}{e}{\pckd{\Delta}}}
    {\Delta(x) = A}
    {\hasty{\Gamma}{\epsilon}{\pi_{\Delta, x}\;e}{A}}
\end{equation}
We will omit $\Delta$ when it is clear from context. We can define a term
$\hasty{\Gamma}{\ms{eff}(\Gamma)}{\ms{packed}(\Gamma)}{[\Gamma]}$ to ``pack'' up the current context
into a record as follows:
\begin{equation}
  \ms{packed}(\cdot) = () \qquad
  \ms{packed}(\Gamma, \thyp{x}{A}{\epsilon}) = (\ms{packed}(\Gamma), x)
\end{equation}
where we define the \emph{effect} of a context $\Gamma$ as follows
\begin{equation}
  \ms{eff}(\cdot) = \bot \qquad 
  \ms{eff}(\Gamma, \thyp{x}{A}{\epsilon}) = \ms{eff}(\Gamma) \sqcup \epsilon
\end{equation}
In particular, we say $\Gamma$ is \emph{pure} if $\ms{eff}(\Gamma) = \bot$. We may now define
packing and unpacking substitutions
$\issubst{\ms{pack}_y^\otimes(\Gamma)}{\Gamma}{\bhyp{y}{\pckd{\Gamma}}}$,
$\issubst{\ms{unpack}_y^\otimes(\Gamma)}{\bhyp{y}{\pckd{\Gamma}}}{\Gamma}$ as follows:
\begin{equation}
  \ms{pack}_y^\otimes(\Gamma) = y \mapsto \ms{packed}(\Gamma) \qquad
  \ms{unpack}_y^\otimes(\Gamma) = (x \mapsto \pi_{\Gamma, x}\;y)_{x \in \Gamma}
\end{equation}
It turns out that our equational theory is sufficient to show that the pack and unpack substitutions
are inverses of each other \emph{for pure contexts}, i.e.
\begin{equation}
  \begin{aligned}
  \issubst
    {[\ms{unpack}_y^\otimes(\Gamma)]\ms{pack}_y^\otimes(\Gamma) 
      &\teqv \ms{id}_{\Gamma}}{\Gamma}{\Gamma}
  \\
  \issubst
    {[\ms{pack}_y^\otimes(\Gamma)]\ms{unpack}_y^\otimes(\Gamma) &\teqv \ms{id}_{\bhyp{y}{[\Gamma]}}}
    {\bhyp{y}{\Gamma}}{\bhyp{y}{\Gamma}}
  \end{aligned}
\end{equation}
Similarly, we may define a ``packing'' operation $\pckd{\cdot}$ on label-contexts as follows:
\begin{equation}
  \pckd{\cdot} = \mb{0} \qquad \pckd{\ms{L}, \ell(A)} = \pckd{\ms{L}} + A
\end{equation}
We may then define injections on enums as follows:
\begin{equation}
  \iota_{(\ms{L}, \ell(A)), \kappa}\;e = \iota_{\ms{L}, \kappa}(\iota_l\;e) \qquad
  \iota_{(\ms{L}, \ell(A)), \ell}\;e = \iota_r\;e
\end{equation}
with typing rule
\begin{equation}
  \prftree[r]{\rle{$\iota_{\ms{enum}}$}}
    {\hasty{\Gamma}{\epsilon}{e}{A}}
    {\ms{L}(\ell) = A}
    {\hasty{\Gamma}{\epsilon}{\iota_{\ms{L}, \ell}\;e}{\pckd{\ms{L}}}}
\end{equation}
Similarly, we can define $n$-ary case-statements on enums as follows:
\begin{equation}
  \begin{aligned}
  \ms{case}_\cdot\;e\;\{\} &= \infty \\
  \ms{case}_{\ms{L}, \kappa(A)}\;e\;\{(\ell_i(x_i): t_i)_i, \kappa(y): s\}
  &= \casestmt{e}{x}{\ms{case}_{\ms{L}}\;e\;\{(\ell_i(A_i): t_i)_i\}}{y}{s}
  \end{aligned}
\end{equation}
with typing rule
\begin{equation}
  \prftree[r]{\rle{$\ms{case}_{\ms{enum}}$}}
    {\hasty{\Gamma}{\epsilon}{e}{[\ms{L}]}}
    {\forall i. \haslb{\Gamma, \bhyp{x_i}{A_i}}{t_i}{\ms{K}}}
    {\haslb{\Gamma}{\ms{case}_{\ms{L}}\;e\;\{(\ell_i(x_i): t_i)_i\}}{\ms{K}}}
\end{equation}
where
\begin{equation}
  \prftree[r]{$\infty$}
    {\haslb{\Gamma}{\infty := \where{\brb{\ell}{()}}{\wbranch{\ell}{x}{\brb{\ell}{x}}}}{\ms{L}}}
\end{equation}
We may now define ``packing'' and ``unpacking'' label-substitutions
$\lbsubst{\cdot}{\ms{pack}_\kappa^+(\ms{L})}{\ms{L}}{\kappa([\ms{L}])}$,
$\lbsubst{\cdot}{\ms{unpack}_\kappa^+(\ms{L})}{\kappa([\ms{L}])}{\ms{L}}$ as follows:
\begin{equation}
  \ms{pack}_\kappa^+(\cdot) = \cdot \qquad
  \ms{pack}_\kappa^+(\ms{L}, \ell(A)) 
  = [\kappa(x) \mapsto \brb{\kappa}{\iota_l\;x}]\ms{pack}_\kappa^+(\ms{L}), 
    \ell(x) \mapsto \brb{\kappa}{\iota_r\;x}  
\end{equation}
\begin{equation}
  \ms{unpack}_\kappa^+(\ms{L}) = \kappa(x) \mapsto \ms{unpack}^+(\ms{L})\;x
\end{equation}
where we have
\begin{equation}
  \prftree[r]{\rle{unpack$^+$}}
    {\hasty{\Gamma}{\epsilon}{e}{\ms{L}}}
    {\haslb{\Gamma}{\ms{unpack}^+(\ms{L})\;e}{\ms{L}}} \qquad
  \prftree[r]{$\infty$}
    {\haslb{\Gamma}{\infty := \where{\brb{\ell}{()}}{\wbranch{\ell}{x}{\brb{\ell}{x}}}}{\ms{L}}}
\end{equation}
defined as follows:
\begin{equation}
  \ms{unpack}^+(\cdot)\;e = \infty \qquad
  \ms{unpack}^+(\ms{L}, \ell(A))\;e = \casestmt{e}{x}{\ms{unpack}^+(\ms{L})\;x}{y}{\brb{\ell}{y}}
\end{equation}
We can further show that, in this case for all label contexts $\ms{L}$, these substitutions are
mutually inverse, i.e., that
\begin{equation}
  \begin{aligned}
  \lbsubst{\Gamma}
    {[\ms{unpack}_\kappa^+(\ms{L})]\ms{pack}_\kappa^+(\ms{L}) 
      &\teqv \ms{id}_{\ms{L}}}{\ms{L}}{\ms{L}}
  \\
  \lbsubst{\Gamma}
    {[\ms{pack}_\kappa^+(\ms{L})]\ms{unpack}_\kappa^+(\ms{L}) &\teqv \ms{id}_{\kappa([\ms{L}])}}
    {\kappa([\ms{L}])}{\kappa([\ms{L}])}
  \end{aligned}
\end{equation}
Finally, fixing a distinguished variable $\invar$, it will be very useful to define 
a ``\emph{variable packing}'' operation on expressions and regions as follows:
\begin{equation}
  \haslb{\Gamma}{r}{\ms{L}} \implies 
  \haslb{\invar : [\Gamma]}{\pckd{r}^\otimes := [\ms{unpack}_{\invar}^\otimes(\Gamma)]r}{\ms{L}}
\end{equation}
In particular, for $\Gamma$ pure, the packing operation $[\cdot]$ is an injection on expressions and
regions w.r.t. the equational theory, since $\ms{unpack}_{\invar}^\otimes(\Gamma)$ has an inverse.
Similarly, fixing a distinguished label $\outlb$, we may define a ``label packing'' operation on
regions as follows:
\begin{equation}
  \haslb{\Gamma}{r}{\ms{L}} \implies 
  \haslb{\Gamma}{\pckd{r}^+ := [\ms{pack}_{\outlb}^+(\ms{L})]r}{\outlb([\ms{L}])}
\end{equation}
Since the packing substitution has an inverse, it similarly follows that label packing is an
injection w.r.t. the equational theory, i.e.
\begin{equation}
  \lbeq{\Gamma}{r}{r'}{\ms{L}} \iff \lbeq{\Gamma}{\pckd{r}^+}{\pckd{r'}^+}{\ms{L}}
\end{equation}
Finally, we can define a ``packing'' operation on regions to be given by label-packing followed by
variable-packing, or vice versa (since it turns out the operations commute), as follows:
\begin{equation}
  \haslb{\Gamma}{r}{\ms{L}} \implies 
  \haslb{\invar : [\Gamma]}
    {\pckd{r} := \pckd{\pckd{r}^+}^\otimes = \pckd{\pckd{r}^\otimes}^+}
    {\outlb(\ms{L})}
\end{equation}
We similarly have that this is an injection for $\Gamma$ pure, i.e.
\begin{equation}
  \lbeq{\Gamma}{r}{r'}{\ms{L}} \iff \lbeq{\Gamma}{\pckd{r}}{\pckd{r'}}{\ms{L}}
\end{equation}

\subsection{B\"ohm-Jacopini for SSA}

\label{ssec:data-control}

Now that we have given our equational theory, we want to show it is ``good enough'' to reason about
the properties inherent to all SSA-based languages. In particular, we
wish to show that we have enough power to reason about interconversion between data-flow and
control-flow. For example, a state machine can be implemented either as a switch on a state value,
or as a set of mutually-tail-recursive functions (i.e., the state can be encoded in the program
counter). We demonstrate this by using some of the machinery from the previous section to state and prove a form
of the B\"ohm-Jacopini theorem \cite{bohm-jacopini} for SSA.

The B\"ohm-Jacopini theorem states that every general control-flow graph program can be rewritten to
an equivalent program which uses only structured control-flow: i.e., it can be rewritten to a program
using only conditional branching, sequencing, and loops. To adapt this result to SSA, we
need to express branching, sequencing and loops as SSA regions $\haslb{\Gamma}{r}{\ms{L}}$, 
so that we can build up an inductive set of structured regions. We will be maximally strict, and
allow branching to an exit label in $\ms{L}$ only as the terminal statement in a structured program
(and, in particular, not from within a loop!). It is obvious that we can represent conditional
branching using \ms{case}-statements, but we need convenient primitives for sequencing and looping.
Sequencing can be expressed using a \ms{where}-block as follows:
\begin{equation}
  \prftree[r]{\rle{seq}}
    {\haslb{\Gamma}{r}{\outlb(A)}}
    {\haslb{\Gamma, \bhyp{\invar}{A}}{s}{\ms{L}}}
    {\haslb{\Gamma}{\ms{seq}(r, s) 
      := (\where{[\outlb(x) \mapsto \brb{\ell}{x}]r}{\wbranch{\ell}{\invar}{s}})}{\ms{L}}}
\end{equation}
where ``$\invar$" is a distinguished input variable, and ``$\outlb$'' is a distinguished output
label. Another, equivalent, way of writing sequencing is: 
\begin{equation}
  \lbeq{\Gamma}{\ms{seq}(r, s)}{[\outlb(\invar) \mapsto s]r}{\ms{L}}
\end{equation}
That is, when we exit $r$, we jump to (a copy of) $s$. 

Expressing structured looping is a bit more complicated, since we do not have
mutable variables and hence cannot directly express while loops. If we recall that loops
can be expressed as tail-recursive procedures which carry the loop state in the argument,
then we can define a ``functional do-while loop'' as follows: 
\begin{equation}
  \prftree[r]{\rle{loop}}
    {\hasty{\Gamma}{\epsilon}{e}{A}}
    {\haslb{\Gamma, \bhyp{\invar}{A}}{r}{\outlb(B + A)}}
    {\haslb{\Gamma}{\ms{loop}(e, r) 
      := (\where{\brb{\ell}{e}}{\wbranch{\ell}{\invar}
        {\ms{seq}(r, \casestmt{\invar}{x}{\brb{\outlb}{x}}{y}{\brb{\ell}{y}})}})}{\outlb(B)}}
\end{equation}
We can now define the inductive predicate
$\shaslb{\Gamma}{r}{\ms{L}}$, which says that $r$ is a structured
region with input variables $\Gamma$ and output labels $\ms{L}$, as in
Figure~\ref{fig:structured-regions}. Note that this defines a subset of the
well-typed regions: every region $r$ such that
$\shaslb{\Gamma}{r}{\ms{L}}$ is also well-typed as $\haslb{\Gamma}{r}{\ms{L}}$.

It now remains to give an algorithm to convert a region $r$ targeting label-context $\ms{L}$ into an
equivalent structured region $\towhile{\ms{L}}{r}$. In particular, we define
\begin{equation}
  \towhile{\ms{L}}{r} = [\ms{unpack}^+_\outlb(\ms{L})]\topwhile{L}{r}
\end{equation}
where we define
\begin{equation}
  \begin{aligned}
    \topwhile{\ms{L}}{\ms{br}\;\ell\;a} &= \ms{pack}^+(\ms{L})_\ell(a) \\
    \topwhile{\ms{L}}{\letstmt{x}{a}{r}} &= \letstmt{x}{a}{\topwhile{\ms{L}}{r}} \\
    \topwhile{\ms{L}}{\letstmt{(x, y)}{a}{r}} &= \letstmt{(x, y)}{a}{\topwhile{\ms{L}}{r}} \\
    \topwhile{\ms{L}}{\casestmt{e}{x}{r}{y}{s}} 
      &= \casestmt{e}{x}{\topwhile{\ms{L}}{r}}{y}{\topwhile{\ms{L}}{s}} \\
    \topwhile{\ms{L}}{\where{r}{(\wbranch{\ell_i}{x_i}{t_i},)_i}} 
      &=
      \ms{seq}(\topwhile{\ms{L}}{r}, \caseexpr{\ms{ua}\;\invar}{x}{\brb{\outlb}{x}\\ & \qquad}{y}
        {\ms{loop}(y, 
        \ms{case}_{\ms{R}}\;\invar\;
          \{\ell_i(x_i) : \ms{seq}(\topwhile{\ms{L}}{t_i}, \brb{\outlb}{(\ms{ua}\;\invar)}\})})
      \\ \qquad \text{where} \qquad \ms{R} = (\ell_i(A_i),)
  \end{aligned}
\end{equation}
and
\begin{equation}
  \ms{ua}_{\ms{L}, \cdot}\;e = e \qquad
  \ms{ua}_{\ms{L}, (\ms{R}, \ell(A))}\;e 
    = \caseexpr{e}{x}{
        \caseexpr{\ms{ua}_{\ms{L}, \ms{R}}\;x}{z}{\iota_l\;z}{w}{\iota_r\;\iota_l\;w}}
        {y}{\iota_r\;(\iota_r\;y)}
\end{equation}
The $\topwhile{\ms{L}}{r}$ function does the actual transformation, and we can see that most of the
cases are trivial except for the $\where{r}{(\wbranch{\ell_i}{x_i}{t_i},)_i}$ case. In this case, we replace the set of labels bound
by the $\ms{where}$ clause with a tag and a while loop containing a case statement branching on the
tag. (This uses the $\ms{ua}$ function, which implements associativity of n-ary coproducts. That is, 
given $\hasty{\Gamma}{\bot}{e}{[L, R]}$, we have that
$\hasty{\Gamma}{\bot}{\ms{ua}_{\ms{L}}\;e}{[L] + [R]}$.)

The B\"ohm-Jacopini theorem for SSA can then
be written:
\begin{theorem}[name=B\"ohm-Jacopini for SSA, restate=bohmjacopini]
  For all $\haslb{\Gamma}{r}{\ms{L}}$, $\shaslb{\Gamma}{\towhile{\ms{L}}{r}}{\ms{L}}$ is structured,
  and $\lbeq{\Gamma}{r}{\towhile{\ms{L}}{r}}{\ms{L}}$. In particular, we have that
  $\shaslb{\Gamma}{\topwhile{\ms{L}}{r}}{\outlb(\pckd{\ms{L}})}$ is structured, and
  $\lbeq{\Gamma}{\pckd{r}^+}{\topwhile{r}}{\outlb(\pckd{\ms{L}})}$.
\end{theorem}
\begin{proof}
  See Appendix~\ref{proof:bohm-jacopini}
\end{proof}

\begin{figure}
  \begin{gather*}
    \prftree[r]{\rle{s-br}}{\hasty{\Gamma}{\bot}{a}{A}}{\ms{L}\;\ell = A}
      {\shaslb{\Gamma}{\brb{\ell}{a}}{\ms{L}}} \qquad
    \prftree[r]{\rle{s-let$_1$-r}}
      {\hasty{\Gamma}{\epsilon}{a}{A}}
      {\shaslb{\Gamma, \bhyp{x}{A}}{r}{\ms{L}}}
      {\shaslb{\Gamma}{\letstmt{x}{a}{r}}{\ms{L}}} \\
    \prftree[r]{\rle{s-let$_2$-r}}
      {\hasty{\Gamma}{\epsilon}{e}{A \otimes B}}
      {\shaslb{\Gamma, \bhyp{x}{A}, \bhyp{y}{B}}{r}{\ms{L}}}
      {\shaslb{\Gamma}{\letstmt{(x, y)}{e}{r}}{\ms{L}}} \\
    \prftree[r]{\rle{s-case-r}}
      {\hasty{\Gamma}{\epsilon}{e}{A + B}}
      {\shaslb{\Gamma, \bhyp{x}{A}}{r}{\ms{L}}}
      {\shaslb{\Gamma, \bhyp{y}{B}}{s}{\ms{L}}}
      {\shaslb{\Gamma}{\casestmt{e}{x}{r}{y}{s}}{\ms{L}}} \\
    \prftree[r]{\rle{s-seq}}
      {\shaslb{\Gamma}{r}{\outlb(A)}}
      {\shaslb{\Gamma, \bhyp{\invar}{A}}{s}{\ms{L}}}
      {\shaslb{\Gamma}{\ms{seq}(r, s)}{\ms{L}}} \qquad
    \prftree[r]{\rle{s-loop}}
      {\hasty{\Gamma}{\epsilon}{e}{\outlb(A)}}
      {\shaslb{\Gamma, \bhyp{\invar}{A}}{r}{\outlb(B + A)}}
      {\shaslb{\Gamma}{\ms{loop}(e, r)}{\outlb(B)}}
  \end{gather*}
  \caption{Typing rules for structured regions}
  \Description{}
  \label{fig:structured-regions}
\end{figure}

\section{Denotational Semantics}

\label{sec:densem}

\subsection{Freyd Elgot Categories}

\citet{moggi-91-monad} showed that the Kleisli category of a strong
monad over a CCC interprets effectful higher-order functional
programs. There are two mismatches in using this as a semantics for
SSA. On one hand, SSA has features not necessarily supported by Moggi
models such as arbitrary, unstructured cyclic control-flow. On the
other hand Moggi models support features SSA (as a first-order
language) does not have, such as higher-order functions and
first-class computation values.

Given that we want to model SSA with some category $\mc{C}$, we hence have to think about what
structure we need $\mc{C}$ to possess so that it has exactly our desired features.
Obviously, we need a way to take the product of two objects, to be able to model contexts as well as
pairs. The usual way to do this is via \emph{monoidal categories}. However, monoidal categories
typically have too many
equations: computations operating on independent data must always commute
(this is the ``sliding rule''), which is obviously not true for effects such as printing, since
$$
\ms{print}(x) ; \ms{print}(y) \neq \ms{print}(y) ; \ms{print}(x)
$$
Instead, we will only require that our category be \emph{premonoidal}: ``monoidal, without
sliding.'' Indeed, the Kleisli category of a strong monad on a CCC is not always monoidal, as
demonstrated by the writer monad on $\ms{Set}$ (which exposes $\ms{print}$), but \emph{is} always
premonoidal. We define a premonoidal category as follows:
\begin{definition}[Symmetric Premonoidal Category]
  We define a \emph{binoidal category} to be a category $\mc{C}$ equipped with a binary operation
  $\otimes : |\mc{C}| \times |\mc{C}| \to |\mc{C}|$ on the objects of $\mc{C}$ and, for each $A, B
  \in |\mc{C}|$, functors $A \otimes -, - \otimes B : \mc{C} \to \mc{C}$. We say a morphism $f : A
  \to A'$ in a binoidal category is \emph{central} if, for all $g : B \to B'$, it satisfies
  \emph{sliding}:
  $$
  f \otimes B ; A' \otimes g = A \otimes g ; f \otimes B' \qquad
  B \otimes f ; g \otimes A' = g \otimes A ; B' \otimes f
  $$
  in which case we may write these morphisms as $f \otimes g : A \otimes B \to A' \otimes B'$ and $g
  \otimes f : B \otimes A \to B' \otimes A'$ respectively. A \emph{premonoidal category} is, then, a
  binoidal category equipped with:
  \begin{itemize}
    \item An \emph{identity} object $I \in |\mc{C}|$
    \item For each triple of objects $A, B, C \in |\mc{C}|$, a central, natural isomorphism
    $\alpha_{A, B, C} : (A \otimes B) \otimes C \to A \otimes (B \otimes C)$, the \emph{associator}
    \item For each object $A$, central, natural isomorphisms $\lambda_A : A \otimes I \to A$ and
    $\rho_A : I \otimes A \to A$, the \emph{left} and \emph{right unitors}
  \end{itemize}
  satisfying the \emph{triangle} and \emph{pentagon identity}
  $$
  \alpha_{A, I, B} ; A \otimes \lambda_B = \rho_A \otimes B \qquad
  \alpha_{A \otimes B, C, D} ; \alpha_{A, B, C \otimes D}
  = \alpha_{A, B, C} \otimes D ; \alpha_{A, B \otimes C, D} ; A \otimes \alpha_{A, B, C}
  $$
  We say a premonoidal category is \emph{symmetric} if it is also equipped with a central, natural
  involution $\sigma_{A, B} : A \otimes B \to B \otimes A$, the \textit{symmetry}, satisfying the
  \emph{hexagon identity}
  $$
  \alpha_{A, B, C} ; \sigma_{A, B \otimes C} ; \alpha_{B, C, A}
  = \sigma_{A, B} \otimes C ; \alpha_{B, A, C} ; B \otimes \sigma_{A, C}
  $$
  We say a premonoidal category is \emph{monoidal} if every morphism is central.
\end{definition}
One important theorem about premonoidal categories is \emph{coherence}:
\begin{theorem}
  The subcategory $\mc{A}$ generated by associators, unitors, and their tensor products is an
  equivalence relation, i.e., for all $A, B : |\mc{A}|$, if
  $f, g : A \to B$ can be constructed using only identity, composition, associators, unitors, and
  their tensor products, then:
  \begin{enumerate}[label=(\alph*)]
    \item $f = g$
    \item $f, g$ are isomorphisms in $\mc{A}$
  \end{enumerate}
  \label{thm:monoidal-coherence}
\end{theorem}
In particular, as a syntactic convenience, we will often simply write ``$\alpha$" for the (unique)
morphism between objects $A$ and $B$, when it exists, satisfying the requirements of
Theorem~\ref{thm:monoidal-coherence}. For example, given $f : A \to B \otimes ((C \otimes I)
\otimes D)$ and $g : ((B \otimes (I \otimes C)) \otimes D) \to E$, we have
$$
f;\alpha;g := 
  f ; 
  B \otimes (\lambda_C \otimes D) ; 
  \alpha_{B, C, D}^{-1} ; 
  (B \otimes \rho_C^{-1}) \otimes D ;
  g 
$$
Just like for higher-order functional languages, we can interpret types $A$ as objects $\dnt{A} :
|\mc{C}|$. Similarly, we can interpret variable contexts $\Gamma$ by taking products of objects, as
follows:
$$
\boxed{\dnt{\Gamma} : |\mc{C}|} \qquad 
  \dnt{\cdot} = I \qquad \dnt{\Gamma, \bhyp{x}{A}} = \dnt{\Gamma} \otimes \dnt{A}
$$
We would like to interpret an expression-in-context $\hasty{\Gamma}{\epsilon}{a}{A}$ as a morphism in
$\mc{C}$ from $\dnt{\Gamma}$ to $\dnt{A}$. However, in standard SSA, it is possible for a variable
to be unused, or used multiple times. Our premonoidal structure, however, does not give us any way
to \emph{project} out of a product type, making it impossible to interpret expressions-in-context
like
$
\hasty{\bhyp{x}{A}, \bhyp{y}{B}}{}{x}{A}
$. 
In order to project out individual variables, we need Cartesian structure, but a premonoidal
category can only interpret \emph{linear} expressions, that is, those which
use every variable exactly once. However, it is too much to require the whole premonoidal category to be Cartesian, because that would validate the sliding rule, which is what we set out to avoid.

In the case of the Kleisli category over a CCC, the Cartesian structure of the CCC becomes the premonoidal structure of the Kleisli category. This lets us use the Cartesian structure to project out the variables, while the product in the Kleisli category still does not satisfy sliding. By analogy with this case, we can suppose that there is a subcategory $\mc{C}_\bot$ of the premonoidal category $\mc{C}$, which has the property that the premonoidal structure in $\mc{C}$ behaves like a Cartesian product in $\mc{C}_\bot$. 

This is called a \emph{Freyd category}, and is defined as follows:

\begin{definition}[Freyd category]
  A \emph{Freyd category} is a premonoidal category $\mc{C}$ equipped with a wide subcategory
  $\mc{C}_\bot \subseteq \mc{C}$ of \emph{pure} morphisms such that
  \begin{itemize}
    \item $\mc{C}_\bot$ contains all associators, unitors, and symmetries
    \item $I$ is a terminal object in $\mc{C}_\bot$
    \item For each $A, B$, $A \otimes B$ is a cartesian product of $A, B$ in $\mc{C}_\bot$
    \item For pure morphisms $f, g$, $f \otimes g = \langle \pi_l; f, \pi_r ; g  \rangle$
  \end{itemize}
  Alternatively, it is equivalent to require
  \begin{itemize}
    \item $\mc{C}_\bot$ contains all associators, unitors, and symmetries
    \item $I$ is a terminal object in $\mc{C}_\bot$
    \item For each $A$, there exists a pure morphism $\dmor{A} : A \to A \otimes A$ forming a
          comonoid with the terminal morphism $\tmor{A} : A \to I$, i.e: $\dmor{A} ; \tmor{A}
          \otimes A ; \rho = \dmor{A} ; A \otimes \tmor {A} ; \lambda = \ms{id}_A$
    \item For every pure morphism $f : A \to_\bot B$, $f ; \dmor{B} = \dmor{A} ; f \otimes f$
          and $f ; !_B = !_A$.  
  \end{itemize}
  In both cases, we have that $\pi_l = A \otimes \tmor{B} ; \lambda$, $\pi_r = \tmor{A}
  \otimes B ; \rho$, $\langle f, g \rangle = \dmor{A} ; f \otimes g$, and $\dmor{A} = \langle
  \ms{id}_A, \ms{id}_A \rangle$
\end{definition}
For convenience, given $f : A \to B$ in a Freyd category, we will define the notation
\begin{equation}
  \lmor{f} := \dmor{A} ; A \otimes f : A \to A \otimes B
\end{equation}
Note that this is pure if and only if $f$ is. This has the following useful properties:
\begin{itemize}
  \item $\lmor{f} ; \pi_r = f$ and, for $f$ pure, $\lmor{f} ; \pi_l = \ms{id}$
  \item $\lmor{f;g} = \lmor{f} ; A \otimes g$ and $\lmor{\ms{id}_A} = \dmor{A}$
  \item $\lmor{\lmor{f}} = \lmor{f} ; \dmor{A} \otimes B ; \alpha_{A, A, B}$, and, therefore, given
  $g: A \otimes B \to C$, we have $\lmor{\lmor{f};g} = \lmor{f} ; \lmor{g} ; \pi_l \otimes C$
\end{itemize}

We now have everything we need to model effectful first-order expressions. For reasoning about
substitution, we will also demand that the denotation of ``\emph{pure}" expressions
$\hasty{\Gamma}{\bot}{a}{A}$ lies in $\mc{C}_\bot(\dnt{\Gamma}, \dnt{A})$; in general, we will write
morphisms in $\mc{C}_\bot$ as $A \to_\bot B$.

At this point, we still have no way to interpret control-flow, i.e. $\ms{case}$-expressions.
Furthermore, if we want to model regions as morphisms, we need some way of modelling label-contexts
$\ms{L}$. At first glance, it seems sufficient for branching control-flow to require the existence
of coproducts, and indeed, assuming the existence of all coproducts and an initial object, we may
model label-contexts as follows:
$$
\boxed{\dnt{\ms{L}} : |\mc{C}|} \qquad 
  \dnt{\cdot} = \mb{0} \qquad \dnt{\ms{L}, \lhyp{\ell}{A}} = \dnt{\ms{L}} + \dnt{A}
$$
Regions can now be interpreted as morphisms in $\mc{C}$ from $\dnt{\Gamma}$ to $\dnt{\ms{L}}$, as
desired. Just like for products, we will write ``$\alpha^+$'' for the (unique) morphism between
objects $A$ and $B$, when it exists, satisfying the requirements of
Theorem~\ref{thm:monoidal-coherence} where coproducts are taken as the monoidal structure; we will
sometimes also write $\alpha^+_B$ for clarity.

It turns out that our coproducts must be \emph{distributive} to allow us to use variables in scope
before a branch. In particular, we define a \emph{distributive} premonoidal category as follows:
\begin{definition}[Distributive category]
  A premonoidal category $\mc{C}$ with all coproducts is \emph{distributive} if, for all $A, B, C$,
  the obvious morphism
  $$
  \delta_{A, B, C} 
    = [(A + \iota_l), (A + \iota_r)] : (A \otimes B) + (A \otimes C) \to A \otimes (B + C)
  $$
  has an inverse $\delta^{-1}$. 
  We will say a Freyd category $\mc{C}$ is distributive if it has all coproducts and the subcategory
  of pure morphisms $\mc{C}_\bot$ is distributive (which implies, in particular, that $\mc{C}$ is
  distributive when taken as a premonoidal category).
\end{definition}
We note in particular that every cartesian closed category with coproducts is a distributive Freyd
category, as is the Kleisli category of a monad over a CCC with coproducts. 
For any finite coproduct $\Sigma_iB_i$, we will introduce the notation $\delta_{\Sigma} : \Sigma_i
(A \otimes B_i) \to A \otimes \Sigma_i B_i$ and $\delta_{\Sigma}^{-1} : A \otimes \Sigma_i B_i \to
\Sigma_i (A \otimes B_i)$ to denote the obvious morphisms.
The last piece of the puzzle, allowing us to model loops in regions, is the existence of an
\emph{iteration operator}, since we currently cannot represent any nonterminating computations. We
would further like this operator to be \emph{strong}; that is, compatible with the distributor, so
that we can soundly access variables in scope in the loop body. Formally, we define a \emph{(strong)
Conway iteration operator} as follows:
\begin{definition}[Conway iteration]
  A category $\mc{C}$ with all coproducts is said to have a \emph{iteration operator} if we can
  define an operator $(-)^\dagger$ taking every morphism $f : A \to B + A$ to a morphism $f^\dagger
  : A \to B$, the \emph{fixpoint} of $f$, with the following property: given $f : A \to B + A$, we
  have $f^\dagger = f;[\ms{id}, f^\dagger]$. We say this operator is a \emph{Conway iteration
  operator} if it additionally satisfies the following properties:
  \begin{itemize}
    \item \emph{Naturality:} given $f : A \to B + A$ and $g : B \to C$, we have
      $
      (f;g + \ms{id})^\dagger = f^\dagger;g : A \to C
      $
    \item \emph{Dinaturality:} given morphisms $g : A \to B + C$ and $h : C \to B + A$, we have that
      $
      (g ; [\iota_l, h])^\dagger = g ; [\ms{id}_B, (h ; [\iota_l, g])^\dagger]
      $
    \item \emph{Codiagonal:} given $f : A \to (B + A) + A$, we have
      $
      (f^\dagger)^\dagger = (f;[\ms{id}, \iota_r])^\dagger : A \to B
      $
  \end{itemize}
  If $\mc{C}$ is distributive, we say this operator is \emph{strong} if
  $$
  \forall f: A \to B + A, (C \otimes f ; \delta^{-1})^\dagger = C \otimes f^\dagger
  $$
  Given a wide subcategory $\mc{K} \subseteq \mc{C}$, we say $\mc{C}$ is \emph{\mc{K}-uniform} if,
  for all $h : A \to_{\mc{K}} B$, $f : B \to C + B$, and $g : A \to C + A$, we have that
  $$
  h ; f = g ; C + h \implies h ; f^\dagger = g^\dagger 
  $$
\end{definition}
We will call Freyd categories equipped with an appropriate Elgot structure \emph{strong Elgot
categories}. In particular, we define
\begin{definition}[(Strong) Elgot structure]
  A Freyd category $\mc{C}$ with all coproducts is said to have an \emph{Elgot structure} if it has
  a Conway iteration operator which is $\mc{C}_\bot$ uniform. If $\mc{C}$ is distributive, we say
  $\mc{C}$ is strong Elgot if its iteration operator is strong. In particular, we say a monad is
  Elgot if its Kliesli category has an Elgot structure. Similarly, we say a \emph{strong} monad is
  strong Elgot if its Kliesli category has a strong Elgot structure.
\end{definition}
It turns out that, to check something is a strong Elgot category, we do not need to
explicitly verify dinaturality. In particularly, we may verify the following:
\begin{proposition}
  If $(-)^\dagger$ is an iteration operator which satisfies naturality and codiagonal and is
  $\mc{K}$-uniform for $\mc{K}$ co-Cartesian, then it also satisfies dinaturality.
\end{proposition}
\begin{proof}
  See Lemma 31 of \citet{goncharov-18-guarded-traced}
\end{proof}
Since in a distributive Freyd category $\mc{C}_\bot$ must be distributive (and hence co-cartesian),
dinaturality follows from the other axioms of an Elgot category.
Given a distributive Freyd category equipped with a (strong) Elgot structure, we will often want to
consider the fixpoint of a morphism $f: R \otimes A \to B + A$, where our ``context'' $R$ does not
change between iterations. To do this, we first need to build up a morphism
\begin{equation}
  \rcase{f} := \lmor{f} ; \pi_l \otimes (B + A) ; \delta^{-1} 
    : R \otimes A \to R \otimes B + R \otimes A
\end{equation}
which computes $f$ and then distributes a copy of the read-only state $R$ to each branch of the
result. We may then define the fixpoint
\begin{equation}
  \rfix{f} := (\rcase{f})^\dagger ; \pi_r : R \otimes A \to B
\end{equation}
We consider some more properties of $\rcase{f}$ and $\rfix{f}$ in Appendix~\ref{apx:environment}.

\subsection{String Diagrams}

\emph{String diagrams} provide a graphical calculus for reasoning about (symmetric) monoidal
categories, which allows us to succinctly express complex morphisms and rewrites. Since both
cartesian and co-cartesian categories are monoidal (with the product and coproduct as tensor,
respectively), we can use string diagrams to reason about both. In the co-cartesian case, string
diagrams behave much like control-flow diagrams, with boxes representing sub-programs, input wires
entry points, and output wires exit points. In particular, a \emph{region} is just a box with a
single input wire. Continuing this analogy, we draw the codiagonal morphism $[\ms{id}_A, \ms{id}_A]$
as joining two wires, and the zero morphism as a wire coming from nowhere, as in
Figure~\ref{fig:coproduct-string-diagrams}. 

\begin{figure}
  \begin{tikzpicture}
    \node[] (S1) at (-0.6, 1) {};
    \node[] (S2) at (0.6, 1) {};
    \node[box=2/0/2/0] (A) at (0, 0) {\text{subprogram 1}};
    \node[box=1/0/2/0] (B) at (-1, -1.5) {\text{subprogram 2}};
    \node[box=2/0/2/0] (C) at (2, -1.5) {\text{subprogram 3}};
    \node[box=3/0/1/0] (D) at (1.5, -3) {\text{subprogram 4}};
    \node[dot] (zero) at ($(C.north.2)+(0, 0.5)$) {};
    \node[dot] (codiag) at (0, -4) {};
    \node[] (end) at (0, -4.5) {};
    \wires{
      S1 = { south = A.north.1 },
      S2 = { south = A.north.2 },
      A = { south.1 = B.north.1, south.2 = C.north.1 },
      B = { south.1 = codiag.west, south.2 = D.north.1 },
      C = { south.1 = D.north.2, south.2 = D.north.3 },
      D = { south = codiag.east },
      zero = { south = C.north.2 },
      codiag = { south = end.north }
    }{}

    \node[] (zerolabel) at (4, 0) {Zero morphism};
    \node[] (regionlabel) at (-4, -1) {Region};
    \node[] (codiaglabel) at (-3, -4) {Codiagonal morphism};

    \draw[gray, dashed, ->] (zerolabel) to (zero);
    \draw[gray, dashed, ->] (regionlabel) to (B);
    \draw[gray, dashed, ->] (codiaglabel) to (codiag);
  \end{tikzpicture}
  \caption{
    A string diagram using the coproduct as symmetric monoidal structure, interpreted as a CFG
  }
  \label{fig:coproduct-string-diagrams}
  \Description{A string diagram using the coproduct as symmetric monoidal structure, showing how to
  draw the codiagonal morphism and zero morphism, as well as showing what the analogue to a
  subprogram and region might look like.}
\end{figure}

The power of string diagrams comes from the fact that many syntactically distinct ways to write
equal values are obviously graphically equivalent by \emph{isotopy}: essentially, moving boxes and
wires around. String diagrams also give us an elegant way to represent and reason about Elgot
structures. It turns out that Elgot structures induce a \emph{trace} on the coproduct
\cite{hasegawa-trace-02}: given $f : A + C \to B + C$, we can define 
\begin{equation}
  \ms{Tr}_{A, B}^C(f) 
  = \iota_l ; [f ; B + \iota_r]^\dagger 
  = \iota_l ; f ; [\ms{id}, (\iota_l;f)^\dagger] : A \to B
\end{equation}
Since this satisfies the axioms of a trace over a symmetric monoidal category, we can draw it, and
therefore the Elgot operator, as in Figure~\ref{fig:elgot-string-diagrams}. Continuing with the
control-flow diagram analogy, such traces can be interpreted as \emph{loops}, with the Elgot axioms,
now drawn as diagrams in Figure~\ref{fig:elgot-ax-string-diagrams}.

\begin{figure}
  \begin{subfigure}{0.5\textwidth}
    \centering
    \begin{tikzpicture}
      \node[] (Atrf) at (0, 2) {$A$};
      \node[box] (trf) at (0, 0) {\ms{Tr}_{A, B}^C(f)};
      \node[] (Btrf) at (0, -2) {$B$};
      \node[] (eq) at (2, 0) {=};
      \node[] (A) at (3.5, 2) {$A$};
      \node[box=2/0/2/0] (f) at (3.5, 0) {\quad f \quad};
      \node[] (B) at (3.5, -2) {$B$};
      \coordinate[label=right:$C$] (C) at (4.7, 0) {};
      \coordinate[] (cup) at (4.25, -1) {};
      \coordinate[] (cap) at (4.25, 1) {};
      \wires{
        Atrf = { south = trf.north },
        trf  = { south = Btrf },
        A    = { south = f.north.1 },
        f    = { south.1 = B, south.2 = cup.west },
        cap  = { west = f.north.2, east = C.north },
        C    = { south = cup.east },
      }{}
    \end{tikzpicture}
    \caption{The trace of $f : A + C \to B + C$}
  \end{subfigure}%
  \begin{subfigure}{0.5\textwidth}
    \centering
    \begin{tikzpicture}
      \node[] (Atrf) at (0, 2) {$A$};
      \node[box] (trf) at (0, 0) {f^\dagger};
      \node[] (Btrf) at (0, -2) {$B$};
      \node[] (eq) at (1.5, 0) {=};
      \node[] (A) at (3, 2) {$A$};
      \node[box=1/0/2/0] (f) at (3, 0) {\quad f \quad};
      \node[] (B) at (3, -2) {$B$};
      \node[dot] (codiag) at (3, 0.75) {};
      \coordinate[label=right:$A$] (C) at (4.2, 0) {};
      \coordinate[] (cup) at (3.75, -1) {};
      \coordinate[] (cap) at (3.75, 1) {};
      \wires{
        Atrf    = { south = trf.north },
        trf     = { south = Btrf },
        A       = { south = codiag.north },
        f       = { south.1 = B, south.2 = cup.west },
        codiag  = { south = f.north, east = cap.west },
        C       = { south = cup.east, north = cap.east },
      }{}
    \end{tikzpicture}
    \caption{
      The fixpoint of $f : A \to B + A$
    }
  \end{subfigure}
  \caption{Representations of the coproduct trace and Elgot structure as string diagrams}
  \label{fig:elgot-string-diagrams}
  \Description{Representations of the coproduct trace and Elgot structure as string diagrams}
\end{figure}

\begin{figure}
  \begin{subfigure}{0.5\textwidth}
    \centering
    \begin{tikzpicture}
      \node[] (A) at (0, 2) {$A$};
      \node[box=1/0/2/0] (f) at (0, 0) {\quad f \quad};
      \node[] (B) at (0, -2) {$B$};
      \node[dot] (codiag) at (0, 0.75) {};
      \coordinate[] (C) at (1.2, 0) {};
      \coordinate[] (cup) at (0.75, -1) {};
      \coordinate[] (cap) at (0.75, 1) {};
      \node[] (eq) at (1.8, 0) {=};
      \node[] (Ad) at (3.6, 2) {$A$};
      \node[box=1/0/2/0] (fd1) at (3.4, 1) {\quad f \quad};
      \node[dot] (codiagd1) at (3.8, 0.25) {};
      \node[box=1/0/2/0] (fd2) at (3.8, -0.5) {\quad f \quad};
      \node[dot] (codiagd2) at (3.4, -1.25) {};
      \node[] (Bd) at (3.6, -2) {$B$};
      \coordinate[] (cupd) at (4.5, -1.25) {};
      \coordinate[] (capd) at (4.5, 0.5) {};
      \wires{
        A         = { south = codiag.north },
        f         = { south.1 = B, south.2 = cup.west },
        codiag    = { south = f.north, east = cap.west },
        C         = { south = cup.east, north = cap.east },
        Ad        = { south = fd1.north },
        fd1       = { south.1 = codiagd2.west, south.2 = codiagd1.north },
        codiagd1  = { south = fd2.north },
        fd2       = { south.1 = codiagd2.north, south.2 = cupd.west },
        codiagd2  = { south = Bd.north },
        cupd      = { east = capd.east },
        capd      = { west = codiagd1.east }, 
      }{}
    \end{tikzpicture}
    \caption{Fixpoint}
  \end{subfigure}%
  \begin{subfigure}{0.5\textwidth}
    \centering
    \begin{tikzpicture}
      \node[] (A) at (0, 2) {$A$};
      \node[box=1/0/2/0] (f) at (0, 0.5) {\quad f \quad};
      \node[box] (g) at (-0.25, -0.5) {g};
      \node[] (B) at (0, -2) {$B$};
      \node[dot] (codiag) at (0, 1.25) {};
      \coordinate[] (C) at (1.2, 0) {};
      \coordinate[] (cup) at (0.75, -1) {};
      \coordinate[] (cap) at (0.75, 1.5) {};
      \coordinate[] (box0) at (1.5, 1.65) {};
      \coordinate[label={[font=\small, text=gray]below:$(f;g + A)^\dagger$}] 
                    (box1) at (1.5, -1.2) {};
      \coordinate[] (box2) at (-1, -1.2) {};
      \coordinate[] (box3) at (-1, 1.65) {};
      \draw [gray, dashed] (box0) -- (box1);
      \draw [gray, dashed] (box1) -- (box2);
      \draw [gray, dashed] (box2) -- (box3);
      \draw [gray, dashed] (box3) -- (box0);
      \node[] (eq) at (1.8, 0) {=};
      \node[] (A2) at (3.6, 2) {$A$};
      \node[box=1/0/2/0] (f2) at (3.6, 0.5) {\quad f \quad};
      \node[box] (g2) at (3.35, -1) {g};
      \node[] (B2) at (3.6, -2) {$B$};
      \node[dot] (codiag2) at (3.6, 1) {};
      \coordinate[] (cup2) at (4.15, -0.2) {};
      \coordinate[] (cap2) at (4.15, 1.25) {};
      \coordinate[] (box02) at (4.8, 1.5) {};
      \coordinate[label={[font=\small, text=gray]below:$f^\dagger$}] 
                    (box12) at (4.8, -0.4) {};
      \coordinate[] (box22) at (2.6, -0.4) {};
      \coordinate[] (box32) at (2.6, 1.5) {};
      \draw [gray, dashed] (box02) -- (box12);
      \draw [gray, dashed] (box12) -- (box22);
      \draw [gray, dashed] (box22) -- (box32);
      \draw [gray, dashed] (box32) -- (box02);
      \wires{
        A         = { south = codiag.north },
        f         = { south.1 = g.north, south.2 = cup.west },
        g         = { south = B.north },
        codiag    = { south = f.north, east = cap.west },
        C         = { south = cup.east, north = cap.east },
        A2         = { south = codiag2.north },
        f2         = { south.1 = g2.north, south.2 = cup2.west },
        g2         = { south = B2.north },
        codiag2    = { south = f2.north, east = cap2.west },
        cup2       = { east = cap2.east },
        cap2       = { west = codiag2.east },
      }{}
    \end{tikzpicture}
    \caption{Naturality}
  \end{subfigure}
  \begin{subfigure}{0.5\textwidth}
    \centering
    \begin{tikzpicture}
      \node[] (A) at (0, 2) {$A$};
      \node[box=1/0/3/0] (f) at (0, 0) {\quad f \quad};
      \node[] (B) at (0, -2) {$B$};
      \node[dot] (codiag) at (0, 0.75) {};
      \coordinate[] (C) at (1.2, 0) {};
      \coordinate[] (cup) at (0.75, -1) {};
      \coordinate[] (cap) at (0.75, 1) {};
      \node[dot] (codiag2) at (0.225, -0.75) {};
      \node[] (eq) at (1.8, 0) {=};
      \node[] (Ad) at (3.6, 2) {$A$};
      \node[box=1/0/3/0] (fd) at (3.6, 0) {\quad f \quad};
      \node[] (Bd) at (3.6, -2) {$B$};
      \node[dot] (codiagd) at (3.6, 1.25) {};
      \node[dot] (codiagd2) at (3.6, 0.75) {};
      \coordinate[] (Cd) at (4.8, 0) {};
      \coordinate[] (cupd) at (4.35, -1.25) {};
      \coordinate[] (capd) at (4.35, 1.25) {};
      \coordinate[] (cupd2) at (4.25, -0.75) {};
      \coordinate[] (capd2) at (4.25, 0.75) {};
      \wires{
        A         = { south = codiag.north },
        f         = { south.1 = B, south.2 = codiag2.west, south.3 = codiag2.east },
        codiag2   = { south = cup.west },
        codiag    = { south = f.north, east = cap.west },
        C         = { south = cup.east, north = cap.east },
        Ad        = { south = codiagd.north },
        codiagd   = { south = codiagd2.north, east = capd.west },
        codiagd2  = { south = fd.north, east = capd2.west },
        fd        = { south.1 = Bd, south.2 = cupd.west, south.3 = cupd2.west },
        cupd      = { east = capd.east },
        cupd2     = { east = capd2.east },
      }{}
    \end{tikzpicture}
    \caption{Codiagonal}
  \end{subfigure}%
  \begin{subfigure}{0.5\textwidth}
    \centering
    \begin{tikzpicture}
      \node[] (A) at (0, 2) {$A$};
      \node[box] (g) at (0, 1.25) {g};
      \node[box=1/0/2/0] (f) at (0, 0.25) {\quad f \quad};
      \node[] (B) at (0, -2) {$B$};
      \node[dot] (codiag) at (0, 0.75) {};
      \coordinate[] (C) at (1.2, 0) {};
      \coordinate[] (cup) at (0.75, -1.5) {};
      \coordinate[] (cap) at (0.75, 1) {};
      \node[box] (g2) at (0.5, -0.75) {g};
      \node[] (eq) at (1.8, 0) {=};
      \node[] (Ad) at (3.6, 2) {$A$};
      \node[box] (gd) at (3.6, 0.5) {g};
      \node[box=1/0/2/0] (fd) at (3.6, -0.5) {\quad f \quad};
      \node[] (Bd) at (3.6, -2) {$B$};
      \node[dot] (codiagd) at (3.6, 1) {};
      \coordinate[] (Cd) at (4.8, 0) {};
      \coordinate[] (cupd) at (4.35, -1.5) {};
      \coordinate[] (capd) at (4.35, 1.5) {};
      \wires{
        A         = { south = g.north },
        g         = { south = codiag.north },
        f         = { south.1 = B, south.2 = g2.north },
        g2        = { south = cup.west },
        codiag    = { south = f.north, east = cap.west },
        C         = { south = cup.east, north = cap.east },
        Ad        = { south = codiagd.north },
        codiagd   = { south = gd.north, east = capd.west },
        gd        = { south = fd.north },
        fd        = { south.1 = Bd, south.2 = cupd.west },
        cupd      = { east = capd.east },
      }{}
    \end{tikzpicture}
    \caption{Dinaturality}
  \end{subfigure}
  \caption{Representations of the Elgot axioms as string diagrams}
  \label{fig:elgot-ax-string-diagrams}
  \Description{Representations of the Elgot axioms as string diagrams}
\end{figure}

Unfortunately, unmodified string diagrams do not work for premonoidal categories, and hence for
Freyd categories. The reason is because, since not all morphisms are central, premonoidal categories
do not in general validate \emph{sliding}. However, this is easy enough to fix: we can postulate a
``state'' wire, drawn in red, which all impure morphisms require as an input and output, as in
Figure~\ref{fig:premonoidal-string-diagram}. Since the state wire linearly threads through all
impure boxes, it establishes a unique order in which they must be executed; this construction is
shown to be sound in \citet{promonad}. Pure morphisms do not have a state wire, so a diagram
representing a pure morphism will simply have a red ``stripe'' on the side. This gives us a
convenient way to distinguish between string diagrams using the monoidal structure induced by the
coproduct and those using the premonoidal structure induced by the tensor product in a category
having both (such as a distributive premonoidal category): the latter will have a state wire, while
the former will not.

\begin{figure}
  \begin{tikzpicture}
    \node[] (SIa) at (-2.5, 3.5) {};
    \node[] (Na) at (-1.5, 3.5) {$\nats$};
    \node[dot] (codiaga) at (-1.5, 3) {};
    \node[box] (adda) at (0.5, 1) {\cdot + 2};
    \node[box=2/0/1/0] (print1a) at (-2, 2) {\texttt{print}};
    \node[box=2/0/1/0] (print2a) at (0, 0) {\texttt{print}};
    \node[] (SOa) at (0, -1) {};
    \node[] (neq) at (2, 1.7) {$\neq$};
    \node[] (SIb) at (3, 3.5) {};
    \node[] (Nb) at (4, 3.5) {$\nats$};
    \node[dot] (codiagb) at (4, 3) {};
    \node[box=2/0/1/0] (print1b) at (5, 1) {\texttt{print}};
    \node[box] (addb) at (5.5, 2) {\cdot + 2};
    \node[box=2/0/1/0] (print2b) at (3, 0) {\texttt{print}};
    \node[] (SOb) at (3, -1) {};
    \wires{
      Na = { south = codiaga.north },
      codiaga = { south = print1a.north.2, east = adda.north },
      adda = { south = print2a.north.2 },
      Nb = { south = codiagb.north },
      codiagb = { south = print2b.north.2, east = addb.north },
      addb = { south = print1b.north.2 },
    }{}
    \wires[red]{
      SIa = { south = print1a.north.1 },
      print1a = { south = print2a.north.1 },
      print2a = { south = SOa.north },
      SIb = { south = print1b.north.1 },
      print1b = { south = print2b.north.1 },
      print2b = { south = SOb.north },
    }{}
  \end{tikzpicture}
  \caption{
    A string diagram in a premonoidal category, demonstrating the necessity of using a state wire
  }
  \label{fig:premonoidal-string-diagram}
  \Description{A string diagram in a premonoidal category, using a red state wire}
\end{figure}

\subsection{Semantics}

We now have all the ingredients we need to give a semantics to \isotopessa{} expressions and regions.
In particular, an \isotopessa{} expression model consists of:
\begin{itemize}
  \item A premonoidal category $\mc{C}$
  \item A complete lattice $E$, and a continuous function $\epsilon \mapsto \mc{C}_\epsilon$ from
  $E$ to wide subcategories of $\mc{C}$, such that $\mc{C}_\top = \mc{C}$, and $(\mc{C},
  \mc{C}_\bot)$ is a distributive Freyd category. Furthermore, we ask that each $\mc{C}_\epsilon$ is
  closed under tensor products and coproducts, and in particular that injections and (inverse)
  distributors are pure and therefore contained in every $\mc{C}_\epsilon$.
\end{itemize}
If an \isotopessa{} expression model additionally has an Elgot structure on $\mc{C}$ (taken as a
Freyd category with pure morphisms $\mc{C}_\bot$), we will refer to it simply as an \isotopessa{}
model.

We will write morphisms $\mc{C}_\epsilon(A, B)$ as $A \to_\epsilon B$, and morphisms $\mc{C}(A, B)$
as $A \to B$. Note that, whenever $\mc{C}$ is a distributive Elgot Freyd category, we can trivially
obtain an \isotopessa{} model by choosing $E = \{\top, \bot\}$; hence, this re-framing does not add
any additional generality, but simplifies bookkeeping. We may also choose $E = \{*\}$ if and only if
$\mc{C}$ is in fact Cartesian.
Fixing an \isotopessa{} expression model, we can model types in the obvious manner: 
\begin{itemize}
  \item The unit type $\mb{1}$ is modelled as the monoidal unit $I$
  \item The empty type $\mb{0}$ is modelled as the initial object $\mb{0}$
  \item Products $A \otimes B$ are modelled as tensor products $\dnt{A} \otimes \dnt{B}$
  \item Sum types $A + B$ are modelled as coproducts $\dnt{A} + \dnt{B}$
\end{itemize}
We'll say an \isotopessa{} model \emph{interprets} a signature $(\mc{T}, \mc{I})$ if we have
\begin{itemize}
  \item For every base type $X \in \mc{T}$, an object $\dnt{X} : |\mc{C}|$
  \item For every primitive instruction $f \in \mc{I}_\epsilon(A, B)$, a morphism $\dnt{f} : \dnt{A}
  \to_\epsilon \dnt{B}$.
\end{itemize}
This gives us all the ingredients to interpret \emph{contexts} and \emph{label contexts}, which, as
in Figure~\ref{fig:ssa-ty-sem}, we will interpret as tensor products and coproducts of the
denotations of their parameters, respectively.

\begin{figure}
  \begin{equation*}
    \boxed{\dnt{A} : |\mc{C}|}
  \end{equation*}
  \begin{align*}
     \dnt{\mb{1}} = I \qquad
      \dnt{A \otimes B} = \dnt{A} \otimes \dnt{B} \qquad
      \dnt{\mb{0}} = \mb{0} \qquad
      \dnt{A + B} = \dnt{A} + \dnt{B} \\
  \end{align*}
  \begin{equation*}
    \boxed{\dnt{\Gamma} : |\mc{C}|}
  \end{equation*}
  \begin{align*}
    \dnt{\cdot} = I \qquad
      \dnt{\Gamma, \thyp{x}{A}{\epsilon}} = \dnt{\Gamma} \otimes \dnt{A} \\
  \end{align*}
  \begin{equation*}
    \boxed{\dnt{\ms{L}} : |\mc{C}|}
  \end{equation*}
  \begin{align*}
      \dnt{\cdot} = 0 \qquad
      \dnt{\ms{L}, \lhyp{\ell}{A}} = \dnt{\ms{L}} + \dnt{A} \\
  \end{align*}
  \begin{equation*}
    \boxed{\dnt{\Gamma \leq \Delta} : \dnt{\Gamma} \to_\bot \dnt{\Delta}}
  \end{equation*}
  \begin{gather*}
    \dnt{\cdot \leq \cdot} = \ms{id}_I \qquad
    \dnt{\Gamma, \thyp{x}{A}{\epsilon} \leq \Delta} = \pi_l;\dnt{\Gamma \leq \Delta} \qquad
    \dnt{\Gamma, \thyp{x}{A}{\epsilon} \leq \Delta, \thyp{x}{A}{\epsilon}}
    = \dnt{\Gamma \leq \Delta} \otimes \dnt{A} \\
  \end{gather*}
  \begin{equation*}
    \boxed{\dnt{\ms{L} \leq \ms{K}} : \dnt{\ms{L}} \to_\bot \dnt{\ms{K}}}
  \end{equation*}
  \begin{gather*}
    \dnt{\cdot \leq \cdot} = \ms{id}_{\mb{0}} \qquad
    \dnt{\ms{L} \leq \ms{K}, \lhyp{\ell}{A}} = \dnt{\ms{L} \leq \ms{K}};\iota_\ell \qquad
    \dnt{\ms{L}, \lhyp{\ell}{A} \leq \ms{K}, \lhyp{\ell}{A}}
    = \dnt{\ms{L} \leq \ms{K}} + \dnt{A}
  \end{gather*}
  \caption{Denotational semantics for \isotopessa{} types, contexts, and weakenings}
  \Description{}
  \label{fig:ssa-ty-sem}
\end{figure}

We can now interpret \isotopessa{} expressions $\hasty{\Gamma}{\epsilon}{a}{A}$ over a given
signature in a model interpreting that signature as morphisms $\dnt{\Gamma} \to_\epsilon \dnt{A}$
using the rules in Figure~\ref{fig:ssa-expr-sem}. Up to this point, both our syntax and semantics
are quite standard; in particular:
\begin{itemize}
  \item Variables $x$ are modelled as projections from the appropriate index in the context's
  denotation $\pi_{\Gamma, x} : \dnt{\Gamma} \to_\bot \dnt{A}$
  \item Applications of primitive instructions $f\;a$ are modelled as $\dnt{f}$ precomposed with the
  denotation of $a$. 
  \item Unary \ms{let}-bindings are modelled by:
  \begin{itemize}
    \item Duplicating the context using the diagonal morphism $\Delta_{\dnt{\Gamma}}$
    \item Passing the right copy of the context through $\dnt{\hasty{\Gamma}{\epsilon}{a}{A}}$ to
    get an input of type $\dnt{\Gamma} \otimes \dnt{A} = \dnt{\Gamma, \bhyp{x}{A}}$
    \item Passing the result of this through $\dnt{\hasty{\Gamma, \bhyp{x}{A}}{\epsilon}{b}{B}}$ to get
    the final result of type $\dnt{B}$
  \end{itemize}
  \item Pairs are modelled by passing the result of the diagonal morphism $\Delta_{\dnt{\Gamma}}$ to
   $
   \dnt{\hasty{\Gamma}{\epsilon}{a}{A}} \ltimes \dnt{\hasty{\Gamma}{\epsilon}{b}{B}}
   $
   i.e., first
   passing the left copy through $\dnt{\hasty{\Gamma}{\epsilon}{a}{A}}$ and then the right copy
   through $\dnt{\hasty{\Gamma}{\epsilon}{b}{B}}$. By the axioms of a Freyd category, for pure
   morphisms, this is the same as simply taking the product
   $
    \langle \dnt{\hasty{\Gamma}{\epsilon}{a}{A}}, \dnt{\hasty{\Gamma}{\epsilon}{b}{B}} \rangle
   $.
  \item Binary \ms{let}-bindings are modelled similarly to unary \ms{let}-bindings, except that after
  passing the right copy of the context through $\dnt{\hasty{\Gamma}{\epsilon}{e}{A \otimes B}}$, we
  re-associate $\dnt{\Gamma} \otimes (\dnt{A} \otimes \dnt{B})$ to $(\dnt{\Gamma} \otimes \dnt{A})
  \otimes \dnt{B} = \dnt{\Gamma, \bhyp{x}{A}, \bhyp{y}{B}}$ before passing the entire result through
  $\dnt{\hasty{\Gamma, \bhyp{x}{A}, \bhyp{y}{B}}{\epsilon}{c}{C}}$.
  \item The unit value $()$ is modelled as the terminal morphism $1_{\dnt{\Gamma}}$, while
  $\ms{abort}\;a$ is modelled as the denotation of $a$ postcomposed with the zero morphism 
  $0_{\dnt{A}}$. Injections, similarly, are simply modelled as the appropriate coproduct injections.
  \item A \ms{case}-expression is modelled by
  \begin{itemize}
    \item Duplicating the context using the diagonal morphism $\Delta_{\dnt{\Gamma}}$
    \item Using the right copy of the context to compute the discriminant using
          $\dnt{\hasty{\Gamma}{\epsilon}{e}{A + B}}$
    \item Applying the inverse distributor $\delta^{-1}$ to obtain a coproduct
          $\dnt{\Gamma} \otimes \dnt{A} + \dnt{\Gamma} \otimes \dnt{B}$
    \item Computing $\dnt{\hasty{\Gamma, x: A}{\epsilon}{a}{C}}$ on the right
          branch and $\dnt{\hasty{\Gamma, y: B}{\epsilon}{b}{C}}$ on the left branch
  \end{itemize}
\end{itemize}

\begin{figure}
  \begin{equation*}
    \boxed{\dnt{\hasty{\Gamma}{\epsilon}{a}{A}} : \dnt{\Gamma} \to_\epsilon \dnt{A}}
  \end{equation*}
  \begin{align*}
    \dnt{\hasty{\Gamma}{\epsilon}{x}{A}} &= \pi_{\Gamma, x} \\
    \dnt{\hasty{\Gamma}{\epsilon}{f\;a}{B}} 
      &= \dnt{f} \circ \dnt{\hasty{\Gamma}{\epsilon}{a}{A}} \\
    \dnt{\hasty{\Gamma}{\epsilon}{\letexpr{x}{a}{b}}{B}}
      &= \lmor{\dnt{\hasty{\Gamma}{\epsilon}{a}{A}}} 
      ; \dnt{\hasty{\Gamma, \bhyp{x}{A}}{\epsilon}{b}{B}}
      \\
    \dnt{\hasty{\Gamma}{\epsilon}{(a, b)}{A \otimes B}} 
      &= \Delta_{\dnt{\Gamma}}
      ; \dnt{\hasty{\Gamma}{\epsilon}{a}{A}} \ltimes \dnt{\hasty{\Gamma}{\epsilon}{b}{B}}
      % ; \dnt{\hasty{\Gamma}{\epsilon}{a}{A}} \otimes \dnt{\Gamma} 
      % ; \dnt{A} \otimes \dnt{\hasty{\Gamma}{\epsilon}{b}{B}}
      \\
    \dnt{\hasty{\Gamma}{\epsilon}{\letexpr{(x, y)}{e}{c}}{C}}
      &= \lmor{\dnt{\hasty{\Gamma}{\epsilon}{e}{A \otimes B}}} ; \alpha
      ; \dnt{\hasty{\Gamma, \bhyp{x}{A}, \bhyp{y}{B}}{\epsilon}{c}{C}}
      \\
    \dnt{\hasty{\Gamma}{\epsilon}{()}{\mb{1}}} &= 1_{\dnt{\Gamma}} \\
    \dnt{\hasty{\Gamma}{\epsilon}{\linl{a}}{A + B}}
      &= \dnt{\hasty{\Gamma}{\epsilon}{a}{A}} ; \iota_l \\
    \dnt{\hasty{\Gamma}{\epsilon}{\linr{b}}{A + B}}
      &= \dnt{\hasty{\Gamma}{\epsilon}{b}{B}} ; \iota_r \\
    \dnt{\hasty{\Gamma}{\epsilon}{\caseexpr{e}{x}{a}{y}{b}}{C}}
      &= \lmor{\dnt{\hasty{\Gamma}{\epsilon}{e}{A + B}}}
        ; \delta^{-1}_{\dnt{\Gamma}} ; \\& \qquad
        [
          \dnt{\hasty{\Gamma, \bhyp{x}{A}}{\epsilon}{a}{C}}, 
          \dnt{\hasty{\Gamma, \bhyp{y}{B}}{\epsilon}{b}{C}}
        ]
      \\
    \dnt{\hasty{\Gamma}{\epsilon}{\labort{a}}{A}} 
      &= \dnt{\hasty{\Gamma}{\epsilon}{a}{\mb{0}}} ; 0_{\dnt{A}}
  \end{align*}
  \begin{equation*}
    \text{where} \quad \boxed{\pi_{\Gamma, x} : \dnt{\Gamma} \to_\bot \dnt{A}} \qquad
    \pi_{(\Gamma, x : A), x} = \pi_r \qquad
    \pi_{(\Gamma, y : B), x} = \pi_l ; \pi_{\Gamma, x}
  \end{equation*}
  \caption{Denotational semantics for \isotopessa{} expressions}
  \Description{Denotational semantics for isotope-SSA expressions}
  \label{fig:ssa-expr-sem}
\end{figure}

Similarly, if we in fact have an \isotopessa{} model, we can interpret \isotopessa{} regions
$\haslb{\Gamma}{r}{\ms{L}}$ as morphisms $\dnt{\Gamma} \to \dnt{\ms{L}}$; note that we don't assume
anything about the effect of these morphisms. As $\ms{L}$ is a coproduct, we can view the result
object of a region $r$ as encoding both \emph{data} and \emph{control-flow} information. In
particular, we interpret a branch $\brb{\ell}{a}$ as simply the injection of the (pure) expression
$\hasty{\Gamma}{\bot}{a}{A}$, our \emph{data}, into the element of the coproduct corresponding to
$\ell$, which encodes the point in control-flow the rest of the program should jump to next. This is
in contrast to \emph{expressions}, which purely encode data, with no particular instructions on how
to use it afterwards. Our interpretation of \ms{let}-statements and \ms{case}-statements, given in
Figure~\ref{fig:ssa-reg-sem}, is exactly the same as that of the corresponding expressions.

Finally, we come to the interpretation of \ms{where}-statements, which is where the Elgot structure
comes in. The semantics of a \ms{where}-block
$\haslb{\Gamma}{\where{r}{(\wbranch{\ell_i}{x_i}{t_i},)_i}}{\ms{L}}$ can be broken down into two
major components:
\begin{itemize}
  \item The \emph{terminator} 
  $\entrymor{\Gamma}{r}{\ms{L}} : \dnt{\Gamma} \to \dnt{\ms{L}} + \Sigma_i\dnt{A_i}$,
  which, given as input the context $\dnt{\Gamma}$, executes $r$ and then re-associates the output.
  The output type $\dnt{\ms{L}} + \Sigma_i\dnt{A_i}$ expresses that $r$ may either:
  \begin{itemize}
    \item Via the left injection, return immediately, jumping to an enclosing label in $\ms{L}$
    \item Via the right injection, jump to the $i^{th}$ basic block $t_i$ in the
    \ms{where}-statement by returning a value in $\dnt{A_i}$
  \end{itemize}
  \item The ``\emph{loop}'' $\loopmor{\Gamma}{(\wbranch{\ell_i}{x_i}{t_i},)_i}{\ms{L}} :
  \dnt{\Gamma} \otimes \Sigma_i\dnt{A_i} \to \dnt{\ms{L}} + \Sigma_i\dnt{A_i}$ which, given as input
  the context $\dnt{\Gamma}$ and an input $\dnt{A_i}$ for the $i^{th}$ basic block $t_i$, executes
  $t_i$ and then re-associates the output. The output type $\dnt{\ms{L}} + \Sigma_i\dnt{A_i}$ again
  expresses that control-flow may either exit the \ms{where}-block (via $\dnt{\ms{L}}$) or jump to
  some other basic block (via $\Sigma_i\dnt{A_i}$).
\end{itemize}
We glue these together in the obvious manner to get the semantics for a where-block:
\begin{itemize}
  \item Compute $\lmor{\entrymor{\Gamma}{r}{\ms{L}}}$, which, given as input the context
  $\dnt{\Gamma}$, copies the context, executes $\entrymor{r}$, and then returns the copied context
  and the output as $\dnt{\Gamma} \otimes (\dnt{\ms{L}} + \Sigma_i\dnt{A_i})$
  \item Distribute the context into each branch of the coproduct, yielding 
  $\dnt{\Gamma} \otimes \dnt{\ms{L}} + \dnt{\Gamma} \otimes \Sigma_i\dnt{A_i}$
  \item If we are in the left branch, project out the result to yield $\dnt{\ms{L}}$, and return
  immediately
  \item Otherwise, compute $\loopmor{\Gamma}{(\wbranch{\ell_i}{x_i}{t_i},)_i}{\ms{L}}$ in a loop,
  passing in a fresh copy of the context each time. This is implemented with $\ms{rfix}$.
\end{itemize}
This simplifies significantly in the case of a \ms{where}-block defining just a single label,
yielding
\begin{equation}
  \dnt{\haslb{\Gamma}{\where{r}{\wbranch{\ell}{x}{s}}}{\ms{L}}}
  = \lmor{\dnt{\haslb{\Gamma}{r}{\ms{L}, \ell(A)}}}
  ; \delta^{-1}
  ; [\pi_r, \rfix{\dnt{\haslb{\Gamma, \bhyp{x}{A}}{s}{\ms{L}, \ell(A)}}}]
\end{equation}
In particular, it is a consequence of label weakening (Lemma~\ref{lem:wk}) that, in the case where
$s$ does not call $\ell$, this simplifies further to
\begin{equation}
  \dnt{\haslb{\Gamma}{\where{r}{\wbranch{\ell}{x}{s}}}{\ms{L}}}
  = \lmor{\dnt{\haslb{\Gamma}{r}{\ms{L}, \ell(A)}}}
  ; \delta^{-1}
  ; [\pi_r, \dnt{\haslb{\Gamma, \bhyp{x}{A}}{s}{\ms{L}}}]
\end{equation}

\begin{figure}
  \begin{equation*}
    \boxed{\dnt{\haslb{\Gamma}{r}{\ms{L}}} : \dnt{\Gamma} \to \dnt{\ms{L}}}
  \end{equation*}
  \begin{align*}
    \dnt{\haslb{\Gamma}{\brb{\ell}{a}}{\ms{L}}} 
      &= \dnt{\hasty{\Gamma}{\bot}{a}{A}} ; \iota_{\ms{L}, \ell}
      \\
    \dnt{\haslb{\Gamma}{\letstmt{x}{a}{r}}{\ms{L}}}
      &= \lmor{\dnt{\hasty{\Gamma}{\epsilon}{a}{A}}}
      ; \dnt{\haslb{\Gamma, \bhyp{x}{A}}{r}{\ms{L}}} 
      \\
    \dnt{\haslb{\Gamma}{\letstmt{(x, y)}{e}{r}}{\ms{L}}}
      &= \lmor{\dnt{\hasty{\Gamma}{\epsilon}{e}{A \otimes B}}} ; \alpha
      ; \dnt{\haslb{\Gamma, \bhyp{x}{A}, \bhyp{y}{B}}{r}{\ms{L}}} 
      \\ 
    \dnt{\haslb{\Gamma}{\casestmt{e}{x}{r}{y}{s}}{\ms{L}}}
      &= \lmor{\dnt{\hasty{\Gamma}{\epsilon}{e}{A + B}}}
      ; \delta^{-1} ;
      \\&\quad\;
      [
        \dnt{\haslb{\Gamma, \bhyp{x}{A}}{r}{\ms{L}}},
        \dnt{\haslb{\Gamma, \bhyp{y}{B}}{s}{\ms{L}}}
      ]
      \\
    \dnt{\haslb{\Gamma}{\where{r}{(\wbranch{\ell_i}{x_i}{t_i},)_i}}{\ms{L}}}
      &= \lmor{\entrymor{\Gamma}{r}{\ms{L}}} 
      ; \delta^{-1} 
      ; [\pi_r, \rfix{\loopmor{\Gamma}{(\wbranch{\ell_i}{x_i}{t_i},)_i}{\ms{L}}} ]
  \end{align*} 
  \begin{align*}
    \text{where} \qquad 
    & \boxed{\iota_{\ms{L}, \ell} : \dnt{A} \to_\bot \dnt{\ms{L}}} \qquad \qquad \qquad
    \iota_{(\ms{L}, \ell(A)), \ell} = \iota_r \qquad
    \iota_{(\ms{L}, \kappa(B)), x} = \iota_l ; \iota_{\ms{L}, \ell} \\
    & \boxed{\entrymor{\Gamma}{r}{\ms{L}} : \dnt{\Gamma} \to \dnt{L} + \Sigma_i\dnt{A_i}} \\
    & \entrymor{\Gamma}{r}{\ms{L}} = \dnt{\haslb{\Gamma}{r}{\ms{L}, (\lhyp{\ell_i}{A_i},)_i}} 
    ; \alpha^+_{\dnt{\ms{L}} + \Sigma_i \dnt{A_i}} \\
    & \boxed{\loopmor{\Gamma}{\wbranch{\ell_i}{x_i}{t_i},)_i}{\ms{L}} 
      : \dnt{\Gamma} \otimes \Sigma_i\dnt{A_i} \to \dnt{\ms{L}} + \Sigma_i\dnt{A_i}}
    \\
    & \loopmor{\Gamma}{\wbranch{\ell_i}{x_i}{t_i},)_i}{\ms{L}} = \delta^{-1}_{\Sigma} 
      ; [ \dnt{\haslb{\Gamma, \bhyp{x_i}{A_i}}{t_i}{\ms{L}, (\lhyp{\ell_j}{A_j},)_j}}, ]_i
      ; \alpha^+_{\dnt{\ms{L}} + \Sigma_i \dnt{A_i}}  
  \end{align*}
  \caption{Denotational semantics for \isotopessa{} regions}
  \Description{Denotational semantics for isotope-SSA regions}
  \label{fig:ssa-reg-sem}
\end{figure}

\subsection{Metatheory}

We can now begin to state the metatheoretic properties of our denotational semantics. Before we do so, we establish the convention that whenever we have an equation involving the interpretation of a derivation (e.g., $\dnt{\mathcal{D}} = \dnt{\mc{D}'}$), we assume that all the derivations (e.g., $\mc{D}$ and $\mc{D}'$) exist and are well-formed. 

We begin with
weakening: as shown in Figure~\ref{fig:ssa-ty-sem}, weakenings are modelled, essentially, as
projections from a larger product $\dnt{\Gamma}$ to a smaller product $\dnt{\Delta}$, while
label-weakenings are modelled as injections from a smaller coproduct $\dnt{\ms{L}}$ to a larger
coproduct $\dnt{\ms{K}}$; in particular, in both cases, the morphisms are pure.  A simple induction
can then be used to derive the following weakening lemmas:
\begin{lemma}[name=(Label) Weakening, restate=weakeninglem]
  Given $\Gamma \leq \Gamma'$ and $\ms{L}' \leq \ms{L}$, $\ms{K}' \leq \ms{K}$, we have
  \begin{enumerate}[label=(\alph*)]
    \item $\dnt{\Gamma \leq \Delta} = \dnt{\Gamma \leq \Gamma'};\dnt{\Gamma' \leq \Delta}$
      \label{itm:varwk}
    \item $\dnt{\ms{L}' \leq \ms{K}} = \dnt{\ms{L}' \leq \ms{L}};\dnt{\ms{L} \leq \ms{K}}$
      \label{itm:lbwk}
    \item $\dnt{\hasty{\Gamma}{\epsilon}{a}{A}} 
      = \dnt{\Gamma \leq \Gamma'};\dnt{\hasty{\Gamma'}{\epsilon}{a}{A}}$
      \label{itm:expwk}
    \item $\dnt{\haslb{\Gamma}{r}{\ms{L}}}
      = \dnt{\Gamma \leq \Gamma'}
      ; \dnt{\haslb{\Gamma'}{r}{\ms{L}'}}
      ; \dnt{\ms{L}' \leq \ms{L}}$
      \label{itm:regwk}
    \item $\dnt{\issubst{\gamma}{\Gamma}{\Delta}}
      = \dnt{\Gamma \leq \Gamma'};\dnt{\issubst{\gamma}{\Gamma'}{\Delta}}$
      \label{itm:substwk}
    \item $\dnt{\lbsubst{\Gamma}{\sigma}{\ms{L}}{\ms{K}}}
      = \dnt{\Gamma \leq \Gamma'} \otimes \dnt{\ms{L}}
      ; \dnt{\lbsubst{\Gamma}{\sigma}{\ms{L}}{\ms{K}'}}
      ; \dnt{\ms{K}' \leq \ms{K}}
      $
      \label{itm:lbsubstwk}
  \end{enumerate}
  \label{lem:wk}
\end{lemma}
\begin{proof}
  See Appendix~\ref{proof:weakening}.
\end{proof}

Our first proper semantic theorem is that \emph{rewriting} is sound: if two (potentially impure!)
substitutions have pointwise equal semantics, then substituting by either yields a result with
the same semantics. That is, value-substitution is a well-defined operation when quotienting by the
semantics, or, more formally:

\begin{theorem}[Soundness (Rewriting)]
  Given substitutions $\issubst{\gamma, \gamma'}{\Gamma}{\Delta}$ such that, for all 
  $\thyp{x}{A}{\epsilon} \in \Delta$, $\dnt{\hasty{\Gamma}{\epsilon}{\gamma\;x}{A}} =
  \dnt{\hasty{\Gamma}{\epsilon}{\gamma'\;x}{A}}$, we have
  \begin{enumerate}[label=(\alph*)]
    \item Given $\hasty{\Delta}{\epsilon}{a}{A}$, 
      $\dnt{\hasty{\Gamma}{\epsilon}{[\gamma]a}{A}} = \dnt{\hasty{\Gamma}{\epsilon}{[\gamma']a}{A}}$
    \item Given $\issubst{\rho}{\Delta}{\Xi}$, $\dnt{\issubst{[\gamma]\rho}{\Gamma}{\Xi}} =
      \dnt{\issubst{[\gamma']\rho}{\Gamma}{\Xi}}$
    \item Given $\haslb{\Delta}{r}{\ms{L}}$, $\dnt{\haslb{\Gamma}{[\gamma]r}{\ms{L}}} =
      \dnt{\haslb{\Gamma}{[\gamma']r}{\ms{L}}}$
    \item Given $\lbsubst{\Delta}{\rho}{\ms{L}}{\ms{K}}$,
      $\dnt{\lbsubst{\Gamma}{[\gamma]\rho}{\ms{L}}{\ms{K}}} =
      \dnt{\lbsubst{\Gamma}{[\gamma']\rho}{\ms{L}}{\ms{K}}}$ 
  \end{enumerate}
  \label{thm:rewriting}
\end{theorem}
\begin{proof}
  By a trivial induction, since our semantics is compositional.
\end{proof}

This result allows us to safely perform peephole rewrites which we can independently prove
semantically valid, including on potentially impure operations. However, for pure substitutions, it
is possible to relate the semantics of the substituted term to the original more precisely. By
\emph{pure} substitutions, we mean substitutions which do not replace variables with impure terms.
More precisely, we can define the \emph{effect} of a context using the lattice structure on effects
as follows
\begin{equation}
  \ms{eff}(\cdot) = \bot \qquad 
  \ms{eff}(\Gamma, \thyp{x}{A}{\epsilon}) = \ms{eff}(\Gamma) \sqcup \epsilon
\end{equation}
We can now give a denotational semantics to substitutions in which a substitution
$\issubst{\gamma}{\Gamma}{\Delta}$ is interpreted as a morphism from $\dnt{\Gamma}$ to
$\dnt{\Delta}$ with effect $\leq \ms{eff}(\Delta)$, as in Figure~\ref{fig:ssa-subst-sem}. Such a
substitution is \emph{pure} if the effect of each of its components is $\bot$, i.e., we have
\begin{equation}
  \substpure{\gamma} \iff \forall x, \dnt{\hasty{\Gamma}{\epsilon}{\gamma\;x}{\Delta\;x}} 
    : \dnt{\Gamma} \to_\bot \dnt{\Delta\;x}
\end{equation}
In particular, we have that
\begin{enumerate}
  \item $\ms{eff}(\Delta) = \bot \implies \substpure{\gamma}$
  \item $\substpure{\gamma} \implies \substpure{\lupg{\gamma}}$ and $\substpure{\rupg{\gamma}}$; 
  in particular, the identity substitution is obviously pure.
\end{enumerate}
We can now state the soundness of variable substitution as follows:
\begin{theorem}[name=Soundness (Substitution), restate=soundnesssubst]
  Given $\dnt{\issubst{\gamma}{\Gamma}{\Delta}} : \dnt{\Gamma} \to \dnt{\Delta}$ pure, we have that
  \begin{enumerate}[label=(\alph*)]
    \item $\dnt{\hasty{\Gamma}{\epsilon}{[\gamma]a}{A}} 
      = \dnt{\issubst{\gamma}{\Gamma}{\Delta}};\dnt{\hasty{\Delta}{\epsilon}{a}{A}}$
      \label{itm:tm-subst-sound}
    \item $\dnt{\haslb{\Gamma}{[\gamma]r}{\ms{L}}}
      = \dnt{\issubst{\gamma}{\Gamma}{\Delta}};\dnt{\haslb{\Delta}{r}{\ms{L}}}$
    \item $\dnt{\issubst{[\gamma]\rho}{\Delta}{\Xi}}
      = \dnt{\issubst{\gamma}{\Gamma}{\Delta}};\dnt{\issubst{\rho}{\Delta}{\Xi}}$
    \item $\dnt{\lbsubst{\Gamma}{[\gamma]\sigma}{\ms{L}}{\ms{K}}}
      = \dnt{\issubst{\gamma}{\Gamma}{\Delta}};\dnt{\lbsubst{\Delta}{\sigma}{\ms{L}}{\ms{K}}}$
  \end{enumerate}
  \label{thm:subst-sound}
\end{theorem}
\begin{proof}
  See Appendix~\ref{proof:soundness-subst}.
\end{proof}
In particular, this implies that, when $\gamma$ is pure, the semantics of substitution composition
$[\gamma]\rho$ is just composition of the denotations of $\gamma, \rho$; note in general that this is
\emph{not} true in general if only the $\rho$ is pure, as can trivially be seen from the example
\(\gamma = x \mapsto \ms{print}; x\), \(\rho = x \mapsto ()\). Note also that we can derive
rewriting (Theorem~\ref{thm:rewriting}) for pure substitutions from soundness of substitution, but
not vice versa; however, rewriting also covers \emph{impure} substitutions; this is similar to the
distinction between \emph{uniformity} (which only holds for pure operations) and \emph{dinaturality}
(which holds in general, but has a more restricted form). We can derive the following important
corollary:
\begin{corollary}[Soundness (Single Substitution)]
  Given $\hasty{\Gamma}{\bot}{a}{A}$, we have
  \begin{enumerate}
    \item Given $\hasty{\Gamma, \bhyp{x}{A}}{\bot}{b}{B}$,
    \begin{equation}
      \dnt{\hasty{\Gamma}{\epsilon}{[a/x]b}{B}}
      = \lmor{\dnt{\hasty{\Gamma}{\bot}{a}{A}}} 
        ; \dnt{\hasty{\Gamma, \bhyp{x}{A}}{\epsilon}{b}{B}}
      = \dnt{\hasty{\Gamma}{\epsilon}{\letexpr{x}{a}{b}}{B}}
    \end{equation}
    \item Given $\haslb{\Gamma, \bhyp{x}{A}}{r}{\ms{L}}$, we have
    \begin{equation}
      \dnt{\haslb{\Gamma}{[a/x]r}{\ms{L}}}
      = \lmor{\dnt{\hasty{\Gamma}{\bot}{a}{A}}}
        ; \dnt{\haslb{\Gamma, \bhyp{x}{A}}{r}{\ms{L}}}
      = \dnt{\haslb{\Gamma}{\letstmt{x}{a}{r}}{\ms{L}}}
    \end{equation}
  \end{enumerate}
  \label{corr:single-subst}
\end{corollary}
\begin{proof}
  Follows immediately from the fact that 
  \begin{equation}
    \dnt{\issubst{\lupg{x \mapsto a}}{\Gamma}{\Gamma, \bhyp{x}{A}}} 
    = \Delta_{\dnt{\Gamma}} ; \dnt{\Gamma} \otimes \dnt{\hasty{\Gamma}{\bot}{a}{A}}
    = \lmor{\dnt{\hasty{\Gamma}{\bot}{a}{A}}}
  \end{equation}
\end{proof}

\begin{figure}
  \begin{equation*}
    \boxed{\dnt{\issubst{\gamma}{\Gamma}{\Delta}} 
      : \dnt{\Gamma} \to_{\ms{eff}(\Delta)} \dnt{\Delta}}
  \end{equation*}
  \begin{gather*}
    \dnt{\issubst{\cdot}{\Gamma}{\cdot}} = \tmor{\dnt{\Gamma}}
    \qquad
    \dnt{\issubst{\gamma, x \mapsto e}{\Gamma}{\Delta, \thyp{x}{A}{\epsilon}}}
    = \dmor{\dnt{\Gamma}};\dnt{\issubst{\gamma}{\Gamma}{\Delta}} 
      \ltimes \dnt{\hasty{\Gamma}{\epsilon}{e}{A}}
  \end{gather*}
  \begin{equation*}
    \boxed{\dnt{\lbsubst{\Gamma}{\kappa}{\ms{L}}{\ms{K}}} 
      : \dnt{\Gamma} \otimes \dnt{\ms{L}} \to \dnt{\ms{K}}}
  \end{equation*}
  \begin{gather*}
    \dnt{\lbsubst{\Gamma}{\cdot}{\cdot}{\ms{K}}} 
      = \tmor{\dnt{\Gamma}} \otimes \mb{0}; \lambda; 0_{\ms{K}}
    \\
    \dnt{\lbsubst{\kappa, \ell(x) \mapsto r}{\Gamma}{\ms{L}, \ell(A)}{\ms{K}}}
      = \delta ; [
      \dnt{\lbsubst{\kappa}{\Gamma}{\ms{L}}{\ms{K}}}, 
      \dnt{\haslb{\Gamma, \bhyp{x}{A}}{r}{\ms{K}}}
    ]
  \end{gather*}
  \caption{Denotational semantics for \isotopessa{} (label) substitutions}
  \Description{}
  \label{fig:ssa-subst-sem} 
\end{figure}

We can now move on to stating the metatheoretic properties of label-substitutions, which,
thankfully, turn out to be a little bit simpler. In particular, in Figure~\ref{fig:ssa-subst-sem},
we interpret label substitutions $\lbsubst{\Gamma}{\sigma}{\ms{L}}{\ms{K}}$ as morphisms taking
a copy of the context $\dnt{\Gamma}$ and an element of the coproduct $\dnt{\ms{L}}$ to an element
of the coproduct $\dnt{\ms{K}}$, with an arbitrary effect. Label substitution is then sound in
general, as stated in the following theorem:
\begin{theorem}[name=Soundness (Label Substitution), restate=soundnesslsubst]
  Given $\lbsubst{\Gamma}{\sigma}{\ms{L}}{\ms{K}}$, we have
  \begin{enumerate}[label=(\alph*)]
    \item $\dnt{\haslb{\Gamma}{[\sigma]r}{\ms{K}}}
      = \lmor{\dnt{\haslb{\Gamma}{r}{\ms{L}}}}
      ; \dnt{\lbsubst{\Gamma}{\sigma}{\ms{L}}{\ms{K}}}$
    \item $\dnt{\lbsubst{\Gamma}{[\sigma]\sigma'}{\ms{M}}{\ms{K}}}
      = \dmor{\dnt{\Gamma}} \otimes \dnt{\ms{L}} ; \alpha
      ; \dnt{\Gamma} \otimes \dnt{\lbsubst{\Gamma}{\sigma'}{\ms{M}}{\ms{L}}}
      ; \dnt{\lbsubst{\Gamma}{\sigma}{\ms{L}}{\ms{K}}}$
  \end{enumerate}
\end{theorem}
\begin{proof}
  See Appendix~\ref{proof:soundness-lsubst}
\end{proof}

\subsection{Equational Theory}

\label{ssec:completeness}

Using the metatheory in the previous section, our goal is now to prove the equational theory given
in Section~\ref{sec:equations} sound with respect to any valid \isotopessa{} model. Stated more
precisely, we have the following:
\begin{theorem}[name=Soundness (Equational Theory), restate=soundnesseqn]
  We have that
  \begin{enumerate}[label=(\alph*)]
    \item $\tmeq{\Gamma}{\epsilon}{a}{a'}{A} \implies 
      \dnt{\hasty{\Gamma}{\epsilon}{a}{A}} = \dnt{\hasty{\Gamma}{\epsilon}{a'}{A}}$
      \label{itm:eqn-sound-expr}
    \item $\lbeq{\Gamma}{r}{r'}{\ms{L}} \implies
      \dnt{\haslb{\Gamma}{r}{\ms{L}}} = \dnt{\haslb{\Gamma}{r'}{\ms{L}}}$
      \label{itm:eqn-sound-region}
    \item $\tmseq{\gamma}{\gamma'}{\Gamma}{\Delta} \implies
      \dnt{\issubst{\gamma}{\Gamma}{\Delta}} = \dnt{\issubst{\gamma'}{\Gamma}{\Delta}}$
      \label{itm:eqn-sound-vsubst}
    \item $\lbseq{\sigma}{\sigma'}{\Gamma}{\ms{L}}{\ms{K}} \implies
      \dnt{\lbsubst{\Gamma}{\sigma}{\ms{L}}{\ms{K}}} 
      = \dnt{\lbsubst{\Gamma}{\sigma'}{\ms{L}}{\ms{K}}}$ 
      \label{itm:eqn-sound-lsubst}
  \end{enumerate}
\end{theorem}
\begin{proof}
  See Appendix~\ref{proof:soundness-eqn}
\end{proof}

Now that we've proved the \emph{soundness} of our equational theory, what remains is to prove that
it is \emph{complete}, i.e., that every equation which holds in all \isotopessa{} models can be
derived from it, or, stated more categorically, that our syntax quotiented by the equational theory
forms an initial \isotopessa{} model.
Our strategy for doing this is as follows:
\begin{enumerate}
  \item We begin by constructing a category of expressions $\ms{Th}^\otimes(\Gamma)$ and a category
  of regions $\ms{Th}(\Gamma, \ms{L})$ quotiented by our equational theory, and constructing a
  functor from the former to the latter.
  \item We then show that the category of expressions and the category of regions have the structure
  of an \isotopessa{} expression model and \isotopessa{} model, respectively, and hence that
  expressions may be interpreted in the former and both expressions and regions in the latter.

  To do so, we will first need some notation to talk about the behaviour of the equivalence classes
  of quotients. Suppose $S$ and $T$ are sets of terms. Then we will write
  $\hasty{\Gamma}{\epsilon}{S}{A}$ to mean that for every $a \in S$, we have
  $\hasty{\Gamma}{\epsilon}{a}{A}$, and similarly we will write $\tmeq{\Gamma}{\epsilon}{S}{T}{A}$
  when for all $a \in S$ and $b \in T$, we have $\tmeq{\Gamma}{\epsilon}{a}{b}{A}$. We generalize in
  the obvious fashion to regions as well as the case when only one side of an equivalence is a set
  of terms.

  Then, it will turn out that, for a distinguished variable $\invar$ and label $\outlb$,
  \begin{equation}
    \begin{gathered}
      \hasty{\Gamma, \invar : \pckd{\Delta}}{\epsilon}
        {\dnt{\hasty{\Delta}{\epsilon}{a}{A}}_{\ms{Th}^\otimes(\Gamma)}}{A}
        \\
      \haslb{\Gamma, \invar : \pckd{\Delta}}
        {\dnt{\hasty{\Delta}{\epsilon}{a}{A}}_{\ms{Th}(\Gamma, \ms{L})}}{\ms{L}, \outlb(A)}
        \\
      \haslb{\Gamma, \invar : \pckd{\Delta}}
        {\dnt{\haslb{\Delta}{r}{\ms{K}}}_{\ms{Th}(\Gamma, \ms{L})}}{\ms{L}, \outlb(\pckd{\ms{K}})}
    \end{gathered}
  \end{equation}
  where packing of contexts is defined as in Section~\ref{ssec:records-enums}.
  
  \item Finally, we refine this result to show that
  \begin{equation}
    \begin{gathered}
    \tmeq{\Gamma, \invar : [\Delta]}{\epsilon}
      {\dnt{\hasty{\Delta}{\epsilon}{a}{A}}_{\ms{Th}^\otimes(\Gamma)}}{\pckd{a}}{A}
    \\
    \lbeq{\Gamma, \invar : [\Delta]}
      {\dnt{\hasty{\Delta}{\epsilon}{a}{A}}_{\ms{Th}(\Gamma, \ms{L})}}{\ms{ret}\;\pckd{a}}
      {\ms{L}, \outlb(A)}
    \\
    \lbeq{\Gamma, \invar : [\Delta]}
      {\dnt{\haslb{\Delta}{r}{\ms{K}}}_{\ms{Th}(\Gamma, \ms{L})}}
      {\pckd{r}}{\ms{L}, \outlb(\pckd{\ms{K}})}
    \end{gathered}
  \end{equation}
  where $\ms{ret}\;a := \brb{\outlb}{a}$. Since the packing operator $\pckd{\cdot}$ on terms and
  regions from Section~\ref{ssec:records-enums} is injective for pure contexts $\ms{eff}(\Gamma) =
  \bot$, and hence in particular for $\Gamma = \cdot$, $\ms{L} = \cdot$, it follows that in this
  case the category of expressions and the category of regions are the initial distributive
  \isotopessa{} expression model and \isotopessa{} model respectively.
\end{enumerate}

\subsection{Expressions}

We'll begin by going over the entire proof of completeness for expressions, which is the simpler
case. In particular, we may define the category $\ms{Th}_\epsilon^\otimes(\Gamma)$ of expressions
with effect $\epsilon$ as follows:
\begin{itemize}
  \item Objects $|\ms{Th}^\otimes(\Gamma)|$ types $A, B, C$
  \item Morphisms $\ms{Th}^\otimes(\Gamma)_\epsilon(A, B) = \{e \mid \hasty{\Gamma,
    \bhyp{\invar}{A}}{\epsilon}{e}{B}\}$ quotiented by $\tmeq{\Gamma,
    \bhyp{\invar}{A}}{\epsilon}{e}{e'}{B}$
  \item Identity $(\hasty{\Gamma, \bhyp{\invar}{A}}{\bot}{\invar}{A}) \in
  \ms{Th}^\otimes(\Gamma)_\epsilon(A, A)$
  \item Composition $e;e' = (\ms{let}\;\invar = e; e')$, which satisfies
  $$
  \hasty{\Gamma, \bhyp{\invar}{A}}{\epsilon}{e}{B}, \quad
  \hasty{\Gamma, \bhyp{\invar}{B}}{\epsilon}{e'}{C} \qquad \implies \qquad
  \hasty{\Gamma, \bhyp{\invar}{A}}{\epsilon}{e;e'}{C}
  $$
  We may verify that this satisfies the axioms of a category w.r.t. our equational theory
\end{itemize}
In general, the category of expressions $\ms{Th}^\otimes(\Gamma)$ is simply then given by
$\ms{Th}_\top^\otimes(\Gamma)$, which can be viewed as the union of all
$\ms{Th}_\epsilon^\otimes(\Gamma)$.

We now need to equip $\ms{Th}_\epsilon^\otimes(\Gamma)$ with the structure of a premonoidal
category. Obviously, we wish to define the tensor product of types $A$ and $B$ to be simply $A
\otimes B$; we can then begin by defining projections
\begin{equation}
  \hasty{\Gamma, \invar : A \otimes B}{\epsilon}{\pi_l := \ms{let}\;(x, y) = \invar; x}{A} \qquad 
  \hasty{\Gamma, \invar : A \otimes B}{\epsilon}{\ms{let}\;(x, y) = \invar; y}{B}
\end{equation}
By simply using the pair constructor as a cartesian product $\langle a, b \rangle = (a, b)$, this
can be shown to endow $\ms{Th}_\bot^\otimes(\Gamma)$ with the structure of a cartesian category,
allowing us to define the associators, symmetries, and unitors in the natural manner. If we then
define tensor functors
\begin{equation}
  - \otimes X : e \mapsto \ms{let}\;(\invar, x) = \invar; (e ; (\invar, x)) \qquad
  X \otimes - : e \mapsto \ms{let}\;(x, \invar) = \invar; (e ; (x, \invar))
\end{equation}
we find that $\ms{Th}_\epsilon^\otimes(\Gamma)$ and hence in particular $\ms{Th}^\otimes(\Gamma)$ is
endowed with the structure of a Freyd category with pure subcategory $\ms{Th}_\bot^\otimes(\Gamma)$.

Similarly, we wish to show that $A + B$ is the coproduct of $A$ and $B$ in
$\ms{Th}_\epsilon^\otimes(\Gamma)$. Since we already have obvious injection morphisms
\begin{equation}
  \hasty{\Gamma, \invar : A}{\epsilon}{\iota_l := \iota_l\;\invar}{A + B} \qquad
  \hasty{\Gamma, \invar : B}{\epsilon}{\iota_r := \iota_r\;\invar}{A + B}
\end{equation}
we can define the coproduct of morphisms $\hasty{\Gamma, \invar : A}{\epsilon}{a}{C}$ and 
$\hasty{\Gamma, \invar : B}{\epsilon}{b}{C}$ to be simply given by
\begin{equation}
  \hasty{\Gamma, \invar : A + B}{\epsilon}{[a, b] := \caseexpr{\invar}{\invar}{a}{\invar}{b}}{C}
\end{equation}
It is straightforward to verify that this indeed induces a coproduct on
$\ms{Th}_\epsilon^\otimes(\Gamma)$ and hence on $\ms{Th}^\otimes(\Gamma)$ All that remains is to
show that $\ms{Th}_\epsilon^\otimes(\Gamma)$ is in fact a \emph{distributive} Freyd category. To do
so, we may define an inverse distributor morphism
\begin{equation}
  \hasty{\Gamma, \invar : A \otimes (B + C)}
    {\epsilon}
    {\delta^{-1} := \ms{let}\;(x, y) = \invar; \caseexpr{y}{z}{\iota_l(x, z)}{z}{\iota_r(x, z)}}
    {A \otimes B + A \otimes C}
\end{equation}
which can easily be shown to be an inverse to the obvious distributor morphism. We may now note that
\begin{equation}
  \dnt{\cdot}_{\ms{Th}^\otimes(\Gamma)} = \mb{1}, \quad
  \dnt{\Delta, \thyp{x}{A}{\epsilon}}_{\ms{Th}^\otimes(\Gamma)} 
    = \dnt{\Delta}_{\ms{Th}^\otimes(\Gamma)} \otimes A
  \qquad \implies \qquad
  \dnt{\Delta}_{\ms{Th}^\otimes(\Gamma)} = \pckd{\Gamma}
\end{equation}
Therefore, it follows that, as expected, that
$
  \hasty{\Gamma, \invar : \pckd{\Delta}}{\epsilon}
          {\dnt{\hasty{\Delta}{\epsilon}{a}{A}}_{\ms{Th}^\otimes(\Gamma)}}{A}
$
and it remains to show that we in fact have
$$
  \tmeq{\Gamma, \invar : \pckd{\Delta}}{\epsilon}
        {\dnt{\hasty{\Delta}{\epsilon}{a}{A}}_{\ms{Th}^\otimes(\Gamma)}}{[a]}{A}
$$
which can be done by a relatively straightforward induction, implying, since $[\cdot]$ is injective
w.r.t. our equational theory for pure contexts, the following theorem:
\begin{theorem}[name=Completeness (Expressions), restate=completenessexpr]
  We have that, for all pure $\ms{eff}(\Gamma) = \bot$,
  $$
    \tmeq{\Gamma}{\epsilon}{e}{e'}{A} 
    \iff \dnt{\hasty{\Gamma}{\epsilon}{e}{A}}_{\ms{Th}^\otimes(\cdot)} 
         = \dnt{\hasty{\Gamma}{\epsilon}{e'}{A}}_{\ms{Th}^\otimes(\cdot)}
  $$
  In particular, this implies that $\ms{Th}(\cdot)$ is the initial \isotopessa{} expression model
  \label{thm:complete-expr}
\end{theorem}
\begin{proof}
  See Appendix~\ref{proof:complete-expr}
\end{proof}

\subsection{Regions}

We define the category $\ms{Th}(\Gamma, \ms{L})$ of regions as follows:
\begin{itemize}
  \item Objects $|\ms{Th}(\Gamma, \ms{L})|$ types $A, B, C$
  \item Morphisms $\ms{Th}(\Gamma, \ms{L})(A, B) = 
    \{r \mid \haslb{\Gamma, \bhyp{\invar}{A}}{r}{\ms{L}, \outlb(B)}\}$ 
    quotiented by $\lbeq{\Gamma, \bhyp{\invar}{A}}{r}{r'}{\ms{L}, \outlb(B)}$
  \item Identity $\haslb{\Gamma, \bhyp{\invar}{A}}{\ms{br}\;\ms{ret}\;\invar}{\ms{L}, \outlb(A)}$
    where $\ms{ret}\;a := \brb{\outlb}{a}$
  \item Composition $r;r' = [\lupg{(\outlb(\invar) \mapsto r')}]r$
\end{itemize}
In particular, we may view $\ms{ret}$ as an identity-on-objects functor
$\ms{Th}^\otimes_\bot(\Gamma) \to \ms{Th}(\Gamma, \ms{L})$ with action on morphisms given by given
by 
\begin{equation}
  \hasty{\Gamma, \invar : A}{\bot}{e}{B} 
  \qquad \mapsto \qquad 
  \haslb{\Gamma, \invar(A)}{\ms{ret}\;e}{\ms{L}, \outlb(B)}
\end{equation}
We will use this to equip $\ms{Th}(\Gamma, \ms{L})$ with the structure of a Freyd category.
In particular, taking our subcategory of pure morphisms to be the image of $\ms{ret}$ in 
$\ms{Th}(\Gamma, \ms{L})$, we may define the obvious tensor functors
\begin{equation}
  - \otimes X : r \mapsto \ms{let}\;(\invar, x) = \invar; (r ; \ms{ret}\;(\invar, x)) \qquad
  X \otimes - : r \mapsto \ms{let}\;(x, \invar) = \invar; (r ; \ms{ret}\;(x, \invar))
\end{equation}
Our premonoidal structure is then completely described by requiring that $\ms{ret}$ preserves all
relevant structure, i.e., that we have
\begin{equation}
  \alpha = \ms{ret}\;\alpha \qquad
  \lambda = \ms{ret}\;\lambda \qquad
  \rho = \ms{ret}\;\rho \qquad
  \sigma = \ms{ret}\;\sigma \qquad
  \Delta = \ms{ret}\;\Delta
\end{equation}
Just like for expressions, we can write the coproduct of 
$\haslb{\Gamma, \bhyp{\invar}{A}}{s}{\ms{L}, \outlb(C)}$ and
$\haslb{\Gamma, \bhyp{\invar}{B}}{t}{\ms{L}, \outlb(C)}$
in $\ms{Th}(\Gamma, \ms{L})$ as
\begin{equation}
  \haslb{\Gamma, \invar : A + B}{[s, t] 
    := \casestmt{\invar}{\invar}{s}{\invar}{t}}{\ms{L}, \outlb(C)}
\end{equation}
with the obvious injections $\iota_l := \ms{ret}\;\iota_l$ and $\iota_r = \ms{ret}\;\iota_r$. It
turns out that in this case $\ms{ret}$ preserves coproducts as well, and we can therefore easily
conclude that our category is distributive by taking inverse distributor $\delta^{-1} =
\ms{ret}\;\delta^{-1}$.

All that remains now is to take $\ms{Th}(\Gamma, \ms{L})$ from an \isotopessa{} expression model
to an \isotopessa{} model by giving it an Elgot structure. We do so by defining the fixpoint
of a morphism $\haslb{\Gamma, \invar : A}{r}{\ms{L}, \outlb(B + A)}$ as follows:
\begin{equation}
  \haslb{\Gamma, \invar : A}
    {r^\dagger := \where{\ms{br}\;\ms{go}\;\invar}{\wbranch{\ms{go}}{\invar : A}
        {r ; \casestmt{\invar}{x}{\ms{ret}\;x}{y}{\ms{br}\;\ms{go}\;y}}}}
    {\ms{L}, \outlb(B)}
\end{equation}
where $\ms{go}$ is an (arbitrary) fresh label. We can verify this indeed satisfies the axioms of
an Elgot structure through a somewhat tedious calculation. We may now note that
\begin{equation}
  \begin{aligned}
    \dnt{\cdot}_{\ms{Th}(\Gamma, \ms{L})} 
      &= \mb{1}, &
    \dnt{\Delta, \thyp{x}{A}{\epsilon}}_{\ms{Th}(\Gamma, \ms{L})} 
      &= \dnt{\Delta}_{\ms{Th}(\Gamma, \ms{L})} \otimes A
    &\qquad \implies \qquad
    \dnt{\Delta}_{\ms{Th}(\Gamma, \ms{L})} &= \pckd{\Gamma} \\
    \dnt{\cdot^+}_{\ms{Th}(\Gamma, \ms{L})} &= \mb{0}, &
    \dnt{\ms{K}, \ell(A)}_{\ms{Th}(\Gamma, \ms{L})} 
      &= \dnt{\ms{K}}_{\ms{Th}(\Gamma, \ms{L})} + A
    &\qquad \implies \qquad
    \dnt{\ms{K}}_{\ms{Th}(\Gamma, \ms{L})} &= \pckd{\ms{K}}
  \end{aligned}
\end{equation}
and hence that 
\begin{equation*}
  \haslb{\Gamma, \invar : \pckd{\Delta}}
        {\dnt{\hasty{\Delta}{\epsilon}{a}{A}}_{\ms{Th}(\Gamma, \ms{L})}}{\ms{L}, \outlb(A)}
      \qquad
  \haslb{\Gamma, \invar : \pckd{\Delta}}
    {\dnt{\haslb{\Delta}{r}{\ms{K}}}_{\ms{Th}(\Gamma, \ms{L})}}{\ms{L}, \outlb(\pckd{\ms{K}})}
\end{equation*}
as expected. It is relatively easy to derive that
\begin{equation*}
  \lbeq{\Gamma, \invar : \pckd{\Delta}}
        {\dnt{\hasty{\Delta}{\epsilon}{a}{A}}_{\ms{Th}(\Gamma, \ms{L})}}
        {\ms{ret}\;\pckd{a}}
        {\ms{L}, \outlb(A)}
\end{equation*}
by a relatively straightforward induction. A much more tedious induction is required to prove that
\begin{equation*}
  \lbeq{\Gamma, \invar : \pckd{\Delta}}
        {\dnt{\haslb{\Delta}{r}{\ms{K}}}_{\ms{Th}(\Gamma, \ms{L})}}
        {[r]}
        {\ms{L}, \outlb(\pckd{\ms{K}})}
\end{equation*}
since the case for \ms{where}-statements is particularly complex. With a little bit more book-keeping (which can be found in the mechanization), we can state the completeness theorem as follows:
\begin{theorem}[name=Completeness (Regions), restate=completenessregions]
  We have that, for all pure $\ms{eff}(\Gamma) = \bot$,
  $$
    \lbeq{\Gamma}{r}{r'}{\ms{L}} 
    \iff \dnt{\haslb{\Gamma}{r}{\ms{L}}}_{\ms{Th}(\cdot, \cdot)} 
        = \dnt{\haslb{\Gamma}{r'}{\ms{L}}}_{\ms{Th}(\cdot, \cdot)} 
  $$
  In particular, this implies that $\ms{Th}(\cdot, \cdot)$ is the initial \isotopessa{} model.
  \label{thm:complete-reg}
\end{theorem}
\begin{proof}
  See Appendix~\ref{proof:complete-reg}
\end{proof}

\section{Concrete Models}

\label{sec:concrete}

\subsection{Monads and Monad Transformers}

As previously stated, every strong monad over a CCC, by virtue of providing a
model of higher-order effectful programming, \emph{also} provides a model of
first-order effectful programming, i.e., (acyclic) SSA. In particular, the Kleisli
category $\ms{C}_{\ms{T}}$ of every such strong monad $\ms{T}$ is a premonoidal category and
induces a Freyd category with pure morphisms given by
\begin{equation}
  (\ms{C}_{\ms{T}})_{\bot}(A, B) = \{f;\eta_B : A \to \ms{T}B \mid f : A \to B\} 
  \subseteq \ms{C}(A, \ms{T}B) = \ms{C}_{\ms{T}}(A, B)
\end{equation}
It follows that every strong Elgot monad over a CCC with all coproducts is an
\isotopessa{} model (since, in particular, every CCC with coproducts is
distributive). This allows us to very quickly amass a large collection of
\isotopessa{} models from the literature on strong Elgot monads.

Probably the simplest example of an \isotopessa{} model is given by the
\emph{option monad} $\ms{Option}\;A = A + \mb{1}$, with fixpoint operation
$$
  f^\dagger = \ms{some}\;b \quad \text{if} \quad \exists n, f^n\;a = \ms{some}\;(\iota_l\;b)
  \qquad
  f^\dagger = \ms{none} \quad \text{otherwise}
$$
where
$$
  f^0\;a = \ms{some}\;(\iota_r\;a) \qquad
  f^{n + 1}\;a = \ms{bind}\;(f^{n}\;a)\;[(\ms{pure};\iota_l), f] 
$$
This model allows us to interpret potentially divergent but otherwise purely functional programs.
Another very important \isotopessa{} model is generated by the \emph{powerset monad} $\mc{P}$ over
\ms{Set}, which has fixpoint operation
$$
  f^\dagger\;a = \bigcup_n\{b \mid \iota_l\;b \in f^n\;a\}
$$
We can use this monad to interpret \emph{nondeterministic} potentially divergent functional
programs. We can further expand our repertoire of \isotopessa{} models by considering \emph{monad
transformers}, many of which preserve Elgot-ness. For example, if $\ms{T}$ is a strong Elgot over
\ms{Set},
\begin{itemize}
  \item For all types $R$, the \emph{reader transformer}
  $\ms{ReaderT}\;R\;\ms{T}\;A = R \to \ms{T}\;A$ is strong Elgot 
  with monad operations
  \begin{equation}
    \eta\;a = \lambda - . \eta_{\ms{T}}\;a \qquad
    \ms{bind}\;a\;f = \lambda r. \ms{bind}_{\ms{T}}\;(a\;r)\;(\lambda a. f\;a\;r)
  \end{equation}
  and fixpoint
  \begin{equation}
    f^\dagger = \lambda r.(\lambda a. f\;a\;r)^{\dagger_{\ms{T}}} 
  \end{equation}
  \item For all monoids $W$, the \emph{writer transformer}
  $\ms{WriterT}\;W\;\ms{T}\;A = \ms{T}\;(A \times W)$ is strong Elgot 
  with monad operations
  \begin{equation}
    \eta\;a = \eta_{\ms{T}}\;(a, 1) \qquad
    \ms{bind}\;a\;f = \ms{bind}_{\ms{T}}\;a\;(\lambda (a, w). 
      \ms{bind}_{\ms{T}}\;(f\;a)\;(\lambda (b, w') . (b, w \cdot w')))
  \end{equation}
  and fixpoint
  \begin{equation}
    f^\dagger\;a = (
        \lambda (a, w). \ms{bind}_{\ms{T}}\;(f\;a)\;
          (\lambda (b, w') . \eta_{\ms{T}}([(\lambda c.\,  (c, w \cdot w')), 
                                          (\lambda c.\, (c, w \cdot w'))]\;b))
      )^{\dagger_{\ms{T}}}
      (a, 1)
  \end{equation}
  \item For all types $S$, the \emph{state transformer}
  $\ms{StateT}\;S\;\ms{T}\;A = S \to \ms{T}(A \times S)$ is strong Elgot 
  with monad operations
  \begin{equation}
    \eta\;a = \lambda s. \eta_{\ms{T}}\;(a, s) \qquad
    \ms{bind}\;a\;f = \lambda s. \ms{bind}_{\ms{T}}\;(a\;s)\;(\lambda (a, s'). f\;a\;s')
  \end{equation}
  and fixpoint
  \begin{equation}
    f^\dagger = \lambda s.(\lambda (a, s'). 
      \ms{bind}_{\ms{T}}\;(f\;a\;s')\;(\lambda (b, s') . 
        \eta_{\ms{T}}([(\lambda c.\, (c, s')), (\lambda c. (c, s'))]\;b)))^{\dagger_{\ms{T}}}
  \end{equation}
\end{itemize}

One other important source of \isotopessa{} models is via \emph{Elgot submonads}
of monads on \ms{Set}, which we define as follows:
\begin{definition}[Elgot submonad]
  A monad $(\ms{S}, \eta', \mu')$ is a submonad of $(\ms{T}, \eta, \mu)$ if
  $\eta = \eta'$, $\forall A, \ms{S}\;A \subseteq \ms{T}\;A$, and, on the set
  $\ms{S}\;(\ms{S}\;A) \subseteq \ms{T}\;(\ms{T}\;A)$, $\mu = \mu'$. We say
  $\ms{S}$ is \emph{strong Elgot} if $\ms{T}$ is strong Elgot and, for all $f
  \in A \to \ms{S}\;(A + B)$, $f^\dagger \in A \to \ms{S}\;B$ (i.e. the fixpoint
  operation sends the Kleisli category of $\ms{S}$ to itself).
\end{definition}
One useful property this definition has is that the intersection of arbitrarily many (Elgot)
submonads forms an (Elgot) submonad, i.e., the (Elgot) submonads form a complete lattice.

\subsection{Trace Models}

A particularly important class of \isotopessa{} models are what we will call
\emph{trace models}, which can be viewed as programs which either diverge or
produce a result, while also producing (potentially nondeterministic,
potentially infinite) traces of events $\epsilon \in \mc{E}$.

\begin{definition}[(Monoidal) Stream Action]
  A \textit{stream action} $(M, I, \cdot, \Sigma)$ is a set $M$ of \emph{finite
  effects} and set $I$ of \emph{divergent effects} equipped with a binary
  function (``action'') $\cdot : M \times I \to I$ and a function $\Sigma : M^\omega \to I$
  mapping infinite streams of $M$ to an element of $I$ satisfying $\Sigma \sigma
  = \sigma_0 \cdot \Sigma_i \sigma_{i + 1}$, where $\Sigma_i \sigma_i = \Sigma
  (\lambda i. \sigma_i)$. If $M$ is a monoid and $(M, I, \cdot)$ a monoid
  action, we say we  have a \emph{monoidal} stream action
\end{definition}
The most important example of a monoidal stream action is the \emph{free}
monoidal stream action over a set $E$, called the set of \emph{events}, where:
\begin{itemize}
  \item $M = E^*$ is the free monoid over a set of ``events" $\epsilon \in E$, i.e., the
  monoid of finite lists of events
  \item $I = E^{\leq \omega}$ is the set of \textit{potentially} infinite
  streams of events $\epsilon \in E$
  \item $m \cdot i$ is defined in the obvious manner, i.e. by prepending the (finite) list $m$ to
  the (potentially infinite) stream $i$
  \item $\Sigma\;\sigma$ is defined as the supremum $\bigsqcup_it_i$, where $t
  \leq t'$ iff $t$ is a prefix of $t'$, $t_0 = \bot$ (the empty stream), and
  $t_{i + 1} = t_i \cdot \sigma_i$
\end{itemize}
Other important examples include:
\begin{itemize}
  \item The free stream action over a set of events $E$, where $M = E$, $I =
  E^\omega$ (the set of infinite streams of events $\epsilon$), $m \cdot i$ is
  merely prepending $m$ to the stream $i$, and $\Sigma \sigma = \sigma$
  \item The terminal stream action over an arbitrary set $M$, where $I
  = \mb{1}$. This is always a monoidal stream action if $M$ is a monoid
\end{itemize}

\begin{definition}[Trace Monad Transformer]
  Given a monoidal stream action $(M, I, \cdot, \Sigma)$, we can define the \emph{trace
  monad transformer} over $\ms{Set}$ as follows:
  $
  \ms{TraceT}\;\Sigma\;\ms{T}\;A = \ms{T}\;(A \times M + I)
  $
  with monad operations
  $$
    \eta = (\lambda a. \iota_l (a, 0)) ; \eta_{\ms{T}} \qquad
    \ms{bind}\;a\;f 
    = \ms{bind}_{\ms{T}}\;a\;[\lambda (a, m). m \cdot f a, \lambda i. \eta_{\ms{T}}(\iota_r\;i)]
  $$
  This yields a family of useful Elgot monads, including
  \begin{itemize}
    \item $\ms{Traces}?\;\Sigma = \ms{TraceT}\;\Sigma\;\mc{P}$, with the fixpoint of 
    $f : A \to B + A$ given by
    \begin{align}
      f^\dagger\;a &= f^\dagger_\infty\;a \cup \bigcup_n f^\dagger_n\;a 
        \qquad \qquad \text{where} \\
      f^\dagger_n\;a &=
          \{\iota_l\;(b, m) \mid \iota_l(\iota_l\;b, m) \in [\ms{id}_B, f]^n(\iota_r\;a)\}
          \cup \{\iota_r\;i \mid \iota_r\;i \in [\ms{id}_B, f]^n(\iota_r\;a)\} \\
      f^\infty\;a &= \{\iota_r \Sigma \sigma \mid \exists \alpha : A^\omega, 
        \alpha_0 = a \land \forall i, \iota_l\;(\iota_r\;\alpha_{i + 1}, \sigma_i) \in f\;\alpha_i \} 
    \end{align}
    We say a morphism $f : A \to \ms{Traces}?\;\Sigma\;B$ is \emph{total} if all $f\;a$ are
    nonempty, and \emph{deterministic} if all $f\;a$ are subsingletons.
    \item $\ms{Traces}\;\Sigma = \ms{TraceT}\;\Sigma\;\mc{P}^+$, which can be viewed as an Elgot
    submonad of $\ms{Traces}?\;\Sigma$ since the fixpoint of a total function is total
    \item $\ms{Trace}?\;\Sigma = \ms{TraceT}\;\Sigma\;\ms{Option}$, which can be viewed as an Elgot
    submonad of $\ms{Traces}?\;\Sigma$ since the fixpoint of a deterministic function is
    deterministic
    \item $\ms{Trace}\;\Sigma = \ms{TraceT}\;\Sigma\;\ms{Id}$, which can be viewed as an Elgot
    submonad of $\ms{Traces}?\;\Sigma$, and in particular the intersection 
    $\ms{Trace}?\;\Sigma \cap \ms{Traces}\;\Sigma$.
  \end{itemize}
\end{definition}

As a simple, concrete example of this type of model this construction lets us define, consider the
\emph{nondeterministic printing monad}, which is simply defined as \(\ms{Print}\;A \equiv
\ms{Traces}\;\Sigma\), where \(\Sigma: (\ms{byte}^*)^\omega \to \ms{byte}^{\leq \omega}\) denotes
concatenation of bytestrings into a (potentially infinite) bytestream. This gives us us an
\isotopessa{} model supporting effectful instructions \(\ms{print}: \ms{byte}^* \to \mb{1}\) and
\(\ms{nondet}: \mb{1} \to A\) (for \(A\) nonempty), with semantics \(\dnt{\ms{print}} = \lambda
b.\{\iota_0 ((), b)\}\), \(\dnt{\ms{nondet}} = \lambda (). \{\iota_0 (a, []) \mid a \in A\}\).

We can make things a little more interesting by adding a heap to the mix, defining our monad
\(\ms{Comp} = \ms{StateT}\;\ms{Heap}\;\ms{Print}\), where \(\ms{Heap} = \nats \rightharpoonup
\nats\) is simply a partial function with finite support. The Kleisli category of this monad is also
a Freyd category and inherits an Elgot structure from that on \(\ms{Set}_{\ms{Print}}\); we
therefore have an \isotopessa{} model supporting the instructions \(\ms{set}: \nats \times \nats \to
\mb{1}\), \(\ms{get}: \nats \to \nats\), \(\ms{alloc}: \nats \to \nats)\), and \(\ms{free}: \nats
\to \mb{1}\). We can assign semantics to \(\ms{set}\) in the standard fashion, with \(\dnt{\ms{set}}
= \lambda (p, v)\;h. \{\iota_0((), [p \mapsto v]h, [])\}\). \(\ms{get}\) is a bit more tricky, since
it is unclear what to do when we try to access uninitialized memory: one option is simply to return
an arbitrary value, with \(\dnt{\ms{get}} = \lambda p\;h. \{\iota_0(v, h, []) \mid v = h\;p \lor p
\notin h\}\). Finally, \(\dnt{\ms{alloc}} = \lambda v\;h.\{\iota_0\;(p, h', []) \mid h' = h \sqcup p
\mapsto v\}\) simply fills a random empty heap cell with the provided value, returning a pointer to
the cell, while \(\dnt{\ms{free}} = \lambda p\;h. \{\iota_0\;((), h \setminus p, [])\}\).

As we can see, the trace monad is a very powerful tool for quickly constructing families of
\isotopessa{} models, at least for relatively simple side-effects. It turns out, however, that the
same techniques can be used to tackle much more complex side-effects; to demonstrate this, we will
show how to build a simple model of TSO weak memory based on that given in \citet{sparky}.

\subsection{TSO weak memory}

\label{ssec:tso}

To start, we need to construct an appropriate stream action monoid to take our traces over. For
simplicity, we will consider an infinite set of named locations \(x, y, z \in \ms{Loc}\) which are
subject to concurrent modification by all threads via TSO reads, writes, and fences. We begin with
the following definitions:
\begin{definition}[Pomset] 
  A \emph{pomset} \(\alpha\) over a set of actions \(\mc{A}\) with a distinguished null action
  "\emph{tick}" \(\delta \in \mc{A}\) is a \textit{nonempty} partially-ordered \emph{carrier set}
  \(P\) such that every \(p \in P\) has finitely many predecessors equipped with a mapping \(\alpha:
  P \to \mc{A}\). A pomset is \emph{finite} if its carrier set is. We will quotient pomsets over the
  removal of arbitrarily many copies of $\delta$ as long as infinite sets are only equated to other
  infinite pomsets. We will write unordered pomsets using multiset notation, e.g. \(\{a, a, b\}\),
  and linearly ordered pomsets using list notation, e.g. \([a, a, b]\).
\end{definition}
(Finite) pomsets form a monoid under sequential composition \(\alpha;\beta\), which is defined by
the function \([\alpha, \beta]: \ms{trim}(P + Q) \to \mc{A}\), where \(P + Q\) is given the
lexicographic ordering and \(\ms{trim}(R)\) removes all elements of \(R\) with infinitely many
predecessors. They also form a monoid under parallel composition, defining \(\alpha || \beta =
[\alpha, \beta]: P + Q \to \mc{A}\) where \(P + Q\) is given the standard partial ordering (with
elements of \(P\) and \(Q\) incomparable). Note in both cases the monoidal unit is \(\{\delta\}\),
since we only consider nonempty pomsets but allow the removal of finitely many ticks. Given any
partially ordered set \(N\) and a family of pomsets \(\alpha_n\) for \(n \in N\), we can define
their \textit{sum} \(\Sigma_n\alpha_n\) to be given by the function \((n, a) \mapsto \alpha_n\;a:
\ms{trim}(\Sigma_nP_n) \to \mc{A}\), where the dependent product \(\Sigma_nP_n\) is given the
lexicographic order. In particular, choosing \(N = \nats\) makes \(\Sigma:
\ms{Pom}_{\ms{fin}}^\omega \to \ms{Pom}\) into a stream action monoid w.r.t the sequential
composition monoid on finite pomsets. 

We define a \emph{program order pomset} to be a pomset with \(\mc{A}_{\ms{PO}} = \mc{A}_w \cup
\mc{A}_r \cup \{\delta\}\), where \(\mc{A}_r\) consists of \textit{reads} of the form \(x = v\) and
\(\mc{A}_w\) consists of \textit{writes} of the form \(x := v\) for locations \(x\) and values \(v
\in \ms{Word}\). We define the \emph{program order monad} \(\ms{PO}\;A = \ms{Traces}\;\Sigma\;A\),
where \(\Sigma\) is taken as a stream action on finite program order pomsets. This yields an
\isotopessa{} model with support for concurrent read and write operations \(\ms{read}_x:
\mc{I}^\varnothing_0(\mb{1}, \ms{Word})\), \(\ms{write}_x: \mc{I}^\varnothing_0(\ms{Word}, \mb{1})\)
with semantics
\(
  \dnt{\ms{read}_x} = R_x^{\ms{PO}} = \lambda (). \{(v, \{x = v\}) | v \in \ms{Word}\}
\),
\(
  \dnt{\ms{write}_x} = W_x^{\ms{PO}} =  \lambda v. \{((), \{x := v\})\}
\).
The semantics of \(\ms{read}\) in particular give a hint as to how this model works: rather than
tracking the state of the heap, we simply \textit{emit} a pomset of the events that would have been
generated by a given execution (in the case of a read, a read event \(x = v\) whenever the read
returns \(v\), and in the case of a write, the single write event \(x := v\) for a write of \(v\)),
and later post-filter to obtain a set of valid executions. By doing so, our semantics remains
compositional, allowing us to reason about each program fragment individually, while still allowing
other program fragments to perform concurrent operations affecting our potential executions. On that
note, we can define the \textit{parallel execution} of two morphisms as follows:
\begin{equation}
  \begin{aligned}
  f_0 || f_1 = \lambda (a_0, a_1). 
  & \{\iota_0 ((b_0, b_1), \alpha_0 || \alpha_1) 
    \mid \iota_0 (b_i, \alpha_i) \in f_i\;a_i\} 
  \\ & \cup \{\iota_1 (\alpha_0 || \alpha_1) 
      \mid (\iota_0 (b_0, \alpha_0) \in f_0\;a \lor \iota_0\;\alpha_0 \in f_0\;a_0) 
      \land \iota_1 \alpha_1 \in f_1\;a_1\} 
  \\ & \cup \{\iota_1 (\alpha_0 || \alpha_1) 
      \mid \iota_1\;\alpha_0 \in f_0\;a_0 
      \land \iota_0 (a_1, \alpha_1) \in f_1\;a_1\}
    : \ms{Set}_{\ms{PO}}(A \otimes A', B \otimes B')
  \end{aligned}
\end{equation}
It's tempting to try to use this to define a tensor product of morphisms to obtain a monoidal
(rather than premonoidal) category, but unfortunately, sliding still fails: \((W_x^{\ms{PO}} ||
\ms{id}) ; (\ms{id} || W_y^{\ms{PO}})\) emits the pomset \([x := a, y := b]\), while \((\ms{id} ||
W_y^{\ms{PO}}) ; (W_x^{\ms{PO}} || \ms{id})\) emits the pomset \([y := b, x := a]\). 
% If we did want this behaviour, we'd need to have a significantly more
% sophisticated model of composition which effectively takes nontermination into
% account, which we will leave to future work.

To graduate from sequential consistency to a genuine, if maximally simple, weak memory model, we
introduce \textit{load buffering} of write actions. We will implement TSO ordering by buffering all
our writes, in which case they are only visible to the local thread. On the other hand, read events
will first attempt to read from the buffer, and, if there is no corresponding write in the buffer,
will read an arbitrary value, in both cases pushing an event to the global pomset. At any point, we
may choose to \textit{flush} some of the buffered writes. We introduce the set \(\mc{A}_b =
\{(\bufloc{x} := v)\}\) of \textit{buffer write actions}, where an action \(\bufloc{x} := v\) by a
thread denotes adding a write \(x := v\) to the thread's write buffer. We may then define the set of
TSO actions \(\mc{A}_{\ms{TSO}} = \mc{A}_{\ms{PO}} \cup \mc{A}_b\); a pomset over this set is called
a \emph{TSO pomset}. A buffer will be defined to be a list of write actions \(\ms{Buf} =
\mc{A}_b^*\), which we will interpret as linear pomsets over \(\mc{A}_{\ms{TSO}}\) ordered by index
(with the empty list corresponding to \(\{\delta\}\)). In particular, we define the monad \(\ms{TSO}
= \ms{StateT}\;\ms{Buf}\;(\ms{Trace}\;\Sigma)\), where \(\Sigma\) is taken as a the stream action of
finite TSO pomsets on TSO pomsets. We can view \(\ms{PO}\) as a submonad of this monad.

Given a buffer \(\ms{Buf}\), we can define the result of reads \([\cdot]_x: \ms{Buf} \to \ms{Word}
\sqcup \{\bot\}\) from the buffer by induction to be:
\(
  (L;\{\bufloc{x} := v\})[x] = v,
\),
\(
  (L;\{\bufloc{y} := v\})[x] = L[x]
\),
\(
  [][x] = \bot
\).
Note that writes at the \textit{end} of the buffer are prioritized, since later writes overwrite
earlier ones! 

The semantics of reads and writes can then be given by
\begin{equation}
  \begin{aligned}
  \dnt{\ms{read}_x} = R_x^{\ms{TSO}} 
    &= \ms{pflush};(\lambda ()\;L. \{(v, L, \{x = v\}) \mid L[x] = v \lor L[x] = \bot\});\ms{pflush} 
  \\
  \dnt{\ms{write}_x} = W_x^{\ms{TSO}}
    &= \ms{pflush};(\lambda v\;L. \{(v, (L;\{\bufloc{x} := v\}), \{x := v\})\});\ms{pflush}
  \end{aligned}
\end{equation}
where buffer flushing is implemented via the morphism 
$$
  \ms{pflush}_A = \lambda a\;L. \{(a, R, \alpha) | L = \alpha;R\}
  : \ms{Set}_{\ms{TSO}}(A, A)
$$
(called \(\ms{split}\) in \cite{sparky}). To be able to perform synchronization, we will also need
to introduce a \(\ms{fence}: \mc{I}^\varnothing_0(\mb{1}, \mb{1})\) instruction, which simply causes
all actions before the fence to be observed before any actions after the fence. For \(\ms{TSO}\),
implementing this is as simple as flushing the buffer, i.e., we define
$$
  \dnt{\ms{fence}} = \lambda ()\;L. \{((), [], L;\{\delta\})\}
$$
We can now define the parallel composition of morphisms in a "fork-join" style as follows: we first
flush the buffer completely (i.e., taking it and sticking it at the beginning of our pomset), then
execute \(f\) and \(g\) in parallel with separate buffers, \textit{filtering out executions which
completely flush the buffer}. Since both threads end up with an empty buffer, the resulting joined
buffer is also empty, giving us the following definition:
\begin{equation}
  \begin{aligned}
    f_0 || f_1 = \lambda (a_0, a_1)\;L. 
    & \{\iota_0 ((b_0, b_1), [], L;(\alpha_0 || \alpha_1)) 
      \mid \iota_0 (b_i, [], \alpha_i) \in f_i\;a_i\;[]\} 
    \\ & \cup \{\iota_1 (\alpha_0 || \alpha_1) 
        \mid (\iota_0 (b_0, [], \alpha_0) \in f_0\;a \lor \iota_0\;\alpha_0 \in f_0\;a_0\;[]) 
        \land \iota_1 \alpha_1 \in f_1\;a_1\;[]\} 
    \\ & \cup \{\iota_1 (\alpha_0 || \alpha_1) 
        \mid \iota_1\;\alpha_0 \in f_0\;a_0\;[] 
        \land \iota_0 (a_1, [], \alpha_1) \in f_1\;a_1\;[]\}
    \end{aligned}
\end{equation}

One issue with this model is that our category currently has ``too many'' operations besides the
ones we want: for example, it is a perfectly valid morphism to simply add an event into the buffer
without any flushing. Of course, one thing we could do is consider the smallest Elgot subcategory
generated by the atomic instructions we want to support, such as \(\ms{read}_x\), \(\ms{write}_x\),
and \(\ms{fence}\), which can be a perfectly valid approach, but is somewhat unwieldy. Another
approach, however, is to think about what ``valid'' morphisms might look like in general; this
also has the benefit of letting us write down specifications that may not be trivially implementable
in terms of actual instructions.

There are many potential properties we might want ``valid'' morphisms to have; in general, the less
valid morphisms we admit, the more equations we have to work with when reasoning. One simple
property we might consider is that a ``valid'' morphism must have at least one execution in which
the buffer is completely flushed, regardless of initial state. This is a perfectly reasonable
property to impose, since it ensures that we do not encounter degenerate cases where, for example,
the parallel composition of two processes has an empty trace  set even though both processes have a
nonempty trace set (since neither empties the buffer). We might even want to go a little bit further
and, representing the fact that a buffer flush can happen at any time between instructions, require
that a nondeterministic buffer flush before or after an instruction does not change its semantics,
i.e., that $\ms{pflush}_A ; f ; \ms{pflush}_B = f$. Unfortunately, this does \emph{not} hold for
the identity, since
\begin{equation}
  \ms{pflush} ; \ms{id} ; \ms{pflush} = \ms{pflush} ; \ms{pflush} = \ms{pflush} \neq \ms{id}
\end{equation}
so this would not give us a subcategory. While one possible solution is to simply add the identity
back in (and this would give us a valid subcategory), we have the issue that, for example, morphisms
like the associator and unitor would still be excluded. Instead, we can take into account the fact
that all $\ms{pflush}_A$ are idempotent and define a new category $\ms{PTSO}$ with objects sets and
morphisms of the form
\begin{equation}
  \ms{PTSO}(A, B) = \{\ms{pflush}_A; f; \ms{pflush}_B \mid f \in \ms{Set}_{\ms{TSO}}(A, B)\} 
\end{equation}
In this category, $\ms{pflush}_A$ is the identity on $A$, since
\begin{equation}
  \begin{aligned}
    \ms{pflush}_A ; \ms{pflush}_A; f; \ms{pflush}_B &= \ms{pflush}_A; f; \ms{pflush}_B \\
    \ms{pflush}_A; f; \ms{pflush}_B; \ms{pflush}_B &= \ms{pflush}_A; f; \ms{pflush}_B
  \end{aligned}
\end{equation}
This is especially useful since it frees us from needing to remember to sprinkle $\ms{pflush}$
everywhere when specifying things (and, of course, when that level of control is desired, we can
specify in $\ms{Set}_{\ms{TSO}}$ and then lift to $\ms{PTSO}$).

It turns out that this is a general construction we can do, which gives us another useful way to
build up \isotopessa{} models that can often be quite difficult to express as simply monads.
In particular, we have the following definition
\begin{definition}[Idempotent envelope category]
  Let $\mc{C}$ be a category and $d$ be a family of morphisms for each $A \in |\mc{C}|$ such that
  each $d_A : \mc{C}(A, A)$ is idempotent, i.e., $d_A ; d_A = d_A$. We may then define the
  subsemicategory $\ms{Ide}(\mc{C}, d)$ of $\mc{C}$ to be given by
  \begin{equation}
    \ms{Ide}(\mc{C}, d)(A, B) 
      = \{f \in \mc{C}(A, B) \mid f = d_A ; f ; d_B\}
      = \{d_A ; f ; d_B \mid f \in \mc{C}(A, B)\}
  \end{equation}
  Note that $\ms{Ide}(\mc{C}, d)$ is \emph{not} a subcategory unless $d_A = \ms{id}_A$ for all $A$,
  in which case $\ms{Ide}(\mc{C}, d) = \mc{C}$. However, $\ms{Ide}(\mc{C}, d)$ can be
  \emph{independently} viewed as a category with identity $d_A$ at $A$.
\end{definition}
We have that, if $\mc{C}$ has coproducts and $d_{A + B} = d_A + d_B$, then $\ms{Ide}(\mc{C}, d)$
inherits coproducts from $\mc{C}$ with injections 
\begin{equation}
  \begin{aligned}
    d_A;\iota_l &= \iota_l;d_{A + B} &= d_A;\iota_l;d_{A + B} &\in \ms{Ide}(\mc{C}, d)(A, A + B) 
    \\
    d_B;\iota_r &= \iota_r;d_{A + B} &= d_B;\iota_r;d_{A + B} &\in \ms{Ide}(\mc{C}, d)(B, A + B)
  \end{aligned}
\end{equation}
since, for $f \in \ms{Ide}(\mc{C}, d)(A, C)$, $g \in \ms{Ide}(\mc{C}, d)(A, C)$, we have that
\begin{equation}
  d_{A + B} ; [f, g] ; d_C 
  = d_A + B ; [f, g] ; d_C
  = [d_A ; f ; d_C, d_B ; g ; d_C]
  = [f, g]
\end{equation}
In particular, it follows that $\ms{Ide}_d$ preserves coproducts. If $\mc{C}$ in addition has an
Elgot structure, then $\ms{Ide}(\mc{C}, d)$ also inherits this structure (and therefore $\ms{Ide}_d$
preserves it), as for $f \in \ms{Ide}(\mc{C}, d)(A, B + A)$, we have
\begin{equation}
  f^\dagger 
  = (d_A ; f ; d_B + d_A)^\dagger 
  = (d_A ; f ; d_B + d_A)^\dagger ; d_B
  = d_A ; (f ; d_B + (d_A ; d_A))^\dagger ; d_B
  = d_A ; f^\dagger ; d_B
\end{equation}
On the other hand, if $\mc{C}$ is equipped with a (symmetric) premonoidal structure and we have that
\begin{equation}
  d_{A \otimes B} = d_A \ltimes d_B = d_A \rtimes d_B 
  \qquad d_{A \otimes I} = d_A \otimes I 
  \qquad d_{I \otimes A} = I \otimes d_A
  \label{eqn:d-premonoidal}
\end{equation}
(note that $d$ does \emph{not} have to be central!), it follows that $\ms{Ide}(\mc{C},
d)$ has (symmetric) premonoidal structure as well, with inherited tensor product functors and
associators, unitors, and symmetries given by
\begin{equation}
  \begin{gathered}
    d_{(A \otimes B) \otimes C} ; \alpha_{A, B, C}
    = \alpha_{A, B, C} ; d_{A \otimes (B \otimes C)}
    = d_{(A \otimes B) \otimes C} ; \alpha_{A, B, C} ; d_{A \otimes (B \otimes C)} \\
    d_{A \otimes I} ; \lambda_A = \lambda_A ; d_A = d_{A \otimes I} ; \lambda_A ; d_A \qquad
    d_{I \otimes A} ; \rho_A = \rho_A ; d_A = d_{I \otimes A} ; \rho_A ; d_A \\
    d_{A \otimes B} ; \sigma_{A, B} 
      = \sigma_{A, B} ; d_{B \otimes A} 
      = d_{A \otimes B} ; \sigma_{A, B} ; d_{B \otimes A}
  \end{gathered}
\end{equation} 
respectively. It follows that if $\mc{C}$ is distributive, so is $\ms{Ide}(\mc{C}, d)$, since we can
verify that the distributor $[d_A \otimes (d_B ; \iota_l), d_A \otimes (d_C ; \iota_r)]$ has the
expected inverse
\begin{equation}
  d_{A \otimes (B + C)} ; \delta^{-1} ; d_{(A \otimes B) + (A \otimes C)}
\end{equation}
If we in fact have that $d_{A \otimes B} = A \otimes d_B = d_A \otimes B$ for all $A,
B$ (which implies Equation~\ref{eqn:d-premonoidal}) and $\mc{C}$ is equipped with the structure
of a Freyd category, then so is $\ms{Ide}(\mc{C}, d)$, with projections and diagonal given by
\begin{equation}
  \begin{gathered}
  d_{A \otimes B} ; \pi_l = \pi_l ; d_A = d_{A \otimes B} ; \pi_l ; d_A \qquad
  d_{A \otimes B} ; \pi_r = \pi_r ; d_B = d_{A \otimes B} ; \pi_r ; d_B \\
  d_A ; \Delta = \Delta ; d_{A \otimes A} = d_A ; \Delta ; d_{A \otimes A}
  \end{gathered}
\end{equation}
Hence, it follows that if $\mc{C}$ is an \isotopessa{} (expression) model, $d_A$ is idempotent,
$d_{A + B} = d_A + d_B$, and $d_{A \otimes B} = A \otimes d_B = d_A \otimes B$, then
$\ms{Ide}(\mc{C}, d)$ is an \isotopessa{} (expression) model as well. It turns out that all these
properties hold for $d = \ms{pflush}$, giving us our final categorical model for TSO weak memory,
$\ms{PTSO} = \ms{Ide}(\ms{Set}_{\ms{TSO}}, \ms{pflush})$.

\section{Discussion and Related Work}

\subsection{SSA, FP and IRs}

Static Single Assignment (SSA) form was first introduced as a compiler intermediate representation
by \citet{alpern-ssa-original-88} and \citet{rosen-gvn-1988}, with the goal of facilitating
effective reasoning about program variable equivalence. To perform optimizations like common
subexpression elimination (CSE) effectively, we need to determine which expressions are equal
\emph{at a given point in time}. By transforming the program into SSA form, we introduce unique
variables for each assignment and use $\phi$-nodes to merge variable values at points where control
paths converge. Since each variable then corresponds to a unique value at a specific point in the
program's execution, SSA unlocks the ability to perform \emph{algebraic reasoning} about variable
values over time. As a result, analyses like CSE become a matter of simple algebraic rewriting based
on variable names.

\citet{cytron-ssa-intro-91} provided the first efficient algorithm for converting 3-address code
programs to SSA form using a minimal number of $\phi$-nodes. They observed that a $\phi$-node only
needs to be introduced at the earliest point where a variable may have different values based on
control flow. This point is computable via the graph-theoretic notion of a \emph{dominance
frontier}, which identifies where control paths merge in the program's control flow graph. Since
then, SSA has become the intermediate representation of choice for most production-grade
compiler toolchains.

One of the most famous optimizations enabled by SSA's support for effective algebraic reasoning is
\emph{sparse conditional constant propagation (SCCP)}, introduced in \citet{wegman-sccp-91}. This
algorithm leverages the SSA property to reason about the possible set of values each variable may
take using a lattice-based data-flow analysis, which models variable values in terms of abstract
states like \emph{unknown}, \emph{constant}, or \emph{overdefined}. Without SSA, one could naively
consider all variable definitions, which would likely detect fewer constants due to imprecision.
Alternatively, achieving more precise results would require a complex and computationally intensive
reaching-definitions analysis.

However, while SSA is an excellent representation for implementing compilers, the semantics of
$\phi$-nodes can be quite unintuitive due to their lack of an obvious operational interpretation,
making them challenging to reason about formally. \citet{kelsey-95-cps} establishes a correspondence
between SSA and a subset of \emph{continuation-passing style (CPS)}, a common intermediate
representation for functional compilers. In particular, while not all CPS programs can be directly
converted to SSA form—those using non-local returns like \texttt{longjmp} or \texttt{call/cc}—the
typical outputs of CPS transformations avoid such features. Kelsey observed that many optimizations
requiring flow analysis in CPS could be performed directly in SSA, often dramatically simplifying
them.

\citet{appel-ssa} builds on this work by informally showing that the functional subset of CPS
programs ``hidden inside" SSA is, in fact, simply nested, mutually tail-recursive functions,
with each function corresponding to a basic block. He makes the key observation that the
dominance-based scoping of SSA corresponds to the lexical scoping of functions. For a variable to be
visible in a function, it must appear in the lexical scope of that variable's definition and
therefore be dominated by it; otherwise, there would be no way to call the function.

In fact, the subset of functional programs identified by \citet{appel-ssa} corresponds to
\emph{A-normal form (ANF)}, another functional intermediate representation advocated by
\citet{flanagan-93-anf}. \citet{chakravarty-functional-ssa-2003} formalizes this correspondence
giving an algorithm to convert SSA programs to ANF, and then showing how SSA optimizations such as
SCCP can be written to operate on ANF programs. The authors highlight that the semantic rigor of
their notation, combined with the well-defined semantics of ANF, make their presentation of SCCP
significantly more amenable to formal analysis.

Going in the other direction, \citet{thorin-12} introduce the Thorin intermediate representation,
which consists of CPS extended with SSA-style dominance-based scoping. As SSA is a first-order
language, it is often difficult to represent programs using closures effectively, and consequently,
difficult to optimize them well. The authors give the example of LLVM, which often needs to generate
a new struct and a large amount of boilerplate code for every closure, which, even after
optimization, is often not significantly reduced. By contrast, Thorin retains the advantages of CPS
for representing functional programs, while enabling SSA-like graph-based use-def analysis by
prohibiting variable shadowing.

The primary focus of the works we have covered so far is establishing the correspondence between
SSA, an imperative intermediate representation, and widely used functional intermediate
representations. However, to advance beyond these correspondences and directly study the
optimization and verification of SSA programs themselves, we require an equational theory
underpinned by formal semantics. One approach found in the literature (and which we also use) is
to relax SSA into forms that support richer substitution principles and well-defined operational
semantics, which can then be used to reason about SSA itself. Two papers which exemplify this
approach are \citet{benton-kennedy-99} and \citet{garbuzov-structural-cfg-2018}.

One of the difficulties of working with ANF as an intermediate representation is that it does not,
in general, satisfy \emph{substitution}: replacing a variable $x$ with a compound expression like a
function call can take a program out of the ANF-fragment. The \emph{monadic intermediate language
(MIL)} introduced by \citet{benton-kennedy-99} for use in the MLj compiler for Standard ML can be
viewed as a relaxation of ANF (and hence, of SSA) to get a nicer substitution principle. Benton and
Kennedy then use this flexibility to build up an equational theory justified by their operational
semantics. One interesting feature is that, unlike Moggi's equational metalanguage
\cite{moggi-91-monad} (and like our \isotopessa{} calculus), MIL enforces a stratification of
\emph{values} and \emph{computations}, without supporting ``computations of computations''
(corresponding to nested monad types $\ms{T}(\ms{T}(A)))$). This hints at Freyd categories, rather
than general Kleisli categories, being a natural model of MIL-like intermediate representations.

Similarly, \citet{garbuzov-structural-cfg-2018} exhibit a correspondence between an operational
semantics for SSA and an operational semantics for call-by-push-value (CBPV) \cite{cbpv}. They then
use the normal form bisimulations of \citet{lassen-bisim} to derive an equational theory for use in
justifying optimizations. In particular, we can view their paper as interpreting CBPV as a
relaxation of SSA more suited for developing an equational theory and for semantics work; in
particular, to be able to take advantage of CBPV to give a \emph{structural} operational semantics
for (unstructured) SSA programs. The semantics of CBPV are widely studied, and hint at a large
variety of potential models for SSA, but the formalization in \cite{garbuzov-structural-cfg-2018}
does not support any effects other than nontermination.

\subsection{Formalizations of SSA}

\subsubsection{Other SSA type systems}

Several attempts have been made to provide a type-theoretic treatment of SSA. The work most similar
to ours is by \citet{typed-effect-ssa-rigon-torrens-vasconcellos-20}, who present a typed
translation from SSA into the lambda calculus using a type-and-effect system, observing that the
algorithm for converting programs to SSA form may also be viewed as a mechanism for transforming
programs in a functional language with unstructured control flow into equivalent expressions.

Similarly to \isotopessa{}, they extend the lambda calculus with a mutually recursive
\texttt{where}-binding, which allows them to directly translate unstructured SSA control flow.
However, their calculus uses only a single syntactic category of expressions, whereas we attempt to
model (generalized) SSA directly by distinguishing between expressions and regions. Their language
also includes support for effect handlers, which are beyond the scope of our current study but
represent an interesting direction for future work. While the authors do not provide an equational
theory or a semantics for their language, we believe that our equational theory could be adapted to
the fragment of their language without effect handlers.

An interesting alternative approach is demonstrated by \citet{ssa-types-matsuno-ohori-06}, who give
a type theory for what appear to be ordinary three-address code programs. However, every well-typed
program can be placed into SSA-form by inserting $\phi$-nodes in a fully type-directed way. This
lets them model SSA without any $\phi$-nodes, letting them use the standard semantics for
three-address code. 

\citet{menon-verified-06} give a type-safe formalization of SSA, along with an operational semantics
and  formal definitions of dominance, definition/use points, and the SSA property for 3-address
code. By augmenting SSA with first-class proof variables, they aim to give a representation which
allows aggressive optimizations to preserve safety information. Their type system requires checking
the SSA property separately from well-typedness, but is proven sound if the SSA property holds.
\citet{hua-explicit-ssa-2010} give another type system and direct operational semantics for standard
SSA, and prove type safety for it.

Many operational semantics for SSA have arisen from compiler verification efforts.
\citet{barthe-compcert-ssa-2014} give an operational semantics as part of the CompCertSSA project,
and give a semantics-preserving translation from three-address code into SSA.
\citet{herklotz-gsa-2023} formalise "gated SSA" and give semantics-preserving translations between
it and ordinary SSA. Going beyond CompCertSSA, \citet{vellvm-12} have studied the semantics of the
LLVM IR itself as part of the Vellvm project. There has been much less work on denotational
semantics for SSA, or directly on its equational theory. \citet{pop-ssa-inout-2009} give an unusual
denotational model of SSA in terms of the iteration structure of a program, which they use to better
understand the loop-closing $\phi$-nodes found both in the gated SSA representation as well as
practical compilers such as GCC.

\subsubsection{Mechanizations of SSA}

CompcertSSA \cite{compcert-ssa-12} is an attempt to extend the CompCert verified compiler
\cite{leroy-compcert-09} with an SSA-based middle-end. This is achieved by generating a pruned SSA
IR from CompCert's RTL format. Instead of verifying the RTL-to-SSA translation within Coq, this
translation is performed by unverified code, and a separate \emph{translation validation} stage uses
a verified checker to ensure correctness. After performing (verified) SSA optimizations, the IR is
naively lowered back to RTL for the rest of CompCert's machinery to work on. \citet{demange-ssa-15} build on this work by providing realistic, verified implementations of Sparse Conditional Constant Propagation (SCCP) \cite{wegman-sccp-91}, Common Subexpression Elimination (CSE), and Global Value Numbering (GVN) \cite{rosen-gvn-1988} within a general framework of flow-insensitive static analysis for CompCertSSA.

While CompCert is a verified implementation of a new compiler in Coq, the Vellvm project
\cite{vellvm-12} attempts to mechanize a subset of the LLVM intermediate representation (IR),
covering the LLVM type system, operational semantics, and the well-formedness and structural
properties of valid LLVM IR. The authors adopt a memory model for LLVM based on CompCert's
\cite{leroy-compcert-09}, allowing them to leverage significant portions of CompCert's Coq
infrastructure. Similarly, \citet{siddharth-24-peephole} attempt to mechanize a well-defined subset
of the Multi-Level Intermediate Representation (MLIR) in the Lean 4 theorem prover. Their
formalization features a user-friendly front end to convert MLIR syntax into their calculus and
scaffolding for defining and verifying peephole rewrites using tactics. Their framework has been
tested on bitvector rewrites from LLVM, structured control flow, and fully homomorphic encryption;
however, as of publication, only structured control flow was fully supported.

\subsection{Compositional (Relaxed) Concurrency}

The most natural way to think about concurrency is often in terms of an operational semantics on an
abstract machine, in which concurrent threads interleave and interact. Such semantics, naturally,
are designed to reason about an entire program at a time. \citet{batty-compositional-17} argues that
a compositional semantics of concurrency is necessary to be able to reason about properties of large
software systems effectively, particularly in the presence of complicating factors such as compiler
optimizations and weak memory semantics. However, it is generally very challenging to reason about a
small component of a concurrent system in isolation since its behaviour may be drastically affected
by other components running in parallel.

One natural approach to constructing denotational models of concurrency is to consider the
extension of an abstract machine's behavior. Each thread performs a sequence of atomic actions,
and so from the outside, the machine can produce any interleaving of the atomic actions of each
thread. Just by itself, considering sets of possible traces is not sufficiently abstract, and so
\citet{brookes-full-abstraction-96} takes the closure of these sets under semantics-preserving
transformations, such as stuttering (introducing extra identity steps) and mumbling (which fuse
sequential atomic operations, thereby hiding implementation details), to obtain a model with good
equational properties such as associativity
($\dnt{\alpha;(\beta;\gamma)} = \dnt{(\alpha;\beta);\gamma}$) and the expected identities for
branches and loops. 

If the machine model is \emph{sequentially consistent} -- i.e., all allowed behaviors arise from
interleavings of sequences of atomic events -- then this style of trace semantics is sufficient.
However, real hardware often exhibits \emph{relaxed behaviour}, in which additional behaviours
which do not correspond to the interleaving of parallel threads are allowed. Even more relaxed
behaviours arise from fundamental compiler optimizations, such as re-ordering of independent reads
and writes. These are perfectly valid in a single-threaded context but can introduce new
behaviours in multithreaded programs — for example, another thread could distinguish the order of
writes. In general, we need tools to reason about interactions with a system that is not
\emph{linearizable} (i.e., in which concurrent operations cannot be reduced to interleaved atomic
operations on a single thread), of which actual hardware is only one example. There are two major
approaches to this problem.

One idea we might have is to augment traces with additional structure, which we can then quotient
away to maintain extensionality by using an appropriate closure operator. What's nice about this
method is that we can often use structures analogous to the additional state in the machine model
which leads to the relaxed behaviour in the first place. For example, in
\citet{jagadeesan-brookes-relaxed-12}, the authors extend \citet{brookes-full-abstraction-96} with
additional state corresponding to the contents of a thread-local buffer, and then take the closure
of their trace-set with respect to buffer operations such as nondeterministic flushing. This gives
them a model of TSO weak memory with good equational properties and a monadic structure, which
remains intuitive, since it has a connection to the original TSO model which can also be framed in
terms of the abstract machine having thread-local buffers. \citet{release-acquire} use this approach
to derive a trace-based semantics for release-acquire atomics inspired by their operational
semantics, showing the viability of this technique even for very complex memory models.

The idea of augmenting traces leads naturally to \emph{games}, which we can view one variant of as
sequential traces of moves between a \emph{proponent} and an \emph{opponent}. Game semantics were
first used with great success to give a semantics to \emph{sequential computation}; for example,
\citet{abramsky-algol-96} give an adequate denotational semantics for sequential Algol using
Hyland-Ong games \cite{hyland-ong-00}. \citet{ghica-08} observe that the sequentiality of Hyland-Ong
games corresponds to a highly constrained, deterministic form of interleaving between concurrent
processes. By generalizing away many of the rules of Hyland-Ong games, and in particular
\emph{alternation} (i.e., the proponent and the opponent must take turns), the authors obtain a form
of game-semantics for concurrent programs, which they use to build a fully abstract semantics for
\emph{fine-grained concurrency} (in contrast to \citet{brookes-full-abstraction-96}, which
implements \emph{coarse-grained} concurrency using the somewhat unrealistic \ms{await} primitive).

In general, generalizing traces to represent true concurrency leads us to another approach:
replacing sets of \emph{linear} traces with (sets of) \emph{partially ordered} structures that
directly represent the concurrency of their component operations. This approach directly captures
the idea that executing two operations concurrently (i.e., without specifying an order between them)
is fundamentally different from executing them in some particular order. The most basic such data
structure is a \emph{partially ordered multiset}, or \emph{pomset}, which we describe in
Section~\ref{ssec:tso} of our paper. Pomsets can naturally model sequential consistency, and, like
traces, can be augmented with additional structure to model relaxed memory models such as TSO weak
memory, as in \citet{sparky}.

One issue with this approach is that it can be very challenging to determine the appropriate
structures to augment pomsets with in order to model more advanced memory models. For example,
\citet{jagadeesan-pwp-20} introduce \emph{pomsets with preconditions} (PwP), which augment pomsets
with logical formulae, to model weak memory behaviors. However, \citet{leaky-semicolon} demonstrate
that sequential composition (the semicolon operator) in PwP is not associative. This lack of
associativity undermines many common program optimizations and breaks the monadic structure
necessary for compositional reasoning. To address this issue, \citet{leaky-semicolon} propose
\emph{pomsets with predicate transformers} (PwT), which restore the associativity of sequential
composition. Despite this improvement, their semantics still do not support loops, and developing a
monadic structure based on PwT that fully supports recursion and iterative constructs remains future
work. 

Another potential generalization of pomsets, inspired by game semantics, is the use of \emph{event
structures}, which introduce a \emph{conflict relation} to represent many potentially conflicting
executions within a single mathematical object. \citet{castellan-16} show how to use event
structures to represent concurrent executions with respect to a memory model. They study a simple
language supporting parallel composition of $n$ linear programs defined as lists of loads, stores,
and arithmetic operations, for which event structures provide a denotational semantics.
\citet{paviotti-modular-relaxed-dep-20} use the event structure approach to give denotational
semantics for relaxed memory concurrency in a more realistic language, including semantics for
branches and loops. However, they employ step indexing in their semantics, and as a result, loop
unrolling is only a refinement rather than an equation. Since event structures satisfy the axioms of
axiomatic domain theory~\cite{fiore-phd-94}, it should be possible to modify this semantics into one
where loop unrolling is an equation, at which point it would be a model of our calculus.

\subsection{Future Work}

\subsubsection{Substructural Types and Effects}

Currently, our treatment of effects is relatively primitive, in that, while we postulate a lattice
of effects, our equational theory only distinguishes between \emph{pure} and \emph{impure}
functions, the latter having effect $\bot$. \citet{fuhrmann-direct-1999} studies languages with
semantics given in an \emph{abstract Kleisli category}. In particular, he introduces the notion of
\begin{itemize}
  \item \emph{Central} operations, which commute with all other operations; we like to call such
  morphisms \emph{linear}
  \item \emph{Duplicable} operations $f$, for which $\letexpr{y}{f(x)}{(y, y)} \teqv (f(x), f(x))$
  \item \emph{Discardable} operations $f$, for which $\letexpr{y}{f(x)}{e} \teqv e$ when $e$ does
  not depend on $y$
\end{itemize}
While, as we make use of in our calculus, pure operations are central, duplicable, and discardable,
it can be useful to study each notion in isolation, as well as to consider morphisms which are
central and duplicable (which we have taken to calling \emph{relevant}) and morphisms which are
both central and discardable (which we have taken to calling \emph{affine}). Examples of morphisms
which are not pure but nonetheless relevant or affine abound; for example,
\begin{itemize}
  \item In a programming language with nondeterminism, nondeterministic functions are \emph{impure}
  (since $\letexpr{y}{f(x)}{(y, y)}$ is a strict refinement of $(f(x), f(x))$), but are nontheless
  \emph{affine}, because discarding the result of a nondeterministic function does not affect the
  program's behavior, assuming they have no other effect.
  \item In a programming language with nontermination, nonterminating functions with no other effect
  are impure (since $\letexpr{y}{f(x)}{e}$ may diverge even if $e$ does not) but are always
  duplicable, because duplicating a nonterminating function does not introduce new effects beyond
  divergence. Depending on whether the language's other effects commute with nontermination, they
  may in fact be relevant.
\end{itemize}
All three of these concepts still make perfect sense when interpreted as-is in a Freyd category, and
are particularly useful when reasoning about models of probabilistic programming, such as Markov
categories (which have many affine morphisms) \cite{nlab:markov-category}.

Given that such relevant and affine side-effects have substitution principles which depend on how
often a variable is used in a term, we call such side-effects \emph{substructural}. We might also
consider how we could support substructural, and in particular linear, \emph{types}. Categorically,
this corresponds to weakening the requirement that our base category be cartesian, instead requiring
only a symmetric monoidal category. The resulting generalization of a Freyd category is called an
\emph{effectful category} by \citet{promonad}.

\subsubsection{Refinement and Enrichment}

Throughout this work, we have focused on the notion of program \emph{equivalence}
$\lbeq{\Gamma}{r}{r'}{\ms{L}}$. In many settings, especially when considering the nondeterminism and
concurrency, we would also like to study program \emph{refinement} $\haslb{\Gamma}{r \sqsubseteq
r'}{\ms{L}}$. Semantically, this corresponds to an enrichment of our model in the category of
partial orders; we conjecture that adding this to our semantics would require minimal changes. We
may also consider how our syntax could be extended to support explicit parallelism and how we could
reflect the corresponding algebraic laws (as given in, e.g., \citet{hoare-parallel-14}), such as
\begin{equation}
  (P \parallel Q) ; (R \parallel S) \sqsubseteq (P ; R) \parallel (Q ; S)
\end{equation}
in our equational theory. More generally, we can consider enrichment with further structures, such
as distributive lattices and dcpos, and the language features these correspond to, such as
nondeterministic choice between programs $P \lor C$. These features are particularly interesting
from the perspective of programs as \emph{specifications}, as they offer the potential for
representing complex specifications within SSA.

\subsubsection{Guarded Iteration}

Elgot categories and monads were introduced in \citet{elgot-elgot-75}, and subsequently formalized
in \citet{adamek-elgot-11}. \citet{coinductive-resumption-levy-goncharov-19} generalize Elgot
categories to \emph{guarded Elgot categories} (and, correspondingly, \emph{guarded Elgot monads},
whose Kleisli categories are guarded Elgot), to support partial non-unique iteration. Generalizing
our language with support for guarded iteration could allow us to study applications of SSA to
domains for which not all fixpoints exist, such as the representation of terminating languages (as,
for example, one would find in a proof assistant) or productive processes. Another potential
advantage of such support is that every category with coproducts can be considered \emph{vacuously}
guarded, with the only fixpoints supported being those of loops which immediately exit. This could
allow us to unify the notion of \isotopessa{} expression models and \isotopessa{} models into the
single notion of \emph{guarded} \isotopessa{} models (with expression models vacuously guarded,
while \isotopessa{} models in the original sense would have all fixpoints).

\citet{goncharov-metalang-21} introduce a metalanguage for guarded iteration based on labelled
iteration as presented by \citet{geron-iteration-16}. Interestingly, they refer to the labels in
their systems as \emph{exceptions}, making the claim that this more closely matches their semantics.
Similar to our work, they provide a denotational semantics for this language using the Kleisli
category of a strong guarded \emph{pre-iterative} monad $\mb{T}$—that is, a monad with a guarded
fixpoint operator—on a distributive category $\mc{C}$. They then demonstrate that this semantics is
sound and adequate with respect to a big-step operational semantics for a specific monad, namely
$\ms{T};X = X \times \mathbb{N}^ + \mathbb{N}^\omega$ over $\ms{Set}$. Notably, unlike our trace
monad, their monad is \emph{guarded iterative} rather than merely iterative. The infinite branch of
the coproduct $\mathbb{N}^\omega$ represents an infinite sequence of natural numbers, effectively
prohibiting non-productive infinite loops — that is, recursion that does not pass through a
\emph{guard}, which in this context is emitting a natural number. Their treatment of iteration is
thus somewhat more general than ours, as they support guarded iteration and require only a
pre-iterative monad. However, they focus on Kleisli categories rather than providing a general
treatment for Freyd categories. Additionally, their language supports functional types, whereas our
language is purely first-order. We are very interested in exploring the connection between guarded
labelled-iteration and SSA implied by the similarities between our respective syntaxes and
semantics.

\subsection*{Acknowledgements}

This work was supported in part by a European Research Council (ERC) Consolidator Grant for the project ``TypeFoundry'', funded under the European Union's Horizon 2020 Framework Programme (grant agreement no. 101002277).

\bibliographystyle{ACM-Reference-Format}
\bibliography{references}

\clearpage 

\appendix

\section{The Environment Comonad}

\label{apx:environment}

If $\mc{C}$ is a Freyd category, then given an object $R \in |\mc{C}|$, the functor $R \otimes
\cdot$ is a comonad with counit $\pi_r : R \otimes A \to A$ and comultiplication $\dmor{R} \otimes A
; \alpha : R \otimes A \to R \otimes (R \otimes A)$, often called the \emph{environment comonad} or
\emph{coreader comonad}. In $\ms{Set}$, the co-Kleisli category of $R \otimes \cdot$ is isomorphic
to the Kleisli category of the reader monad $R \to \cdot$, and thus can be equipped with the
structure of a distributive Elgot category. We might therefore intuit that, in general, if $\mc{C}$
is a distributive Elgot category, the co-Kleisli category of $R \otimes \cdot$ has this structure,
even if e.g. $\mc{C}$ lacks exponentials. The rest of this section is dedicated to proving this,
which, beyond providing yet another class of \isotopessa{} models, will also give us a some useful
equations and notation to use in our proofs later in the appendix.
We begin by giving an explicit definition:
\begin{definition}[Environment Comonad]
  Given a Freyd category $\mc{C}$ and an object (the \emph{environment}) $R \in |\mc{C}|$, the
  \emph{environment comonad} $\envcom{\mc{C}}{R}$ is given by the functor $A \mapsto R \otimes A$
  and has counit $\pi_r$ and comultiplication $\Delta_R \otimes A$. It follows that composition of
  co-Kleisli morphisms $f : R \otimes A \to B$, $g : R \otimes B \to C$ is given by
  \begin{equation}
    \rseq{R}{f}{g} := \dmor{R} ; \alpha ; R \otimes f ; g : R \otimes A \to C
  \end{equation}
  In particular, we can define an identity on objects functor $\ms{env}_R : \mc{C} \to
  \envcom{\mc{C}}{R}$ as follows:
  \begin{equation}
    \toenv{R}{f} := \pi_r ; f
  \end{equation}
  Where it is clear from context, we will leave out the environment $R$.
\end{definition}
Given $f : R \otimes A \to B$, we will introduce the syntax sugar
\begin{equation}
  \rlmor{f} := \dmor{R} ; \alpha ; R \otimes f : R \otimes A \to R \otimes B
\end{equation}
In particular, we have that $\rseq{}{f}{g} = \rlmor{f} ; g$ and $\rlmor{\toenv{}{f}} := R \otimes
f$. It is trivial to verify (for example, by drawing string diagrams) that
\begin{equation}
  \begin{gathered}
  \rlmor{\rlmor{f} ; g} = \rlmor{f} ; \rlmor{g} 
  \implies \rseq{}{f}{(\rseq{}{g}{h})} = \rseq{}{(\rseq{}{f}{g})}{h} \\
  \rlmor{f} ; \pi_r = f \implies \rseq{}{f}{\pi_r} = f \qquad
  \rlmor{\pi_r} = \ms{id} \implies \rseq{}{\pi_r}{f} = f
\end{gathered}
\end{equation}
and hence that $\envcom{\mc{C}}{R}$ is indeed a category. To show it is in fact a \emph{Freyd
category}, we can define tensor functors
\begin{equation}
  \begin{aligned}
  \envtn{R}{f}{X} &:= \alpha^{-1} ; f \otimes X 
    : R \otimes (A \otimes X) \to B \otimes X \\
  \envtn{R}{X}{f} &:= R \otimes \sigma ; \envtn{R}{f}{X} ; \sigma
    : R \otimes (X \otimes A) \to X \otimes B
  \end{aligned}
\end{equation}
We can define
\begin{itemize}
  \item Associators $\alpha^R_{A, B, C} := \toenv{R}{\alpha_{A, B, C}}$
  \item Unitors $\lambda^R_A := \toenv{R}{\lambda_A}$ and $\rho^R_A := \toenv{R}{\rho_A}$
  \item Symmetries $\sigma^R_{A, B} := \toenv{R}{\sigma_{A, B}}$
  \item Projections $\pi_l^R := \toenv{R}{\pi_l}$ and $\pi_r^R := \toenv{R}{\pi_r}$
  \item Terminal morphisms $!^R_A := \toenv{R}{!_A} = !_{R \otimes A}$
  \item Diagonals $\Delta^R_A := \toenv{R}{\Delta_A}$
  \item Pure morphisms ${\envcom{\mc{C}}{R}}_\bot(A, B) = \mc{C}_{\bot}(R \otimes A, B)$
\end{itemize}
\begin{lemma}
  We can always equip $\envcom{\mc{C}}{R}$ with the structure of a Freyd category as described above
  ; furthermore, $\toenv{R}{\cdot}$ strictly preserves premonoidal structure.
\end{lemma}
\begin{proof}
  We have that
  \begin{equation}
    \begin{aligned}
    \toenv{R}{f \otimes A} &= \pi_r ; f \otimes A = \alpha^{-1} ; (\pi_r ; f) \otimes A 
                            = \toenv{R}{f} \otimes_R A \\
    \toenv{R}{A \otimes f} &= \pi_r ; A \otimes f 
                            = R \otimes \sigma ; \alpha^{-1} ; (\pi_r ; f) \otimes A ; \sigma
                            = A \otimes_R \toenv{R}{f}
    \end{aligned} 
  \end{equation}
  and hence that $\ms{env}_R$ preserves products. Since $\ms{env}_R$ is a functor, to check we
  indeed have a Freyd category, it suffices to show that:
  \begin{itemize}
    \item $\alpha^R_{A, B, C}$ is natural: we have that, given 
    $f: R \otimes A \to A'$, $g : R \otimes B \to B'$, $h : R \otimes C \to C'$
    \begin{align}
      \rseq{}{(f \otimes_R B) \otimes_R C}{\alpha_{A', B, C}} 
        &= \rseq{}{\alpha_{A, B, C}}{f \otimes_R (B \otimes C)} \label{eqn:env-a-nat-1} \\
      \rseq{}{(A \otimes_R g) \otimes C}{\alpha_{A, B', C}}
        &= \rseq{}{\alpha_{A, B, C}}{A \otimes_R (g \otimes_R C)} \label{eqn:env-a-nat-2} \\
      \rseq{}{(A \otimes B) \otimes_R h}{\alpha_{A, B, C'}}
        &= \rseq{}{\alpha_{A, B, C}}{A \otimes_R (B \otimes_R h)} \label{eqn:env-a-nat-3} \\
    \end{align}
    since, for each equation, both sides are equal to the same string diagram in
    Figure~\ref{fig:env-a-nat-1}, \ref{fig:env-a-nat-2}, and \ref{fig:env-a-nat-3}
    respectively.
    \item $\lambda^R_{A}$ is natural: we have that, for $f : R \otimes A \to A'$,
    \begin{equation}
      \begin{aligned}
      \rseq{}{f \otimes_R \mb{1}}{\lambda^R_{A'}}
      &= \dmor{R} ; \alpha ; R \otimes (\alpha^{-1} ; f \otimes \mb{1}) ; \pi_r ; \pi_l \\
      &= R \otimes \pi_l ; f \\
      &= \dmor{R} ; \alpha ; R \otimes (\pi_r ; \pi_l) ; f
      &= \rseq{}{\lambda^R_{A}}{f}
      \end{aligned}
    \end{equation}
    In general, we note that, for all $g : R \otimes A \to A' \otimes \mb{1}$, 
    $\rseq{}{g}{\lambda^R} = g ; \pi_l$.
    \item $\rho^R_{A}$ is natural:
    \begin{equation}
      \begin{aligned}
      \rseq{}{\mb{1} \otimes_R f}{\rho^R_{A'}}
      & = \dmor{R} \otimes (\mb{1} \otimes A) ; \alpha 
        ; R \otimes (R \otimes \sigma ; \alpha^{-1} ; f \otimes \mb{1} ; \sigma) 
        ; \pi_r ; \pi_r \\
      & = R \otimes \pi_r ; f \\
      & = \dmor{R} ; \alpha ; R \otimes (\pi_r ; \pi_r) ; f
      & = \rseq{}{\rho^R_{A}}{f}
      \end{aligned}
    \end{equation}
    In general, we note that, for all $g : R \otimes A \to \mb{1} \otimes A'$, 
    $\rseq{}{g}{\rho^R} = g ; \pi_r$.
    \item $\sigma^R_{A, B}$ is natural: given $f : R \otimes A \to A'$ and $g : R \otimes B \to B'$, 
    we have that
    \begin{equation}
      \begin{aligned}
      \rseq{}{f \otimes_R B}{\sigma^R_{A', B}}  
        & = \dmor{R} \otimes (A \otimes B) ; \alpha 
          ; R \otimes (\alpha^{-1} ; f \otimes B) ; \pi_r ; \sigma \\
        & = \alpha^{-1} ; f \otimes B ; \sigma \\
        & = R \otimes \sigma ; R \otimes \sigma ; f \otimes B ; \sigma \\
        & = \dmor{R} \otimes (B \otimes A) ; R \otimes (\pi_r ; \sigma) 
          ; R \otimes \sigma ; \alpha^{-1} ; f \otimes B ; \sigma
        & = \rseq{}{\sigma^R_{A, B}}{B \otimes_R f}
      \end{aligned}
    \end{equation}
    and
    \begin{equation}
      \begin{aligned}
      \rseq{}{A \otimes g}{\sigma^R_{A, B'}} 
        &= \dmor{R} \otimes (A \otimes B) ; \alpha 
        ; R \otimes (R \otimes \sigma ; \alpha^{-1} ; g \otimes A ; \sigma)
        ; \pi_r ; \sigma \\
        &= R \otimes \sigma ; \alpha^{-1} ; g \otimes A ; \cancel{\sigma ; \sigma} \\
        &= \dmor{R} \otimes (A \otimes B) ; \alpha ; R \otimes (\pi_r ; \sigma) 
          ; \alpha^{-1} ; g \otimes A
        &= \rseq{}{\sigma^R_{A, B}}{g \otimes_R A}
      \end{aligned}
    \end{equation}
    \item Pure morphisms are central: given $f : R \otimes A \to A'$ and $g : R \otimes B \to B'$
    pure, we have that
    \begin{equation}
      \begin{aligned}
      \rseq{}{f \otimes_R B}{A' \otimes_R g}
      & = \Delta_R ; \alpha ; R \otimes (\alpha^{-1} ; f \otimes B ; \sigma) 
        ; \alpha^{-1} ; g \otimes A' ; \sigma \\
      & = \Delta_R ; \alpha ; R \otimes (R \otimes \sigma ; \alpha^{-1} ; g \otimes A ; \sigma) 
      ; \alpha^{-1} ; f \otimes B' \\
      & = \rseq{}{A \otimes_R g}{f \otimes_R B'}
      \end{aligned}
    \end{equation}
    since both sides correspond to the diagram in Figure~\ref{fig:env-slide}. We will write this as $f
    \otimes_R g$.
    \item $!^R_A$ is terminal for pure morphisms: yes, since $!^R_A$ is just $!_{R \otimes A}$ which
    is terminal in $\mc{C}$
    \item $\Delta^R_A$ duplicates pure morphisms: we have
    \begin{equation}
      \begin{aligned}
      \rseq{}{f}{\Delta^R_B} 
        & = \dmor{R} ; \alpha ; R \otimes f ; \pi_r ; \dmor{A'} \\
        & = f ; \dmor{A'} = \dmor{R \otimes A} ; f \otimes f \\
        & = \dmor{R} ; \alpha ; R \otimes (\alpha^{-1} ; f \otimes A') 
          ; R \otimes \sigma ; \alpha^{-1} ; f \otimes A' ; \sigma
        & = \rseq{}{\Delta^R_A}{f \otimes_R f}
      \end{aligned}
    \end{equation}
  \end{itemize}
\end{proof}

\begin{figure}
  \begin{subfigure}{0.3\textwidth}
    \begin{tikzpicture}
      \node[] (E) at (0, 0) {};
      \node[] (R) at (0.5, 0) {$R$};
      \node[] (A) at (1, 0) {$A$};
      \node[] (B) at (1.5, 0) {$B$};
      \node[] (C) at (2, 0) {$C$};
      \node[box=3/0/2/0] (f) at (0.75, -1.75) {\quad f \quad};
      \node[] (E') at (0.5, -3.5) {};
      \node[] (A') at (1.5, -3.5) {$A'$};
      \node[] (B') at (2, -3.5) {$B$};
      \node[] (C') at (2.5, -3.5) {$C$};
      \wires{
        R = { south = f.north.2 },
        A = { south = f.north.3 },
        B = { south = B'.north },
        C = { south = C'.north },
        f = { south.2 = A'.north },
      }{}
      \wires[red]{
        E = { south = f.north.1 },
        f = { south.1 = E'.north }
      }{}
    \end{tikzpicture}
    \caption{Equation~\ref{eqn:env-a-nat-1}}
    \label{fig:env-a-nat-1}
  \end{subfigure}%
  \begin{subfigure}{0.3\textwidth}
    \begin{tikzpicture}
      \node[] (E) at (0.5, 0) {};
      \node[] (R) at (1, 0) {$R$};
      \node[] (A) at (1.5, 0) {$A$};
      \node[] (B) at (2, 0) {$B$};
      \node[] (C) at (2.5, 0) {$C$};
      \node[box=3/0/2/0] (g) at (2, -2) {\quad g \quad};
      \node[] (E') at (0.5, -3.5) {};
      \node[] (A') at (1, -3.5) {$A$};
      \node[] (B') at (2, -3.5) {$B'$};
      \node[] (C') at (3, -3.5) {$C$};
      \wires{
        R = { south = g.north.2 },
        A = { south = A'.north },
        B = { south = g.north.3 },
        C = { south = C'.north },
        g = { south.2 = B'.north },
      }{}
      \wires[red]{
        E = { south = g.north.1 },
        g = { south.1 = E'.north }
      }{}
    \end{tikzpicture}
    \caption{Equation~\ref{eqn:env-a-nat-2}}
    \label{fig:env-a-nat-2}
  \end{subfigure}%
  \begin{subfigure}{0.3\textwidth}
    \begin{tikzpicture}
      \node[] (E) at (0.25, 0) {};
      \node[] (R) at (1, 0) {$R$};
      \node[] (A) at (1.5, 0) {$A$};
      \node[] (B) at (2, 0) {$B$};
      \node[] (C) at (2.5, 0) {$C$};
      \node[box=3/0/2/0] (h) at (2.5, -2) {\quad h \quad};
      \node[] (E') at (0.5, -3.5) {};
      \node[] (A') at (1, -3.5) {$A$};
      \node[] (B') at (1.5, -3.5) {$B$};
      \node[] (C') at (2, -3.5) {$C'$};
      \wires{
        R = { south = h.north.2 },
        A = { south = A'.north },
        B = { south = B'.north },
        C = { south = h.north.3 },
        h = { south.2 = C'.north },
      }{}
      \wires[red]{
        E = { south = h.north.1 },
        h = { south.1 = E'.north }
      }{}
    \end{tikzpicture}
    \caption{Equation~\ref{eqn:env-a-nat-3}}
    \label{fig:env-a-nat-3}
  \end{subfigure}
  \caption{Naturality of the associator in the co-Kleisli category of the environment comonad}
  \Description{}
  \label{fig:env-a-nat}
\end{figure} %
\begin{figure}
  \begin{tikzpicture}
    \node[] (R) at (-0.5, 0) {$R$};
    \node[dot] (dR) at (0, -1) {};
    \node[] (A) at (1, 0) {$A$};
    \node[] (B) at (2, 0) {$B$};
    \node[box=2/0/1/0] (f) at (1, -1.75) {\quad f \quad};
    \node[box=2/0/1/0] (g) at (2, -3) {\quad g \quad};
    \node[] (A') at (0, -4) {$A'$};
    \node[] (B') at (2, -4) {$B'$};
    \coordinate[] (cont) at (0.5, -2.5) {};
    \wires{
      R = { south = dR.north },
      dR = { east = f.north.1, west = cont.west },
      cont = { east = g.north.1 },
      A  = { south = f.north.2 },
      B = { south = g.north.2 },
      f  = { south = A' },
      g  = { south = B' },
    }{}

    \node[] (eq) at (3, -2) {$=$};

    \node[] (R2) at (3.5, 0) {$R$};
    \node[dot] (dR2) at (4, -1) {};
    \node[] (A2) at (5, 0) {$A$};
    \node[] (B2) at (6, 0) {$B$};
    \node[box=2/0/1/0] (f2) at (4.5, -2.5) {\quad f \quad};
    \node[box=2/0/1/0] (g2) at (6, -1.75) {\quad g \quad};
    \node[] (A2') at (4, -4) {$A'$};
    \node[] (B2') at (5, -4) {$B'$};
    \wires{
      R2 = { south = dR2.north },
      dR2 = { west = f2.north.1, east = g2.north.1 },
      A2  = { south = f2.north.2 },
      B2 = { south = g2.north.2 },
      f2  = { south = A2' },
      g2  = { south = B2' },
    }{}
  \end{tikzpicture}
  \caption{Centrality of pure morphisms in the environment comonad's co-Kleisli category}
  \Description{}
  \label{fig:env-slide}
\end{figure} %
Now, assume $\mc{C}$ is distributive. We wish to show that $A + B$ is a coproduct in
$\envcom{\mc{C}}{R}$, and furthermore that $\envcom{\mc{C}}{R}$ is distributive; note that even the
former may not be the case without distributivity! We proceed as follows:
\begin{lemma}
  If $\mc{C}$ is a distributive Freyd category, then $\envcom{\mc{C}}{R}$ is also distributive Freyd
\end{lemma}
\begin{proof}
  Given $f : R \otimes A \to B$ and $g : R \otimes A \to C$, we may define the coproduct and
  injections
  \begin{equation}
    \envcop{R}{f}{g} := \delta^{-1} ; [f, g] : R \otimes A \to B + C \qquad
    \envinl{R} := \pi_r ; \iota_r : R \otimes A \to A + B \qquad
    \envinr{R} := \pi_r ; \iota_r : R \otimes B \to A + B
  \end{equation}
  To verify this is indeed a coproduct, we can check that
  \begin{equation}
    \begin{aligned}
    \rseq{}{\envinl{R}}{\envcop{R}{f}{g}} 
      &= \dmor{R} ; \alpha ; R \otimes (\pi_r ; \iota_r) ; \delta^{-1} ; [f, g]
      = R \otimes \iota_r ; \delta^{-1} ; [f, g]
      = f \\
    \rseq{}{\envinr{R}}{\envcop{R}{f}{g}} 
    &= \dmor{R} ; \alpha ; R \otimes (\pi_r ; \iota_l) ; \delta^{-1} ; [f, g]
    = R \otimes \iota_l ; \delta^{-1} ; [f, g]
    = g
    \end{aligned}
  \end{equation}
  This morphism is obviously unique, since we have for all $h : R \otimes (A + B) \to C$
  \begin{equation}
    \envcop{R}{\rseq{}{\envinl{R}}{h}}{\rseq{}{\envinr{R}}{h}}
    = \delta^{-1} 
    ; [\dmor{R} ; \alpha ; R \otimes \iota_l ; \pi_r ; h,
      \dmor{R} ; \alpha ; R \otimes \iota_r ; \pi_r ; h]
    = \cancel{\delta^{-1} ; [R \otimes \iota_r, R \otimes \iota_l]} ; h
  \end{equation}
  To show that $\envcom{R}{\mc{C}}$ is indeed distributive, we need $\delta^{R-1} :=
  \toenv{R}{\delta^{-1}} = \pi_r ; \delta^R$ to be the inverse to $\delta^R := \envcop{R}{A \otimes
  \envinl{R}}{A \otimes \envinr{R}}$, which can easily be derived from the functoriality of
  $\toenv{R}{\cdot}$, since
  \begin{equation}
    \toenv{R}{\delta} = \pi_r ; [\iota_l ; \iota_r] 
                      = \delta^{-1} ; [\pi_r ; \iota_l, \pi_r \iota_r]
                      = \delta^R
  \end{equation}
\end{proof}
Note in particular that this allows us to define sums
\begin{equation}
  \begin{aligned}
    f +_R g 
    & = \envcop{R}{\rseq{}{f}{\envinl{R}}}{\rseq{}{g}{\envinr{R}}} \\
    & = \delta^{-1} ; [
      \dmor{R} \otimes (A + B) ; \alpha ; R \otimes f ; \pi_r ; \iota_l,
      \dmor{R} \otimes (A + B) ; \alpha ; R \otimes g ; \pi_r ; \iota_r
    ]
    & = \delta^{-1} ; f + g
  \end{aligned}
\end{equation}
Our final task is to show that, if $\mc{C}$ is a strong Elgot category, then so is
$\envcom{R}{\mc{C}}$, and, in particular, with fixpoint operator $\rfix{f}$. Note that, just like we
needed distributivity to have coproducts at all, we will need strength to have an Elgot structure.
We begin by stating some generally useful properties of $\rcase{f}$ and $\rfix{f}$. In
particular, we have that: For $f: R \otimes A \to B + C$:
\begin{itemize}
  \item Given $h : R \otimes X \to A$, we have
  \begin{equation}
    \begin{aligned}
    \rlmor{h} ; \rcase{f} 
    & = \dmor{R} \otimes X ; \alpha ; R \otimes h 
      ; \dmor{R} \otimes A ; \alpha ; R \otimes f ; \delta^{-1} \\
    & = \dmor{R} \otimes X ; \alpha 
      ; R \otimes (\dmor{R} \otimes X ; \alpha ; R \otimes h ; f) ; \delta^{-1} \\
    & = \rcase{\rlmor{h} ; f} & = \rcase{\rseq{}{h}{f}}
    \end{aligned}
  \end{equation}
  \item Given $g : R \otimes B \to X$, $h : R \otimes C \to Y$, we have
  \begin{equation}
    \begin{aligned}
    \rcase{\rseq{}{f}{g +_R h}}
    & = \rlmor{\rseq{}{f}{g +_R h}} ; \delta^{-1} \\
    & = \rlmor{f} ; \rlmor{g +_R h} ; \delta^{-1} \\
    & = \rlmor{f} 
      ; \dmor{R} \otimes (B + C) ; \alpha ; R \otimes (\delta^{-1} ; g + h) ; \delta^{-1} \\
    & = \rlmor{f} 
      ; \dmor{R} \otimes (B + C) ; \alpha ; R \otimes \delta^{-1} ; \delta^{-1} 
      ; (R \otimes g) + (R \otimes h) \\
    & = \rlmor{f} ; \delta^{-1} ; \rlmor{g} + \rlmor{h} \\
    & = \rcase{f} ; \rlmor{g} + \rlmor{h} \\
    & = \rlmor{f} ; \rlmor{g} +_R \rlmor{h}
    \end{aligned}
  \end{equation}
\end{itemize}
Similarly, for $f : R \otimes A \to B + A$:
\begin{itemize}
  \item Given $h : R' \to_\bot R$, $h \otimes A ; \rfix{f} = \rfix{h \otimes A ; f}$: we have that
  \begin{equation}
    \begin{aligned}
    h \otimes A ; \rcase{f} 
      & = h \otimes A ; \dmor{R} \otimes A ; \alpha ; R \otimes f ; \delta^{-1} \\
      & = \dmor{R'} \otimes A ; \alpha ; h \otimes (h \otimes A ; f) ; \delta^{-1} \\
      & = \dmor{R'} \otimes A ; \alpha ; R \otimes (h \otimes A ; f) ; \delta^{-1} 
        ; h \otimes B + h \otimes A \\
      & = \rcase{h \otimes A ; f} ; h \otimes B + h \otimes A
    \end{aligned}
  \end{equation}
  and hence by uniformity that
  \begin{equation}
    \begin{aligned}
    h \otimes A ; \rfix{f} 
      & = h \otimes A ; (\rcase{f})^\dagger ; \pi_r \\
      & = (\rcase{h \otimes A ; f} ; h \otimes B + R' \otimes A)^\dagger ; \pi_r \\
      & = (\rcase{h \otimes A ; f})^\dagger ; h \otimes B ; \pi_r 
      & = \rfix{h \otimes A ; f}
    \end{aligned}
  \end{equation}
  \item $\rlmor{\rfix{f}} = (\rcase{f})^\dagger$: we have that
  \begin{equation}
    \begin{aligned}
      & (\dmor{R} \otimes A ; \alpha)
        ; (R \otimes \rcase{f} 
        ; \delta^{-1} 
        ; R \otimes \pi_r + R \otimes (R \otimes A)) \\
      & = \dmor{R} \otimes A ; \alpha
        ; R \otimes (\lmor{f} ; \pi_l \otimes (B + A) ; \delta^{-1}) 
        ; \delta^{-1} 
        ; R \otimes \pi_r + R \otimes (R \otimes A) \\
      & = \dmor{R} \otimes A ; \alpha
        ; R \otimes f
        ; \dmor{R} \otimes (B + A) ; \alpha 
        ; R \otimes \delta^{-1} ; \delta^{-1}
        ; R \otimes \pi_r + R \otimes (R \otimes A) \\
      & = \dmor{R} \otimes A ; \alpha
        ; R \otimes f
        ; \delta^{-1}
        ; (\dmor{R} \otimes A ; \alpha ; R \otimes \pi_r) + (\dmor{R} \otimes A ; \alpha) \\
      & = \dmor{R \otimes A}
        ; (R \otimes A) \otimes f
        ; \pi_l \otimes (B + A)
        ; \delta^{-1}
        ; R \otimes A + (\dmor{R} \otimes A ; \alpha) \\
      &= \rcase{f} ; R \otimes A + (\dmor{R} \otimes A ; \alpha)
    \end{aligned}
  \end{equation}
  It follows by uniformity and strength that
  \begin{equation}
    \begin{aligned}
    \rlmor{\rfix{f}} = \lmor{\rfix{f}} ; \pi_l \otimes B 
      &= \dmor{R \otimes A} ; (R \otimes A) \otimes \rfix{f} ; \pi_l \otimes B \\
      &= \dmor{R} \otimes A ; \alpha ; R \otimes \rfix{f} \\
      &= \dmor{R} \otimes A ; \alpha ; R \otimes ((\rcase{f})^\dagger ; \pi_r) \\
      &= \dmor{R} \otimes A ; \alpha
        ; (R \otimes \rcase{f} ; \delta^{-1} ; R \otimes \pi_r + R \otimes (R \otimes A))^\dagger \\
      &= (\rcase{f})^\dagger
    \end{aligned}
  \end{equation}
  as desired.
\end{itemize}
We also state some generally useful properties of $\rlmor{f}$:
\begin{itemize}
  \item For all $f : R \otimes A \to C$, $g : R \otimes B \to C$, 
  \begin{equation}
    \begin{aligned}
        \rlmor{\delta^{-1} ; [f, g]} 
        & = \dmor{R} \otimes (A + B) ; \alpha ; R \otimes (\delta^{-1} ; [f, g]) \\
        & = \delta^{-1} ; [\dmor{R} \otimes A ; \alpha ; f, \dmor{R} \otimes B ; \alpha ; g] \\
        & = \delta^{-1} ; [\rlmor{f}, \rlmor{g}]
    \end{aligned}
  \end{equation}
\end{itemize}
We may now state our desired result as follows:
\begin{lemma}
  If $\mc{C}$ is a premonoidal strong Elgot category, then so is $\envcom{\mc{C}}{R}$ with
  fixpoint operator $\rfix{f}$.
\end{lemma}
\begin{proof}
  We check each of the Elgot axioms as follows:
  \begin{itemize}
    \item \emph{Fixpoint:} given $f : R \otimes A \to B + A$, we have that
    \begin{equation}
      \begin{aligned}
        \rfix{f} 
        & = (\rcase{f})^\dagger ; \pi_r = (\rlmor{f} ; \delta^{-1})^\dagger ; \pi_r \\  
        & = \rlmor{f} ; \delta^{-1} 
          ; [\ms{id}_{R \otimes B} , (\rlmor{f} ; \delta^{-1})^\dagger] ; \pi_r \\
        & = \rlmor{f} ; \delta^{-1} ; [\pi_r, (\rcase{f})^\dagger ; \pi_r]
        & = \rseq{}{f}{\envcop{R}{\pi_r}{\rfix{f}}}
      \end{aligned}
    \end{equation}
    as desired.
    \item \emph{Naturality:} given $f : R \otimes A \to B + A$, $g: R \otimes B \to C$, we have that
    \begin{equation}
      \begin{aligned}
        \rfix{\rseq{}{f}{g +_R A}} 
        & = (\rlmor{\dmor{R} \otimes A ; \alpha ; R \otimes f ; \delta^{-1} ; g + \pi_r} 
            ; \delta^{-1})^\dagger 
          ; \pi_r \\
        & = (\rlmor{\dmor{R} \otimes A ; \alpha ; R \otimes f ; \delta^{-1}} 
          ; \delta^{-1}; R \otimes g + R \otimes \pi_r)^\dagger 
          ; \pi_r \\
        & = (\dmor{R} \otimes A
          ; \alpha
          ; R \otimes (\dmor{R} \otimes A ; \alpha ; R \otimes f ; \delta^{-1})
          ; \delta^{-1}; (\pi_r ; g) + R \otimes \pi_r)^\dagger \\
        & = (\dmor{R} \otimes A
          ; \alpha 
          ; R \otimes f ; \delta^{-1}
          ; (\dmor{R} \otimes B ; \alpha ; \pi_r ; g) 
          + (\dmor{R} \otimes A ; \alpha ; R \otimes \pi_r))^\dagger \\
        & = (\dmor{R} \otimes A
        ; \alpha 
        ; R \otimes f ; \delta^{-1}
        ; g 
        + R \otimes A)^\dagger \\
        & = (\dmor{R} \otimes A
        ; \alpha 
        ; R \otimes f ; \delta^{-1})^\dagger
        ; g
        = (\rcase{f})^\dagger ; g
      \end{aligned}
    \end{equation}
    On the other hand, we have that
    \begin{equation}
      \begin{aligned}
        \rseq{}{\rfix{f}}{g} 
        & = \dmor{R} \otimes A ; \alpha ; R \otimes (\rcase{f} ; \pi_r + R \otimes A)^\dagger ; g \\
        & = \dmor{R} \otimes A ; \alpha 
          ; (R \otimes (\rcase{f} ; \pi_r + R \otimes A) ; \delta^{-1})^\dagger ; g \\
        & = \dmor{R} \otimes A ; \alpha 
          ; (R \otimes 
            (\dmor{R} \otimes A ; \alpha ; R \otimes f ; \delta^{-1} ; \pi_r + R \otimes A) 
            ; \delta^{-1})^\dagger ; g \\
        & = \dmor{R} \otimes A ; \alpha 
          ; (R \otimes 
            (\dmor{R} \otimes A ; \alpha ; R \otimes f ; \delta^{-1}) 
            ; \delta^{-1} ; R \otimes \pi_r + R \otimes (R \otimes A))^\dagger ; g
      \end{aligned}
    \end{equation}
    By uniformity, it hence suffices to show that
    \begin{equation}
      \begin{aligned}
        & (\dmor{R} \otimes A ; \alpha)
          ; (R \otimes 
          (\dmor{R} \otimes A ; \alpha ; R \otimes f ; \delta^{-1}) 
          ; \delta^{-1} ; R \otimes \pi_r + R \otimes (R \otimes A)) \\
        & = \dmor{R} \otimes A ; \alpha ; R \otimes f ; \delta^{-1}
          ; (\dmor{R} \otimes B ; \alpha ; R \otimes \pi_r)
          + (\dmor{R} \otimes A ; \alpha) \\
        & = \rcase{f} ; (R \otimes B) + (\dmor{R} \otimes A ; \alpha)
      \end{aligned}
    \end{equation}
    to yield the desired result.
    \item \emph{Codiagonal:} given $f : R \otimes A \to (B + A) + A$, we have
    \begin{equation}
      \begin{aligned}
        \rfix{\rfix{f}} 
        & = (\dmor{R} \otimes A ; \alpha \
            ; R \otimes ((\rcase{f})^\dagger ; \pi_r) ; \delta^{-1})^\dagger 
          ; \pi_r \\
        & = (\dmor{R} \otimes A ; \alpha 
            ; R \otimes (\rcase{f} ; \pi_r + R \otimes A)^\dagger
            ; \delta^{-1})^\dagger 
          ; \pi_r \\
        & = (\dmor{R} \otimes A ; \alpha 
            ; (R \otimes (\rcase{f} ; \pi_r + R \otimes A) ; \delta^{-1})^\dagger
            ; \delta^{-1})^\dagger 
          ; \pi_r \\
        & = (\dmor{R} \otimes A ; \alpha 
            ; (
              R \otimes (\rcase{f} ; \pi_r + R \otimes A) ; \delta^{-1} 
                ; \delta^{-1} + R \otimes (R \otimes A))^\dagger
            )^\dagger 
          ; \pi_r \\
        & = (\dmor{R} \otimes A ; \alpha 
            ; (
              R \otimes (\rcase{f} ; \pi_r + R \otimes A) ; \delta^{-1} 
                ; \delta^{-1} + R \otimes (R \otimes A))^\dagger
              ; \pi_r + R \otimes A
            )^\dagger \\
        & = (\dmor{R} \otimes A ; \alpha 
            ; (
              R \otimes (\rcase{f} ; \pi_r + R \otimes A) ; \delta^{-1} 
                ; (\delta^{-1} ; (\pi_r + R \otimes A)) + R \otimes (R \otimes A)
              )^\dagger
            )^\dagger
      \end{aligned}
    \end{equation}
    On the other hand, we have
    \begin{equation}
      \begin{aligned}
        \rfix{\rseq{}{f}{\envcop{R}{\pi_r}{\envinl{R}}}}
        & = \rfix{
          \dmor{R} \otimes A ; \alpha ; R \otimes f ; \delta^{-1} ; [\pi_r, \pi_r ; \iota_l]
        } \\
        & = \rfix{f ; [\ms{id}, \iota_l]} \\
        & = (\dmor{R} \otimes A ; \alpha ; R \otimes (f ; [\ms{id}, \iota_l]) ; \delta^{-1})^\dagger 
          ; \pi_r \\
        & = (\dmor{R} \otimes A ; \alpha ; R \otimes f ; \delta^{-1}  
            ; \delta^{-1} + R \otimes A ; [\ms{id}, \iota_l])^\dagger 
          ; \pi_r \\
        & = ((\dmor{R} \otimes A ; \alpha ; R \otimes f ; \delta^{-1}  
        ; \delta^{-1} + R \otimes A)^\dagger)^\dagger ; \pi_r \\
        & = ((\dmor{R} \otimes A ; \alpha ; R \otimes f ; \delta^{-1}  
        ; \delta^{-1} + R \otimes A)^\dagger ; \pi_r + R \otimes A)^\dagger \\
        & = ((\dmor{R} \otimes A ; \alpha ; R \otimes f ; \delta^{-1}  
        ; \delta^{-1} + R \otimes A ; (\pi_r + R \otimes A) + R \otimes A)^\dagger)^\dagger \\
      \end{aligned}
    \end{equation}
    By uniformity, it therefore suffices to show that
    \begin{equation}
      \begin{aligned}
        & (\dmor{R} \otimes A ; \alpha )
          ; (
            R \otimes (\rcase{f} ; \pi_r + R \otimes A) ; \delta^{-1} 
              ; (\delta^{-1} ; (\pi_r + R \otimes A)) + R \otimes (R \otimes A)
          ) \\
        & = \dmor{R} \otimes A ; \alpha
          ; R \otimes (\dmor{R} \otimes A ; \alpha ; R \otimes f ; \delta^{-1} 
            ; \pi_r + R \otimes A) 
          ; \delta^{-1} ; (\delta^{-1} ; (\pi_r + R \otimes A)) + R \otimes (R \otimes A) \\
        & = \dmor{R} \otimes A ; \alpha
          ; R \otimes (\dmor{R} \otimes A ; \alpha ; R \otimes f ; \delta^{-1}) ; \\ & \qquad 
            \delta^{-1} ; R \otimes \pi_r + R \otimes (R \otimes A) 
          ; (\delta^{-1} ; (\pi_r + R \otimes A)) + R \otimes (R \otimes A) \\
        & = \dmor{R} \otimes A ; \alpha
          ; R \otimes f
          ; \delta^{-1}
          ; \delta^{-1} + (R \otimes A)
          ; (\pi_r + R \otimes A) + (\dmor{R} \otimes A ; \alpha)
      \end{aligned}
    \end{equation}
    to obtain the desired result.
    \item \emph{Uniformity:} 
    Given $f : R \otimes A \to B + A$, $g : R \otimes X \to B + X$ and $h : R \otimes X \to_\bot A$
    such that $\rseq{}{h}{f} = \rseq{}{g}{B +_R h}$, we have that
    \begin{equation}
      \begin{aligned}
      \rlmor{h} ; \rcase{f} 
      & = \rcase{\rseq{}{h}{f}} \\
      & = \rcase{\rseq{}{g}{B +_R h}} \\ 
      & = \rcase{g} ; (R \otimes B) + \rlmor{h}        
      \end{aligned}
    \end{equation}
    and hence by uniformity that
    \begin{equation}
      \rseq{}{h}{\rfix{f}} 
      = \rlmor{h} ; (\rcase{f})^\dagger ; \pi_r
      = (\rcase{g})^\dagger ; \pi_r 
      = \rfix{g}
    \end{equation}
  \end{itemize}
\end{proof}

\section{Proofs}

\subsection{Strict SSA}

\anfconversion*

\label{proof:anf-conversion}

\ssaconversion*

\begin{proof}
  We may show that $\ms{ANF}(r)$ is always in ANF and that $\ms{ANF}_{\ms{let}}(x, a, r)$ is in ANF
  if $r$ is by a straightforward induction. To prove the rest of the lemma, we begin proving the
  correctness of $\ms{ANF}_{\ms{let}}(x, a, r)$ by induction on expressions $a$:
  \begin{itemize}
    \item If $\ahasty{\Gamma}{\epsilon}{a}{A}$ is atomic, we are done; otherwise
    \item If $a = f\;e$, then, since $e$ is a subterm of $a$, by induction, we have
    \begin{equation}
      \letstmt{x}{a}{r} 
      \teqv \letstmt{y}{e}{\letstmt{x}{f\;y}{r}}
      \teqv \letanf{y}{e}{\letstmt{x}{f\;y}{r}}
      = \letanf{x}{a}{r}
    \end{equation}
    as desired.
    \item The cases for pairs, injections, and aborts containing expressions are analogous
    \item If $a = \caseexpr{e}{y}{b}{z}{c}$, then we may rewrite
    \begin{equation}
      \letstmt{x}{a}{r} 
      \teqv \letstmt{w}{e}{\casestmt{w}{y}{\letexpr{b}{x}{r}}{z}{\letexpr{c}{x}{r}}}
    \end{equation}
    Since $r$ is in ANF, by induction (since $b, c$ are subterms of $a$), this is equivalent to
    \begin{equation}
      \letstmt{w}{e}
        {\casestmt{w}{y}{\letanf{b}{x}{r}}{z}{\letanf{c}{x}{r}}}
    \end{equation}
    which is equal to $\letanf{x}{a}{r}$, as desired.
    \item The cases for unary and binary $\ms{let}$ are analogous to the above
  \end{itemize}
  We may now prove the correctness of $\toanf{r}$ by a straightforward induction on $r$:
  \begin{itemize}
    \item If $r = \brb{\ell}{a}$, by $\beta$-reduction, we have that
    \begin{equation}
      r \teqv \letexpr{x}{a}{\brb{\ell}{x}} \teqv \letanf{x}{a}{\brb{\ell}{x}} = \toanf{r}
    \end{equation}
    \item If $r = \letstmt{x}{a}{r'}$, then by induction we have that
    \begin{equation}
      r \teqv \letstmt{x}{a}{\toanf{r'}} \teqv \letanf{x}{a}{\toanf{r'}} = \toanf{r}
    \end{equation}
    as desired.
    \item If $r = \letstmt{(x, y)}{a}{r'}$, then by induction we have that
    \begin{equation}
      \begin{aligned}
        r & \teqv \letstmt{z}{a}{\letstmt{(x, y)}{z}{r'}} \\
          & \teqv \letstmt{z}{a}{\letstmt{(x, y)}{z}{\ms{ANF}(r')}} \\
          & \teqv \letanf{z}{a}{\letstmt{(x, y)}{z}{\ms{ANF}(r')}}
            \teqv \toanf{r}
      \end{aligned}
    \end{equation}
    The proof for \ms{case}-statements is analogous
    \item The case for control-flow graphs follows trivially by induction
  \end{itemize}
\end{proof}

\label{proof:ssa-conversion}

\begin{proof}
  Given that $\ssawhere{r}{G} \teqv \where{r}{G}$ for $r$ in ANF, it is trivial to see
  that
  \begin{equation}
    \tossa{r} := \ssawhere{\toanf{r}}{\cdot} \teqv (\where{\toanf{r}}{\cdot}) \teqv \toanf{r} 
      \teqv r
  \end{equation}
  We hence only need to prove the second part of the lemma. We proceed by induction on ANF regions
  $r$ as follows:
  \begin{itemize}
    \item If $r$ is a terminator, this holds trivially by reflexivity
    \item If $r = \letstmt{x}{a}{r'}$, then by induction, we have that
    \begin{equation}
      \ssawhere{r}{G} 
        := \letstmt{x}{a}{\ssawhere{r'}{G}} 
        \teqv \letstmt{x}{a}{\where{r'}{G}} 
        \teqv \where{\letstmt{x}{a}{r'}}{G}
        \teqv \where{r}{G}
    \end{equation}
    as desired. The case for binary \ms{let}-statements is analogous.
    \item If $r = \casestmt{a}{x}{s}{y}{t}$, then by induction, we have that
    \begin{equation}
      \begin{aligned}
      \ssawhere{r}{G} 
        &:= \where{(\casestmt{a}{x}{\brb{\ell_l}{x}}{y}
        {\brb{\ell_r}{y}})}{G, \wbranch{\ell_l}{x}{\tossa{s}}, \wbranch{\ell_r}{y}{\tossa{t}}} \\
        &:= \where{(\casestmt{a}{x}{\brb{\ell_l}{x}}{y}
        {\brb{\ell_r}{y}})}{G, \wbranch{\ell_l}{x}{s}, \wbranch{\ell_r}{y}{t}} \\
        &\teqv \where{(\where{(\casestmt{a}{x}{\brb{\ell_l}{x}}{y}
        {\brb{\ell_r}{y}})}{\wbranch{\ell_l}{x}{s}, \wbranch{\ell_r}{y}{t}})}{G} \\
        &\teqv \where{\casestmt{a}{x}{s}{y}{t}}{G} \teqv \where{r}{G}
      \end{aligned}
    \end{equation}
    \item If $r = \where{r'}{\wbranch{\ell_i}{x_i}{t_i},)_i}$, then by induction, we have that
    \begin{equation}
      \begin{aligned}
        \ssawhere{r}{G} 
          &:= \ssawhere{r'}{G, (\wbranch{\ell_i}{x_i}{\tossa{t_i}},)_i} \\
          &\teqv \where{r'}{G, (\wbranch{\ell_i}{x_i}{t_i},)_i} \\
          &\teqv \where{(\where{r'}{(\wbranch{\ell_i}{x_i}{t_i},)_i}}{G} \teqv \where{r}{G} \\
      \end{aligned}
    \end{equation}
  \end{itemize}
\end{proof}

\cfgperminvar*

\label{proof:cfg-perm-invar}

\begin{proof}
  We proceed by induction on the size of $G$. If $G$ consists only of an entry block, there is only
  one permutation, so we are done. Otherwise, assume $G = \beta, (\wbranch{\ell_i}{x_i}{t_i},)_i$.
  Furthermore, let:
  \begin{itemize}
    \item $\wbranch{\ell_i}{x_i}{\beta_i}$ be the children of $\beta$ in $G$ in order
    \item $G_i = (\wbranch{\kappa_{i, j}}{y_{i, j}}{t_{i, j}},)_j$ be the CFG composed of the
    descendants of $\beta_i$ in $G$, with $\beta_i$ as entry block, in order
    \item $\wbranch{\ell_i'}{x_i'}{\beta_i'}$ be the children of $\beta$ in $G'$ in order
    \item $G_i'$ be the CFG composed of the descendants of $\beta_i'$ in $G'$, with $\beta_i'$ as
    entry block, in order
  \end{itemize}
  By assumption, there exists some permutation $\sigma$ such that $G' = \beta,
  (\wbranch{\ell_{\sigma_i}}{x_{\sigma_i}}{t_{\sigma_i}},)_i$. It follows that, since the dominance
  relation on labels is permutation-invariant, there exists some permutation $\rho$ such that
  $\wbranch{\ell_i'}{x_i'}{\beta_i'} = \wbranch{\ell_{\rho_i}}{x_{\rho_i}}{\beta_{\rho_i}}$, as well
  as permutations $\tau_i$ such that $G_i' = (\wbranch{\ell_{\rho_i, \tau_{ij}}}{x_{\rho_i,
  \tau_{ij}}}{t_{\rho_i, \tau_{ij}}},)_j$, implying in particular that $G_i' \simeq G_{\rho_i}$. By
  induction, we hence have that, for all $i$, $\toreg{G_i'} \teqv \toreg{G_{\rho_i}}$, and hence
  that
  \begin{equation}
    \begin{aligned}
    \toreg{G'} &:= \adddom{\beta}{(\wbranch{\ell_{\rho_i}}{x_{\rho_i}}{\toreg{G_i'}},)_i} \\
    &\teqv \adddom{\beta}{(\wbranch{\ell_{\rho_i}}{x_{\rho_i}}{\toreg{G_{\rho_i}}},)_i} \\
    &\teqv \adddom{\beta}{(\wbranch{\ell_i}{x_i}{\toreg{G_i}},)_i} \teqv \toreg{G}
    \end{aligned}
  \end{equation}
\end{proof}

\cfgconversion*

\label{proof:cfg-conversion}

\begin{proof}
  We will proceed by induction on the length of $G = \tocfg{r}$. If $G$ consists of only an entry
  block $\beta$, then we trivially have that $r = \toreg{\tocfg{r}}$, and so we are done. Otherwise,
  assume $r = \adddom{\beta}{(\wbranch{\ell_i}{x_i}{t_i},)_i}$. Clearly, $\tocfg{r}$ will have entry
  block $\beta$; moreover, since every block in $\tocfg{r}$ other than those of the form
  $\toentry{t_i}$ can only be reached from within the region $t_i$ (due to lexical scoping of
  labels), we have $\beta_i = \toentry{t_{\rho_i}}$ for some injection $\rho$, where $\beta_i$ are
  the children of $\beta$ in the dominance tree of $G$. In particular, we can write
  $\{\wbranch{\ell_i}{x_i}{\toentry{t_i}},\}_i$ as the disjoint union of:
  \begin{itemize}
    \item The immediate children of $\beta$, $\wbranch{\kappa_i}{y_i}{\beta_i}$, where 
      $\kappa_i = \ell_{\rho_i}$, $y_i = x_{\rho_i}$, and $\beta_i = \toentry{t_{\rho_i}}$
    \item The nodes dominated by each $\beta_i$ but not immediately dominated by $\beta$; we will
    write the collection of such nodes for each $i$ as $\wbranch{\kappa_{i, j}}{y_{i, j}}{\beta_{i,
    j}}$, where $\kappa_{i, j} = \ell_{\rho_{i, j}}$, $y_{i, j} = x_{\rho_{i, j}}$, and $\beta_{i,
    j} = \toentry{t_{\rho_{i, j}}}$.
  \end{itemize}
  As in the algorithm to compute $\toreg{\cdot}$, let $G_i$ denote the control-flow graph with 
  entry block $\beta_i$ consisting of all blocks dominated by $\beta_i$. We may write
  \begin{equation}
    G \simeq \beta, (\wbranch{\kappa_i}{y_i}{G_i},)_i \qquad
    \forall i, G_i = \beta_i, (\wbranch{\kappa_{i, j}}{y_{i, j}}{\beta_{i, j}},)_j, R_i
  \end{equation}
  where $R_i$ is the ``remainder'' of the control-flow graph $G_i$. Note the equations for $G_i$ are
  actually equalities, rather than being equivalence up to permutation $\simeq$, since the
  $\beta_{i, j}$, appear in $G$ before any other elements, being immediate children of $\beta$. It
  follows that
  \begin{equation}
    \toreg{\tocfg{r}}
    = \toreg{\tocfg{\adddom{\beta}{(\wbranch{\ell_i}{x_i}{t_i},)_i}}}
    = \adddom{\beta}{(\wbranch{\kappa_i}{y_i}{\toreg{G_i}},)_i}
  \end{equation}
  Now, define the strict regions
  \begin{equation}
    r_i = \adddom{\beta_i}{
      ((\wbranch{\kappa_{i, j}}{y_{i, j}}{t_{\rho_{i, j}}},)_j, \todom{t_i})
    }
  \end{equation}
  It is easy to show that $\tocfg{r_i} \simeq G_i$, since every basic block dominated by $\beta_i$
  must be dominated via either some $t_{\rho_{i, j}}$ or some child of $t_i$. Since by induction
  $\toreg{\tocfg{r_i}} \teqv r_i$, and by the previous lemma $\toreg{\tocfg_r_i} \teqv \toreg{G_i}$,
  it follows that
  \begin{equation}
    \toreg{\tocfg{r}}
    \teqv \adddom{\beta}{(\wbranch{\kappa_i}{y_i}{r_i},)_i}
  \end{equation}
  Define
  \begin{equation}
    T_i = (\wbranch{\kappa_j}{y_j}{r_j},)_{j < i}, 
      (\wbranch{\kappa_j}{y_j}{t_j}, 
      (\wbranch{\kappa_{j, k}}{y_{j, k}}{t_{\rho_{j, k}}},)_k,)_{j \geq i}
  \end{equation}
  Using Equation~\ref{eqn:pull-where}, we may show that, for all $i$, $\adddom{\beta}{T_{i + 1}}
  \teqv \adddom{\beta}{T_i}$, since
  \begin{itemize}
    \item For all $j$, $t_{\rho_{i, j}}$ cannot use variables defined in $\beta_i$, or $r$ would not
    typecheck
    \item For all $j$, if $r_k$ or $t_{\rho_{k, j'}}$ calls $\kappa_{i, j}$, then $k = i$, or
    $\beta_{i, j}$ would not be dominated by $\beta_i$
    \item Similarly, $\beta$ cannot call $\kappa_{i, j}$ or $\beta_{i, j}$ would be a direct
    descendant of $\beta$
  \end{itemize}
  We hence have by induction that
  \begin{equation}
    \toreg{\tocfg{r}} \teqv \adddom{\beta}{T_N} \teqv \adddom{\beta}{T_0} \teqv r
  \end{equation}
  as desired, since $T_0$ is a permutation of $(\wbranch{\ell_i}{x_i}{t_i},)_i$.
\end{proof}

\subsection{B\"ohm-Jacopini}

\begin{lemma}
  The following facts hold:
  \begin{itemize}
    \item Given $\haslb{\Gamma}{r}{\outlb(A)}$ and $\haslb{\Gamma, \bhyp{\invar}{A}}{s}{\ms{L}}$,
    we have that
    \begin{equation}
      \begin{aligned}
      & \dnt{\haslb{\Gamma}{\ms{seq}(r, s) 
        := (\where{[\outlb(x) \mapsto \brb{\ell}{x}]r}{\wbranch{\ell}{\invar}{s}})}{\ms{L}}} \\
      & = \lmor{\dnt{\haslb{\Gamma}{[\outlb(x) \mapsto \brb{\ell}{x}]r}{\ell(A)}}} 
        ; \dnt{\Gamma} \otimes \alpha^+_A 
        ; \dnt{\haslb{\Gamma, \bhyp{\invar}{A}}{s}{\ms{L}}} \\
      & = \lmor{\dnt{\haslb{\Gamma}{r}{\ell(A)}}} 
        ; \dnt{\Gamma} \otimes \alpha^+_A 
        ; \dnt{\haslb{\Gamma, \bhyp{\invar}{A}}{s}{\ms{L}}}
      \end{aligned}
    \end{equation}
    \item Given $\hasty{\Gamma}{\epsilon}{e}{A}$, %
                $\haslb{\Gamma, \bhyp{\invar}{A}}{r}{\outlb(B + A)}$,
    we have that
    \begin{equation}
      \begin{aligned}
      & \dnt{\haslb{\Gamma}{\ms{loop}(e, r) 
        := (\where{\brb{\ell}{e}}{\wbranch{\ell}{\invar}
        {\ms{seq}(r, \casestmt{\invar}{x}{\brb{\outlb}{x}}{y}{\brb{\ell}{y}})}})}{\outlb(B)}} \\
      & = \lmor{\dnt{\hasty{\Gamma}{\epsilon}{e}{A}}}
        ; \rfix{\dnt{\haslb{\Gamma, \bhyp{\invar}{A}}{r}{\outlb(B + A)}} ; \alpha^+_{B + A}}
      \end{aligned}
    \end{equation}
    \item Given $\ms{R} = (\ell_i(A_i),)_i$ and $\haslb{\Gamma, \bhyp{x_i}{A_i}}{t_i}{\ms{L}}$, we
    have
    \begin{equation}
      \begin{aligned}
        & \dnt{
          \haslb{\Gamma, \bhyp{c}{\pckd{\ms{R}}}}{\ms{case}_{\ms{R}}\;c\; \{\ell_i(x_i) : t_i\}}
                {\ms{L}}
          } \\
        & = \dnt{\Gamma} \otimes \alpha^+_{\Sigma_i\dnt{A_i}} ; \delta^{-1}_{\Sigma}
          ; [\dnt{\haslb{\Gamma, \bhyp{x_i}{A_i}}{t_i}{\ms{L}}},]_i
      \end{aligned}
    \end{equation}
  \end{itemize}
  \item We have that
  \begin{equation}
    \dnt{\hasty{\Gamma, \bhyp{\invar}{\bot}}{[\ms{L}]}{\ms{ua}\;\invar}
      {\pckd{\ms{L}} + \pckd{\ms{R}}}}
    = \pi_r ; \alpha^+_{\dnt{\ms{L}} + \dnt{\ms{R}}}
  \end{equation}
  \item We have that
  \begin{equation}
    \dnt{\lbsubst{\Gamma}{\ms{pack}_\kappa^+(\ms{L})}{\ms{L}}{\kappa(\pckd{\ms{L}})}}   
      = \pi_r ; \alpha^+_{\mb{0} + \dnt{\ms{L}}}
  \end{equation}
  and
  \begin{equation}
    \dnt{\lbsubst{\Gamma}{\ms{unpack}_\kappa^+(\ms{L})}{\kappa(\pckd{\ms{L}})}{\ms{L}}}
      = \pi_r ; \alpha^+_{\dnt{\ms{L}}}
  \end{equation}
  In particular, given $\haslb{\Gamma}{r}{\ms{L}}$, we have that
  \begin{equation}
    \dnt{\haslb{\Gamma}{\pckd{r}^+}{\outlb(\pckd{\ms{L}})}}
    = \dnt{\haslb{\Gamma}{r}{\ms{L}}} ; \alpha^+_{\mb{0} + \dnt{\ms{L}}}
  \end{equation}
\end{lemma}

\bohmjacopini*

\label{proof:bohm-jacopini}

\begin{proof}
  We begin by showing the correctness of $\topwhile{\ms{L}}{r}$ by induction on $r$:
  \begin{itemize}
    \item If $r = \brb{\ell}{a}$, we have that 
    $\topwhile{\ms{L}}{a} = \ms{pack}^+(\ms{L})_\ell(a) = [\brb{\ell}{a}]^+$
    by definition
    \item If $r = \letstmt{x}{a}{s}$, we have by induction that
    $$
    \topwhile{\ms{L}}{\letstmt{x}{a}{s}} 
    = \letstmt{x}{a}{\topwhile{\ms{L}}{s}}
    \teqv \letstmt{x}{a}{[s]^+}
    = [\letstmt{x}{a}{s}]^+
    $$
    \item If $r = \letstmt{(x, y)}{a}{s}$, we have by induction that
    $$
    \topwhile{\ms{L}}{\letstmt{(x, y)}{a}{s}} 
    = \letstmt{(x, y)}{a}{\topwhile{\ms{L}}{s}}
    \teqv \letstmt{(x, y)}{a}{\pckd{s}^+}
    = \pckd{\letstmt{(x, y)}{a}{s}}^+
    $$
    \item If $r = \casestmt{a}{x}{s}{y}{t}$, we have by induction that
    $$
    \begin{aligned}
    \topwhile{\ms{L}}{\casestmt{a}{x}{s}{y}{t}}
    &= \casestmt{a}{x}{\topwhile{\ms{L}}{s}}{y}{\topwhile{\ms{L}}{t}} \\
    &\teqv \casestmt{a}{x}{\pckd{s}^+}{y}{\pckd{t}^+} \\
    &= \pckd{\casestmt{a}{x}{s}{y}{t}}^+ \\
    \end{aligned}
    $$
    \item Assume $r = \where{s}{(\wbranch{\ell_i}{t_i}{x_i},)_i}$. %
    Define $\ms{R} = (\ell_i(A_i),)_i$.
    By induction, we have that $\forall i, \topwhile{\ms{L}, \ms{R}}{t_i} \teqv \pckd{t_i}^+$, and
    hence by soundness
    \begin{equation}
      \dnt{\haslb{\Gamma, \bhyp{x_i}{A_i}}{\topwhile{\ms{L}}{t_i}}{\outlb(\pckd{\ms{L}, \ms{R}})}}
      = \dnt{\haslb{\Gamma, \bhyp{x_i}{A_i}}{\pckd{t_i}^+}{\outlb(\pckd{\ms{L}, \ms{R}})}}
    \end{equation}
    Now, define
    $$
    D = \ms{case}_{\ms{L}}\;\invar\;
    \{\ell_i(x_i) : \ms{seq}(\topwhile{\ms{L}}{t_i}, \brb{\outlb}{(\ms{ua}\;\invar)}\})
    $$
    and $L = \ms{loop}(y, D)$. %
    It follows that
    \begin{equation}
      \begin{aligned}
      & \dnt{\haslb{\Gamma, \invar : [R]}{D}{\outlb(\pckd{L} + \pckd{R})}} \\
      & = \dnt{\Gamma} \otimes \alpha^+_{\Sigma_i\dnt{A_i}} ; \delta^{-1}_{\Sigma}
      ; [\dnt{\haslb{\Gamma, \bhyp{x_i}{A_i}} 
                {\ms{seq}(\topwhile{\ms{L}}{t_i}, \brb{\outlb}{(\ms{ua}\;\invar))}}
                {\outlb(\pckd{\ms{L}} + \pckd{\ms{R}})}},]_i \\
      & = \dnt{\Gamma} \otimes \alpha^+_{\Sigma_i\dnt{A_i}} ; \delta^{-1}_{\Sigma}
      ; [\dnt{\haslb{\Gamma, \bhyp{x_i}{A_i}}{\topwhile{\ms{L}}{t_i}}{\outlb(\pckd{\ms{L}, \ms{R}})}} 
      ; \alpha^+_{\mb{0} + (\dnt{\ms{L}} + \dnt{\ms{R}})},]_i \\
      & = \dnt{\Gamma} \otimes \alpha^+_{\Sigma_i\dnt{A_i}} ; \delta^{-1}_{\Sigma}
      ; [\dnt{\haslb{\Gamma, \bhyp{x_i}{A_i}}{\pckd{t_i}^+}{\outlb(\pckd{\ms{L}, \ms{R}})}} 
      ; \alpha^+_{\mb{0} + (\dnt{\ms{L}} + \dnt{\ms{R}})},]_i\\
      & = \dnt{\Gamma} \otimes \alpha^+_{\Sigma_i\dnt{A_i}} ; \delta^{-1}_{\Sigma}
      ; [\dnt{\haslb{\Gamma, \bhyp{x_i}{A_i}}{t_i}{\ms{L}, \ms{R}}} 
      ; \alpha^+_{\mb{0} + (\dnt{\ms{L}} + \dnt{\ms{R}})},]_i
      \end{aligned}
    \end{equation}
    and therefore that
    \begin{equation}
      \begin{aligned}
      & \dnt{\haslb{\Gamma, y : \pckd{R}}{L}{\outlb(\pckd{\ms{L}})}} \\
      & = \lmor{\dnt{\hasty{\Gamma, y : \pckd{R}}{\bot}{y}{A}}}
      ; \rfix{\dnt{\haslb{\Gamma, \invar : \pckd{R}}{D}{\outlb(\pckd{L} + \pckd{R})}} 
      ; \alpha^+_{(\mb{0} + \dnt{\ms{L}}) + \dnt{\ms{R}}}} \\
      & = \rfix{\dnt{\Gamma} \otimes \alpha^+_{\Sigma_i\dnt{A_i}} ; \delta^{-1}_{\Sigma}
      ; [\dnt{\haslb{\Gamma, \bhyp{x_i}{A_i}}{t_i}{\pckd{\ms{L}, \ms{R}}}} 
      ; \alpha^+_{(\mb{0} + \dnt{\ms{L}}) + \dnt{\ms{R}}},]_i} \\
      & = \dnt{\Gamma} \otimes \alpha^+_{\Sigma_i\dnt{A_i}} 
        ; \rfix{\loopmor{\Gamma}{(\wbranch{\ell_i}{x_i}{t_i},)_i}{\ms{L}}}
        ; \alpha^+_{\mb{0} + \dnt{\ms{L}}}
      \end{aligned}
    \end{equation}
    Hence, we have that
    \begin{equation}
      \begin{aligned}
      & \dnt{\haslb 
        {\Gamma, \bhyp{\invar}{\pckd{\ms{L}, \ms{R}}}}
        {\caseexpr{\ms{ua}\;\invar}{x}{\brb{\outlb}{x}}{y}{L}}
        {\outlb([\ms{L}])}} \\
      & = \dnt{\Gamma} \otimes \alpha^+_{\dnt{\ms{L}} + \dnt{\ms{R}}}
      ; \delta^{-1}
      ; [
        \dnt{\haslb{\Gamma, x : \pckd{\ms{L}}}{\brb{\outlb}{x}}{\outlb(\pckd{\ms{L}})}},
        \dnt{\haslb{\Gamma, y : \pckd{\ms{R}}}{L}{\outlb(\pckd{\ms{L}})}}
      ] \\
      & = \dnt{\Gamma} \otimes \alpha^+_{\dnt{\ms{L}} + \dnt{\ms{R}}}
      ; \delta^{-1}
      ; [
        \pi_r ; \alpha^+_{\mb{0} + \dnt{\ms{L}}}, 
        \dnt{\Gamma} \otimes \alpha^+_{\Sigma_i\dnt{A_i}} 
        ; \rfix{\loopmor{\Gamma}{(\wbranch{\ell_i}{x_i}{t_i},)_i}{\ms{L}}} ; 
        \alpha^+_{\mb{0} + \dnt{\ms{L}}}
      ] \\
      & = \dnt{\Gamma} \otimes \alpha^+_{\dnt{\ms{L}} + \Sigma_i\dnt{A_i}}
      ; \delta^{-1}
      ; [
        \pi_r,  
        \rfix{\loopmor{\Gamma}{(\wbranch{\ell_i}{x_i}{t_i},)_i}{\ms{L}}}
      ]
      ; \alpha^+_{\mb{0} + \dnt{\ms{L}}}
      \end{aligned}
    \end{equation}
    It hence suffices by completeness (Theorem~\ref{thm:complete-reg}) to show that
    \begin{equation}
      \begin{aligned}
        & \dnt{\haslb{\Gamma}{\ms{seq}(
          \topwhile{\ms{L}}{r}, \caseexpr{\ms{ua}\;\invar}{x}{\brb{\outlb}{x}}{y}{L})
        }{\outlb(\pckd{\ms{L}})}} \\
        & = \lmor{\dnt{\haslb{\Gamma}{\topwhile{\ms{L}}{r}}{\outlb(\pckd{\ms{L}, \ms{R}})}}}
          ; \dnt{\Gamma} \otimes \alpha^+_{\dnt{\ms{L, R}}} 
          ; \\ & \qquad \dnt{\haslb 
          {\Gamma, \bhyp{\invar}{\pckd{\ms{L}, \ms{R}}}}
          {\caseexpr{\ms{ua}\;\invar}{x}{\brb{\outlb}{x}}{y}{L}}
          {\outlb(\pckd{\ms{L}})}} \\
        & = \lmor{\dnt{\haslb{\Gamma}{\topwhile{\ms{L}}{r}}{\outlb(\pckd{\ms{L}, \ms{R}})}}}
          ; \dnt{\Gamma} \otimes \alpha^+_{\dnt{\ms{L}} + \Sigma_i\dnt{A_i}}
          ; \delta^{-1}
          ; \\ & \qquad [
            \pi_r,  
            \rfix{\loopmor{\Gamma}{(\wbranch{\ell_i}{x_i}{t_i},)_i}{\ms{L}}}
          ] ; \alpha^+_{\mb{0} + \dnt{\ms{L}}} \\
        & = \dnt{\haslb{\Gamma}{\where{r}{(\wbranch{\ell_i}{x_i}{t_i},)_i}}{\ms{L}}}
          ; \alpha^+_{\mb{0} + \dnt{\ms{L}}} \\
        & = \dnt{\haslb{\Gamma}{[\where{r}{(\wbranch{\ell_i}{x_i}{t_i},)_i}]^+}{\ms{L}}}
      \end{aligned}
    \end{equation}
  \end{itemize}
\end{proof}

\subsection{Substitution}

\begin{lemma}
  The following facts hold:
  \begin{enumerate}[label=(\alph*)]
    \item For all $f: A \to B \otimes C$, $g : (A \otimes B) \otimes C \to D$, we have
    \begin{equation}
      \lmor{f} ; \alpha ; \lmor{g} ; (\pi_l ; \pi_l) \otimes D =
      \lmor{\lmor{f} ; \alpha ; g}
    \end{equation}
    \item For all $f: A \to B + C$, $g : A \otimes B \to D$, $h : A \otimes C \to D$, we
    have
    \begin{equation}
      \lmor{f} ; \delta^{-1} ; [\lmor{g} ; \pi_l \otimes D, \lmor{h} ; \pi_l \otimes D]
      = \lmor{\lmor{f} ; \delta^{-1} ; [g, h]}
    \end{equation}
    \item For all $f_i : R \otimes A_i \to B$, we have
    \begin{equation}
      \delta_{\Sigma}^{-1} ; [\lmor{f_i} ; \pi_l \otimes B,]_i
      = \lmor{\delta_{\Sigma}^{-1} ; [f_i]_i} ; \pi_l \otimes B
    \end{equation}
    i.e.,
    \begin{equation}
      \delta_{\Sigma}^{-1} ; [\rlmor{f_i},]_i
      = \rlmor{\delta_{\Sigma}^{-1} ; [f_i]_i}
    \end{equation}
    \item For all $f : A \to B + C$, $g : A \otimes B \to B'$, $h : A \otimes C \to C'$, we have
    \begin{equation}
      \begin{aligned}
      \lmor{\lmor{f} ; \delta^{-1} ; g + h}
      &= \lmor{f} ; \delta^{-1} 
        ; [\lmor{g} ; \pi_l \otimes \iota_l, \lmor{h} ; \pi_l \otimes \iota_r] \\
      &= \lmor{f} ; \delta^{-1} 
        ; (\lmor{g} ; \pi_l \otimes B') + (\lmor{h} ; \pi_l \otimes C')
        ; \delta \\
      &= \lmor{f} ; \delta^{-1} 
      ; \rlmor{g} + \rlmor{h}
      ; \delta
      \end{aligned}
    \end{equation}
    In particular, we have
    \begin{equation}
      \begin{aligned}
      \lmor{\lmor{f} ; \delta^{-1} ; \pi_r + h}
      &= \lmor{f} ; \delta^{-1} 
        ; [A \otimes \iota_l, \lmor{h} ; \pi_l \otimes \iota_r] \\
      &= \lmor{f} ; \delta^{-1} 
        ; (A \otimes B) + (\lmor{h} ; \pi_l \otimes C')
        ; \delta\\
      &= \lmor{f} ; \delta^{-1} 
        ; (A \otimes B) + \rlmor{h}
        ; \delta
      \end{aligned}
    \end{equation}
    and
    \begin{equation}
      \begin{aligned}
      \lmor{\lmor{f} ; \delta^{-1} ; g + \pi_r}
      &= \lmor{f} ; \delta^{-1} 
        ; [\lmor{g} ; \pi_l \otimes \iota_l, A \otimes \iota_r] \\
      &= \lmor{f} ; \delta^{-1} 
        ; (\lmor{g} ; \pi_l \otimes B') + (A \otimes C)
        ; \delta \\
      &= \lmor{f} ; \delta^{-1} 
        ; \rlmor{g} + (A \otimes C)
        ; \delta
      \end{aligned}
    \end{equation}
  \end{enumerate}
\end{lemma}

\weakeninglem*

\label{proof:weakening}

\begin{proof}
  To show that variable weakenings compose \ref{itm:varwk}, we proceed by induction on the
  derivation of $\Gamma \leq \Gamma'$ as follows:
  \begin{itemize}
    \item \brle{wk-nil}: if $\Gamma, \Gamma' = \cdot$, then $\Delta = \cdot$, and so we trivially
    have $\dnt{\cdot \leq \cdot} ; \dnt{\cdot \leq \cdot} = \dnt{\cdot \leq \cdot} = \ms{id}$
    \item \brle{wk-skip}: we have $\Gamma = \Xi, \thyp{x}{A}{\epsilon}$, and so by induction
    \begin{equation}
      \dnt{\Xi, \thyp{x}{A}{\epsilon} \leq \Gamma'} ; \dnt{\Gamma' \leq \Delta}
      = \pi_l ; \dnt{\Xi \leq \Gamma'} ; \dnt{\Gamma' \leq \Delta}
      = \pi_l ; \dnt{\Xi \leq \Delta}
      = \dnt{\Xi, \thyp{x}{A}{\epsilon} \leq \Delta}
    \end{equation}
    as desired
    \item \brle{wk-cons}: we have $\Gamma = \Xi, \thyp{x}{A}{\epsilon}$ and $\Gamma' = \Xi',
    \thyp{x}{A}{\epsilon'}$. We proceed by case analysis on $\Gamma' \leq \Delta$:
    \begin{itemize}
      \item \brle{wk-skip}: we have
      \begin{equation}
        \begin{aligned}
        \dnt{\Xi, \thyp{x}{A}{\epsilon} \leq \Xi', \thyp{x}{A}{\epsilon'}} ; \dnt{\Xi' \leq \Delta}
        & = \dnt{\Xi \leq \Xi'} \otimes \dnt{A} ; \dnt{\Xi', \thyp{x}{A}{\epsilon'} \leq \Delta} \\
        & = \dnt{\Xi \leq \Xi'} \otimes \dnt{A} ; \pi_l ; \dnt{\Xi' \leq \Delta} \\
        & = \pi_l ; \dnt{\Xi \leq \Xi'} ; \dnt{\Xi' \leq \Delta} \\
        & = \pi_l ; \dnt{\Xi \leq \Delta} \\
        & = \dnt{\Xi, \thyp{x}{A}{\epsilon} \leq \Delta}
        \end{aligned}
      \end{equation}
      as desired
      \item \brle{wk-cons}: we have $\Delta = \Delta' , \thyp{x}{A}{\epsilon''}$, and so by
      induction
      \begin{equation}
        \begin{aligned}
        \dnt{\Xi, \thyp{x}{A}{\epsilon} \leq \Xi', \thyp{x}{A}{\epsilon'}} ; 
        \dnt{\Xi', \thyp{x}{A}{\epsilon'} \leq \Delta', \thyp{x}{A}{\epsilon''}}
        &= \dnt{\Xi \leq \Xi'} \otimes \dnt{A} ; \dnt{\Xi' \leq \Delta'} \otimes \dnt{A} \\
        &= (\dnt{\Xi \leq \Xi'} ; \dnt{\Xi' \leq \Delta'}) \otimes \dnt{A} \\
        &= \dnt{\Xi \leq \Delta'} \otimes \dnt{A} \\
        &= \dnt{\Xi, \thyp{x}{A}{\epsilon} \leq \Delta}
        \end{aligned}
      \end{equation}
      as desired
    \end{itemize}
  \end{itemize}
  We can analogously show \ref{itm:lbwk} (i.e., that label weakenings compose) by induction on the
  derivation of $\ms{L} \leq \ms{K}$ as follows:
  \begin{itemize}
    \item \brle{lwk-nil}: if $\ms{L} = \ms{K} = \cdot$, then $\ms{L}' = \cdot$, so the result
    follows trivially from the fact that $\dnt{\cdot} = \mb{\ms{L}'}$ is the initial object
    \item \brle{lwk-skip}: we have $\ms{K} = \ms{K}', \ell(A)$, and therefore
    \begin{equation}
      \dnt{\ms{L}' \leq \ms{L}} ; \dnt{\ms{L} \leq \ms{K}', \ell(A)}
      = \dnt{\ms{L}' \leq \ms{L}} ; \dnt{\ms{L} \leq \ms{K}'} ; \iota_r
      = \dnt{\ms{L}' \leq \ms{K}'} ; \iota_r
      = \dnt{\ms{L}' \leq \ms{K}}
    \end{equation}
    \item \brle{lwk-cons}: we have $\ms{L} = \ms{R}, \ell(A)$ and $\ms{K} = \ms{K}', \ell(A)$. We
    proceed by case splitting on $\ms{L} \leq \ms{K}'$:
    \begin{itemize}
      \item \brle{lwk-skip}: we have
      \begin{equation}
        \begin{aligned}
        \dnt{\ms{L}' \leq \ms{R}, \ell(A)} ; \dnt{\ms{R}, \ell(A) \leq \ms{K}', \ell(A)}
        & = \dnt{\ms{L}' \leq \ms{R}} ; \iota_r ; \dnt{\ms{R} \leq \ms{K}'} + \dnt{A} \\
        & = \dnt{\ms{L}' \leq \ms{R}} ; \dnt{\ms{R} \leq \ms{K}'} ; \iota_r \\
        & = \dnt{\ms{L}' \leq \ms{K}'} ; \iota_r \\
        & = \dnt{\ms{L}' \leq \ms{K}', \ell(A)}
        \end{aligned}
      \end{equation}
      \item \brle{lwk-cons}: we have that $\ms{L}' = \ms{R}', \ell(A)$, and therefore
      \begin{equation}
        \begin{aligned}
        \dnt{\ms{R}', \ell(A) \leq \ms{R}, \ell(A)} ; \dnt{\ms{R}, \ell(A) \leq \ms{K}', \ell(A)}
        & = \dnt{\ms{R}' \leq \ms{R}} + \dnt{A} ; \dnt{\ms{R} \leq \ms{K}'} + \dnt{A} \\
        & = (\dnt{\ms{R}' \leq \ms{R}} ; \dnt{\ms{R} \leq \ms{K}'}) + \dnt{A} \\
        & = \dnt{\ms{R}' \leq \ms{K}'} + \dnt{A} \\
        & = \dnt{\ms{R}', \ell(A) \leq \ms{K}', \ell(A)}
        \end{aligned}
      \end{equation}
    \end{itemize}
  \end{itemize}
  We can now show weakening for expressions $\hasty{\Delta}{\epsilon}{a}{A}$ \ref{itm:expwk} by 
  induction on the typing derivation as follows:
  \begin{itemize}
    \item \brle{var}: we need to show that
    \begin{equation}
      \dnt{\Gamma \leq \Delta} ; \dnt{\hasty{\Delta}{\epsilon}{x}{A}}
      = \dnt{\Gamma \leq \Delta} ; \pi_{\Delta, x}
      = \dnt{\hasty{\Gamma}{\epsilon}{x}{A}}
      = \pi_{\Gamma, x}
    \end{equation}
    we proceed by induction on $\Gamma \leq \Delta$:
    \begin{itemize}
      \item \brle{wk-nil}: this case yields a contradiction, since if $\Delta = \cdot$ it cannot
      define $x$.
      \item \brle{wk-cons}: given $\Gamma = \Gamma', \thyp{y}{B}{\epsilon''}$, $\Delta = \Delta',
      \thyp{y}{B}{\epsilon'}$,
      \begin{itemize}
        \item If $x = y$, then $B = A$, and
        \begin{equation}
          \dnt{\Gamma', \thyp{x}{A}{\epsilon''} \leq \Delta', \thyp{x}{A}{\epsilon'}} 
          ; \pi_{(\Delta, \thyp{x}{A}{\epsilon'}), x}
          = \dnt{\Gamma' \leq \Delta'} \otimes \dnt{A} ; \pi_r
          = \pi_r
          = \pi_{(\Gamma, \thyp{x}{A}{\epsilon''}), x}
        \end{equation}
        as desired.
        \item Otherwise, we have by induction that
        \begin{equation}
          \begin{aligned}
          \dnt{\Gamma', \thyp{y}{B}{\epsilon''} \leq \Delta', \thyp{y}{B}{\epsilon'}}
          ; \pi_{(\Delta', \thyp{y}{B}{\epsilon'}), x}
          & = \dnt{\Gamma' \leq \Delta'} \otimes \dnt{B} ; \pi_l ; \pi_{\Delta', x} \\
          & = \pi_l ; \dnt{\Gamma' \leq \Delta'} ; \pi_{\Delta', x}
          = \pi_l ; \pi_{\Gamma', x}
          = \pi_{\Gamma, x}
          \end{aligned}
        \end{equation}
      \end{itemize}
      \item \brle{wk-skip}: we have $\Gamma = \Gamma', \thyp{y}{B}{\epsilon'}$, and hence
      \begin{equation}
        \begin{aligned}
          \dnt{\Gamma', \thyp{y}{B}{\epsilon'} \leq \Delta} ; \pi_{\Delta, x}
          = \pi_l ; \dnt{\Gamma' \leq \Delta} ; \pi_{\Delta, x}
          = \pi_l ; \pi_{\Gamma', x}
          = \pi_{\Gamma, x}
        \end{aligned}
      \end{equation}
    \end{itemize}
    \item \brle{let$_1$}: we have
    \begin{equation}
      \begin{aligned}
        & \dnt{\Gamma \leq \Delta} ; \dnt{\hasty{\Delta}{\epsilon}{\letexpr{x}{a}{b}}{B}} \\
        & = \dnt{\Gamma \leq \Delta} 
          ; \lmor{\dnt{\hasty{\Delta}{\epsilon}{a}{A}}}
          ; \dnt{\hasty{\Delta, \bhyp{x}{A}}{\epsilon}{b}{B}} \\
        &= \lmor{\dnt{\Gamma \leq \Delta} ; \dnt{\hasty{\Delta}{\epsilon}{a}{A}}}
        ; \dnt{\Gamma \leq \Delta} \otimes \dnt{A}
        ; \dnt{\hasty{\Delta, \bhyp{x}{A}}{\epsilon}{b}{B}} \\
        &= \lmor{\dnt{\hasty{\Gamma}{\epsilon}{a}{A}}}
        ; \dnt{\Gamma, \bhyp{x}{A} \leq \Delta, \bhyp{x}{A}}
        ; \dnt{\hasty{\Delta, \bhyp{x}{A}}{\epsilon}{b}{B}} \\
        &= \lmor{\dnt{\hasty{\Gamma}{\epsilon}{a}{A}}}
        ; \dnt{\hasty{\Gamma, \bhyp{x}{A}}{\epsilon}{b}{B}} \\
        &= \dnt{\hasty{\Gamma}{\epsilon}{\letexpr{x}{a}{b}}{B}}
      \end{aligned}
    \end{equation}
    \item \brle{unit}: follows immediately since weakenings are pure and $\mb{1}$ is the terminal
    object.
    \item \brle{pair}: we have
    \begin{equation}
      \begin{aligned}
        & \dnt{\Gamma \leq \Delta} ; \dnt{\hasty{\Delta}{\epsilon}{(a, b)}{A \otimes B}} \\
        & = \dnt{\Gamma \leq \Delta} 
          ; \dmor{\Delta}
          ; \dnt{\hasty{\Delta}{\epsilon}{a}{A}} \otimes \dnt{\hasty{\Delta}{\epsilon}{b}{B}} \\
        & = \dmor{\Gamma}
          ; (\dnt{\Gamma \leq \Delta} ; \dnt{\hasty{\Delta}{\epsilon}{a}{A}}) \otimes
          (\dnt{\Gamma \leq \Delta} ; \dnt{\hasty{\Delta}{\epsilon}{b}{B}}) \\
        & = \dmor{\Gamma}
          ; \dnt{\hasty{\Gamma}{\epsilon}{a}{A}} \otimes \dnt{\hasty{\Gamma}{\epsilon}{b}{B}} \\
        & = \dnt{\hasty{\Gamma}{\epsilon}{(a, b)}{A \otimes B}}
      \end{aligned}
    \end{equation}
    \item \brle{let$_2$}: we have
    \begin{equation}
      \begin{aligned}
        & \dnt{\Gamma \leq \Delta} ; \dnt{\hasty{\Delta}{\epsilon}{\letexpr{(x, y)}{a}{b}}{B}} \\
        & = \dnt{\Gamma \leq \Delta} 
          ; \lmor{\dnt{\hasty{\Delta}{\epsilon}{a}{A \otimes B}}}
          ; \alpha 
          ; \dnt{\hasty{\Delta, \bhyp{x}{A}, \bhyp{y}{B}}{\epsilon}{b}{B}} \\
        & = \lmor{\dnt{\Gamma \leq \Delta} ; \dnt{\hasty{\Delta}{\epsilon}{a}{A \otimes B}}}
          ; \alpha
          ; \dnt{\Gamma \leq \Delta} \otimes \dnt{A} \otimes \dnt{B}
          ; \dnt{\hasty{\Delta, \bhyp{x}{A}, \bhyp{y}{B}}{\epsilon}{b}{B}} \\
        & = \lmor{\dnt{\hasty{\Gamma}{\epsilon}{a}{A \otimes B}}}
          ; \alpha
          ; \dnt{\Gamma, \bhyp{x}{A}, \bhyp{y}{B} \leq \Delta, \bhyp{x}{A}, \bhyp{y}{B}}
          ; \dnt{\hasty{\Delta, \bhyp{x}{A}, \bhyp{y}{B}}{\epsilon}{b}{B}} \\
        & = \lmor{\dnt{\hasty{\Gamma}{\epsilon}{a}{A \otimes B}}}
          ; \alpha
          ; \dnt{\hasty{\Gamma, \bhyp{x}{A}, \bhyp{y}{B}}{\epsilon}{b}{B}} \\
        & = \dnt{\hasty{\Gamma}{\epsilon}{\letexpr{(x, y)}{a}{b}}{B}}
      \end{aligned}
    \end{equation}
    \item \brle{case}: we have
    \begin{equation}
      \begin{aligned}
        & \dnt{\Gamma \leq \Delta} ; \dnt{\hasty{\Delta}{\epsilon}{\caseexpr{a}{x}{s}{y}{t}}{B}} \\
        & = \dnt{\Gamma \leq \Delta} 
          ; \lmor{\dnt{\hasty{\Delta}{\epsilon}{a}{A}}}
          ; \delta^{-1}
          ; [\dnt{\hasty{\Delta, \bhyp{x}{A}}{\epsilon}{s}{B}}, 
             \dnt{\hasty{\Delta, \bhyp{y}{A}}{\epsilon}{t}{B}}] \\
        & = \lmor{\dnt{\Gamma \leq \Delta} ; \dnt{\hasty{\Gamma}{\epsilon}{a}{A}}}
          ; \delta^{-1}
          ; \\ & \qquad [
              \dnt{\Gamma \leq \Delta} \otimes \dnt{A} 
                ; \dnt{\hasty{\Delta, \bhyp{x}{A}}{\epsilon}{s}{C}},
              \dnt{\Gamma \leq \Delta} \otimes \dnt{B} 
                ; \dnt{\hasty{\Delta, \bhyp{y}{B}}{\epsilon}{t}{C}}
             ] \\
        & = \lmor{\dnt{\hasty{\Delta}{\epsilon}{a}{A}}}
          ; \\ & \qquad [
            \dnt{\Gamma, \bhyp{x}{A} \leq \Delta, \bhyp{x}{A}} 
              ; \dnt{\hasty{\Delta, \bhyp{x}{A}}{\epsilon}{s}{C}},
            \dnt{\Gamma, \bhyp{y}{B} \leq \Delta, \bhyp{y}{B}}
              ; \dnt{\hasty{\Delta, \bhyp{y}{B}}{\epsilon}{t}{C}}
           ] \\
        & = \lmor{\dnt{\hasty{\Delta}{\epsilon}{a}{A}}}
          ; [
            \dnt{\hasty{\Gamma, \bhyp{x}{A}}{\epsilon}{s}{C}},
            \dnt{\hasty{\Gamma, \bhyp{y}{B}}{\epsilon}{t}{C}}
           ]
      \end{aligned}
    \end{equation}
    \item \brle{op}: we have
    \begin{equation}
      \begin{aligned}
        \dnt{\Gamma \leq \Delta} ; \dnt{\hasty{\Gamma}{\epsilon}{f\;a}{B}} 
        & = \dnt{\Gamma \leq \Delta} 
          ; \dnt{\hasty{\Delta}{\epsilon}{a}{A}} 
          ; \dnt{\isop{f}{A}{B}{\epsilon}} \\
        & = \dnt{\hasty{\Gamma}{\epsilon}{a}{A}} 
          ; \dnt{\isop{f}{A}{B}{\epsilon}} \\
        & = \dnt{\hasty{\Gamma}{\epsilon}{f\;a}{B}}
      \end{aligned}
    \end{equation}
    \item \brle{inl}, \brle{inr}: analogous to the \brle{op} case
  \end{itemize}
  Similarly, we can show weakening for regions $\haslb{\Delta}{r}{\ms{L}}$ \ref{itm:regwk} by
  induction on the typing derivation as follows:
  \begin{itemize}
    \item \brle{br}: we have that
    \begin{equation}
      \begin{aligned}
      \dnt{\Gamma \leq \Delta}
        ; \dnt{\haslb{\Delta}{\brb{\ell}{a}}{\ms{L}}}
        ; \dnt{\ms{L} \leq \ms{K}}
      & = \dnt{\Gamma \leq \Delta} 
        ; \dnt{\hasty{\Delta}{\bot}{a}{A}}
        ; \iota_{\ms{L}, \ell}
        ; \dnt{\ms{L} \leq \ms{K}} \\
      & = \dnt{\hasty{\Gamma}{\bot}{a}{A}}
        ; \iota_{\ms{L}, \ell}
        ; \dnt{\ms{L} \leq \ms{K}}
      \end{aligned}
    \end{equation}
    It hence suffices to show that 
    $\iota_{\ms{L}, \ell} ; \dnt{\ms{L} \leq \ms{K}} = \iota_{\ms{K}, \ell}$, which we can do by
    induction on $\ms{L} \leq \ms{K}$:
    \begin{itemize}
      \item \brle{lwk-nil}: this case yields a contradiction, since if $\ms{K} = \cdot$ it cannot
      define the label $\ell$.
      \item \brle{lwk-skip}: we have $\ms{K} = \ms{K}', \kappa(B)$, and hence
      \begin{equation}
        \iota_{\ms{L}', \ell} ; \dnt{\ms{L} \leq \ms{K}', \kappa(B)}
        = \iota_{\ms{L}', \ell} ; \dnt{\ms{L} \leq \ms{K}', \kappa(B)} ; \iota_l
        = \iota_{\ms{K}', \ell} ; \iota_l
        = \iota_{\ms{K}, \ell}
      \end{equation}
      \item \brle{lwk-cons}: we have $\ms{L} = \ms{L}', \kappa(B)$ and $\ms{K} = \ms{K}', \kappa(B)$
      .
      \begin{itemize}
        \item If $\kappa = \ell$, then $B = A$ and
        \begin{equation}
          \iota_{(\ms{L}', \ell(A)), \ell} 
          ; \dnt{\ms{L}, \ell(A) \leq \ms{K}', \ell(A)}
          = \iota_r ; \dnt{\ms{L}' \leq \ms{K}'} + \dnt{A}
          = \iota_r
           = \iota_{\ms{K}, \ell}
        \end{equation}
        \item Otherwise, we have by induction that
        \begin{equation}
          \begin{aligned}
          \iota_{(\ms{L}', \kappa(B)), \ell}
          ; \dnt{\ms{L}, \kappa(B) \leq \ms{K}', \kappa(B)}
          & = \iota_{\ms{L}', \ell} ; \iota_l ; \dnt{\ms{L}' \leq \ms{K}'} + \dnt{B} \\
          & = \iota_{\ms{L}', \ell} ; \dnt{\ms{L}' \leq \ms{K}'} ; \iota_l \\
          & = \iota_{\ms{K}', \ell} ; \iota_l & = \iota_{\ms{K}, \ell}
          \end{aligned}
        \end{equation}
      \end{itemize}
    \end{itemize}
    \item \brle{let$_1$-r}: we have by induction that
    \begin{equation}
      \begin{aligned}
        & \dnt{\Gamma \leq \Delta} 
          ; \dnt{\haslb{\Delta}{\letstmt{x}{a}{r}}{\ms{L}}} 
          ; \dnt{\ms{L} \leq \ms{K}} \\
        & = \dnt{\Gamma \leq \Delta}
          ; \lmor{\dnt{\hasty{\Delta}{\epsilon}{a}{A}}}
          ; \dnt{\haslb{\Delta, \bhyp{x}{A}}{r}{\ms{L}}}
          ; \dnt{\ms{L} \leq \ms{K}} \\
        & = \lmor{\dnt{\Gamma \leq \Delta} ; \dnt{\hasty{\Delta}{\epsilon}{a}{A}}}
          ; \dnt{\Gamma \leq \Delta} \otimes \dnt{A}
          ; \dnt{\haslb{\Delta, \bhyp{x}{A}}{r}{\ms{L}}}
          ; \dnt{\ms{L} \leq \ms{K}} \\
        & = \lmor{\dnt{\hasty{\Gamma}{\epsilon}{a}{A}}}
          ; \dnt{\Gamma, \bhyp{x}{A} \leq \Delta, \bhyp{x}{A}}
          ; \dnt{\haslb{\Delta, \bhyp{x}{A}}{r}{\ms{L}}}
          ; \dnt{\ms{L} \leq \ms{K}} \\
        & = \lmor{\dnt{\hasty{\Gamma}{\epsilon}{a}{A}}}
          ; \dnt{\haslb{\Gamma, \bhyp{x}{A}}{r}{\ms{K}}} \\
        & = \dnt{\haslb{\Gamma}{\letstmt{x}{a}{r}}{\ms{K}}}
      \end{aligned}
    \end{equation}
    \item \brle{let$_2$-r}: we have by induction that
    \begin{equation}
      \begin{aligned}
        & \dnt{\Gamma \leq \Delta} 
          ; \dnt{\haslb{\Delta}{\letstmt{(x, y)}{a}{r}}{\ms{L}}} 
          ; \dnt{\ms{L} \leq \ms{K}} \\
        & = \dnt{\Gamma \leq \Delta}
          ; \lmor{\dnt{\hasty{\Delta}{\epsilon}{a}{A \otimes B}}}
          ; \alpha
          ; \dnt{\haslb{\Delta, \bhyp{x}{A}, \bhyp{y}{B}}{r}{\ms{L}}}
          ; \dnt{\ms{L} \leq \ms{K}} \\
        & = \lmor{\dnt{\Gamma \leq \Delta} ; \dnt{\hasty{\Delta}{\epsilon}{a}{A \otimes B}}}
          ; \alpha
          ; \dnt{\Gamma \leq \Delta} \otimes \dnt{A} \otimes \dnt{B}
          ; \dnt{\haslb{\Delta, \bhyp{x}{A}, \bhyp{y}{B}}{r}{\ms{L}}}
          ; \dnt{\ms{L} \leq \ms{K}} \\
        & = \lmor{\dnt{\hasty{\Gamma}{\epsilon}{a}{A \otimes B}}}
          ; \alpha
          ; \dnt{\Gamma, \bhyp{x}{A}, \bhyp{y}{B} \leq \Delta, \bhyp{x}{A}, \bhyp{y}{B}}
          ; \dnt{\haslb{\Delta, \bhyp{x}{A}, \bhyp{y}{B}}{r}{\ms{L}}}
          ; \dnt{\ms{L} \leq \ms{K}} \\
        & = \lmor{\dnt{\hasty{\Gamma}{\epsilon}{a}{A \otimes B}}}
          ; \alpha
          ; \dnt{\haslb{\Gamma, \bhyp{x}{A}, \bhyp{y}{B}}{r}{\ms{K}}}
      \end{aligned}
    \end{equation}
    \item \brle{case-r}: we have by induction that
    \begin{equation}
      \begin{aligned}
        & \dnt{\Gamma \leq \Delta} 
          ; \dnt{\haslb{\Delta}{\caseexpr{a}{x}{r}{y}{s}}{\ms{L}}}
          ; \dnt{\ms{L} \leq \ms{K}} \\
        & = \dnt{\Gamma \leq \Delta}
          ; \lmor{\dnt{\hasty{\Delta}{\epsilon}{a}{A}}}
          ; \delta^{-1}
          ; \\ & \qquad [ 
             \dnt{\haslb{\Delta, \bhyp{x}{A}}{s}{\ms{L}}}, 
             \dnt{\haslb{\Delta, \bhyp{y}{B}}{t}{\ms{L}}}
          ]
          ; \dnt{\ms{L} \leq \ms{K}} \\
        & = \lmor{\dnt{\Gamma \leq \Delta} ; \dnt{\hasty{\Delta}{\epsilon}{a}{A}}}
          ; \delta^{-1}
          ; [ \\ & \qquad 
             \dnt{\Gamma \leq \Delta} \otimes \dnt{A}
               ; \dnt{\haslb{\Delta, \bhyp{x}{A}}{s}{\ms{L}}}
               ; \dnt{\ms{L} \leq \ms{K}}, \\ & \qquad
             \dnt{\Gamma \leq \Delta} \otimes \dnt{B}
               ; \dnt{\haslb{\Delta, \bhyp{y}{B}}{t}{\ms{L}}}
               ; \dnt{\ms{L} \leq \ms{K}}
          ] \\
        & = \lmor{\dnt{\Gamma \leq \Delta} ; \dnt{\hasty{\Delta}{\epsilon}{a}{A}}}
           ; \delta^{-1}
           ; [ \\ & \qquad 
              \dnt{\Gamma, \bhyp{x}{A} \leq \Delta, \bhyp{x}{A}}
                ; \dnt{\haslb{\Delta, \bhyp{x}{A}}{s}{\ms{L}}}
                ; \dnt{\ms{L} \leq \ms{K}}, \\ & \qquad
              \dnt{\Gamma, \bhyp{y}{B} \leq \Delta, \bhyp{y}{B}}
                ; \dnt{\haslb{\Delta, \bhyp{y}{B}}{t}{\ms{L}}}
                ; \dnt{\ms{L} \leq \ms{K}}
           ]
            \\
        & = \lmor{\dnt{\hasty{\Gamma}{\epsilon}{a}{A}}}
          ; \delta^{-1}
          ; [ 
            \dnt{\haslb{\Gamma, \bhyp{x}{A}}{s}{\ms{K}}},
            \dnt{\haslb{\Gamma, \bhyp{y}{B}}{t}{\ms{K}}}
          ] \\
        & = \dnt{\haslb{\Gamma}{\caseexpr{a}{x}{s}{y}{t}}{\ms{K}}}
      \end{aligned}
    \end{equation}
    \item \brle{cfg}: Let $L = \loopmor{\Delta}{(\wbranch{\ell_i}{x_i}{t_i},)_i}{\ms{L}}$ and
    $\ms{R} = (\ell_i(A_i),)_i$. We have by induction that
    \begin{equation}
      \begin{aligned}
        & \dnt{\Gamma \leq \Delta} ; L ; \dnt{\ms{L} \leq \ms{K}} + \Sigma_i\dnt{A_i} \\
        & = \dnt{\Gamma \leq \Delta} 
          ; \delta^{-1}_{\Sigma}
          ; [(\dnt{\haslb{\Delta, \bhyp{x_i}{A_i}}{t_i}{\ms{L}, \ms{R}}},)_i]
          ; \alpha^+_{\dnt{\ms{L}} + \Sigma_i\dnt{A_i}}
          ; \dnt{\ms{L} \leq \ms{K}} + \Sigma_i\dnt{A_i} \\
        & = \delta^{-1}_{\Sigma}
          ; [
            \dnt{\Gamma \leq \Delta} \otimes \dnt{A_i}
            ; (\dnt{\haslb{\Delta, \bhyp{x_i}{A_i}}{t_i}{\ms{L}, \ms{R}}}
            ; \alpha_{\dnt{\ms{L}} + \dnt{\ms{R}}}
            ; \dnt{\ms{L} \leq \ms{K}} + \dnt{\ms{R}},)_i]
          ; \alpha^+_{\dnt{\ms{L}} + \Sigma_i\dnt{A_i}} \\
        & = \delta^{-1}_{\Sigma}
          ; [
            (\dnt{\Gamma, \bhyp{x_i}{A_i} \leq \Delta, \bhyp{x_i}{A_i}}
            ; \dnt{\haslb{\Delta, \bhyp{x_i}{A_i}}{t_i}{\ms{L}, \ms{R}}}
            ; \dnt{\ms{L}, \ms{R} \leq \ms{K}, \ms{R}},)_i]
          ; \alpha^+_{\dnt{\ms{L}} + \Sigma_i\dnt{A_i}} \\
        & = \delta^{-1}_{\Sigma}
          ; [(\dnt{\haslb{\Gamma, \bhyp{x_i}{A_i}}{t_i}{\ms{K}, \ms{R}}},)_i]
          ; \alpha^+_{\dnt{\ms{L}} + \Sigma_i\dnt{A_i}} \\
        & = \loopmor{\Gamma}{(\wbranch{\ell_i}{x_i}{t_i},)_i}{\ms{K}}
      \end{aligned}
    \end{equation}
    It follows that
    \begin{equation}
      \begin{aligned}
        & \dnt{\Gamma \leq \Delta} 
          ; \dnt{\haslb{\Delta}{\where{r}{(\wbranch{\ell_i}{x_i}{t_i},)_i}}{\ms{L}}}
          ; \dnt{\ms{L} \leq \ms{K}} \\
        & = \dnt{\Gamma \leq \Delta} 
          ; \lmor{\dnt{\haslb{\Delta}{r}{\ms{L}, \ms{R}}}}  
          ; \dnt{\Delta} \otimes \alpha^+_{\dnt{L} + \Sigma_i\dnt{A_i}})
          ; \delta^{-1}
          ; [\pi_r, \rfix{L}]
          ; \dnt{\ms{L} \leq \ms{K}} \\
        & = \lmor{\dnt{\Gamma \leq \Delta} 
            ; \dnt{\haslb{\Delta}{r}{\ms{L}, \ms{R}}}
            ; \alpha^+_{\dnt{L} + \Sigma_i\dnt{A_i}}}
          ; \delta^{-1}
          ; [\pi_r, \dnt{\Gamma \leq \Delta} ; \rfix{L}]
          ; \dnt{\ms{L} \leq \ms{K}} \\
        & = \lmor{\dnt{\haslb{\Gamma}{r}{\ms{L}, \ms{R}}}; \alpha^+_{\dnt{L} + \Sigma_i\dnt{A_i}}}
          ; [ \pi_r ; \dnt{\ms{L} \leq \ms{K}}, 
              \rfix{\loopmor{\Gamma}{(\wbranch{\ell_i}{x_i}{t_i},)_i}{\ms{K}}}
          ]
          \\
        & = \lmor{\dnt{\haslb{\Gamma}{r}{\ms{L}, \ms{R}}} 
            ; \alpha_{\dnt{\ms{L}} + \Sigma_i\dnt{A_i}}
            ; \dnt{\ms{L} \leq \ms{K}} + \Sigma_i\dnt{A_i}}
          ; [ \pi_r, 
              \rfix{\loopmor{\Gamma}{(\wbranch{\ell_i}{x_i}{t_i},)_i}{\ms{K}}}
          ]
          \\
        & = \lmor{\dnt{\haslb{\Gamma}{r}{\ms{L}, \ms{R}}} 
            ; \alpha_{\dnt{\ms{L}} + \Sigma_i\dnt{A_i}}
            ; \dnt{\ms{L} \leq \ms{K}} + \Sigma_i\dnt{A_i}}
          ; [ \pi_r, 
              \rfix{\loopmor{\Gamma}{(\wbranch{\ell_i}{x_i}{t_i},)_i}{\ms{K}}}
          ]
          \\
        & = \lmor{\dnt{\haslb{\Gamma}{r}{\ms{L}, \ms{R}}}
            ; \dnt{\ms{L}, \ms{R} \leq \ms{K}, \ms{R}} 
            ; \alpha_{\dnt{\ms{K}} + \Sigma_i\dnt{A_i}}}
          ; [ \pi_r, 
              \rfix{\loopmor{\Gamma}{(\wbranch{\ell_i}{x_i}{t_i},)_i}{\ms{K}}}
          ]
          \\
        & = \lmor{\dnt{\haslb{\Gamma}{r}{\ms{K}, \ms{R}}}
            ; \alpha_{\dnt{\ms{K}} + \Sigma_i\dnt{A_i}}}
          ; [ \pi_r, 
              \rfix{\loopmor{\Gamma}{(\wbranch{\ell_i}{x_i}{t_i},)_i}{\ms{K}}}
          ]
          \\
        & = \dnt{\haslb{\Gamma}{\where{r}{(\wbranch{\ell_i}{x_i}{t_i},)_i}}{\ms{K}}}
      \end{aligned}
    \end{equation}
  \end{itemize}
  Weakening for substitutions \ref{itm:substwk} and label substitututions \ref{itm:lbsubstwk}
  then follow by a trivial induction.
\end{proof}

\begin{lemma}[Substitution Projection]
  For $\issubst{\gamma}{\Gamma}{\Delta}$, $\ms{eff}(\Delta) = \bot$ and $\Delta(x) = A$, we have
  \begin{equation}
    \dnt{\issubst{\gamma}{\Gamma}{\Delta}};\pi_{\Delta, x} 
    = \dnt{\hasty{\Gamma}{\epsilon}{[\gamma]x}{A}}
  \end{equation}
  \label{lem:subst-proj}
\end{lemma}
\begin{proof}
  We proceed by induction on $\Delta$:
  \begin{itemize}[leftmargin=*]
    \item If $\Delta = \cdot$, then $\Delta(x) = A$ is a contradiction.
    \item If $\Delta = \Delta', x : A$, then $\gamma = \gamma', x \mapsto [\gamma]x$, so we have 
    \begin{equation}
      \dnt{\issubst{\gamma}{\Gamma}{\Delta}};\pi_{\Delta, x} =
      \Delta_{\dnt{\Gamma}}
      ; \dnt{\issubst{\gamma'}{\Gamma}{\Delta'}} \otimes \dnt{\hasty{\Gamma}{\bot}{[\gamma]x}{A}}
      ; \pi_r
      = \dnt{\hasty{\Gamma}{\epsilon}{[\gamma]x}{A}}
    \end{equation}
    as desired.
    \item If $\Delta = \Delta', y : B$ (with $y \neq x$), then $\gamma = \gamma', y \mapsto
    [\gamma]y$, so by induction we have
    \begin{equation}
      \begin{aligned}
      \dnt{\issubst{\gamma}{\Gamma}{\Delta}};\pi_{\Delta, x} 
      & = \Delta_{\dnt{\Gamma}}
       ; \dnt{\issubst{\gamma'}{\Gamma}{\Delta'}} \otimes \dnt{\hasty{\Gamma}{\bot}{[\gamma]y}{B}}
       ; \pi_l ; \pi_{\Delta', x} \\
      & = \dnt{\issubst{\gamma'}{\Gamma}{\Delta'}} ; \pi_{\Delta', x}
        = \dnt{\hasty{\Gamma}{\epsilon}{[\gamma']x}{A}}
        = \dnt{\hasty{\Gamma}{\epsilon}{[\gamma]x}{A}}
      \end{aligned}
    \end{equation}
  \end{itemize}
\end{proof}

\soundnesssubst*

\label{proof:soundness-subst}

\begin{proof}
  Fix $\issubst{\gamma}{\Gamma}{\Delta}$ with $\ms{eff}(\Delta) = \bot$. We will begin by showing
  the soundness of substitution for expressions \ref{itm:tm-subst-sound}: we proceed by induction
  on the derivation $\hasty{\Delta}{\epsilon}{e}{E}$:
  \begin{itemize}[leftmargin=*]
    \item If $e = x$ is a variable, then by Lemma~\ref{lem:subst-proj}, we have
    \begin{equation}
      \dnt{\issubst{\gamma}{\Gamma}{\Delta}} ; \dnt{\hasty{\Delta}{\epsilon}{x}{A}} 
      = \dnt{\issubst{\gamma}{\Gamma}{\Delta}} ; \pi_{\Delta, x}
      = \dnt{\hasty{\Gamma}{\epsilon}{[\gamma]x}{A}}
    \end{equation}
    as desired.
    \item If $e = f\;a$ is an operation, then by induction we have that
    \begin{equation}
      \begin{aligned}
      \dnt{\issubst{\gamma}{\Gamma}{\Delta}} ; \dnt{\hasty{\Gamma}{\epsilon}{f\;a}{B}}
      &= \dnt{\issubst{\gamma}{\Gamma}{\Delta}} 
      ; \dnt{\hasty{\Delta}{\epsilon}{a}{A}} 
      ; \dnt{\isop{f}{\epsilon}{A}{B}}
      \\ &= \dnt{\hasty{\Gamma}{\epsilon}{[\gamma]a}{A}}
      ; \dnt{\isop{f}{\epsilon}{A}{B}}
      = \dnt{\hasty{\Gamma}{\epsilon}{[\gamma](f\;a)}{B}}
      \end{aligned}
    \end{equation}
    as desired. The cases for left injections, right injections, and \ms{abort} are analogous
    \item If $e = (\letexpr{x}{a}{b})$ is a unary \ms{let}-binding, then by induction we have that
    \begin{equation}
      \begin{aligned}
      & \dnt{\issubst{\gamma}{\Gamma}{\Delta}} \otimes \dnt{A} 
      ; \dnt{\hasty{\Delta, \bhyp{x}{A}}{\epsilon}{b}{B}}
      \\ &= \dnt{\issubst{(\gamma, x \mapsto x)}{(\Gamma, \bhyp{x}{A})}{(\Delta, \bhyp{x}{A})}}
      ; \dnt{\hasty{\Delta, \bhyp{x}{A}}{\epsilon}{b}{B}}
      \\ &= \dnt{\hasty{\Gamma, \bhyp{x}{A}}{\epsilon}{[\gamma, x \mapsto x]b}{B}}
          = \dnt{\hasty{\Gamma, \bhyp{x}{A}}{\epsilon}{[\gamma]b}{B}}
      \end{aligned}
    \end{equation}
    as we can assume that $x$ is a fresh variable. Hence, it follows that
    \begin{equation}
      \begin{aligned}
        & \dnt{\issubst{\gamma}{\Gamma}{\Delta}} 
        ; \dnt{\hasty{\Delta}{\epsilon}{\letexpr{x}{a}{b}}{B}} \\
        &= \dnt{\issubst{\gamma}{\Gamma}{\Delta}}
        ; \Delta_{\dnt{\Delta}} ; \dnt{\Delta} \otimes \dnt{\hasty{\Delta}{\epsilon}{a}{A}}
        ; \dnt{\hasty{\Delta, \bhyp{x}{A}}{\epsilon}{b}{B}} \\
        &= \Delta_{\dnt{\Gamma}} 
        ; \dnt{\Gamma} 
          \otimes \dnt{\hasty{\Gamma}{\epsilon}{[\gamma]a}{A}}
        ; \dnt{\issubst{\gamma}{\Gamma}{\Delta}} \otimes \dnt{A}
        ; \dnt{\hasty{\Delta, \bhyp{x}{A}}{\epsilon}{b}{B}} \\
        &= \Delta_{\dnt{\Gamma}}
        ; \dnt{\Gamma} \otimes \dnt{\hasty{\Gamma}{\epsilon}{[\gamma]a}{A}}
        ; \dnt{\hasty{\Gamma, \bhyp{x}{A}}{\epsilon}{[\gamma]b}{B}} \\
        &= \dnt{\hasty{\Gamma}{\epsilon}{[\gamma](\letexpr{x}{a}{b})}{B}}
      \end{aligned}
    \end{equation}
    as desired.
    \item If $e = (a, b)$ is a pair, then by induction we have that
    \begin{equation}
      \begin{aligned}
      &\dnt{\issubst{\gamma}{\Gamma}{\Delta}} ; \dnt{\hasty{\Delta}{\epsilon}{(a, b)}{A \times B}} 
      \\
      &= \dnt{\issubst{\gamma}{\Gamma}{\Delta}} 
        ; \Delta_{\dnt{\Delta}} 
        ; \dnt{\hasty{\Delta}{\epsilon}{a}{A}} \ltimes \dnt{\hasty{\Delta}{\epsilon}{b}{B}}
      \\
      &= \Delta_{\dnt{\Gamma}} ; 
        (\dnt{\issubst{\gamma}{\Gamma}{\Delta}} \otimes \dnt{\hasty{\Delta}{\epsilon}{a}{A}})
        \ltimes 
        (\dnt{\issubst{\gamma}{\Gamma}{\Delta}} \otimes \dnt{\hasty{\Delta}{\epsilon}{b}{B}})
      \\
      &= \Delta_{\dnt{\Gamma}} ; 
        \dnt{\hasty{\Gamma}{\epsilon}{[\gamma]a}{A}} 
        \ltimes \dnt{\hasty{\Gamma}{\epsilon}{[\gamma]b}{B}}
      = \dnt{\hasty{\Gamma}{\epsilon}{[\gamma](a, b)}{A \times B}}
      \end{aligned}
    \end{equation}
    \item If $e = (\letexpr{(x, y)}{a}{b})$ is a binary \ms{let}-binding, then by induction we have
    that
    \begin{equation}
      \begin{aligned}
      & \dnt{\issubst{\gamma}{\Gamma}{\Delta}} \otimes (\dnt{A} \otimes \dnt{B}) 
          ; \alpha ; \dnt{\hasty{\Delta, \bhyp{x}{A}, \bhyp{y}{B}}{\epsilon}{b}{C}}
      \\ &= \alpha ; \dnt{\issubst{(\gamma, x \mapsto x, y \mapsto y)}
        {(\Gamma, \bhyp{x}{A}, \bhyp{y}{B})}{(\Delta, \bhyp{x}{A}, \bhyp{y}{B})}}
          ; \dnt{\hasty{\Delta, \bhyp{x}{A}, \bhyp{y}{B}}{\epsilon}{b}{C}}
      \\ &= \dnt{\hasty{\Gamma, \bhyp{x}{A}, \bhyp{y}{B}}
                        {\epsilon}{[\gamma, x \mapsto x, y \mapsto y]b}{C}}
          = \dnt{\hasty{\Gamma, \bhyp{x}{A}, \bhyp{y}{B}}{\epsilon}{[\gamma]b}{C}}
      \end{aligned}
    \end{equation}
    as we can assume that $x, y$ are fresh variables. Hence, it follows that
    \begin{equation}
      \begin{aligned}
        & \dnt{\issubst{\gamma}{\Gamma}{\Delta}} 
        ; \dnt{\hasty{\Delta}{\epsilon}{\letexpr{(x, y)}{a}{b}}{C}} \\
        &= \dnt{\issubst{\gamma}{\Gamma}{\Delta}}
        ; \Delta_{\dnt{\Delta}} 
        ; \dnt{\Delta} \otimes \dnt{\hasty{\Delta}{\epsilon}{a}{A \otimes B}}
        ; \alpha ; \dnt{\hasty{\Delta, \bhyp{x}{A}, \bhyp{y}{B}}{\epsilon}{b}{C}} \\
        &= \Delta_{\dnt{\Gamma}} 
        ; \dnt{\Gamma} \otimes \dnt{\hasty{\Gamma}{\epsilon}{[\gamma]a}{A \otimes B}}
        ; \dnt{\issubst{\gamma}{\Gamma}{\Delta}} \otimes (\dnt{A} \otimes \dnt{B}) ; \alpha
        ; \dnt{\hasty{\Delta, \bhyp{x}{A}, \bhyp{y}{B}}{\epsilon}{b}{C}} \\
        &= \Delta_{\dnt{\Gamma}}
        ; \dnt{\Gamma} \otimes \dnt{\hasty{\Gamma}{\epsilon}{[\gamma]a}{A \otimes B}}
        ; \dnt{\hasty{\Gamma, \bhyp{x}{A}, \bhyp{y}{B}}{\epsilon}{[\gamma]b}{C}} \\
        &= \dnt{\hasty{\Gamma}{\epsilon}{[\gamma](\letexpr{(x, y)}{a}{b})}{C}}
      \end{aligned}
    \end{equation}
    as desired.
    \item If $e = \caseexpr{a}{x}{b}{y}{c}$ is a \ms{case}-expression, then by induction we have
    that
    \begin{equation}
      \begin{aligned}
      & \dnt{\issubst{\gamma}{\Gamma}{\Delta}} \otimes \dnt{A} 
      ; \dnt{\hasty{\Delta, \bhyp{x}{A}}{\epsilon}{b}{C}}
      \\ &= \dnt{\issubst{(\gamma, x \mapsto x)}{(\Gamma, \bhyp{x}{A})}{(\Delta, \bhyp{x}{A})}}
      ; \dnt{\hasty{\Delta, \bhyp{x}{A}}{\epsilon}{b}{C}}
      \\ &= \dnt{\hasty{\Gamma, \bhyp{x}{A}}{\epsilon}{[\gamma, x \mapsto x]b}{C}}
          = \dnt{\hasty{\Gamma, \bhyp{x}{A}}{\epsilon}{[\gamma]b}{C}}
      \end{aligned}
    \end{equation}
    and
    \begin{equation}
      \begin{aligned}
      & \dnt{\issubst{\gamma}{\Gamma}{\Delta}} \otimes \dnt{B} 
      ; \dnt{\hasty{\Delta, \bhyp{y}{B}}{\epsilon}{c}{C}}
      \\ &= \dnt{\issubst{(\gamma, y \mapsto y)}{(\Gamma, \bhyp{y}{B})}{(\Delta, \bhyp{y}{B})}}
      ; \dnt{\hasty{\Delta, \bhyp{y}{B}}{\epsilon}{c}{C}}
      \\ &= \dnt{\hasty{\Gamma, \bhyp{y}{B}}{\epsilon}{[\gamma, x \mapsto x]c}{C}}
          = \dnt{\hasty{\Gamma, \bhyp{y}{B}}{\epsilon}{[\gamma]c}{C}}
      \end{aligned}
    \end{equation}
    as we can assume that $x, y$ are fresh variables. Hence, it follows that
    \begin{equation}
      \begin{aligned}
        & \dnt{\issubst{\gamma}{\Gamma}{\Delta}}
        ; \Delta_{\dnt{\Delta}} 
        ; \dnt{\Delta} \otimes \dnt{\hasty{\Delta}{\epsilon}{a}{A + B}}
        ; \delta^{-1}_{\dnt{\Delta}} 
        ; [
          \dnt{\hasty{\Delta, \bhyp{x}{A}}{\epsilon}{b}{C}}, 
          \dnt{\hasty{\Delta, \bhyp{y}{B}}{\epsilon}{c}{C}}
        ] 
        \\ &
        = \Delta_{\dnt{\Gamma}}
        ; \dnt{\Gamma} \otimes \dnt{\hasty{\Gamma}{\epsilon}{[\gamma]a}{A + B}}
        ; \delta^{-1}_{\dnt{\Gamma}}
        ;  \\ & \qquad
        [
          \dnt{\issubst{\gamma}{\Gamma}{\Delta}} \otimes \dnt{A} 
            ; \dnt{\hasty{\Delta, \bhyp{x}{A}}{\epsilon}{b}{C}}, 
          \dnt{\issubst{\gamma}{\Gamma}{\Delta}} \otimes \dnt{B} 
            ; \dnt{\hasty{\Delta, \bhyp{y}{B}}{\epsilon}{c}{C}}
        ] 
        \\ &
        = \Delta_{\dnt{\Gamma}}
        ; \dnt{\Gamma} \otimes \dnt{\hasty{\Gamma}{\epsilon}{[\gamma]a}{A + B}}
        ; \delta^{-1}_{\dnt{\Gamma}} 
        ; [
          \dnt{\hasty{\Gamma, \bhyp{x}{A}}{\epsilon}{[\gamma]b}{C}}, 
          \dnt{\hasty{\Gamma, \bhyp{y}{B}}{\epsilon}{[\gamma]c}{C}}
        ]
        \\ &
        = \dnt{\hasty{\Gamma}{\epsilon}{[\gamma](\caseexpr{a}{x}{b}{y}{c})}{C}}
      \end{aligned}
    \end{equation} 
    as desired.
    \item If $e = ()$ is the null expression, since $\dnt{\issubst{\gamma}{\Gamma}{\Delta}}$ is
    pure, the desired result holds trivially since $\mb{1}$ is the terminal object in the category
    of pure morphisms.
  \end{itemize}
  We may now prove the soundness of substitution for regions as follows: assuming
  $\haslb{\Gamma}{r}{\ms{L}}$, we proceed by induction on $r$ as follows:
  \begin{itemize}[leftmargin=*]
    \item If $r = \brb{\ell}{a}$, then we have that
    \begin{equation}
      \begin{aligned}
      \dnt{\issubst{\gamma}{\Gamma}{\Delta}} ; \dnt{\haslb{\Delta}{\brb{\ell}{A}}{\ms{L}}}
      &= \dnt{\issubst{\gamma}{\Gamma}{\Delta}}  
      ; \dnt{\hasty{\Delta}{\bot}{a}{A}} 
      ; \iota_{\ms{L}, \ell} \\
      &= \dnt{\hasty{\Gamma}{\bot}{[\gamma]a}{A}}
      ; \iota_{\ms{L}, \ell}
      = \dnt{\haslb{\Gamma}{[\gamma](\brb{\ell}{a})}{\ms{L}}} 
      \end{aligned}
    \end{equation}
    \item If $r = (\letstmt{x}{a}{t})$, then we have that, by induction,
    \begin{equation}
      \begin{aligned}
      & \dnt{\issubst{\gamma}{\Gamma}{\Delta}} \otimes \dnt{A} 
      ; \dnt{\haslb{\Delta, \bhyp{x}{A}}{t}{\ms{L}}}
      \\ &= \dnt{\issubst{(\gamma, x \mapsto x)}{(\Gamma, \bhyp{x}{A})}{(\Delta, \bhyp{x}{A})}}
      ; \dnt{\haslb{\Delta, \bhyp{x}{A}}{t}{\ms{L}}}
      \\ &= \dnt{\haslb{\Gamma, \bhyp{x}{A}}{[\gamma, x \mapsto x]t}{\ms{L}}}
          = \dnt{\haslb{\Gamma, \bhyp{x}{A}}{[\gamma]t}{\ms{L}}}
      \end{aligned}
    \end{equation}
    since $x$ can be taken to be a free variable. Hence,
    \begin{equation}
      \begin{aligned}
        & \dnt{\issubst{\gamma}{\Gamma}{\Delta}} 
        ; \dnt{\haslb{\Delta}{\letexpr{x}{a}{t}}{\ms{L}}} \\
        &= \dnt{\issubst{\gamma}{\Gamma}{\Delta}}
        ; \Delta_{\dnt{\Delta}} ; \dnt{\Delta} \otimes \dnt{\hasty{\Delta}{\epsilon}{a}{A}}
        ; \dnt{\haslb{\Delta, \bhyp{x}{A}}{t}{\ms{L}}} \\
        &= \Delta_{\dnt{\Gamma}} 
        ; \dnt{\Gamma} 
          \otimes \dnt{\hasty{\Gamma}{\epsilon}{[\gamma]a}{A}}
        ; \dnt{\issubst{\gamma}{\Gamma}{\Delta}} \otimes \dnt{A}
        ; \dnt{\haslb{\Delta, \bhyp{x}{A}}{t}{\ms{L}}} \\
        &= \Delta_{\dnt{\Gamma}}
        ; \dnt{\Gamma} \otimes \dnt{\hasty{\Gamma}{\epsilon}{[\gamma]a}{A}}
        ; \dnt{\haslb{\Gamma, \bhyp{x}{A}}{[\gamma]t}{\ms{L}}} \\
        &= \dnt{\haslb{\Gamma}{[\gamma](\letexpr{x}{a}{t})}{\ms{L}}}
      \end{aligned}
    \end{equation}
    \item If $r = (\letstmt{(x, y)}{a}{t})$, then we have that, by induction,
    \begin{equation}
      \begin{aligned}
      & \dnt{\issubst{\gamma}{\Gamma}{\Delta}} \otimes (\dnt{A} \otimes \dnt{B}) 
          ; \alpha ; \dnt{\haslb{\Delta, \bhyp{x}{A}, \bhyp{y}{B}}{t}{\ms{L}}}
      \\ &= \alpha ; \dnt{\issubst{(\gamma, x \mapsto x, y \mapsto y)}
        {(\Gamma, \bhyp{x}{A}, \bhyp{y}{B})}{(\Delta, \bhyp{x}{A}, \bhyp{y}{B})}}
          ; \dnt{\haslb{\Delta, \bhyp{x}{A}, \bhyp{y}{B}}{t}{\ms{L}}}
      \\ &= \dnt{\haslb{\Gamma, \bhyp{x}{A}, \bhyp{y}{B}}
                        {[\gamma, x \mapsto x, y \mapsto y]t}{\ms{L}}}
          = \dnt{\haslb{\Gamma, \bhyp{x}{A}, \bhyp{y}{B}}{[\gamma]t}{\ms{L}}}
      \end{aligned}
    \end{equation}
    since $x, y$ can be taken to be free variables. Hence,
    \begin{equation}
      \begin{aligned}
        & \dnt{\issubst{\gamma}{\Gamma}{\Delta}} 
        ; \dnt{\haslb{\Delta}{\letexpr{(x, y)}{a}{t}}{\ms{L}}} \\
        &= \dnt{\issubst{\gamma}{\Gamma}{\Delta}}
        ; \Delta_{\dnt{\Delta}} 
        ; \dnt{\Delta} \otimes \dnt{\hasty{\Delta}{\epsilon}{a}{A \otimes B}}
        ; \alpha ; \dnt{\haslb{\Delta, \bhyp{x}{A}, \bhyp{y}{B}}{t}{\ms{L}}} \\
        &= \Delta_{\dnt{\Gamma}} 
        ; \dnt{\Gamma} \otimes \dnt{\hasty{\Gamma}{\epsilon}{[\gamma]a}{A \otimes B}}
        ; \dnt{\issubst{\gamma}{\Gamma}{\Delta}} \otimes (\dnt{A} \otimes \dnt{B}) ; \alpha
        ; \dnt{\haslb{\Delta, \bhyp{x}{A}, \bhyp{y}{B}}{t}{\ms{L}}} \\
        &= \Delta_{\dnt{\Gamma}}
        ; \dnt{\Gamma} \otimes \dnt{\hasty{\Gamma}{\epsilon}{[\gamma]a}{A \otimes B}}
        ; \dnt{\haslb{\Gamma, \bhyp{x}{A}, \bhyp{y}{B}}{[\gamma]t}{\ms{L}}} \\
        &= \dnt{\haslb{\Gamma}{[\gamma](\letexpr{(x, y)}{a}{t})}{\ms{L}}}
      \end{aligned}
    \end{equation}
    \item If $r = \casestmt{a}{x}{s}{y}{t}$, then we have that, by induction
    \begin{equation}
      \begin{aligned}
      & \dnt{\issubst{\gamma}{\Gamma}{\Delta}} \otimes \dnt{A} 
      ; \dnt{\haslb{\Delta, \bhyp{x}{A}}{s}{\ms{L}}}
      \\ &= \dnt{\issubst{(\gamma, x \mapsto x)}{(\Gamma, \bhyp{x}{A})}{(\Delta, \bhyp{x}{A})}}
      ; \dnt{\haslb{\Delta, \bhyp{x}{A}}{s}{\ms{L}}}
      \\ &= \dnt{\haslb{\Gamma, \bhyp{x}{A}}{[\gamma, x \mapsto x]s}{\ms{L}}}
          = \dnt{\haslb{\Gamma, \bhyp{x}{A}}{[\gamma]s}{\ms{L}}}
      \end{aligned}
    \end{equation}
    and
    \begin{equation}
      \begin{aligned}
      & \dnt{\issubst{\gamma}{\Gamma}{\Delta}} \otimes \dnt{B} 
      ; \dnt{\haslb{\Delta, \bhyp{y}{B}}{t}{\ms{L}}}
      \\ &= \dnt{\issubst{(\gamma, y \mapsto y)}{(\Gamma, \bhyp{y}{B})}{(\Delta, \bhyp{y}{B})}}
      ; \dnt{\haslb{\Delta, \bhyp{y}{B}}{t}{\ms{L}}}
      \\ &= \dnt{\haslb{\Gamma, \bhyp{y}{B}}{[\gamma, y \mapsto y]t}{\ms{L}}}
          = \dnt{\haslb{\Gamma, \bhyp{y}{B}}{[\gamma]t}{\ms{L}}}
      \end{aligned}
    \end{equation}
    since $x, y$ can be taken to be free variables. Hence,
    \begin{equation}
      \begin{aligned}
        & \dnt{\issubst{\gamma}{\Gamma}{\Delta}}
        ; \Delta_{\dnt{\Delta}} 
        ; \dnt{\Delta} \otimes \dnt{\hasty{\Delta}{\epsilon}{a}{A + B}}
        ; \delta^{-1}_{\dnt{\Delta}} 
        ; [
          \dnt{\haslb{\Delta, \bhyp{x}{A}}{s}{\ms{L}}}, 
          \dnt{\haslb{\Delta, \bhyp{y}{B}}{t}{\ms{L}}}
        ] 
        \\ &
        = \Delta_{\dnt{\Gamma}}
        ; \dnt{\Gamma} \otimes \dnt{\hasty{\Gamma}{\epsilon}{[\gamma]a}{A + B}}
        ; \delta^{-1}_{\dnt{\Gamma}}
        ;  \\ & \qquad
        [
          \dnt{\issubst{\gamma}{\Gamma}{\Delta}} \otimes \dnt{A} 
            ; \dnt{\haslb{\Delta, \bhyp{x}{A}}{s}{\ms{L}}}, 
          \dnt{\issubst{\gamma}{\Gamma}{\Delta}} \otimes \dnt{B} 
            ; \dnt{\haslb{\Delta, \bhyp{y}{B}}{t}{\ms{L}}}
        ] 
        \\ &
        = \Delta_{\dnt{\Gamma}}
        ; \dnt{\Gamma} \otimes \dnt{\hasty{\Gamma}{\epsilon}{[\gamma]a}{A + B}}
        ; \delta^{-1}_{\dnt{\Gamma}} 
        ; [
          \dnt{\haslb{\Gamma, \bhyp{x}{A}}{[\gamma]s}{\ms{L}}}, 
          \dnt{\haslb{\Gamma, \bhyp{y}{B}}{[\gamma]t}{\ms{L}}}
        ]
        \\ &
        = \dnt{\haslb{\Gamma}{[\gamma](\casestmt{a}{x}{s}{y}{t})}{\ms{L}}}
      \end{aligned}
    \end{equation}
    as desired.
    \item Assume $r = \where{s}{(\wbranch{\ell_i}{x_i}{t_i},)_i}$. Define $\ms{R} =
    (\ell_i(A_i),)_i$ and $S = \dnt{\issubst{\gamma}{\Gamma}{\Delta}}$. We have by induction that,
    for all $i$,
    \begin{equation}
      \begin{aligned}
      & \dnt{\issubst{\gamma}{\Gamma}{\Delta}} \otimes \dnt{A_i} 
      ; \dnt{\haslb{\Delta, \bhyp{x_i}{A_i}}{t_i}{\ms{L}, \ms{R}}}
      \\ &= \dnt{\issubst{(\gamma, x_i \mapsto x_i)}  
                  {(\Gamma, \bhyp{x_i}{A_i})}{(\Delta, \bhyp{x_i}{A_i})}}
      ; \dnt{\haslb{\Delta, \bhyp{x_i}{A_i}}{t_i}{\ms{L}, \ms{R}}}
      \\ &= \dnt{\haslb{\Gamma, \bhyp{x_i}{A_i}}{[\gamma, x_i \mapsto x_i]t_i}{\ms{L}, \ms{R}}}
          = \dnt{\haslb{\Gamma, \bhyp{x_i}{A_i}}{[\gamma]t_i}{\ms{L}, \ms{R}}}
      \end{aligned}
    \end{equation}
    and therefore that
    \begin{equation}
      \begin{aligned}
        & S \otimes \Sigma_i\dnt{A_i} 
        ; \loopmor{\Delta}{(\wbranch{\ell_i}{x_i}{t_i},)_i}{\ms{L}} \\
        & = S \otimes \Sigma_i\dnt{A_i} 
        ; \delta^{-1}_{\Sigma} ; 
        [ 
          \dnt{\haslb{\Delta, \bhyp{x_i}{A_i}}{t_i}{\ms{L}, \ms{R}}},
        ]_i
        ; \alpha^+_{\dnt{\ms{L}} + \Sigma_i \dnt{A_i}} \\
        & = \delta^{-1}_{\Sigma} ; 
        [ 
          S \otimes \dnt{A_i} ;
          \dnt{\haslb{\Delta, \bhyp{x_i}{A_i}}{t_i}{\ms{L}, \ms{R}}},
        ]_i
        ; \alpha^+_{\dnt{\ms{L}} + \Sigma_i \dnt{A_i}} \\
        & = \delta^{-1}_{\Sigma} ; 
        [ 
          \dnt{\haslb{\Gamma, \bhyp{x_i}{A_i}}{[\gamma]t_i}{\ms{L}, \ms{R}}},
        ]_i
        ; \alpha^+_{\dnt{\ms{L}} + \Sigma_i \dnt{A_i}} \\
        & = \loopmor{\Gamma}{(\wbranch{\ell_i}{x_i}{[\gamma]t_i},)_i}{\ms{L}}
      \end{aligned}
    \end{equation}
    It follows that
    \begin{equation}
      \begin{aligned}
        & S
          ; \dnt{\haslb{\Delta}{\where{s}{(\wbranch{\ell_i}{x_i}{t_i},)_i}}{\ms{L}}} \\
        & = S
          ; \lmor{\dnt{\haslb{\Delta}{s}{\ms{L}, \ms{R}}} 
            ; \alpha^+_{\dnt{\ms{L}} + \Sigma_i \dnt{A_i}}}
          ; \delta^{-1}
          ; [\pi_r, \loopmor{\Delta}{(\wbranch{\ell_i}{x_i}{t_i},)_i}{\ms{L}}] \\
        & = \lmor{S ; \dnt{\haslb{\Delta}{s}{\ms{L}, \ms{R}}} 
            ; \alpha^+_{\dnt{\ms{L}} + \Sigma_i \dnt{A_i}}}
          ; S \otimes (\dnt{\ms{L}} + \Sigma_i \dnt{A_i})
          ; \delta^{-1}
          ; [\pi_r, \loopmor{\Delta}{(\wbranch{\ell_i}{x_i}{t_i},)_i}{\ms{L}}] \\
        & = \lmor{\dnt{\haslb{\Gamma}{[\gamma]s}{\ms{L}, \ms{R}}} 
            ; \alpha^+_{\dnt{\ms{L}} + \Sigma_i \dnt{A_i}}}
          ; \delta^{-1}
          ; [\pi_r, 
            S \otimes (\dnt{\ms{L}} + \Sigma_i \dnt{A_i}) 
            ; \loopmor{\Delta}{(\wbranch{\ell_i}{x_i}{t_i},)_i}{\ms{L}}] \\
        & = \lmor{\dnt{\haslb{\Gamma}{[\gamma]s}{\ms{L}, \ms{R}}} 
            ; \alpha^+_{\dnt{\ms{L}} + \Sigma_i \dnt{A_i}}}
          ; \delta^{-1}
          ; [\pi_r, \loopmor{\Gamma}{(\wbranch{\ell_i}{x_i}{[\gamma]t_i},)_i}{\ms{L}}] \\
        & = \dnt{\haslb{\Gamma}
          {\where{[\gamma]s}
          {(\wbranch{\ell_i}{x_i}{[\gamma]t_i},)_i}}{\ms{L}}} \\
      \end{aligned}
    \end{equation}
    as desired.
  \end{itemize}
  Composition of substitutions then follows by a trivial induction, as does substitution for label
  substitutions.
\end{proof}

\subsection{Label Substitution}

\begin{lemma}[Label Substitution Injection]
  For $\lbsubst{\Gamma}{\sigma}{\ms{L}}{\ms{K}}$ and $\ms{L}(\ell) = A$, we have
  \begin{equation}
    \dnt{\Gamma} \otimes \iota_{\ms{L}, \ell} ; \dnt{\lbsubst{\Gamma}{\sigma}{\ms{L}}{\ms{K}}}
    = \dnt{\haslb{\Gamma, x : A}{[\gamma](\brb{\ell}{x})}{\ms{K}}}
  \end{equation}
  \label{lem:lsubst-inj}
\end{lemma}

\begin{proof}
  We proceed by induction on $\ms{L}$:
  \begin{itemize}[leftmargin=*]
    \item If $\ms{L} = \cdot$, then $\ms{L}(\ell) = A$ is a contradiction
    \item If $\ms{L} = \ms{L}', \ell(A)$, then 
    $\sigma = \sigma', \ell(x) \mapsto [\sigma](\brb{\ell}{x})$, so we have
    \begin{equation}
      \begin{aligned}
      & \dnt{\Gamma} \otimes \iota_{\ms{L}, \ell} 
        ; \dnt{\lbsubst{\Gamma}{\sigma}{\ms{L}, \ell(A)}{\ms{K}}} \\
      & = \dnt{\Gamma} \otimes \iota_r ; \delta^{-1} ; [
          \dnt{\lbsubst{\Gamma}{\sigma'}{\ms{L}}{\ms{K}}},
          \dnt{\haslb{\Gamma, \bhyp{x}{A}}{[\sigma](\brb{\ell}{x})}{\ms{K}}}
        ] \\
      & = \dnt{\haslb{\Gamma, \bhyp{x}{A}}{[\sigma](\brb{\ell}{x})}{\ms{K}}}
      \end{aligned}
    \end{equation}
    \item If $\ms{L} = \ms{L}', \kappa(A)$, then 
    $\sigma = \sigma', \kappa(x) \mapsto [\sigma](\brb{\kappa}{x})$, so by
    induction we have
    \begin{equation}
      \begin{aligned}
      & \dnt{\Gamma} \otimes \iota_{\ms{L}, \ell} 
        ; \dnt{\lbsubst{\Gamma}{\sigma}{\ms{L}, \ell(A)}{\ms{K}}} \\
      & = \dnt{\Gamma} \otimes (\iota_{\ms{L}', \ell} ; \iota_l) ; \delta^{-1} ; [
          \dnt{\lbsubst{\Gamma}{\sigma'}{\ms{L}}{\ms{K}}},
          \dnt{\haslb{\Gamma, \bhyp{x}{B}}{[\sigma](\brb{\kappa}{x})}{\ms{K}}}
        ] \\
      & = \dnt{\Gamma} \otimes \iota_{\ms{L}', \ell} 
        ; \dnt{\lbsubst{\Gamma}{\sigma'}{\ms{L}}{\ms{K}}} \\
      & = \dnt{\haslb{\Gamma, \bhyp{x}{A}}{[\sigma'](\brb{\ell}{x})}{\ms{K}}} \\
      & = \dnt{\haslb{\Gamma, \bhyp{x}{A}}{[\sigma](\brb{\ell}{x})}{\ms{K}}}
      \end{aligned}
    \end{equation}
  \end{itemize}
\end{proof}

\begin{lemma}[Label Substitution Splitting]
  For $\lbsubst{\Gamma}{\sigma}{\ms{L}}{\ms{K}}$, where $\ms{L} = (\ell_i(A_i),)_i$, we have
  \begin{equation}
    \dnt{\lbsubst{\Gamma}{\sigma}{\ms{L}}{\ms{K}}} 
    = \dnt{\Gamma} \otimes \alpha_{\Sigma_i\dnt{A_i}} 
    ; \delta^{-1}_{\Sigma}
    ; [\dnt{\haslb{\Gamma, \bhyp{x_i}{A_i}}{\sigma_i\;x_i}{\ms{K}}},]_i
  \end{equation}
  and therefore
  \begin{equation}
    \delta^{-1}_{\Sigma}
    ; [\dnt{\haslb{\Gamma, \bhyp{x_i}{A_i}}{\sigma_i\;x_i}{\ms{K}}},]_i
    = \alpha_{\dnt{\ms{L}}}
    ; \dnt{\lbsubst{\Gamma}{\sigma}{\ms{L}}{\ms{K}}} 
  \end{equation}
  In particular, we have that, given $\forall i, \haslb{\Gamma, \bhyp{x_i}{A_i}}{t_i}{\ms{K}}$,
  \begin{equation}
    \dnt{\lbsubst{\Gamma}{(\wbranch{\ell_i}{x_i}{A_i},)_i)}{\ms{L}}{\ms{K}}}
    = \dnt{\Gamma} \otimes \alpha_{\Sigma_i\dnt{A_i}} 
    ; \delta^{-1}_{\Sigma}
    ; [\dnt{\haslb{\Gamma, \bhyp{x_i}{A_i}}{t_i}{\ms{K}}},]_i
  \end{equation}
  \label{lem:lsubst-distrib}
\end{lemma}

\begin{lemma}[Label Substitution Extension]
  Given $\sigma = \sigma_l, \sigma_r$,
    $\lbsubst{\Gamma}{\sigma}{\ms{L}, \ms{R}}{\ms{L}', \ms{R}'}$,
    $\lbsubst{\Gamma}{\sigma_l}{\ms{L}}{\ms{L}'}$,
    $\lbsubst{\Gamma}{\sigma_r}{\ms{R}}{\ms{R}'}$, we have
  \begin{equation}
    \dnt{\lbsubst{\Gamma}{\sigma}{\ms{L}, \ms{R}}{\ms{L}', \ms{R}'}}
    = \dnt{\Gamma} \otimes \alpha_{\dnt{\ms{L}} + \dnt{\ms{R}}} 
    ; \dnt{\lbsubst{\Gamma}{\sigma_l}{\ms{L}}{\ms{L}'}} 
    + \dnt{\lbsubst{\Gamma}{\sigma_r}{\ms{R}}{\ms{R}'}}
    ; \alpha_{\dnt{\ms{L}', \ms{R}'}}
  \end{equation}
  In particular, for $\lbsubst{\Gamma}{\sigma}{\ms{L}}{\ms{K}}$, we have
  \begin{equation}
    \dnt{\lbsubst{\Gamma}{\lupg{\sigma}}{\ms{R}, \ms{L}}{\ms{R}, \ms{K}}}
    = \dnt{\Gamma} \otimes \alpha^+_{\dnt{\ms{R}} + \dnt{\ms{L}}} 
      ; \delta^{-1} 
      ; \pi_r + \dnt{\lbsubst{\Gamma}{\sigma}{\ms{L}}{\ms{K}}}
      ; \alpha^+_{\dnt{\ms{R}, \ms{K}}}
  \end{equation}
  \begin{equation}
    \dnt{\lbsubst{\Gamma}{\rupg{\sigma}}{\ms{L}, \ms{R}}{\ms{K}, \ms{R}}}
    = \dnt{\Gamma} \otimes \alpha^+_{\dnt{\ms{L}} + \dnt{\ms{R}}} 
      ; \delta^{-1} 
      ; \dnt{\lbsubst{\Gamma}{\sigma}{\ms{L}}{\ms{K}}} + \pi_r
      ; \alpha^+_{\dnt{\ms{K}, \ms{R}}}
  \end{equation}
  since
  \begin{equation}
    \lbsubst{\Gamma}{\ms{id}}{\ms{R}}{\ms{R}} = \pi_r
  \end{equation}
\end{lemma}

% \begin{lemma}
%   Given $\forall i. \haslb{\Gamma, \bhyp{x_i}{A_i}}{t_i}{\ms{L}, (\lhyp{\ell_j}{A_j},)_j}$, we have
%   that
%   \begin{multline}
%     \dnt{
%       \lbsubst{\Gamma}
%         {\cfgsubst{(\wbranch{\ell_i}{x_i}{t_i},)_i}}{\ms{L}, (\lhyp{\ell_j}{A_j},)_j}{\ms{L}}
%      } \\ = \dnt{\Gamma} \otimes \alpha_{\dnt{\ms{L}} + \Sigma_i\dnt{A_i}} ; [
%       \pi_r, \rfix{\loopmor{(\wbranch{\ell_i}{x_i}{t_i},)_i}}
%      ]
%   \end{multline}
% \end{lemma}

\soundnesslsubst*

\label{proof:soundness-lsubst}

\begin{proof}
  Fix $\lbsubst{\Gamma}{\sigma}{\ms{L}}{\ms{K}}$. We begin by proving the soundness of label
  substitution for regions as follows: assuming $\haslb{\Gamma}{r}{\ms{L}}$, we proceed by induction
  on $r$ as follows:
  \begin{itemize}[leftmargin=*]
    \item If $r = \brb{\ell}{a}$, then by Lemma~\ref{lem:lsubst-inj} we have that
    \begin{equation}
      \begin{aligned}
        \dnt{\haslb{\Gamma}{[\sigma](\brb{\ell}{a})}{\ms{K}}}
        & = \dnt{\haslb{\Gamma}{[a/x][\sigma](\brb{\ell}{x})}{\ms{K}}} \\
        & = \lmor{\dnt{\hasty{\Gamma}{\bot}{a}{A}}}
          ; \dnt{\haslb{\Gamma, \bhyp{x}{A}}{[\sigma](\brb{\ell}{x})}{\ms{K}}} \\
        & = \lmor{\dnt{\hasty{\Gamma}{\bot}{a}{A}}}
          ; \dnt{\Gamma} \otimes \iota_{\ms{L}, \ell} 
          ; \dnt{\lbsubst{\Gamma}{\sigma}{\ms{L}}{\ms{K}}}
          \\
        & = \lmor{\dnt{\hasty{\Gamma}{\bot}{a}{A}} ; \iota_{\ms{L}, \ell}}
          ; \dnt{\lbsubst{\Gamma}{\sigma}{\ms{L}}{\ms{K}}} \\
        & = \lmor{\dnt{\haslb{\Gamma}{\brb{\ell}{a}}{\ms{L}}}}
          ; \dnt{\lbsubst{\Gamma}{\sigma}{\ms{L}}{\ms{K}}}
      \end{aligned}
    \end{equation}
    as desired.
    \item If $r = (\letstmt{x}{a}{t})$, then we have by induction that
    \begin{equation}
      \begin{aligned}
        & \dnt{\haslb{\Gamma}{[\sigma](\letstmt{x}{a}{t})}{\ms{K}}} 
          = \dnt{\haslb{\Gamma}{\letstmt{x}{a}{[\sigma]t}}{\ms{K}}} \\
        & = \lmor{\dnt{\hasty{\Gamma}{\epsilon}{a}{A}}}
          ; \dnt{\haslb{\Gamma, \bhyp{x}{A}}{[\sigma]t}{\ms{K}}} \\
        & = \lmor{\dnt{\hasty{\Gamma}{\epsilon}{a}{A}}}
          ; \lmor{\dnt{\haslb{\Gamma, \bhyp{x}{A}}{t}{\ms{L}}}}
          ; \dnt{\lbsubst{\Gamma, \bhyp{x}{A}}{\sigma}{\ms{L}}{\ms{K}}} \\
        & = \lmor{\dnt{\hasty{\Gamma}{\epsilon}{a}{A}}}
          ; \lmor{\dnt{\haslb{\Gamma, \bhyp{x}{A}}{t}{\ms{L}}}}
          ; \pi_l \otimes \dnt{\ms{L}} 
          ; \dnt{\lbsubst{\Gamma}{\sigma}{\ms{L}}{\ms{K}}} \\
        & = \lmor{
            \lmor{\dnt{\hasty{\Gamma}{\epsilon}{a}{A}}}
            ; \dnt{\haslb{\Gamma, \bhyp{x}{A}}{t}{\ms{L}}}
          }
          ; \dnt{\lbsubst{\Gamma}{\sigma}{\ms{L}}{\ms{K}}} \\
        & = \lmor{\dnt{\haslb{\Gamma}{\letstmt{x}{a}{t}}{\ms{L}}}}
          ; \dnt{\lbsubst{\Gamma}{\sigma}{\ms{L}}{\ms{K}}} \\
      \end{aligned}
    \end{equation}
    as desired.
    \item If $r = (\letstmt{(x, y)}{a}{t})$, then we have by induction that
    \begin{equation}
      \begin{aligned}
        & \dnt{\haslb{\Gamma}{[\sigma](\letstmt{(x, y)}{a}{t})}{\ms{K}}} 
          = \dnt{\haslb{\Gamma}{\letstmt{(x, y)}{a}{[\sigma]t}}{\ms{K}}} \\
        & = \lmor{\dnt{\hasty{\Gamma}{\epsilon}{a}{A \otimes B}}} ; \alpha
          ; \dnt{\haslb{\Gamma, \bhyp{x}{A}, \bhyp{y}{B}}{[\sigma]t}{\ms{K}}} \\
        & = \lmor{\dnt{\hasty{\Gamma}{\epsilon}{a}{A \otimes B}}} ; \alpha
          ; \lmor{\dnt{\haslb{\Gamma, \bhyp{x}{A}, \bhyp{y}{B}}{t}{\ms{L}}}}
          ; \dnt{\lbsubst{\Gamma, \bhyp{x}{A}, \bhyp{y}{B}}{\sigma}{\ms{L}}{\ms{K}}} \\
          & = \lmor{\dnt{\hasty{\Gamma}{\epsilon}{a}{A \otimes B}}} ; \alpha
            ; \lmor{\dnt{\haslb{\Gamma, \bhyp{x}{A}, \bhyp{y}{B}}{t}{\ms{L}}}}
            ; (\pi_l ; \pi_l) \otimes \dnt{\ms{L}}
            ; \dnt{\lbsubst{\Gamma}{\sigma}{\ms{L}}{\ms{K}}} \\
        & = \lmor{\lmor{\dnt{\hasty{\Gamma}{\epsilon}{a}{A \otimes B}}} ; \alpha 
            ; \dnt{\haslb{\Gamma, \bhyp{x}{A}, \bhyp{y}{B}}{t}{\ms{L}}}}
          ; \dnt{\lbsubst{\Gamma}{\sigma}{\ms{L}}{\ms{K}}} \\
        & = \lmor{\dnt{\haslb{\Gamma}{\letstmt{(x, y)}{a}{t}}{\ms{L}}}}
          ; \dnt{\lbsubst{\Gamma}{\sigma}{\ms{L}}{\ms{K}}}
      \end{aligned}
    \end{equation}
    as desired, since 
    \item If $r = \casestmt{a}{x}{s}{y}{t}$, then we have by induction that
    \begin{equation}
      \begin{aligned}
        & \dnt{\haslb{\Gamma}{[\sigma](\casestmt{a}{x}{s}{y}{t})}{\ms{K}}} 
          = \dnt{\haslb{\Gamma}{\casestmt{a}{x}{[\sigma]s}{y}{[\sigma]t}}{\ms{K}}} \\
        & = \lmor{\dnt{\hasty{\Gamma}{\epsilon}{a}{A + B}}} ; \delta^{-1}
          ; [\dnt{\haslb{\Gamma, \bhyp{x}{A}}{[\sigma]s}{\ms{K}}}, 
              \dnt{\haslb{\Gamma, \bhyp{y}{B}}{[\sigma]t}{\ms{K}}}] \\
        & = \lmor{\dnt{\hasty{\Gamma}{\epsilon}{a}{A + B}}} ; \delta^{-1} ; [
        \\ & \qquad 
            \lmor{\dnt{\haslb{\Gamma, \bhyp{x}{A}}{s}{\ms{L}}}} 
              ; \dnt{\lbsubst{\Gamma, \bhyp{x}{A}}{\sigma}{\ms{L}}{\ms{K}}},
        \\ & \qquad
            \lmor{\dnt{\haslb{\Gamma, \bhyp{y}{B}}{t}{\ms{L}}}}
              ; \dnt{\lbsubst{\Gamma, \bhyp{y}{B}}{\sigma}{\ms{L}}{\ms{K}}}
          ] \\
        & = \lmor{\dnt{\hasty{\Gamma}{\epsilon}{a}{A + B}}} ; \delta^{-1} ;
        \\ & \qquad [
              \lmor{\dnt{\haslb{\Gamma, \bhyp{x}{A}}{s}{\ms{L}}}} ; \pi_l \otimes \dnt{\ms{L}}, 
              \lmor{\dnt{\haslb{\Gamma, \bhyp{y}{B}}{t}{\ms{L}}}} ; \pi_l \otimes \dnt{\ms{L}}] 
          ; \dnt{\lbsubst{\Gamma}{\sigma}{\ms{L}}{\ms{K}}} \\
        & = \lmor{
            \lmor{\dnt{\hasty{\Gamma}{\epsilon}{a}{A + B}}} 
            ; \delta^{-1} 
            ; [\dnt{\haslb{\Gamma, \bhyp{x}{A}}{s}{\ms{L}}}, 
              \dnt{\haslb{\Gamma, \bhyp{y}{B}}{t}{\ms{L}}}]}
          ; \dnt{\lbsubst{\Gamma}{\sigma}{\ms{L}}{\ms{K}}} \\
        & = \lmor{\dnt{\haslb{\Gamma}{\casestmt{a}{x}{s}{y}{t}}{\ms{L}}}}
          ; \dnt{\lbsubst{\Gamma}{\sigma}{\ms{L}}{\ms{K}}}
      \end{aligned}
    \end{equation}
    as desired.
    \item If $r = \where{s}{(\wbranch{\ell_i}{x_i}{t_i},)_i}$, then we have by induction,
    taking $\ms{R} = (\ell_i(A_i),)_i$,
    \begin{equation}
      \begin{aligned}
        & \entrymor{\Gamma}{[\rupg{\sigma}]s}{\ms{K}} 
          = \dnt{\haslb{\Gamma}{[\rupg{\sigma}]s}{\ms{K}, \ms{R}}}
          ; \alpha^+_{\dnt{\ms{L}} + \Sigma_i \dnt{A_i}} \\
        & = \lmor{\dnt{\haslb{\Gamma}{s}{\ms{L}, \ms{R}}}}
          ; \dnt{\lbsubst{\Gamma}{\rupg{\sigma}}{\ms{L}, \ms{R}}{\ms{K}, \ms{R}}}
          ; \alpha^+_{\dnt{\ms{K}} + \Sigma_i \dnt{A_i}} \\
        & = \lmor{\dnt{\haslb{\Gamma}{s}{\ms{L}, \ms{R}}}}
          ; \dnt{\Gamma} \otimes \alpha^+_{\dnt{\ms{L}} + \dnt{\ms{R}}} ; \delta^{-1} 
          ; \dnt{\lbsubst{\Gamma}{\sigma}{\ms{L}}{\ms{K}}} + \pi_r
          ; \alpha^+_{\dnt{\ms{L}} + \Sigma_i \dnt{A_i}} \\
        & = \lmor{\dnt{\haslb{\Gamma}{s}{\ms{L}, \ms{R}}}}
          ; \dnt{\Gamma} \otimes \alpha^+_{\dnt{\ms{L}} + \Sigma_i \dnt{A_i}} ; \delta^{-1} 
          ; \dnt{\lbsubst{\Gamma}{\sigma}{\ms{L}}{\ms{K}}} + \pi_r\\
        & = \lmor{\entrymor{\Gamma}{s}{\ms{L}}} 
          ; \delta^{-1} ; \dnt{\lbsubst{\Gamma}{\sigma}{\ms{L}}{\ms{K}}} + \pi_r
      \end{aligned}
    \end{equation}
    and
    \begin{equation}
      \begin{aligned}       
        & \loopmor{\Gamma}{(\wbranch{\ell_i}{x_i}{[\rupg{\sigma}]t_i},)_i}{\ms{K}} 
          = \delta^{-1}_{\Sigma} 
          ; [ \dnt{\haslb{\Gamma, \bhyp{x_i}{A_i}}{[\rupg{\sigma}]t_i}{\ms{K}, \ms{R}}} ]_i
          ; \alpha^+_{\dnt{\ms{K}} + \Sigma_i \dnt{A_i}} \\
        & = \delta^{-1}_{\Sigma} 
        ; [ 
            \lmor{\dnt{\haslb{\Gamma, \bhyp{x_i}{A_i}}{t_i}{\ms{L}, \ms{R}}}} 
            ; \dnt{\lbsubst{\Gamma, \bhyp{x_i}{A_i}}{\rupg{\sigma}}{\ms{L}, \ms{R}}{\ms{K}, \ms{R}}}
          ]_i
        ; \alpha^+_{\dnt{\ms{K}} + \Sigma_i \dnt{A_i}} \\
        & = \delta^{-1}_{\Sigma} 
        ; [ 
            \lmor{\dnt{\haslb{\Gamma, \bhyp{x_i}{A_i}}{t_i}{\ms{L}, \ms{R}}}} 
            ; \pi_l \otimes \dnt{\ms{L}, \ms{R}}
          ]_i
        ; \dnt{\lbsubst{\Gamma}{\rupg{\sigma}}{\ms{L}, \ms{R}}{\ms{K}, \ms{R}}}
        ; \alpha^+_{\dnt{\ms{K}} + \Sigma_i \dnt{A_i}} \\
        & = \delta^{-1}_{\Sigma} 
        ; [ 
            \lmor{\dnt{\haslb{\Gamma, \bhyp{x_i}{A_i}}{t_i}{\ms{L}, \ms{R}}}} 
            ; \pi_l \otimes \alpha^+_{\dnt{\ms{L}} + \Sigma_j \dnt{A_j}}
          ]_i
        ; \delta^{-1}
        ; \dnt{\lbsubst{\Gamma}{\sigma}{\ms{L}}{\ms{K}}} + \pi_r \\
        & = \lmor{
          \delta^{-1}_{\Sigma}
            ; [\dnt{\haslb{\Gamma, \bhyp{x_i}{A_i}}{t_i}{\ms{L}, \ms{R}}} 
                ; \alpha^+_{\dnt{\ms{L}} + \Sigma_j \dnt{A_j}},]_i }
          ; \pi_l \otimes -
          ; \delta^{-1} 
          ; \dnt{\lbsubst{\Gamma}{\sigma}{\ms{L}}{\ms{K}}} + \pi_r \\
        & = \lmor{
          \delta^{-1}_{\Sigma}
            ; [\dnt{\haslb{\Gamma, \bhyp{x_i}{A_i}}{t_i}{\ms{L}, \ms{R}}},]_i 
            ; \alpha^+_{\dnt{\ms{L}} + \Sigma_i \dnt{A_i}}
            } 
          ; \pi_l \otimes -
          ; \delta^{-1} 
          ; \dnt{\lbsubst{\Gamma}{\sigma}{\ms{L}}{\ms{K}}} + \pi_r \\
        & = \lmor{\loopmor{\Gamma}{(\wbranch{\ell_i}{x_i}{t_i},)_i}{\ms{L}}} 
          ; \pi_l \otimes -
          ; \delta^{-1} 
          ; \dnt{\lbsubst{\Gamma}{\sigma}{\ms{L}}{\ms{K}}} + \pi_r
      \end{aligned}
    \end{equation}
    Letting $L = \loopmor{\Gamma}{(\wbranch{\ell_i}{x_i}{t_i},)_i}{\ms{L}}$ and
    $S = \dnt{\lbsubst{\Gamma}{\sigma}{\ms{L}}{\ms{K}}}$, we have that
    \begin{equation}
      \begin{aligned}
        & \rcase{\loopmor{\Gamma}{(\wbranch{\ell_i}{x_i}{[\rupg{\sigma}]t_i},)_i}{\ms{K}}} \\
        & = \lmor{
            \lmor{L} 
            ; \pi_l \otimes (\dnt{\ms{L}} + \Sigma_i\dnt{A_i})
            ; \delta^{-1} ; S + \pi_r}
          ; \pi_l \otimes (\dnt{\ms{K}} + \Sigma_i\dnt{A_i}) ; \delta^{-1}
          \\
        & = \lmor{\lmor{L}} 
          ; (\dnt{\Gamma} \otimes \Sigma_i\dnt{A_i}) 
            \otimes (\pi_l \otimes (\dnt{\ms{L}} + \Sigma_i\dnt{A_i}) ; \delta^{-1} ; S + \pi_r)
          ; \pi_l \otimes (\dnt{\ms{K}} + \Sigma_i\dnt{A_i}) ; \delta^{-1} \\
        & = \lmor{L}
          ; \dmor{\dnt{\Gamma} \otimes \Sigma_i\dnt{A_i}} \otimes -
          ; \alpha
          ; (\dnt{\Gamma} \otimes \Sigma_i\dnt{A_i}) 
          \otimes (\pi_l \otimes (\dnt{\ms{L}} + \Sigma_i\dnt{A_i}) ; \delta^{-1} ; S + \pi_r)
          ; \pi_l \otimes - ; \delta^{-1} \\
        & = \lmor{L}
          ; (\pi_l ; \dmor{\dnt{\Gamma}}) \otimes (\dnt{\ms{L}} + \Sigma_i\dnt{A_i})
          ; \alpha
          ; \dnt{\Gamma}
          \otimes (\delta^{-1} ; S + \pi_r)
          ; \delta^{-1} \\
        & = \lmor{L}
          ; (\pi_l ; \dmor{\dnt{\Gamma}}) \otimes (\dnt{\ms{L}} + \Sigma_i\dnt{A_i})
          ; \alpha
          ; \dnt{\Gamma}
          \otimes \delta^{-1}
          ; \delta^{-1} 
          ; (\dnt{\Gamma} \otimes \dnt{\Gamma}) \otimes S 
            + (\dnt{\Gamma} \otimes \dnt{\Gamma}) \otimes \pi_r \\
        & = \lmor{L} ; \pi_l \otimes (\dnt{\ms{L}} + \Sigma_i\dnt{A_i})
          ; \delta^{-1}
          ; (\dmor{\dnt{\Gamma}} \otimes \dnt{\ms{L}} ; \alpha ; \dnt{\Gamma} \otimes S) 
          + (\dmor{\dnt{\Gamma}} \otimes \Sigma_i\dnt{A_i}
              ; \alpha ; \dnt{\Gamma} \otimes \pi_r) \\
          & = \rcase{L}
            ; \dmor{\dnt{\Gamma}} \otimes \dnt{\ms{L}} ; \alpha ; \dnt{\Gamma} \otimes S
            + \dnt{\Gamma} \otimes \Sigma_i\dnt{A_i}
      \end{aligned}
    \end{equation}
    implying by naturality that
    \begin{equation}
      \begin{aligned}
        \rfix{\loopmor{\Gamma}{(\wbranch{\ell_i}{x_i}{[\rupg{\sigma}]t_i},)_i}{\ms{K}}}
        & = (\rcase{L}
        ; \dmor{\dnt{\Gamma}} \otimes \dnt{\ms{L}} ; \alpha ; \dnt{\Gamma} \otimes S
            + \dnt{\Gamma} \otimes \Sigma_i\dnt{A_i})^\dagger ; \pi_r \\
        & = (\rcase{L})^\dagger 
          ; \dmor{\dnt{\Gamma}} \otimes \dnt{\ms{L}} ; \alpha ; \dnt{\Gamma} \otimes S ; \pi_r \\
        & = (\rcase{L})^\dagger ; S
      \end{aligned}
    \end{equation}
    Hence, we have that
    \begin{equation}
      \begin{aligned}
        & \dnt{\haslb{\Gamma}{[\sigma](\where{s}{(\wbranch{\ell_i}{x_i}{t_i},)_i})}{\ms{K}}}
          = \dnt{\haslb{\Gamma}
                  {\where{[\rupg{\sigma}]s}
                  {(\wbranch{\ell_i}{x_i}{[\rupg{\sigma}]t_i},)_i}}{\ms{K}}} \\
        & = \lmor{\entrymor{\Gamma}{[\rupg{\sigma}]s}{\ms{K}}} ; \delta^{-1} 
          ; [\pi_r , \rfix{\loopmor{\Gamma}{(\wbranch{\ell_i}{x_i}{[\rupg{\sigma}]t_i},)_i}{\ms{K}}}] \\
        & = \lmor{\lmor{\entrymor{\Gamma}{s}{\ms{L}}} ; \delta^{-1} ; S + \pi_r}
          ; \delta^{-1}
          ; [\pi_r, (\rcase{L})^\dagger ; S] \\
        & = \lmor{\lmor{\entrymor{\Gamma}{s}{\ms{L}}}} ; \dnt{\Gamma} \otimes (\delta^{-1} ; S + \pi_r)
          ; \delta^{-1}
          ; [\pi_r, (\rcase{L})^\dagger ; S] \\
        & = \lmor{\lmor{\entrymor{\Gamma}{s}{\ms{L}}}} ; \dnt{\Gamma} \otimes \delta^{-1}
          ; \delta^{-1}
          ; [
              \dnt{\Gamma} \otimes S ; \pi_r, 
              \dnt{\Gamma} \otimes \pi_r ; (\rcase{L})^\dagger ; S
            ] \\
          & = \lmor{\lmor{\entrymor{\Gamma}{s}{\ms{L}}}} 
            ; \dnt{\Gamma} \otimes \delta^{-1} ; \delta^{-1}
            ; [\pi_r, \dnt{\Gamma} \otimes \pi_r ; (\rcase{L})^\dagger]
            ; S \\
        & = \lmor{\entrymor{\Gamma}{s}{\ms{L}}} 
          ; \dmor{\dnt{\Gamma}} \otimes (\dnt{\ms{L}} + \Sigma_i\dnt{A_i}) ; \alpha 
          ; \dnt{\Gamma} \otimes \delta^{-1} ; \delta^{-1}
          ; [\pi_r, \dnt{\Gamma} \otimes \pi_r ; (\rcase{L})^\dagger] ; S 
          \\
        & = \lmor{\entrymor{\Gamma}{s}{\ms{L}}} ; \delta^{-1} ; [
              \Delta_{\dnt{\Gamma}} \otimes \dnt{\ms{L}} ; \alpha ; \pi_r, 
              \Delta_{\dnt{\Gamma}} \otimes \Sigma_i\dnt{A_i}  
                ; \alpha ; \dnt{\Gamma} \otimes \pi_r ; (\rcase{L})^\dagger
            ]
          ; S \\
        & = \lmor{\entrymor{\Gamma}{s}{\ms{L}}} ; \delta^{-1} ; [
              \Delta_{\dnt{\Gamma}} \otimes \dnt{\ms{L}} ; \alpha ; \dnt{\Gamma} \otimes \pi_r, 
              (\rcase{L})^\dagger
            ]
          ; S \\
        & = \lmor{\entrymor{\Gamma}{s}{\ms{L}}} ; \delta^{-1} ; [
              \Delta_{\dnt{\Gamma}} \otimes \dnt{\ms{L}} ; \alpha ; \dnt{\Gamma} \otimes \pi_r, 
              \lmor{\rfix{L}} ; \pi_l \otimes (\dnt{\ms{L}} + \Sigma_i\dnt{A_i})
            ]
          ; S \\
        & = \lmor{\entrymor{\Gamma}{s}{\ms{L}}} ; \delta^{-1} ; [
              \Delta_{\dnt{\Gamma}} \otimes \dnt{\ms{L}} ; \alpha ; \dnt{\Gamma} \otimes \pi_r, 
              \Delta_{\dnt{\Gamma}} \otimes \Sigma_i\dnt{A_i} ; \alpha 
                ; \dnt{\Gamma} \otimes \rfix{L}
            ]
          ; S \\
        & = \lmor{\entrymor{\Gamma}{s}{\ms{L}}} 
          ; \Delta_{\dnt{\Gamma}} \otimes (\dnt{\ms{L}} + \Sigma_i\dnt{A_i}) ; \alpha
          ; \dnt{\Gamma} \otimes \delta^{-1} ; \delta^{-1}
          ; [
              \dnt{\Gamma} \otimes \pi_r, 
              \dnt{\Gamma} \otimes \rfix{L}
            ]
          ; S \\
        & = \lmor{\entrymor{\Gamma}{s}{\ms{L}}} 
          ; \Delta_{\dnt{\Gamma}} \otimes (\dnt{\ms{L}} + \Sigma_i\dnt{A_i}) ; \alpha
          ; \dnt{\Gamma} \otimes (\delta^{-1} ; [\pi_r, \rfix{L}])
          ; S \\
        & = \lmor{\lmor{\entrymor{\Gamma}{s}{\ms{L}}}}
          ; \dnt{\Gamma} \otimes (\delta^{-1} ; [\pi_r, \rfix{L}])
          ; S \\
        & = \lmor{\lmor{\entrymor{\Gamma}{s}{\ms{L}}} ; \delta^{-1} ; [\pi_r, \rfix{L}] }
          ; S \\
        & = \lmor{\dnt{\haslb{\Gamma}{\where{s}{(\wbranch{\ell_i}{x_i}{t_i},)_i}}{\ms{L}}}}
          ; \dnt{\lbsubst{\Gamma}{\sigma}{\ms{L}}{\ms{K}}}
      \end{aligned}
    \end{equation}
    as desired.
  \end{itemize}
  Composition of label substitutions then follows by a trivial induction.
\end{proof}

\subsection{Equational Theory}

\soundnesseqn*

\label{proof:soundness-eqn}

\begin{proof}
  We begin our proof by showing soundness of the equational theory for expressions, i.e.
  \ref{itm:eqn-sound-expr}, by rule induction.
  \begin{itemize}[leftmargin=*]
    \item \emph{Congruence}: these follow trivially by induction
    \item \brle{initial}, \brle{terminal}: both of these follow trivially from the universal
    property of the initial/terminal object, respectively.
    \item \brle{let$_1$-$\beta$}: this follows directly from Corollary~\ref{corr:single-subst}
    \item \brle{let$_1$-$\eta$}: we have
    \begin{align*}
      & \dnt{\hasty{\Gamma}{\epsilon}{\letexpr{x}{a}{x}}{A}} \\
      & = \dmor{\dnt{\Gamma}} 
        ; \dnt{\Gamma} \otimes \dnt{\hasty{\Gamma}{\epsilon}{a}{A}}
        ; \dnt{\hasty{\Gamma, \bhyp{x}{A}}{\epsilon}{x}{A}} \\
      & = \dmor{\dnt{\Gamma}} 
      ; \dnt{\Gamma} \otimes \dnt{\hasty{\Gamma}{\epsilon}{a}{A}}
      ; \pi_r \\
      &= \dnt{\hasty{\Gamma}{\epsilon}{a}{A}}
    \end{align*}
    as desired.
    \item \brle{let$_1$-op}: we have
    \begin{align*}
      & \dnt{\hasty{\Gamma}{\epsilon}{\letexpr{x}{a}{\letexpr{y}{f\;x}{c}}}{C}} \\
      & = \lmor{\dnt{\hasty{\Gamma}{\epsilon}{a}{A}}}
        ; \lmor{\dnt{\hasty{\Gamma, \bhyp{x}{A}}{\epsilon}{f\;x}{B}}}
        ; \dnt{\hasty{\Gamma, \bhyp{x}{A}, \bhyp{y}{B}}{\epsilon}{c}{C}} \\
      & = \lmor{\dnt{\hasty{\Gamma}{\epsilon}{a}{A}}}
        ; \lmor{\pi_r ; \dnt{f}}
        ; \pi_l \otimes \dnt{B}
        ; \dnt{\hasty{\Gamma, \bhyp{y}{B}}{\epsilon}{c}{C}} \\
      & = \lmor{
            \lmor{\dnt{\hasty{\Gamma}{\epsilon}{a}{A}}} ; \pi_r ; \dnt{f}
          } ; \dnt{\hasty{\Gamma, \bhyp{y}{B}}{\epsilon}{c}{C}} \\
      & = \lmor{\dnt{\hasty{\Gamma}{\epsilon}{a}{A}} ; \dnt{f}}
        ; \dnt{\hasty{\Gamma, \bhyp{y}{B}}{\epsilon}{c}{C}} \\
      & = \dnt{\hasty{\Gamma}{\epsilon}{\letexpr{y}{f\;a}{c}}{C}}
    \end{align*}
    as desired. The \brle{let$_1$-abort} case is analogous.
    \item \brle{let$_1$-let$_1$}: we have that
    \begin{align*}
      & \dnt{\hasty{\Gamma}{\epsilon}{\letexpr{y}{(\letexpr{x}{a}{b}}{c})}{C}} \\
      & = \lmor{\lmor{\dnt{\hasty{\Gamma}{\epsilon}{a}{A}}} 
        ; \dnt{\hasty{\Gamma, \bhyp{x}{A}}{\epsilon}{b}{B}}}
        ; \dnt{\hasty{\Gamma, \bhyp{y}{B}}{\epsilon}{c}{C}} \\
      & = \lmor{\dnt{\hasty{\Gamma}{\epsilon}{a}{A}}}
        ; \lmor{\dnt{\hasty{\Gamma, \bhyp{x}{A}}{\epsilon}{b}{B}}}
        ; \pi_l \otimes \dnt{B}
        ; \dnt{\hasty{\Gamma, \bhyp{y}{B}}{\epsilon}{c}{C}} \\
      & = \lmor{\dnt{\hasty{\Gamma}{\epsilon}{a}{A}}}
      ; \lmor{\dnt{\hasty{\Gamma, \bhyp{x}{A}}{\epsilon}{b}{B}}}
      ; \dnt{\hasty{\Gamma, \bhyp{x}{A}, \bhyp{y}{B}}{\epsilon}{c}{C}} \\
      & = \dnt{\hasty{\Gamma}{\epsilon}{\letexpr{x}{a}{\letexpr{y}{b}{c}}}{C}}
    \end{align*}
    as desired.
    \item \brle{let$_1$-let$_2$}: we have that
    \begin{align*}
      & \dnt{\hasty{\Gamma}{\epsilon}{\letexpr{z}{(\letexpr{(x, y)}{e}{c})}{d}}{D}} \\
      & = \lmor{\lmor{\dnt{\hasty{\Gamma}{\epsilon}{e}{A \otimes B}}}
          ; \alpha
          ; \dnt{\hasty{\Gamma, \bhyp{x}{A}, \bhyp{y}{B}}{\epsilon}{c}{C}}}
        ; \dnt{\hasty{\Gamma, \bhyp{z}{C}}{\epsilon}{d}{D}} \\
      & = \lmor{\hasty{\Gamma}{\epsilon}{e}{A \otimes B}}
        ; \alpha
        ; \lmor{\dnt{\hasty{\Gamma, \bhyp{x}{A}, \bhyp{y}{B}}{\epsilon}{c}{C}}}
        ; (\pi_l ; \pi_l) \otimes \dnt{\ms{C}}
        ; \dnt{\hasty{\Gamma, \bhyp{z}{C}}{\epsilon}{d}{D}} \\
      & = \lmor{\hasty{\Gamma}{\epsilon}{e}{A \otimes B}}
      ; \alpha
      ; \lmor{\dnt{\hasty{\Gamma, \bhyp{x}{A}, \bhyp{y}{B}}{\epsilon}{c}{C}}}
      ; \dnt{\hasty{\Gamma, \bhyp{x}{A}, \bhyp{y}{B}, \bhyp{z}{C}}{\epsilon}{d}{D}} \\
      & = \dnt{\hasty{\Gamma}{\epsilon}{\letexpr{(x, y)}{e}{\letexpr{z}{c}{d}}}{D}}
    \end{align*}
    as desired.
    \item \brle{let$_1$-case}: follows from the properties of the coproduct; in particular, we have
    that
    \begin{align*}
      & \dnt{\hasty{\Gamma}{\epsilon}{\caseexpr{e}{x}{\letexpr{z}{a}{d}}{y}{\letexpr{z}{b}{d}}}{D}}
      \\ &= \Delta ; \dnt{\Gamma} \otimes \dnt{\hasty{\Gamma}{\epsilon}{e}{A + B}} ; \delta^{-1} ; [
      \\ & \qquad \Delta
            ; \dnt{\Gamma, \bhyp{x}{A}} \otimes \dnt{\hasty{\Gamma, \bhyp{x}{A}}{\epsilon}{a}{C}}
            ; \dnt{\hasty{\Gamma, \bhyp{x}{A}, \bhyp{z}{C}}{\epsilon}{d}{D}}
            ,
      \\  & \qquad \Delta
            ; \dnt{\Gamma, \bhyp{y}{B}} \otimes \dnt{\hasty{\Gamma, \bhyp{y}{B}}{\epsilon}{b}{C}}
            ; \dnt{\hasty{\Gamma, \bhyp{y}{B}, \bhyp{z}{C}}{\epsilon}{d}{D}}
        ]
      \\ &= \Delta ; \dnt{\Gamma} \otimes \dnt{\hasty{\Gamma}{\epsilon}{e}{A + B}} ; \delta^{-1} ; [
      \\ & \qquad \Delta \otimes \dnt{A} ; \alpha
            ; \dnt{\Gamma} \otimes \dnt{\hasty{\Gamma, \bhyp{x}{A}}{\epsilon}{a}{C}}
            ; \dnt{\hasty{\Gamma, \bhyp{z}{C}}{\epsilon}{d}{D}}
            ,
      \\  & \qquad \Delta \otimes \dnt{B} ; \alpha
            ; \dnt{\Gamma} \otimes \dnt{\hasty{\Gamma, \bhyp{y}{B}}{\epsilon}{b}{C}}
            ; \dnt{\hasty{\Gamma, \bhyp{z}{C}}{\epsilon}{d}{D}}
        ]
      \\ &= \Delta ; \dnt{\Gamma} \otimes \dnt{\hasty{\Gamma}{\epsilon}{e}{A + B}} ; \delta^{-1} ; [
      \\ & \qquad \Delta \otimes \dnt{A} ; \alpha
            ; \dnt{\Gamma} \otimes \dnt{\hasty{\Gamma, \bhyp{x}{A}}{\epsilon}{a}{C}}
            ,
      \\  & \qquad \Delta \otimes \dnt{B} ; \alpha
            ; \dnt{\Gamma} \otimes \dnt{\hasty{\Gamma, \bhyp{y}{B}}{\epsilon}{b}{C}}
        ] ; \dnt{\hasty{\Gamma, \bhyp{z}{C}}{\epsilon}{d}{D}}
      \\ &= \Delta ; \Delta \otimes \dnt{\hasty{\Gamma}{\epsilon}{e}{A + B}} ; \delta^{-1} ; 
      \\ & \qquad [
            \alpha ; \dnt{\Gamma} \otimes \dnt{\hasty{\Gamma, \bhyp{x}{A}}{\epsilon}{a}{C}},
            \alpha ; \dnt{\Gamma} \otimes \dnt{\hasty{\Gamma, \bhyp{y}{B}}{\epsilon}{b}{C}}
        ] ; \dnt{\hasty{\Gamma, \bhyp{z}{C}}{\epsilon}{d}{D}}
      \\ &= \Delta ; \Delta \otimes \dnt{\hasty{\Gamma}{\epsilon}{e}{A + B}}  
                   ; \alpha ; \dnt{\Gamma} \otimes (\delta^{-1} ; [
                        \dnt{\hasty{\Gamma, \bhyp{x}{A}}{\epsilon}{a}{C}},
                        \dnt{\hasty{\Gamma, \bhyp{y}{B}}{\epsilon}{b}{C}}
                    ]) ;
      \\ & \qquad \dnt{\hasty{\Gamma, \bhyp{z}{C}}{\epsilon}{d}{D}}
      \\ &= \Delta ; \dnt{\Gamma} \otimes (
                  \Delta ; \dnt{\Gamma} \otimes \dnt{\hasty{\Gamma}{\epsilon}{e}{A + B}} ;
                  \delta^{-1} ; [
                        \dnt{\hasty{\Gamma, \bhyp{x}{A}}{\epsilon}{a}{C}},
                        \dnt{\hasty{\Gamma, \bhyp{y}{B}}{\epsilon}{b}{C}}
                    ]) ;
      \\ & \qquad \dnt{\hasty{\Gamma, \bhyp{z}{C}}{\epsilon}{d}{D}}
      \\ &= \dnt{\hasty{\Gamma}{\epsilon}{\letexpr{z}{(\caseexpr{e}{x}{a}{y}{b})}{d}}{D}}
    \end{align*}
    \item \brle{let$_2$-bind}: we have
    \begin{align*}
      & \dnt{\hasty{\Gamma}{\epsilon}{\letexpr{z}{e}{\letexpr{(x, y)}{z}{c}}}{C}} \\
      & = \lmor{\dnt{\hasty{\Gamma}{\epsilon}{e}{A \otimes B}}}
        ; \lmor{\dnt{\hasty{\Gamma, \bhyp{z}{A \otimes B}}{\epsilon}{z}{A \otimes B}}}
        ; \alpha
        ; \dnt{\hasty{\Gamma, \bhyp{z}{A \otimes B}, \bhyp{x}{A}, \bhyp{y}{B}}{\epsilon}{c}{C}} \\
      & = \lmor{\dnt{\hasty{\Gamma}{\epsilon}{e}{A \otimes B}}}
        ; \lmor{\pi_r}
        ; \pi_l \otimes (\dnt{A} \otimes \dnt{B})
        ; \alpha
        ; \dnt{\hasty{\Gamma, \bhyp{x}{A}, \bhyp{y}{B}}{\epsilon}{c}{C}} \\
      & = \lmor{\lmor{\dnt{\hasty{\Gamma}{\epsilon}{e}{A \otimes B}}} ; \pi_r}
        ; \alpha
        ; \dnt{\hasty{\Gamma, \bhyp{x}{A}, \bhyp{y}{B}}{\epsilon}{c}{C}} \\
      & = \lmor{\dnt{\hasty{\Gamma}{\epsilon}{e}{A \otimes B}}}
        ; \alpha
        ; \dnt{\hasty{\Gamma, \bhyp{x}{A}, \bhyp{y}{B}}{\epsilon}{c}{C}} \\
      & = \dnt{\hasty{\Gamma}{\epsilon}{\letexpr{(x, y)}{e}{c}}{C}}
    \end{align*}
    \item \brle{let$_2$-$\eta$}: follows from the properties of the product; in particular,
    we have that
    \begin{align*}
      &\dnt{\hasty{\Gamma}{\epsilon}{\letexpr{(x, y)}{e}{(x, y)}}{A \otimes B}} \\
      &= \Delta ; \dnt{\Gamma} \otimes \dnt{\hasty{\Gamma}{\epsilon}{e}{A \otimes B}}
                ; \Delta ; 
                \dnt{\hasty{\Gamma, \bhyp{x}{A}, \bhyp{y}{B}}{\bot}{x}{A}} \otimes
                \dnt{\hasty{\Gamma, \bhyp{x}{A}, \bhyp{y}{B}}{\bot}{y}{B}} \\
      &= \Delta ; \dnt{\Gamma} \otimes \dnt{\hasty{\Gamma}{\epsilon}{e}{A \otimes B}}
                ; \Delta ; (\pi_l ; \pi_r) \otimes \pi_r
       = \dnt{\hasty{\Gamma}{\epsilon}{e}{A \otimes B}}
    \end{align*}
    \item \brle{case-inl}: follows from the properties of the coproduct and inverse distributor; in
    particular, we have that
    \begin{align*}
      & \dnt{\hasty{\Gamma}{\epsilon}{\caseexpr{\linl{a}}{x}{c}{y}{d}}{C}}
      \\ &= \Delta 
      ; \dnt{\Gamma} \otimes (\dnt{\hasty{\Gamma}{\epsilon}{a}{A}} ; \iota_l)
      ; \delta^{-1} ; [
        \dnt{\hasty{\Gamma, \bhyp{x}{A}}{\epsilon}{c}{C}}, 
        \dnt{\hasty{\Gamma, \bhyp{y}{B}}{\epsilon}{d}{C}}
      ]
      \\ &= \Delta 
      ; \dnt{\Gamma} \otimes \dnt{\hasty{\Gamma}{\epsilon}{a}{A}}
      ; \iota_l ; [
        \dnt{\hasty{\Gamma, \bhyp{x}{A}}{\epsilon}{c}{C}}, 
        \dnt{\hasty{\Gamma, \bhyp{y}{B}}{\epsilon}{d}{C}}
      ]
      \\ &= \Delta 
      ; \dnt{\Gamma} \otimes \dnt{\hasty{\Gamma}{\epsilon}{a}{A}}
      ; \iota_l ; \dnt{\hasty{\Gamma, \bhyp{x}{A}}{\epsilon}{c}{C}}
      \\ &= \dnt{\hasty{\Gamma}{\epsilon}{\letexpr{x}{a}{c}}{C}}
    \end{align*}
    We can validate \brle{case-inr} analogously
    \item \brle{case-$\eta$}: follows from the properties of the coproduct and distributor; in
    particular, we have
    \begin{align*}
      & \dnt{\hasty{\Gamma}{\epsilon}{\caseexpr{e}{x}{\linl{x}}{y}{\linr{y}}}{A + B}} \\
      &= \Delta ; \dnt{\Gamma} \otimes \dnt{\hasty{\Gamma}{\epsilon}{e}{A + B}} ; \delta^{-1} ; [
        \dnt{\hasty{\Gamma, \bhyp{x}{A}}{\epsilon}{x}{A}};\iota_l,
        \dnt{\hasty{\Gamma, \bhyp{y}{B}}{\epsilon}{y}{B};\iota_r}
      ] \\
      &= \Delta ; \dnt{\Gamma} \otimes \dnt{\hasty{\Gamma}{\epsilon}{e}{A + B}} 
                ; \delta^{-1} ; (\pi_r + \pi_r)
      = \dnt{\hasty{\Gamma}{\epsilon}{e}{A + B}} 
    \end{align*}
  \end{itemize}
  We now proceed to tackle the equational theory for regions $r$ in the same manner. In particular,
  we proceed by rule induction as follows:
  \begin{itemize}[leftmargin=*]
    \item \brle{cfg-$\beta_1$}: 
    Define $P = \dnt{\hasty{\Gamma}{\bot}{a}{A_k}}$, %
           $L = \loopmor{\Gamma}{(\wbranch{\ell_i}{x_i}{t_i},)_i}{\ms{L}}$ and %
           $\ms{R} = (\ell_i(A_i),)_i$. %
    We have that
    \begin{equation}
      \begin{aligned}
        & \dnt{\haslb{\Gamma}{\where{\brb{\ell_k}{a}}{(\wbranch{\ell_i}{x_i}{t_i},)_i}}{\ms{L}}} \\
        & = \lmor{
              P ; \iota_{(\ms{L}, \ms{R}), \ell_k} ;
              \alpha^+_{\dnt{\ms{L}} + \Sigma_i \dnt{A_i}}
          } ; \delta^{-1} ; [\pi_r, \rfix{L}] \\
        & = \lmor{P ; \iota_k} ; \dnt{\Gamma} \otimes \iota_r ; \delta^{-1} ; [\pi_r, \rfix{L}] \\
        & = \lmor{P ; \iota_k} ; \rfix{L} \\
        & = \lmor{P ; \iota_k} ; \rcase{L} ; [\pi_r, \rfix{L}] \\
        & = \lmor{P} ; \rcase{\dnt{\Gamma} \otimes \iota_k ; L} ; [\pi_r, \rfix{L}] \\
        & = \lmor{P} 
          ; \rcase{\dnt{\haslb{\Gamma, \bhyp{x_k}{A_k}}{t_k}{\ms{L}, \ms{R}}} 
                    ; \alpha^+_{\dnt{L} + \Sigma_i\dnt{A_i}}} 
          ; [\pi_r, \rfix{L}] \\
        & = \lmor{P} 
          ; \rlmor{\dnt{\haslb{\Gamma, \bhyp{x_k}{A_k}}{t_k}{\ms{L}, \ms{R}}} 
                    ; \alpha^+_{\dnt{L} + \Sigma_i\dnt{A_i}}} 
          ; \delta^{-1}
          ; [\pi_r, \rfix{L}] \\
        & = \lmor{\lmor{P}  
              ; \dnt{\haslb{\Gamma, \bhyp{x_k}{A_k}}{t_k}{\ms{L}, \ms{R}}} 
              ; \alpha^+_{\dnt{L} + \Sigma_i\dnt{A_i}}}
          ; \delta^{-1}
          ; [\pi_r, \rfix{L}] \\
        & = \lmor{
            \dnt{\haslb{\Gamma}{\letstmt{x_k}{a}{t_k}}{\ms{L}, \ms{R}}} 
            ; \alpha^+_{\dnt{L} + \Sigma_i\dnt{A_i}}}
          ; \delta^{-1}
          ; [\pi_r, \rfix{L}] \\
        & = \dnt{
          \haslb{\Gamma}{\where{\letstmt{x_k}{a}{t_k}}{(\wbranch{\ell_i}{x_i}{A_i},)_i}}{\ms{L}}
        }
      \end{aligned}
    \end{equation}
    as desired.
    \item \brle{cfg-$\beta_2$}: 
    Define $P = \dnt{\hasty{\Gamma}{\bot}{b}{B}}$,
            $L = \loopmor{\Gamma}{(\wbranch{\ell_i}{x_i}{t_i},)_i}{\ms{L}}$ and %
            $\ms{R} = (\ell_i(A_i),)_i$. %
    We have that
    \begin{equation}
      \begin{aligned}
        & \dnt{\haslb{\Gamma}{\where{\brb{\kappa}{b}}{(\wbranch{\ell_i}{x_i}{t_i},)_i}}{\ms{L}}} \\
        & = \lmor{P 
            ; \iota_{(\ms{L}, \ms{R}), \kappa} 
            ; \alpha^+_{\dnt{\ms{L}} + \Sigma_i \dnt{A_i}}
          } ; \delta^{-1} ; [\pi_r, \rfix{L}] \\
        & = \lmor{P ; \iota_{\ms{L}, \kappa}} 
          ; \dnt{\Gamma} \otimes \iota_l ; \delta^{-1} ; [\pi_r, \rfix{L}] \\
        & = \lmor{P ; \iota_{\ms{L}, \kappa}} ; \pi_r
          = P ; \iota_{\ms{L}, \kappa}
          = \dnt{\haslb{\Gamma}{\brb{\kappa}{b}}{\ms{L}}}
      \end{aligned}
    \end{equation}
    as desired.
    \item \brle{cfg-$\eta$}: 
    Let $\ms{L} = (\kappa_j(B_j),)_j$, and define %
      $\ms{R} = (\ell_i(A_i),)_i$ and %
      $L = \loopmor{\Gamma}{(\wbranch{\ell_i}{x_i}{t_i},)_i}{\ms{L}}$. %
    Using the \brle{cfg-$\beta_1$} and \brle{cfg-$\beta_2$} cases proved above, we have that
    \begin{equation}
      \begin{aligned}
        & \dnt{
            \lbsubst{\Gamma}
              {\cfgsubst{(\wbranch{\ell_i}{x_i}{t_i},)_i}}{\ms{L}, \ms{R}}{\ms{L}}
          } \\
        & = \dnt{\Gamma} \otimes \alpha_{\dnt{\ms{L}} + \Sigma_i\dnt{A_i}} ; \delta^{-1}
          ; [ \\ & \qquad
            \dnt{
              \lbsubst{\Gamma}
                {\cfgsubst{(\wbranch{\ell_i}{x_i}{t_i},)_i}}{\ms{L}}{\ms{L}}
            }, \\ & \qquad
            \dnt{
              \lbsubst{\Gamma}
                {\cfgsubst{(\wbranch{\ell_i}{x_i}{t_i},)_i}}{\ms{R}}{\ms{L}}
            }
          ] \\
        & = \dnt{\Gamma} \otimes \alpha_{\Sigma_i\dnt{B_i} + \Sigma_i\dnt{A_i}}; \delta^{-1}
          ; [ \\ & \qquad
            \delta^{-1}_{\Sigma} 
            ; [\dnt{\haslb{\Gamma, \bhyp{y_j}{B_j}}
                {\where{\brb{\kappa_j}{y_j}}{(\wbranch{\ell_i}{x_i}{t_i},)_i}}{\ms{L}}},]_j, 
                \\ & \qquad
            \delta^{-1}_{\Sigma} 
            ; [\dnt{\haslb{\Gamma, \bhyp{x_j}{A_j}}
            {\where{\brb{\ell_j}{x_j}}{(\wbranch{\ell_i}{x_i}{t_i},)_i}}{\ms{L}}},]_j
          ] \\
        & = \dnt{\Gamma} \otimes \alpha_{\Sigma_i\dnt{B_i} + \Sigma_i\dnt{A_i}}; \delta^{-1}
          ; [ \\ & \qquad
            \delta^{-1}_{\Sigma} 
            ; [\dnt{\hasty{\Gamma, \bhyp{y_j}{B_j}}{\bot}{y_j}{B_j}} ; \iota_{\ms{L}, \kappa_j},]_j, 
                \\ & \qquad
            \delta^{-1}_{\Sigma} 
            ; [
              \lmor{\dnt{\hasty{\Gamma, \bhyp{x_j}{A_j}}{\bot}{x_j}{A_j}} ; \iota_j} 
              ; \rfix{\pi_l \otimes \Sigma_k\dnt{A_k} ; L}
            ]_j
          ] \\
        & = \dnt{\Gamma} \otimes \alpha_{\Sigma_i\dnt{B_i} + \Sigma_i\dnt{A_i}}; \delta^{-1}
          ; [
            \delta^{-1}_{\Sigma} ; [\pi_r ; \iota_{\ms{L}, \kappa_j},]_j, 
            \delta^{-1}_{\Sigma} ; [\lmor{\pi_r ; \iota_j} ; \pi_l \otimes \Sigma_k\dnt{A_k}]_j 
            ; \rfix{L}
          ] \\
        & = \dnt{\Gamma} \otimes \alpha_{\Sigma_i\dnt{B_i} + \Sigma_i\dnt{A_i}}; \delta^{-1}
          ; [
            \pi_r ; \alpha_{\ms{L}}, 
            \rlmor{\delta^{-1}_{\Sigma} ; [\pi_r ; \iota_j]_j}
            ; \rfix{L}
          ] \\
        & = \dnt{\Gamma} \otimes \alpha_{\ms{L} + \Sigma_i\dnt{A_i}}; \delta^{-1}
          ; [
            \pi_r, 
            \rlmor{\pi_r} ; \rfix{L}
          ] \\
        & = \dnt{\Gamma} \otimes \alpha_{\ms{L} + \Sigma_i\dnt{A_i}}; \delta^{-1} 
          ; [\pi_r, \rfix{L}]
      \end{aligned}
    \end{equation}
    It follows by label-substitution that that
    \begin{equation}
      \begin{aligned}
        & \dnt{\haslb{\Gamma}{[\cfgsubst{(\wbranch{\ell_i}{x_i}{t_i},)_i}]r}{\ms{L}}}  \\
        & = \lmor{\dnt{\haslb{\Gamma}{r}{\ms{L}, \ms{R}}}}
          ; \dnt{
            \lbsubst{\Gamma}
              {\cfgsubst{(\wbranch{\ell_i}{x_i}{t_i},)_i}}{\ms{L}, \ms{R}}{\ms{L}}
          } \\
        & = \lmor{\dnt{\haslb{\Gamma}{r}{\ms{L}, \ms{R}}}}
          ; \dnt{\Gamma} \otimes \alpha_{\ms{L} + \Sigma_i\dnt{A_i}}; \delta^{-1} 
          ; [\pi_r, \rfix{L}] \\
        & = \entrymor{\Gamma}{r}{\ms{L}}; \delta^{-1} ; [\pi_r, \rfix{L}] \\
        & \dnt{\haslb{\Gamma}{\where{r}{(\wbranch{\ell_i}{x_i}{t_i},)_i}}{\ms{L}}}
      \end{aligned}
    \end{equation}
    as desired.
    \item \brle{codiag}: 
    Define $R = \dnt{\haslb{\Gamma}{r}{\ms{L}, \ell(A)}}$, and %
           $S = \dnt{\haslb{\Gamma, \bhyp{y}{A}}{s}{\ms{L}, \ell(A), \kappa(A)}}$ %
    We have that
    \begin{equation}
      \begin{aligned}
      & \dnt{\haslb{\Gamma}
        {\where{r}{\wbranch{\ell}{x}{\where{\brb{\kappa}{x}}{\wbranch{\kappa}{y}{s}}}}}{\ms{L}}} \\
      & = \lmor{R} ; \delta^{-1} ; [\ms{id}, 
        \rfix{\dnt{\haslb{\Gamma, \bhyp{x}{A}}
          {\where{\brb{\kappa}{x}}{\wbranch{\kappa}{y}{s}}}{\ms{L}, \ell(A)}}}] \\
      & = \lmor{R} ; \delta^{-1} ; [\ms{id}, 
        \rfix{\lmor{\pi_r ; \iota_r} ; \delta^{-1} 
            ; [\pi_r, 
            \rfix{\dnt{\haslb{\Gamma, \bhyp{x}{A}, \bhyp{y}{A}}{s}{\ms{L}, \ell(A), \kappa(A)}}}]
          }] \\
      & = \lmor{R} ; \delta^{-1} ; [\ms{id}, 
        \rfix{\lmor{\pi_r ; \iota_r} ; \delta^{-1} 
            ; [\pi_l \otimes A ; \pi_r, 
            \rfix{\pi_l \otimes \dnt{A} ; S}]
          }] \\
      & = \lmor{R} ; \delta^{-1} ; [\ms{id}, 
        \rfix{\lmor{\pi_r ; \iota_r} ; \pi_l \otimes \dnt{A} ; \delta^{-1} ; [\pi_r, \rfix{S}]
          }] \\
      & = \lmor{R} ; \delta^{-1} ; [\ms{id}, 
          \rfix{\rseq{}{\envinr{\dnt{\Gamma}}}{\envcop{\dnt{\Gamma}}{\pi_r}{\rfix{S}}}}] \\
      & = \lmor{R} ; \delta^{-1} ; [\ms{id}, \rfix{\rfix{S}}]
        = \lmor{R} ; \delta^{-1} ; 
          [\ms{id}, \rfix{\rseq{}{S}{\envcop{\dnt{\Gamma}}{\pi_r}{\envinr{\dnt{\Gamma}}}}}] \\
      & = \lmor{R} ; \delta^{-1} ; 
          [\ms{id}, \rfix{\lmor{S} ; \delta^{-1} ; [\pi_r, \pi_r ; \iota_r]}]
        = \lmor{R} ; \delta^{-1} ; 
          [\ms{id}, \rfix{S ; [\ms{id}, \iota_r]}] \\
      & = \lmor{R} ; \delta^{-1} ; 
      [\ms{id}, \rfix{\dnt{\haslb{\Gamma, \bhyp{y}{A}}{[\ell/\kappa]s}{\ms{L}, \ell(A)}}}]
        = \dnt{\haslb{\Gamma}{\where{r}{\wbranch{\ell}{y}{[\ell/\kappa]s}}}{\ms{L}}}
      \end{aligned}
    \end{equation}
    as desired.

    \item \brle{uni}: 
    Define $R = \dnt{\haslb{\Gamma}{r}{\ms{L}, \ell(A)}}$, %
           $E = \dnt{\hasty{\Gamma, \bhyp{x}{A}}{\bot}{e}{B}}$, %
           $S = \dnt{\haslb{\Gamma, \bhyp{y}{B}}{s}{\ms{L}, \kappa(B)}}$, and %
           $T = \dnt{\haslb{\Gamma, \bhyp{x}{A}}{s}{\ms{L}, \ell(A)}}$. %
    We have by induction that
    \begin{equation}
      \begin{aligned}
        \dnt{\haslb{\Gamma, \bhyp{x}{A}}{\letstmt{y}{e}{s}}{\ms{L}, \kappa(B)}}
        &= \rlmor{E} ; S = \\
        \dnt{\haslb{\Gamma, \bhyp{x}{A}}
          {\where{t}{\wbranch{\ell}{x}{\brb{\kappa}{e}}}}{\ms{L}, \kappa(B)}}
        &= \rcase{T} ; \dnt{\ms{L}} + E  
      \end{aligned}
    \end{equation}
    It follows in particular that
    \begin{equation}
      \rlmor{E} ; \rfix{S} = \rfix{T}
    \end{equation}
    and hence that
    \begin{equation}
      \begin{aligned}
        & \dnt{
          \haslb{\Gamma}
          {\where{(\where{r}{\wbranch{\ell}{x}{\brb{\kappa}{e}}})}{\wbranch{\kappa}{y}{s}}}
          {\ms{L}}} \\
        & = \lmor{\lmor{R} ; \delta^{-1} ; \pi_r + E}
          ; \delta^{-1} ; [\pi_r, \rfix{S}] \\
        & = \lmor{R} ; \delta^{-1} 
          ; (\dnt{\Gamma} \otimes \dnt{\ms{L}}) + \rlmor{E} 
          ; [\pi_r, \rfix{S}]
          = \lmor{R} ; \delta^{-1}
          ; [\pi_r, \rlmor{E} ; \rfix{S}] \\
        & = \lmor{R} ; \delta^{-1}
          ; [\pi_r, \rfix{T}]
          = \dnt{\haslb{\Gamma}{\where{r}{\wbranch{\ell}{x}{t}}}{\ms{L}}} \\
      \end{aligned}
    \end{equation}
    as desired.
    
    \item \brle{dinat}: 
    We define 
      $\ms{R} = (\ell_i(A_i),)_i$, $\ms{R}' = (\kappa_j(B_j),)_j$, %
      $S = \dnt{\lbsubst{\Gamma}{\sigma}{\ms{R}}{\ms{R}'}}$, and %
      $S' = \dnt{\Gamma} \otimes \alpha^+_{\dnt{R}} ; S ; \alpha^+_{\Sigma_j\dnt{B_j}}$. %
    We have that
    \begin{equation}
      \begin{aligned}
        & \entrymor{\Gamma}{[\lupg{\sigma}]r}{\ms{L}}
          = \dnt{\haslb{\Gamma}{[\lupg{\sigma}]r}{\ms{L}, \ms{R}'}}
          ; \alpha^+_{\dnt{\ms{L}} + \Sigma_j \dnt{B_j}} \\
        & = \lmor{\dnt{\haslb{\Gamma}{r}{\ms{L}, \ms{R}}}}
          ; \dnt{\Gamma} \otimes \alpha^+_{\dnt{\ms{L}} + \dnt{\ms{R}}}
          ; \delta^{-1}
          ; \pi_r 
            + (S ; \alpha^+_{\Sigma_j\dnt{B_j}}) \\
        & = \lmor{\entrymor{\Gamma}{r}{\ms{L}}} ; \delta^{-1} ; \pi_r + S' \\
      \end{aligned}
    \end{equation}
    and, writing $L = \loopmor{\Gamma}{(\wbranch{\ell_i}{x_i}{t_i},)_i}{\ms{L}}$,
    \begin{equation}
      \begin{aligned}
        & \lmor{(\wbranch{\kappa_i}{x_i}{[\lupg{\sigma}]t_i},)_i}
          = \delta^{-1}_\Sigma
          ; [ \dnt{\haslb{\Gamma, \bhyp{x_i}{B_i}}{[\lupg{\sigma}]t_i}{\ms{L}, \ms{R}'}}, ]_i
          ; \alpha^+_{\dnt{\ms{L}} + \Sigma_j \dnt{B_j}} \\
        & = \delta^{-1}_\Sigma
          ; [ \lmor{\dnt{\haslb{\Gamma, \bhyp{x_i}{B_i}}{t_i}{\ms{L}, \ms{R}}}} 
            ; \dnt{\lbsubst{\Gamma, \bhyp{x_i}{B_i}}
                    {\lupg{\sigma}}{\ms{L}, \ms{R}}{\ms{L}, \ms{R}'}},
          ]_i
          ; \alpha^+_{\dnt{\ms{L}} + \Sigma_j \dnt{B_j}} \\
        & = \delta^{-1}_\Sigma
          ; [ \rlmor{\dnt{\haslb{\Gamma, \bhyp{x_i}{B_i}}{t_i}{\ms{L}, \ms{R}}}},
          ]_i
          ; \dnt{\lbsubst{\Gamma}{\lupg{\sigma}}{\ms{L}, \ms{R}}{\ms{L}, \ms{R}'}}
          ; \alpha^+_{\dnt{\ms{L}} + \Sigma_j \dnt{B_j}} \\
        & = \rlmor{\delta^{-1}_\Sigma
          ; [ \dnt{\haslb{\Gamma, \bhyp{x_i}{B_i}}{t_i}{\ms{L}, \ms{R}}}
          ]_i}
          ; \dnt{\Gamma} \otimes \alpha^+_{\dnt{\ms{L}} + \dnt{\ms{R}}}
          ; \delta^{-1}
          ; \pi_r + (
            S ; \alpha^+_{\Sigma_j\dnt{B_j}}
          )
           \\
        & = \rlmor{
            \delta^{-1}_\Sigma
            ; [ \dnt{\haslb{\Gamma, \bhyp{x_i}{B_i}}{t_i}{\ms{L}, \ms{R}}},]_i
            ; \alpha^+_{\dnt{\ms{L}} + \Sigma_j\dnt{B_j}}}
          ; \delta^{-1} ; \pi_r + S'
          \\
        & = \rcase{L} ; \pi_r + S'
          \\
      \end{aligned}
    \end{equation}
    Furthermore, we have that, letting 
    $G = \dnt{\lbsubst{\Gamma}{(\kappa_j(x_j) \mapsto t_j,)_j}{\ms{R}'}{\ms{L}, \ms{R}}}$,
    \begin{equation}
      \begin{aligned}
        & \dnt{\haslb
          {\Gamma, \bhyp{x_i}{A_i}}{[\lupg{(\kappa_j(x_j) \mapsto t_j,)_j}](\sigma_i\;x_i)}
          {\ms{L}, \ms{R}}} \\
        & = \lmor{\dnt{\haslb{\Gamma, \bhyp{x_i}{A_i}}{\sigma_i\;x_i}{\ms{L}, \ms{R}'}}}
          ; \dnt{\lbsubst{\Gamma, \bhyp{x_i}{A_i}}{\lupg{(\kappa_j(x_j) \mapsto t_j,)_j}}
            {\ms{L}, \ms{R}'}{\ms{L}, \ms{R}}} \\
        & = \lmor{\dnt{\haslb{\Gamma, \bhyp{x_i}{A_i}}{\sigma_i\;x_i}{\ms{L}, \ms{R}'}}}
          ; \pi_l \otimes \dnt{\ms{L}, \ms{R'}} 
          ; \dnt{\lbsubst{\Gamma}{\lupg{(\kappa_j(x_j) \mapsto t_j,)_j}}
            {\ms{L}, \ms{R}'}{\ms{L}, \ms{R}}} \\
        & = \rlmor{\dnt{\haslb{\Gamma, \bhyp{x_i}{A_i}}{\sigma_i\;x_i}{\ms{L}, \ms{R}'}}}
        ; \dnt{\Gamma} \otimes \alpha^+_{\dnt{\ms{L}} + \dnt{\ms{R}'}}
        ; \delta^{-1}
        ; [\pi_r ; \iota_l ; \alpha^+_{\dnt{\ms{L}, \ms{R}}} , G]
         \\
        & = \rlmor{
          \dnt{\haslb{\Gamma, \bhyp{x_i}{A_i}}{\sigma_i\;x_i}{\ms{L}, \ms{R}'}}
          ; \alpha^+_{\dnt{\ms{L}} + \dnt{\ms{R}'}}
        } ; \delta^{-1} 
        ; [\pi_r ; \iota_l ; \alpha^+_{\dnt{\ms{L}, \ms{R}}} , G] \\
        & = \rcase{
          \dnt{\haslb{\Gamma, \bhyp{x_i}{A_i}}{\sigma_i\;x_i}{\ms{L}, \ms{R}'}}
          ; \alpha^+_{\dnt{\ms{L}} + \dnt{\ms{R}'}}
        }
        ; [\pi_r ; \iota_l ; \alpha^+_{\dnt{\ms{L}, \ms{R}}} , G] \\
        & = \rcase{
          \dnt{\haslb{\Gamma, \bhyp{x_i}{A_i}}{\sigma_i\;x_i}{\ms{R}'}}
          ; \iota_r
        }
        ; [\pi_r ; \iota_l ; \alpha^+_{\dnt{\ms{L}, \ms{R}}} , G] \\
        & = \rlmor{
          \dnt{\haslb{\Gamma, \bhyp{x_i}{A_i}}{\sigma_i\;x_i}{\ms{R}'}}
        } ; G \\
      \end{aligned}
    \end{equation}
    It follows that
    \begin{equation}
      \begin{aligned}
        & \loopmor{\Gamma}
          {(\wbranch{\ell_i}{x_i}{[\lupg{(\kappa_j(x_j) \mapsto t_j,)}](\sigma_i\;x_i)},)_i}
          {\ms{L}}
        \\
        & = \delta^{-1}_\Sigma 
          ; [ \dnt{\haslb{\Gamma, \bhyp{x_i}{A_i}}
                  {[\lupg{(\kappa_j(x_j) \mapsto t_j,)_j}](\sigma_i\;x_i)}{\ms{L}, \ms{R}}}, ]_i
          ; \alpha^+_{\dnt{\ms{L}} + \Sigma_i \dnt{A_i}}   \\
        & = \delta^{-1}_\Sigma 
          ; [ \rlmor{
                \dnt{\haslb{\Gamma, \bhyp{x_i}{A_i}}{\sigma_i\;x_i}{\ms{R}'}}
              } ; G ]_i
          ; \alpha^+_{\dnt{\ms{L}} + \Sigma_i \dnt{A_i}}   \\
        & = \delta^{-1}_\Sigma 
          ; [ \rlmor{
                \dnt{\haslb{\Gamma, \bhyp{x_i}{A_i}}{\sigma_i\;x_i}{\ms{R}'}}
              } ; G ]_i
          ; \alpha^+_{\dnt{\ms{L}} + \Sigma_i \dnt{A_i}}   \\
        & = \delta^{-1}_\Sigma 
          ; [\rlmor{
                \dnt{\haslb{\Gamma, \bhyp{x_i}{A_i}}{\sigma_i\;x_i}{\ms{R}'}}
              }]_i ; G ; \alpha^+_{\dnt{\ms{L}} + \Sigma_i \dnt{A_i}}  \\
        & = \rlmor{\delta^{-1}_\Sigma 
          ; [
                \dnt{\haslb{\Gamma, \bhyp{x_i}{A_i}}{\sigma_i\;x_i}{\ms{R}'}}
            ]_i} ; G ; \alpha^+_{\dnt{\ms{L}} + \Sigma_i \dnt{A_i}}   \\
        & = \rlmor{\delta^{-1}_\Sigma 
          ; [
                \dnt{\haslb{\Gamma, \bhyp{x_i}{A_i}}{\sigma_i\;x_i}{\ms{R}'}}
            ]_i} ; G ; \alpha^+_{\dnt{\ms{L}} + \Sigma_i \dnt{A_i}}   \\
        & = \rlmor{\dnt{\Gamma} \otimes \alpha_{\dnt{\ms{R}}} ; S} ; G 
          ; \alpha^+_{\dnt{\ms{L}} + \Sigma_i \dnt{A_i}} \\
        & = \rlmor{S'} ; \dnt{\Gamma} \otimes \alpha_{\dnt{\ms{R'}}} 
          ; \dnt{\lbsubst{\Gamma}{(\kappa_j(x_j) \mapsto t_j,)_j}{\ms{R}'}{\ms{L}, \ms{R}}}
          ; \alpha^+_{\dnt{\ms{L}} + \Sigma_i \dnt{A_i}} \\
        & = \rlmor{S'} ; \delta^{-1}_{\Sigma}
          ; [\dnt{\haslb{\Gamma, \bhyp{x_i}{A_i}}{t_i}{\ms{R}'}{\ms{L}, \ms{R}}},]_i
          ; \alpha^+_{\dnt{\ms{L}} + \Sigma_i \dnt{A_i}} \\
        & = \rlmor{S'} ; L \\
      \end{aligned}
    \end{equation}
    We therefore have
    \begin{equation}
      \begin{aligned}
        & \dnt{\haslb{\Gamma}
          {\where{[\lupg{\sigma}]r}{(\wbranch{\kappa_i}{x_i}{[\lupg{\sigma}]t_i},)_i}}
          {\ms{L}}} \\
        & = \lmor{\entrymor{\Gamma}{[\lupg{\sigma}]r}{\ms{L}}} ; \delta^{-1} ; [\pi_r,
          \rfix{\loopmor{\Gamma}{(\wbranch{\kappa_i}{x_i}{[\lupg{\sigma}]t_i},)_i}{\ms{L}}}] \\
        & = \lmor{\lmor{\entrymor{\Gamma}{r}{\ms{L}}} ; \delta^{-1} ; \pi_r + S'} ; \delta^{-1}
          ; [
              \pi_r,
              \rfix{\rcase{L} ; \pi_r + S'}
          ] \\
        & = \lmor{\entrymor{\Gamma}{r}{\ms{L}}} ; \delta^{-1} 
          ; \dnt{\Gamma} \otimes \dnt{\ms{L}} + \rlmor{S}
          ; [
            \pi_r,
            \rfix{\rcase{L} ; \pi_r + S'}
          ] \\
        & = \lmor{\entrymor{\Gamma}{r}{\ms{L}}} ; \delta^{-1} 
        ; [
          \pi_r,
          \rlmor{S'}
          ; \rfix{\rcase{L} ; \pi_r + S'}
        ] \\
        & = \lmor{\entrymor{\Gamma}{r}{\ms{L}}} ; \delta^{-1} 
        ; [
          \pi_r,
          \rfix{\rlmor{S'} ; L}
        ] \\
        & = \lmor{\entrymor{\Gamma}{r}{\ms{L}}} ; \delta^{-1} 
        ; [
          \pi_r,
          \rfix{
            \loopmor{\Gamma}{(\wbranch{\ell_i}{x_i}
                    {[\lupg{(\kappa_j(x_j) \mapsto t_j,)}](\sigma_i\;x_i)},)_i}{\ms{L}}}
        ] \\
        & = \dnt{\haslb{\Gamma}
          {\where{r}{[\lupg{(\kappa_j(x_j) \mapsto t_j,)}](\sigma_i\;x_i)},)_i}{\ms{L}}}
      \end{aligned}
    \end{equation}
    as desired.
  \end{itemize}
  All other cases are analogous to those for expressions, and so omitted.
\end{proof}

\begin{lemma}[\ms{where}-fusion]
  The rewrite rules \brle{cfg-fuse$_1$} (Eqn.~\ref{eqn:where-fusion-1}) and \brle{cfg-fuse$_2$}
  (Eqn.~\ref{eqn:where-fusion-2}) are sound.
  \label{lem:where-fusion}
\end{lemma}
\begin{proof}
  We begin with \brle{cfg-fuse$_1$}. Let %
  $\ms{R} = (\ell_j(A_j),)_j$, $\ms{K} = (\kappa_k(B_k),)_k$, %
  $S = (\wbranch{\kappa_i}{y_i}{s_i},)_i$,
  $T = (\wbranch{\ell_i}{x_i}{t_i},)_j$, and %
  $D_S = \loopmor{\Gamma}{S}{\ms{L}, \ms{R}}$, %
  $D_T = \loopmor{\Gamma}{T}{\ms{L}}$,
  $D_G = \loopmor{\Gamma}{S, T}{\ms{L}}$ %
  We have that
  \begin{equation}
    \begin{aligned}
      & \dnt{\haslb{\Gamma}{\where{(\where{r}{S})}{T}}{\ms{L}}} \\
      & = \lmor{
        \lmor{\entrymor{\Gamma}{r}{(\ms{L}, \ms{R})}} ; \delta^{-1} ;
        [\pi_r, \rfix{D_S}] ;
        \alpha^+_{\dnt{\ms{L}} + \Sigma_j\dnt{A_j}}
      } ; \delta^{-1} ; [\pi_r, \rfix{D_T}] \\
      & = \lmor{\entrymor{\Gamma}{r}{(\ms{L}, \ms{R})}}
        ; \rlmor{\delta^{-1} ; [\pi_r, \rfix{D_S}] ; \alpha^+_{\dnt{\ms{L}} + \Sigma_j\dnt{A_j}}}
        ; \delta^{-1} ; [\pi_r, \rfix{D_T}] \\
      & = \lmor{\entrymor{\Gamma}{r}{(\ms{L}, \ms{R})}}
        ; \delta^{-1}
        ; [
          \rlmor{\pi_r ; \alpha^+_{\dnt{\ms{L}} + \Sigma_j\dnt{A_j}}},
          \rlmor{\rfix{D_S} ; \alpha^+_{\dnt{\ms{L}} + \Sigma_j\dnt{A_j}}}
        ]
        ; \delta^{-1} ; [\pi_r, \rfix{D_T}] \\
      & = \lmor{\entrymor{\Gamma}{r}{(\ms{L}, \ms{R})}}
        ; \envcop{\dnt{\Gamma}}
            { \\ & \qquad
              \rseq{}
                {\pi_r ; \alpha^+_{\dnt{\ms{L}} + \Sigma_j\dnt{A_j}}}
                {\envcop{\dnt{\Gamma}}{\pi_r}{\rfix{D_T}}}
            }
            { \\ & \qquad
              \rseq{}
                {\rfix{D_S} ; \alpha^+_{\dnt{\ms{L}} + \Sigma_j\dnt{A_j}}}
                {\envcop{\dnt{\Gamma}}{\pi_r}{\rfix{D_T}}}
            } \\
      & = \lmor{\entrymor{\Gamma}{r}{(\ms{L}, \ms{R})}}
        ; \envcop{\dnt{\Gamma}}
            { \\ & \qquad
              \rseq{}
                {\alpha^{R+}_{\dnt{\ms{L}} + \Sigma_j\dnt{A_j}}}
                {\envcop{\dnt{\Gamma}}{\pi_r}{\rfix{D_T}}}
            }
            { \\ & \qquad
              \rseq{}
                {\rfix{D_S} ; \alpha^+_{\dnt{\ms{L}} + \Sigma_j\dnt{A_j}}}
                {\envcop{\dnt{\Gamma}}{\pi_r}{\rfix{D_T}}}
            } \\
      & = \lmor{\haslb{\Gamma}{r}{\ms{L}, \ms{R}, \ms{K}}}
        ; \rseq{}
          {\alpha^{\ms{R}+}_{\dnt{\ms{L}} + (\dnt{R} + \Sigma_j\dnt{A_j})}}{
            \envcop{\dnt{\Gamma}}{\pi_r}{ \\ & \qquad
              \envcop{\dnt{\Gamma}}{
                \rfix{D_T}
              }{
                \rseq{}
                  {\rfix{D_S} ; \alpha^+_{\dnt{\ms{L}} + \Sigma_j\dnt{A_j}}}
                  {\envcop{\dnt{\Gamma}}{\pi_r}{\rfix{D_T}}}
              }
            }
          } \\
      & = \lmor{
          \haslb{\Gamma}{r}{\ms{L}, \ms{R}, \ms{K}} 
          ; \alpha^{\ms{R}+}_{\dnt{\ms{L}} + (\dnt{R} + \Sigma_j\dnt{A_j})}}
        ; \delta^{-1}
        ; \\ & \qquad [\pi_r, 
            \envcop{\dnt{\Gamma}}{
              \rfix{D_T}
            }{
              \rseq{}
                {\rfix{D_S} ; \alpha^+_{\dnt{\ms{L}} + \Sigma_j\dnt{A_j}}}
                {\envcop{\dnt{\Gamma}}{\pi_r}{\rfix{D_T}}}
            }
        ] \\
        & = \lmor{\entrymor{\Gamma}{r}{\ms{L}}}
        ; \delta^{-1}
        ; [\pi_r, 
            \envcop{\dnt{\Gamma}}{
              \rfix{D_T}
            }{
              \rseq{}
                {\rfix{D_S} ; \alpha^+_{\dnt{\ms{L}} + \Sigma_j\dnt{A_j}}}
                {\envcop{\dnt{\Gamma}}{\pi_r}{\rfix{D_T}}}
            }
        ] \\
    \end{aligned}
  \end{equation}
  For this to be equal to $\dnt{\haslb{\Gamma}{\where{r}{S, T}}{\ms{L}}}$, it therefore suffices to
  show that
  \begin{equation}
    \begin{aligned}
      \rfix{D_G} &= \envcop{\dnt{\Gamma}}{
        \rfix{D_T}
      }{
        \rseq{}
          {\rfix{D_S} ; \alpha^+_{\dnt{\ms{L}} + \Sigma_j\dnt{A_j}}}
          {\envcop{\dnt{\Gamma}}{\pi_r}{\rfix{D_T}}}
      }
    \end{aligned}
    \label{eqn:dgdt}
  \end{equation}
  We note that, by re-association and weakening, we have that
  \begin{equation}
    \begin{aligned}
      D_G 
      & = \dnt{\Gamma} \otimes \alpha^+_{\dnt{\ms{R}} + \dnt{\ms{K}}} ; \delta^{-1} ; [
        \delta^{-1}_\Sigma ; [
          \dnt{\haslb{\Gamma, \bhyp{x_i}{A_i}}{t_i}{\ms{L}, \ms{R}, \ms{K}}}
          ; \alpha^+_{\dnt{\ms{L}} + \Sigma_j\dnt{C_j}},
        ]_i,
        D_S
      ] \\
      & = \dnt{\Gamma} \otimes \alpha^+_{\dnt{\ms{R}} + \dnt{\ms{K}}} ; \delta^{-1} ; [
        D_T ; \dnt{\ms{L}} 
          + (\alpha^+_{\ms{R}} ; \dnt{\ms{R} \leq \ms{R, K}} ; \alpha^+_{\Sigma_i\dnt{C_i}}),
        D_S
      ] \\
      & = \dnt{\Gamma} \otimes \alpha^+_{\dnt{\ms{R}} + \dnt{\ms{K}}} ; \delta^{-1} ; [
        D_T ; \dnt{\ms{L}} 
          + (\iota_{l, \Sigma_i\dnt{A_i} + \Sigma_i\dnt{B_i}} ; \alpha^+_{\Sigma_i\dnt{C_i}}),
        D_S
      ] \\
      & = \rseq{}{\alpha^{+\dnt{\Gamma}}_{\ms{R} + \ms{K}}}{\envcop{\dnt{\Gamma}}
        {\rseq{}{D_T}{\dnt{\ms{L}} +_R (\iota_l^{\dnt{\Gamma}} ; \alpha^+_{\Sigma_i\dnt{C_i}})}}
        {}}
    \end{aligned}
    \label{eqn:dg-rhs}
  \end{equation}
  where $C_{1..k} = A_{1..k}$ and $C_{k+1..n} = B_{1..n-k}$. We note in particular that $\dnt{\ms{R}
  \leq \ms{R, K}}$ is up to isomorphism the left injection $\iota_l : \Sigma_i\dnt{A_i} \to
  \Sigma_i\dnt{A_i} + \Sigma_i\dnt{B_i}$. We can now derive Equation~\ref{eqn:dgdt}, and hence the
  soundness of \brle{cfg-fuse$_1$}, via the string-diagrams in
  Figure~\ref{fig:string-diagram-fusion}, which are drawn in the co-Kleisli category inducted by
  $\dnt{\Gamma}$. \brle{cfg-fuse$_2$} then follows by repeated application of \brle{cfg-fuse$_1$}, as
  desired.
\end{proof}

\begin{figure} 
  \begin{subfigure}[t]{.45\textwidth}
    \centering
    \begin{tikzpicture}
      \node[] (R) at (0, 0) {$\dnt{\ms{R}}$};
      \node[] (K) at (1, 0) {$\dnt{\ms{K}}$};
      \node[dot] (dR) at (0, -0.75) {}; 
      \node[dot] (dK) at (1, -1.25) {}; 
      \node[box=1/0/2/0] (DT) at (0, -2.5) {\quad D_T \quad};
      \node[box=1/0/3/0] (DS) at (2, -2.5) {\quad D_S \quad};
      \node[dot] (dL) at (0.25, -3.5) {};
      \node[] (L) at (1, -6) {$\dnt{\ms{L}}$};
      \node[dot] (dR2) at (1.25, -3.5) {};
      \coordinate[] (cupS) at (3, -3.5) {};
      \coordinate[] (capS) at (3, -1) {};
      \coordinate[] (cupT) at (3, -4) {};
      \coordinate[] (capT) at (3, -0.5) {};
      \wires{
        R = { south = dR.north },
        K = { south = dK.north },
        dR = { south = DT.north.1 },
        dK = { south = DS.north.1 },
        DT = { south.1 = dL.west, south.2 = dR2.west },
        DS = { south.1 = dL.east, south.2 = dR2.east, south.3 = cupS.west },
        dL = { south = L.north },
        dR2 = { south = cupT.west },
        cupS = { east = capS.east },
        capS = { west = dK.east },
        cupT = { east = capT.east },
        capT = { west = dR.east },
      }{}
    \end{tikzpicture}
    \caption{$\rfix{D_G}$, with $D_G$ drawn as per the right of Eqn.~\ref{eqn:dg-rhs}}
    \label{fig:rfix-dg}
  \end{subfigure}
  \quad
  \begin{subfigure}[t]{.45\textwidth}
    \centering
    \begin{tikzpicture}
      \node[] (R) at (0, 0) {$\dnt{\ms{R}}$};
      \node[] (K) at (1, 0) {$\dnt{\ms{K}}$};
      \node[dot] (dR) at (0.75, -2.5) {}; 
      \coordinate[] (cR) at (0.25, -3) {};
      \node[dot] (dR2) at (0.25, -3.25) {};
      \node[dot] (dK) at (1.5, -1) {}; 
      \node[box=1/0/2/0] (DT) at (0.25, -3.75) {\quad D_T \quad};
      \node[box=1/0/3/0] (DS) at (1.5, -1.5) {\quad D_S \quad};
      \coordinate[] (cL) at (2, -3.5) {};
      \node[dot] (dL) at (1, -5.25) {};
      \node[] (L) at (1, -6) {$\dnt{\ms{L}}$};
      \coordinate[] (cupS) at (2.25, -2.05) {};
      \coordinate[] (capS) at (2.25, -0.95) {};
      \coordinate[] (cupT) at (1, -4.3) {};
      \coordinate[] (capT) at (1, -3.2) {};
      \wires{
        R = { south = dR.west },
        K = { south = dK.north },
        dK = { south = DS.north.1 },
        DS = { south.1 = cL.north, south.2 = dR.east, south.3 = cupS.west },
        dR = { south = cR.north },
        cR = { south = dR2.north },
        dR2 = { south = DT.north.1 },
        DT = { south.1 = dL.west, south.2 = cupT.west },
        cL = { south = dL.east },
        dL = { south = L.north },
        cupS = { east = capS.east },
        capS = { west = dK.east },
        cupT = { east = capT.east },
        capT = { west = dR2.east },
      }{} 
      \coordinate (dt0) at (-0.8, -3) {};
      \coordinate (dt1) at (1.5, -3) {};
      \coordinate (dt2) at (1.5, -4.5) {};
      \coordinate (dt3) at (-0.8, -4.5) {};
      \draw[red] (dt0) -- (dt1) -- (dt2) -- (dt3) -- (dt0);
    \end{tikzpicture}
    \caption{
      Equivalent to \ref{fig:rfix-dg} by isotopy and associativity of the codiagonal. 
      We highlight $\rfix{D_T}$, which is used in the next step.
    }
    \label{fig:rfix-isotopy}
  \end{subfigure}
  \begin{subfigure}[t]{\textwidth}
    \centering
    \begin{tikzpicture}
      \node[] (R) at (0, 0) {$\dnt{\ms{R}}$};
      \node[] (K) at (1, 0) {$\dnt{\ms{K}}$};
      \node[dot] (dT1) at (0.25, -3.25) {};
      \node[dot] (dK) at (1.5, -1) {}; 
      \node[box=1/0/3/0] (DS) at (1.5, -1.5) {\quad D_S \quad};
      \node[box=1/0/2/0] (DT1) at (0.25, -3.75) {\quad D_T \quad};
      \node[box=1/0/2/0] (DT2) at (2.5, -3.25) {\quad D_T \quad};
      \coordinate[] (cL) at (1.5, -3.5) {};
      \node[dot] (dL1) at (0.5, -4.75) {};
      \node[dot] (dL2) at (1, -5.25) {};
      \node[] (L) at (1, -6) {$\dnt{\ms{L}}$};
      \coordinate[] (cupS) at (2.25, -2.05) {};
      \coordinate[] (capS) at (2.25, -0.95) {};
      \coordinate[] (cupT1) at (1, -4.3) {};
      \coordinate[] (capT1) at (1, -3.2) {};
      \coordinate[] (cupT2) at (3.2, -3.8) {};
      \coordinate[] (capT2) at (3.2, -2.7) {};
      \node[dot] (dT2) at (2.5, -2.75) {};
      \wires{
        R = { south = dT1.north },
        K = { south = dK.north },
        dK = { south = DS.north.1 },
        DS = { south.1 = cL.north, south.2 = dT2.north, south.3 = cupS.west },
        dT1 = { south = DT1.north },
        DT1 = { south.1 = dL1.west, south.2 = cupT1.west },
        cL = { south = dL2.east },
        dL1 = { south = dL2.west },
        dL2 = { south = L.north },
        cupS = { east = capS.east },
        capS = { west = dK.east },
        cupT1 = { east = capT1.east },
        capT1 = { west = dR2.east },
        dT2 = { south = DT2.north },
        DT2 = { south.1 = dL1.east, south.2 = cupT2.west },
        cupT2 = { east = capT2.east },
        capT2 = { west = dT2.east },
      }{} 
      \coordinate (dt10) at (-0.8, -3) {};
      \coordinate (dt11) at (1.5, -3) {};
      \coordinate (dt12) at (1.5, -4.5) {};
      \coordinate (dt13) at (-0.8, -4.5) {};
      \coordinate (dt20) at (1.5, -2.5) {};
      \coordinate (dt21) at (3.8, -2.5) {};
      \coordinate (dt22) at (3.8, -4) {};
      \coordinate (dt23) at (1.5, -4) {};
      \draw[red] (dt10) -- (dt11) -- (dt12) -- (dt13) -- (dt10);
      \draw[red] (dt20) -- (dt21) -- (dt22) -- (dt23) -- (dt20);
    \end{tikzpicture}
    \caption{
      $\envcop{\dnt{\Gamma}}{
        \rfix{D_T}
      }{
        \rseq{}
          {\rfix{D_S} ; \alpha^+_{\dnt{\ms{L}} + \Sigma_j\dnt{A_j}}}
          {\envcop{\dnt{\Gamma}}{\pi_r}{\rfix{D_T}}}
      }$. \\
      Equivalent to \ref{fig:rfix-isotopy} by duplication of $\rfix{D_T}$.
    }
    \label{fig:rfix-lhs}
  \end{subfigure}%
  \caption{
    String diagrams validating the soundness of \brle{cfg-fuse$_1$}
    (Eqn.~\ref{eqn:where-fusion-1}), drawn in the co-Kleisli category induced by $\dnt{\Gamma}$
  } 
  \Description{}
  \label{fig:string-diagram-fusion}
\end{figure}

\subsection{Completeness}

\completenessexpr*

\label{proof:complete-expr}

\begin{proof}
  We begin by showing $\ms{Th}^\otimes(\Gamma)$ is an \isotopessa{} expression model by validating
  all the equations for a distributive Freyd category. These are formalized as follows:
  \begin{itemize}
    \item $\ms{Th}^\otimes(\Gamma)$ is a category: formalized in
    \texttt{Rewrite/Term/Compose/Seq.lean}
    \item $\ms{Th}^\otimes(\Gamma)$ is a Freyd category: formalized in
    \texttt{Rewrite/Term/Compose/Product.lean}
    \item $\ms{Th}^\otimes(\Gamma)$ has coproducts: formalized in
    \texttt{Rewrite/Term/Compose/Sum.lean}
    \item $\ms{Th}^\otimes(\Gamma)$ is distributive: formalized in
    \texttt{Rewrite/Term/Compose/Distrib.lean}
  \end{itemize}
  We then prove that the packing of each type constructor is equivalent to its denotational
  semantics in $\ms{Th}^\otimes(\Gamma)$ in \texttt{Rewrite/Term/Compose/Completeness.lean}.
  Finally, we show that packing and unpacking are mutual inverses for $\Gamma$ pure in
  \texttt{Term.Eqv.packed_unpacked} and \texttt{Term.Eqv.unpacked_packed} in
  \texttt{Rewrite/Term/Structural/Product.lean}, completing the proof of initiality.
\end{proof}

\completenessregions*

\label{proof:complete-reg}

\begin{proof}
  We begin by showing $\ms{Th}(\Gamma, \ms{L})$ is an \isotopessa{} region model by validating all
  the equations for a distributive Elgot category. These are formalized as follows:
  \begin{itemize}
    \item $\ms{Th}(\Gamma, \ms{L})$ is a category: formalized in
    \texttt{Rewrite/Region/Compose/Seq.lean}
    \item $\ms{Th}(\Gamma, \ms{L})$ is a Freyd category: formalized in
    \texttt{Rewrite/Region/Compose/Product.lean}
    \item $\ms{Th}(\Gamma, \ms{L})$ has coproducts: formalized in
    \texttt{Rewrite/Region/Compose/Sum.lean}
    \item $\ms{Th}(\Gamma, \ms{L})$ is distributive: formalized in
    \texttt{Rewrite/Region/Compose/Distrib.lean}
    \item $\ms{Th}(\Gamma, \ms{L})$ is a strong Elgot category: formalized in
    \texttt{Rewrite/Region/Compose/Elgot.lean}
  \end{itemize}
  We then prove that the packing of each type constructor is equivalent to its denotational
  semantics in $\ms{Th}(\Gamma, \ms{L})$ in \texttt{Rewrite/Region/Compose/Completeness.lean}.
  Finally, we show that packing and unpacking are mutual inverses for $\Gamma$ pure in
  \texttt{Region.Eqv.packed_unpacked} and \texttt{Region.Eqv.unpacked_packed} in
  \texttt{Rewrite/Region/Structural.lean}, completing the proof of initiality.
\end{proof}

\end{document}